\documentclass[11pt,a4paper]{amsart}

\usepackage{lipsum}
\makeatletter
\g@addto@macro{\endabstract}{\@setabstract}
\newcommand{\authorfootnotes}{\renewcommand\thefootnote{\@fnsymbol\c@footnote}}%
\makeatother

\usepackage{amssymb}
\usepackage{mathtools}
\usepackage{mathrsfs}
\usepackage{mathptmx}
\usepackage[utf8]{inputenc}
\usepackage{bm}
\usepackage{bbm}
\usepackage{amsmath}
\usepackage{amsfonts}
\makeatletter
\def\amsbb{\use@mathgroup \M@U \symAMSb}
\makeatother

\usepackage{comment}
\usepackage{graphicx}
\usepackage{grffile}
\usepackage{caption}
\usepackage{subcaption}

\usepackage{bbold}

\usepackage{mathptmx}

\newcommand{\supp}{\textrm{supp}\,}

\newcommand{\bga}{\begin{aligned}}
\newcommand{\ena}{\end{aligned}}

\newcommand{\bge}{\begin{enumerate}}
\newcommand{\ene}{\end{enumerate}}

\newcommand{\red}[1]{{\color{red} #1}}

\newcommand{\magenta}[1]{{\color{magenta} #1}}
\newcommand{\hide}[1]{}

 \usepackage{dutchcal}

\usepackage{fancyhdr}
\usepackage{subcaption}

\usepackage{newtxtext,newtxmath}
\usepackage{pgfplots}
\pgfplotsset{compat=1.15}

\usepackage{xcolor}
\definecolor{webgreen}{rgb}{0,.5,0}
\definecolor{webbrown}{rgb}{.6,0,0}
\definecolor{RoyalBlue}{cmyk}{1, 0.50, 0, 0}
\usepackage[colorlinks=true, breaklinks=true, urlcolor=webbrown, linkcolor=RoyalBlue, citecolor=webgreen,backref=page]{hyperref}

\usepackage{tikz}
\usepackage{pict2e}
\usepackage{todonotes}
\usetikzlibrary{decorations.markings}

\DeclareSymbolFont{bbold}{U}{bbold}{m}{n}
\DeclareSymbolFontAlphabet{\mathbbold}{bbold}

\newcommand{\R}{{\mathbb R}}
\newcommand{\C}{{\mathbb C}}

\newcommand{\N}{{\mathbb N}}

\newcommand{\al}{\alpha}
\newcommand{\be}{\beta}

\newcommand{\ga}{\gamma}
\newcommand{\Ga}{\Gamma}
\newcommand{\La}{\Lambda}

\newcommand{\ep}{\varepsilon}
\newcommand{\de}{\delta}

\newcommand{\De}{\Delta}
\newcommand{\om}{\omega}
\newcommand{\Om}{\Omega}
\newcommand{\ze}{\zeta}

\newcommand{\sg}{\sigma}

\newcommand{\di}{\displaystyle}

\newcommand{\ii}{\textrm{i}}
\newcommand{\dd}{\textrm{d}}

\newcommand{\qasq}{\quad \text{as} \quad}
\newcommand{\qandq}{\quad \text{and} \quad}

	\newtheorem{theorem}{Theorem}
	
	\newtheorem{definition}[theorem]{Definition}
	
	\newtheorem{remark}[theorem]{Remark}
	\newtheorem{lemma}[theorem]{Lemma}
	\newtheorem{proposition}[theorem]{Proposition}
	\newtheorem{corollary}{Corollary}[theorem]
	\numberwithin{equation}{section}
	\numberwithin{theorem}{section}
    \numberwithin{notation}{section}

\usepackage[T1]{fontenc}

\usepackage{mathtools}

\makeatletter
\DeclareRobustCommand\widecheck[1]{{\mathpalette\@widecheck{#1}}}
\def\@widecheck#1#2{%
	\setbox\z@\hbox{\m@th$#1#2$}%
	\setbox\tw@\hbox{\m@th$#1%
		\widehat{%
			\vrule\@width\z@\@height\ht\z@
			\vrule\@height\z@\@width\wd\z@}$}%
	\dp\tw@-\ht\z@
	\@tempdima\ht\z@ \advance\@tempdima2\ht\tw@ \divide\@tempdima\thr@@
	\setbox\tw@\hbox{%
		\raise\@tempdima\hbox{\scalebox{1}[-1]{\lower\@tempdima\box
				\tw@}}}%
	{\ooalign{\box\tw@ \cr \box\z@}}}
\makeatother

\begin{document}

\tikzset{middlearrow/.style={
			decoration={markings,
				mark= at position 0.6 with {\arrow{#1}} ,
			},
			postaction={decorate}
		}
	}\tikzset{middlearrow/.style={
			decoration={markings,
				mark= at position 0.6 with {\arrow{#1}} ,
			},
			postaction={decorate}
		}
	}

	\tikzset{->-/.style={decoration={
				markings,
				mark=at position #1 with {\arrow{latex}}},postaction={decorate}}}
	
	\tikzset{-<-/.style={decoration={
				markings,
				mark=at position #1 with {\arrowreversed{latex}}},postaction={decorate}}}

\title{Phase diagram and Topological Expansion in the complex Quartic Random Matrix Model}

\maketitle
	
\begin{center}
    \authorfootnotes
  Pavel Bleher\footnote{Department of Mathematical Sciences, Indiana University-Purdue University Indianapolis, 402 N. Blackford St., Indianapolis, IN 46202, Blackford St., Indianapolis, IN 46202, USA. e-mail: pbleher@iupui.edu},
  Roozbeh Gharakhloo\footnote{Department of Mathematics, Colorado State University, Fort Collins, CO 80521, USA, E-mail: roozbeh.gharakhloo@colostate.edu}, Kenneth T-R  McLaughlin\footnote{Department of Mathematics, Colorado State University, Fort Collins, CO 80521, USA, E-mail: kenmcl@rams.colostate.edu}
  \par \bigskip
\end{center}
\begin{abstract} We use the Riemann-Hilbert approach, together with string and Toda equations, to study the topological expansion in the quartic random matrix model. The coefficients of the
topological expansion are generating functions for 
 the numbers $\mathscr{N}_j(g)$ of $4$-valent 
 connected graphs with $j$ vertices on a compact Riemann surface of genus $g$. We explicitly evaluate these numbers for Riemann surfaces of genus $0,1,2,$ and $3$. Also, for a Riemann surface of an arbitrary genus $g$, we 
 calculate the leading term in the  asymptotics of $\mathscr{N}_j(g)$ as the number of vertices tends to infinity. Using the theory of quadratic differentials, we  characterize the critical contours in the complex parameter plane where phase transitions in the quartic model take place, thereby proving a result of David \cite{DAVID}. These phase transitions are of the following four types: a) one-cut to two-cut through the  splitting of the cut at the origin, b) two-cut to three-cut through the birth of a new cut at the origin, c) one-cut to three-cut through the splitting of the cut at two symmetric points, and d) one-cut to three-cut through the birth of two symmetric cuts.    
\end{abstract}

\tableofcontents
		
	\section{Introduction and Main Results}
	
	Our starting point is 
	the unitary ensemble of $n \times n$ Hermitian random matrices,
	\begin{equation} \label{int1}
	\dd \mu_{nN}(M;u) = \frac{1}{\Tilde{\mathcal{Z}}_{nN}(u)} e^{-N \mathrm{Tr}\,\mathscr{V}(M;u)} \dd M,
	\end{equation}
	with the quartic potential
	\begin{equation}\label{int2}
	\mathscr{V}(z;u) = \frac{z^2}{2}+\frac{uz^4}{4},
	\end{equation}
	where $u>0$ and $N>0$ are {\it parameters} of the model, and
	\begin{equation} \label{int3}
	\Tilde{\mathcal{Z}}_{nN}(u) = \int_{\mathcal{H}_n} e^{-N \mathrm{Tr}\,\mathscr{V}(M;u)} \dd M.
	\end{equation}
	is the {\it partition function}. 
	
	As well known (see, e.g., \cite{BL}), the ensemble of eigenvalues of $M$,
	\[
	Me_k=z_ke_k,\;k=1,\ldots,n,
	\]
	is given by the probability distribution
	\begin{equation}\label{int4}
	\bga
	\dd \mu_{nN}(z;u)&=\frac{1}{\mathcal{Z}_{nN}(u)}
	\prod_{1\leq j<k\leq n} (z_j-z_k)^2\prod^{n}_{j=1} \exp\left[-N\left(\frac{z^2_j}{2}+\frac{u z^4_j}{4}\right)\right] \dd z_1 \cdots \dd z_n,\\
	&z=\left\{ z_1,\ldots,z_n \right\},
	\ena
	\end{equation}
	where 
	\begin{equation}\label{int5}
	\mathcal{Z}_{nN}(u) = \int^{\infty}_{-\infty} \cdots \int^{\infty}_{-\infty} \prod_{1\leq j<k\leq n} (z_j-z_k)^2\prod^{n}_{j=1} \exp\left[-N\left(\frac{z^2_j}{2}+\frac{u z^4_j}{4}\right)\right] \dd z_1 \cdots \dd z_n,
	\end{equation}
	is the {\it eigenvalue partition function}. The partition functions
	$\tilde {\mathcal{Z}}_{nN}(u)$ and 
	$ \mathcal{Z}_{nN}(u)$ are related by the formula,
	\begin{equation}\label{int6}
	\frac{ \mathcal{Z}_{nN}(u)}{\tilde {\mathcal{Z}}_{nN}(u)} = 
	\frac{1}{\pi^{n(n-1)/2}}\prod_{k=1}^nk!
	\end{equation}
	(see, e.g., \cite{BL}).
	
	We define the {\it free energy} of the  unitary ensemble of $n \times n$ Hermitian random matrices as 
	\begin{equation}\label{int7}
	\mathscr{F}_{nN}(u)=\frac{1}{n^2} \ln\,   \frac{\tilde {\mathcal{Z}}_{nN}(u)}{ \tilde{\mathcal{Z}}_{nN}(0)} \,.
	\end{equation}
	Observe that by \eqref{int6},
	\begin{equation}\label{int8}
	\mathscr{F}_{nN}(u)=\frac{1}{n^2} \ln\,   \frac{ {\mathcal{Z}}_{nN}(u)}{ {\mathcal{Z}}_{nN}(0)} \,.
	\end{equation}
	The quantity
	\begin{equation}\label{int8a}
	{\mathcal{Z}}_{nN}(0)=
	\int^{\infty}_{-\infty} \cdots \int^{\infty}_{-\infty} \prod_{1\leq j<k\leq n} (z_j-z_k)^2\prod^{n}_{j=1} \exp\left(\frac{-Nz^2_j}{2}\right) \dd z_1 \cdots \dd z_n
	\end{equation}
	is the partition function of the Gaussian unitary ensemble (GUE), and 
	it is equal to
	\begin{equation}\label{intb}
	{\mathcal{Z}}_{nN}(0)={\mathcal{Z}}^{\rm GUE}_{nN}=\left(\frac{n}{N}\right)^{n^2}
	\frac{(2\pi)^{n/2}}{(2n)^{n^2/2}}\prod_{k=1}^n k!.
	\end{equation}
	We will be especially interested in the free energy in the case when $n=N$.
	The free energy
	$\mathscr{F}_{NN}(u) $ admits the asymptotic expansion,
	\begin{equation}\label{int9}
	\mathscr{F}_{NN}(u)\sim \sum_{g=0}^\infty 
	\frac{\mathcal{f}_{2g}(u)}{N^{2g}}\,,
	\end{equation}
	in the sense that for any integer $M>0$,  as $N\to\infty$,
	\begin{equation}\label{int10}
	\mathscr{F}_{NN}(u)= \sum_{g=0}^M \frac{\mathcal{f}_{2g}(u)}{N^{2g}}+ O\left(N^{-2(M+1)}\right)\,.
	\end{equation}
	In addition, the coefficients $\mathcal{f}_{2g}(u)$ are analytic functions of $u$ in a neighborhood of the origin independent of $g$.  This was proven by Bleher and Its in \cite{BleherIts2005}, for any $u>0$, and for general real 1-cut potentials $V$ in \cite{ErcolaniMcLaughlin}.  More recently, probabilistic arguments have been used to study partition functions for generalized $\beta$ ensembles (again with real 1-cut potentials) in \cite{Borot-Guionnet}. Moreover, the asymptotics of the partition function for the real Gaussian-type, Laguerre-type, and Jacobi-type 1-cut potentials $V$ were found using Riemann-Hilbert analysis in \cite{Charlier}, and \cite{CharlierGharakhloo}.
	
	Asymptotic expansion \eqref{int9} is called the {\it topological expansion}, and its study was initiated in the classical work of Bessis, Itzykson, and Zuber \cite{BIZ}. As shown in \cite{BIZ}, the functions 
	$\mathcal{f}_{2g}(u)$ are generating functions for the number of topologically different 4-valent graphs with $j$ vertices, $\mathscr {N}_g(j)$, on a closed Riemannian surface of genus $g$. We will discuss this remarkable fact later.
	
	To evaluate the asymptotics of the Taylor coefficients of the functions
	$\mathcal{f}_{2g}(u)$, we will study an {\it analytic continuation} 
	of the partition function  $\mathcal{Z}_{NN}(u)$ to the complex plane in $u$ and singularities of the  analytic continuation. Observe that integral \eqref{int5} defining the eigenvalue partition $\mathcal{Z}_{NN}(u)$
	converges for $\Re u>0$, and we will prove that topological expansion \eqref{int9} is valid for any $u$ with $\Re u>0$. Also,
	we will prove that all the  functions
	$\mathcal{f}_{2g}(u)$ are {\it analytic} in the half-plane $\Re u>0$.
	
	To extend the partition function  $\mathcal{Z}_{nN}(u)$ to 
	$\Re u\le 0$, we will use a {\it regularization} of
	$\mathcal{Z}_{nN}(u)$. Assume first that $u>0$, and
	let us make the change of  variables
	\begin{equation}\label{int11}
	z_j= \sg^{1/2}\ze_j \qandq u=\sigma^{-2}\,,\quad \sg>0,
	\end{equation}
	in the integral in \eqref{int5}:
	\begin{equation}\label{int12}
	\bga
	\mathcal{Z}_{nN}(u)& = \int^{\infty}_{-\infty} \cdots \int^{\infty}_{-\infty} \prod_{1\leq j<k\leq n} (z_j-z_k)^2\prod^{n}_{j=1} \exp\left[-N\left(\frac{z^2_j}{2}+\frac{u z^4_j}{4}\right)\right] \dd z_1 \cdots \dd z_n\\
	&=\sigma^{\frac{n^2}{2}}\int^{\infty}_{-\infty} \cdots \int^{\infty}_{-\infty} \prod_{1\leq j<k\leq n} (\ze_j-\ze_k)^2\prod^{n}_{j=1} \exp\left[-N\left(\frac{\sg\ze^2_j}{2}
	+\frac{ \ze^4_j}{4}\right)\right] \dd \ze_1 \cdots \dd \ze_n.
	\ena
	\end{equation}
	Define now the quartic polynomial
	\begin{equation}\label{int13}
	V(\ze;\sigma):= \mathscr{V}(\sg^{1/2}\ze;\sigma^{-2})
	= \frac{\sg\ze^2}{2} + \frac{\ze^4}{4},
	\end{equation}
	The corresponding partition function of eigenvalues is given by
	\begin{equation}\label{int14}
	Z_{nN}(\sigma) = \int^{\infty}_{-\infty} \cdots \int^{\infty}_{-\infty} \prod_{1\leq j<k\leq n} (\ze_j-\ze_k)^2\prod^{n}_{j=1}  \exp\left[-N\left(\frac{\sg\ze^2_j}{2}
	+\frac{ \ze^4_j}{4}\right)\right] \dd \ze_1 \cdots \dd \ze_n,
	\end{equation}
	which {\it converges for all} $\sigma \in \C$ and defines $Z_{nN}(\sigma)$
	as an entire function on the complex plane. Note that \begin{equation}\label{int15}
	\mathcal{Z}_{nN}(u) = \sigma^{\frac{n^2}{2}}Z_{nN}(\sigma),
	\quad\sg=u^{-1/2}.
	\end{equation}
	This formula gives  an analytic continuation of the partition function $\mathcal{Z}_{nN}(u)$ to the two-sheet covering of the complex plane.
	
	Similar to \eqref{int8} we define the free energy for the quartic polynomial \eqref{int13} as
	\begin{equation}\label{int15a}
	F_{nN}(\sg)=\frac{1}{n^2} \ln\,   \frac{ Z_{nN}(\sg)}{ {\mathcal{Z}}_{nN}^{\rm GUE}} \,,
	\end{equation}
	where the value of ${\mathcal{Z}}_{nN}^{\rm GUE} $ is given in formula \eqref{intb}. From formulae \eqref{int15} and \eqref{int8} we obtain the relation between the free energies ${\mathscr F}_{nN}(u)$ and
	$F_{nN}(\sg)$:
	\begin{equation}\label{int15b}
	\mathscr{F}_{nN}(u)=\frac{\ln\sg}{2}+
	F_{nN}(\sg),\quad \sg=u^{-1/2}.
	\end{equation}

	In this work our goal will be 
	\bge
	\item to find and calculate critical curves of the matrix model with the quartic polynomial
	$V(z;\sg)$ on the complex plane $\sg\in\C$, and
	\item to prove the topological expansion of the free energy 
	$F_{NN}(\sg)$ in the one-cut region on the complex plane $\sg$ and calculate the coefficients of the topological expansion.
	\ene
	Below we formulate our main results.
	\subsection{Phase Diagram}
	
	\begin{figure}[b]
		\centering
		\includegraphics[scale=0.27]{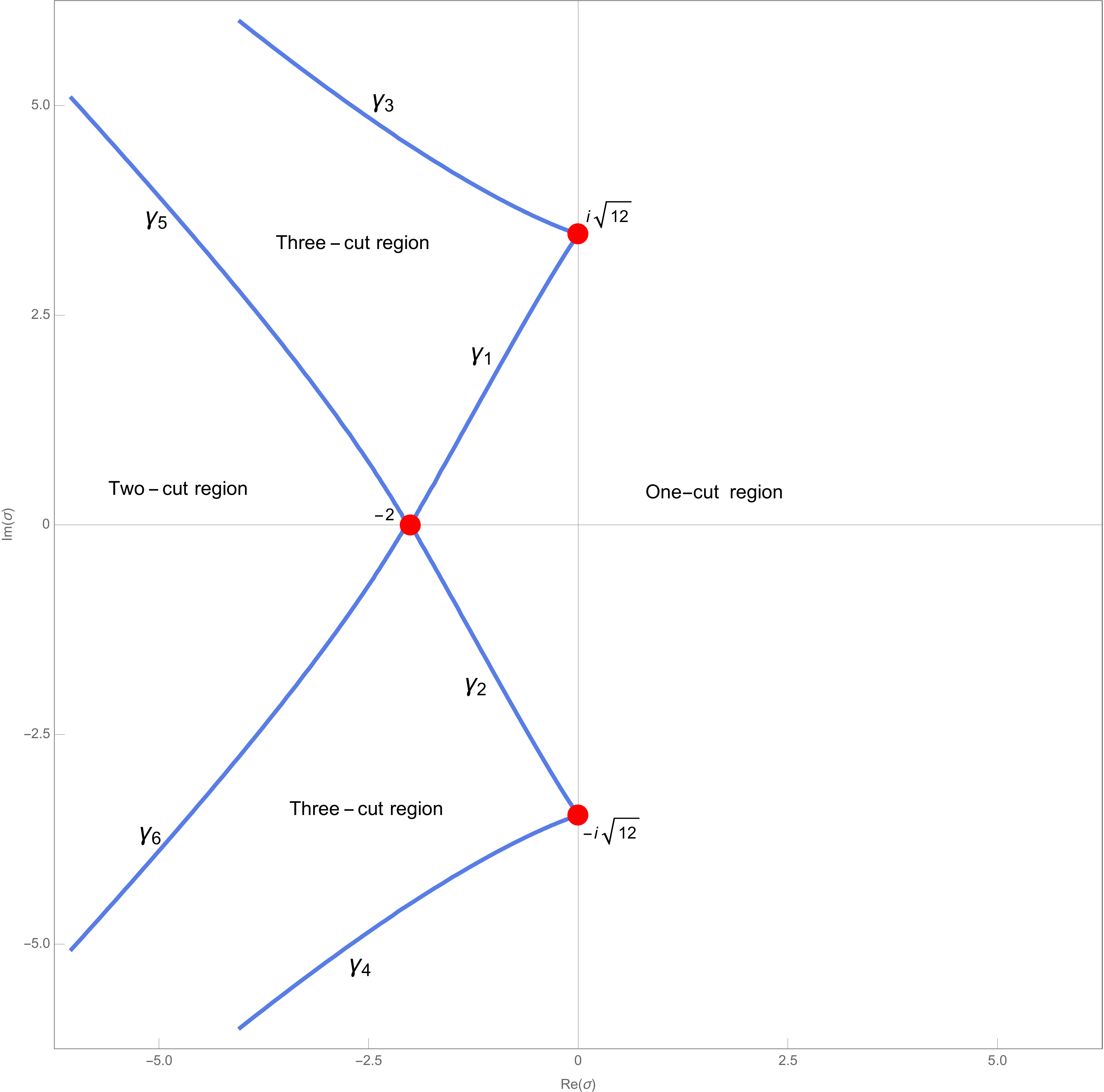}
		\caption{The phase diagram of the complex quartic random matrix model in the $\sigma$-plane. This phase diagram first appeared in the work \cite{DAVID} of David. The Painlev{\'e} II double scaling limit corresponding to the multi-critical point $\sigma=-2$ was studied in \cite{BleherIts2003}, while the Painlev{\'e} I double scaling limit associated to the multi-critical points $\sigma=\pm \ii \sqrt{12}$ was investigated in \cite{DK}. The borders labeled by $\ga_1$ through $\ga_4$ separating the one-cut region from the three-cut region are the same lines shown in Figure \ref{fig:Crit Traj 1to3 sigma plane} by labels I, XII, VII, and IX; and the borders labeled by $\ga_5$ and $\ga_6$ separating the two-cut region from the three-cut region are the same lines shown in Figure \ref{fig:Crit Traj 2to3 sigma plane} by labels 1 and 2.}
		\label{fig:Phase Diagram}
	\end{figure}
	
	The phase diagram of the complex quartic matrix model first appeared in the work \cite{DAVID} of David (See Figure \ref{fig:Phase Diagram} below and Figure 5 of \cite{DAVID}). Later in the work \cite{BertolaTovbis2015}, Bertola and Tovbis found the phase diagram in the two-sheeted $u$-plane based on numerical computations, which under the change of parameters \eqref{int11} is equivalent to Figure \ref{fig:Phase Diagram} (See  Figure 6 in \cite{BertolaTovbis2015}). They also considered several other cases for the contour of integration in \eqref{int5} other than the real axis, and among other things, in each case found the phase diagram (see Figures 4 through 9 in \cite{BertolaTovbis2015})  by computer-assisted methods. In \cite{BertolaTovbis2015} the authors
	did not provide a rigorous description of the phase diagram characterizing the boundaries of regions with different numbers of cuts.  However, in \cite{BertolaTovbis2016}, they developed such an analysis for a different configuration of contours of orthogonality, obtained explicit equations which provide 
	an implicit characterization of the boundaries, and a proof that the regular case is open with respect to the parameters.

	In the present work, one objective
	is to provide an \textit{explicit characterization} of all the boundary curves shown in Figure \ref{fig:Phase Diagram}, in terms of critical trajectories of new auxiliary quadratic differentials in the parameter space, originally discovered in \cite{BDY} for the case of a cubic potential.  Along the way we
	do provide an independent proof that the regular one-cut,  two-cut, and three-cut regimes are open, which is more straightforward than the approach of \cite{BertolaTovbis2016} because it is tailored to the quartic situation.

	The phase diagram of the matrix model with the quartic polynomial
	$V(z;\sg)$ on the complex plane $\sg\in\C$ is described in terms
	of the underlying {\it equilibrium measure}
	\begin{equation}\label{int15c}
	\dd \nu_{\rm eq}(z;\sg)=\frac{1}{\pi \ii}\,[Q(z;\sg)^{1/2}]_+\dd z.
	\end{equation}
	Here $Q(z;\sg)$ is a polynomial in $z$ of degree 6 and  the intervals of the support
	of the equilibrium measure (the cuts) are critical trajectories of the quadratic differential $Q(z;\sg)\,\dd z^2$. See the work of Kuijlaars and Silva \cite{KS} and  \S \ref{Sec crit graphs z plane} below. 
	For the quartic polynomial \eqref{int13} the equilibrium measure  can have 1, 2, or 3 cuts. See  the works of Bertola and Tovbis \cite{BertolaTovbis2015,BertolaTovbis2016}
	and \S \ref{Sec crit graphs z plane} below. Figure \ref{fig:Phase Diagram} depicts the phase regions on the complex plane $\sg$ corresponding to different numbers of cuts. There are three critical points on the phase diagram,
	\begin{equation}\label{int16}
	\sg_1=-2\quad \textrm{ and}\quad  \sg_{2,3}=\pm \sqrt {12}\,\ii,
	\end{equation}
	and six critical curves, separating different phase regions. We will denote by $\mathcal C[a,b]$ a critical curve connecting  the points $a$ and $b$. We will also
	denote by $\mathcal C[a, e^{\ii \theta}\infty]$ a critical curve which goes from the point $a$ to $\infty$ on the complex plane approaching the direction with angle $\theta$ at infinity.

	Observe that the phase diagram is symmetric with respect to the real axis $\sg$, and the critical curves on  the phase diagram are of the two types:
	\bge
	\item split of a cut, and
	\item birth of a cut.
	\ene
	Notice that on the phase diagram \ref{fig:Phase Diagram} the curves 
	\begin{equation}\label{int17}
	\ga_{1,2} =\mathcal C[-2,\pm \sqrt {12}\,\ii]
	\end{equation}
	correspond to the
	split of a cut, and the ones
	\begin{equation}\label{int18}
	\ga_{3,4} =\mathcal C[\pm \sqrt {12}\,\ii, e^{\pi \mp \pi/4}\infty],\quad 
	\ga_{5,6} =\mathcal C[-2, e^{\pi \mp \pi/4}\infty], 
	\end{equation}
	to the birth of a cut. 
	
	Our first main result in this paper is a description of the critical curves $\ga_1,\ldots,\ga_6$ in terms of critical trajectories of some {\it auxiliary quadratic differentials}. 
	We will denote by $\Ga[a,b]$ a trajectory of a quadratic differential connecting the points $a$ and $b$.

	\begin{theorem}\label{main thm 1}
		Part I. Critical curves separating one-cut and three-cut regions. {\it Let us make the substitution
			\begin{equation}\label{int19}
			\sg=-\frac{3\be}{4}+\frac{4}{\be}\,.
			\end{equation}
			Then the critical curves
			$\ga_1,\ga_2,\ga_3,\ga_4$ mapped to  the $\be$-plane are critical trajectories of the quadratic differential $\Xi(\be)\normalfont{\dd}\be^2$, where
			\begin{equation}\label{int20}
			\Xi(\be) =-\frac{(3\be^2+16)^3(\be^2-16)}{1024\, \be^6}\,.
			\end{equation}}
		Part II. Critical curves separating two-cut and three-cut regions. {\it The curves $\ga_5$ and $\ga_6$ are critical trajectories of the quadratic differential $\Upsilon(\sigma)\normalfont{\dd}\sg^2$, where
			\begin{equation}\label{int21}
			\Upsilon(\sigma)=\frac{\sg^2}{4}-1.
			\end{equation}}
	\end{theorem}

	Let us comment on Theorem \ref{main thm 1}. According to formula \eqref{int20}, the quadratic differential 
	$\Xi(\be)\normalfont{\dd}\be^2$ has five finite critical points: one pole, $\be_0=0$, and four zeros, $\be_1,\be_2,\be_3,\be_4$,
	where
	\begin{equation}\label{int22}
	\be_{1,2}=\pm 4,\quad \be_{3,4}=\pm\frac{4\ii}{\sqrt{3}}\,,
	\end{equation}
	as shown in Figure \ref{fig:Crit Traj be plane}. The pole $\be_0$ is of degree 6, the zeros $\be_{3,4}=\di\pm \frac{4\ii}{\sqrt{3}}$ are of degree 3, and the zeros $\be_{1,2}=\pm 4$ of degree 1.
	Correspondingly, there are five critical trajectories of the quadratic differential emanating from  $\be_3$ and $\be_4$, at the angle of $72^\circ$ to each other, and there are three critical trajectories emanating from each of the simple critical points, $\be_1$ and $\be_2$, at the angle of $120^\circ$ to each other.
	Finally, there are four critical trajectories emanating from the origin, at the angle of $90^\circ$ to each other. See Figure \ref{fig:Crit Traj be plane}.
	
	Observe that substitution \eqref{int19} is a scaled Joukowski transformation. It maps the points as follows:
	\begin{equation}\label{int23}
	\be=\pm 4 \mapsto \sg =\mp 2,\quad 
	\be=\pm\frac{4\ii}{\sqrt{3}} \mapsto \sg= \mp \sqrt{12}\,\ii\,.
	\end{equation}
	Respectively, it maps the critical trajectories of the quadratic differential  $\Xi(\be)\dd \be^2$ to the critical curves as follows:
	\begin{equation}\label{int24}
	\begin{aligned}
	&\Ga\left[4,\mp \frac{4\ii}{\sqrt{3}}\right]
	\to
	\mathcal C \left[-2, \pm\sqrt{12}\,\ii\right]=\ga_{1,2},\\
	&\Ga\left[\mp\frac{4\ii}{\sqrt{3}},0\right]\; (\Re \be\ge 0)
	\to \mathcal C[\pm \sqrt {12}\,\ii, e^{\pi \mp \pi/4}\infty]=\ga_{3,4}.
	\end{aligned}
	\end{equation}
	This gives the critical curves $\ga_1,\ga_2,\ga_3,\ga_4$
	as the Joukowski type map of the critical trajectories
	of the quadratic differential  $\Xi(\be)\dd \be^2$. Observe that
	the one-cut region on the $\be$-plane is bounded by the 
	trajectories 
	\[
	\Ga\left[4,\pm \frac{4\ii}{\sqrt{3}}\right] \quad\textrm{and}
	\quad 
	\Ga\left[\pm \frac{4\ii}{\sqrt{3}},0\right]\; (\Re \be\ge 0).
	\]
	See Figure \ref{fig:Crit Traj be plane}.
	
	Furthermore, the critical curves $\ga_5,\;\ga_6$, separating two-cut and three-cut regions, are critical trajectories of the quadratic differential $\Upsilon(\sg)\dd \sg^2$, where $\Upsilon(\sg)$ is given in \eqref{int21}. Observe that $\Upsilon(\sg)$ has two simple critical points $\sg_{1,2}=\pm 2$, and the critical curves $\ga_5,\;\ga_6$ are the critical trajectories of the quadratic differential $\Upsilon(\sg)\dd \sg^2$ labeled by $1$ and $2$ in Figure \ref{fig:Crit Traj 2to3 sigma plane}. We prove Theorem \ref{main thm 1} in \S \ref{Sec Aux QD}.
	
	Our second main result in this paper, which we prove in \S \ref{Sec crit graphs z plane}, is a description
	of the equilibrium measure in different phase regions on the phase diagram.
	\begin{theorem}\label{thm 1.2}
		Part I. One-cut region.  Let $\mathcal{O}_{1}$ be the open set on the complex plane $\sg$ lying to the right of the curves $\ga_1,\;\ga_2,\ga_3$, and $\ga_4$, see Figure \ref{fig:Phase Diagram}. Then for all $\sg\in\mathcal{O}_1$ 
		\bge
		\item
		The equilibrium measure $\nu_{\rm eq}=\nu_{\rm eq}(\sg)$ is regular. 
		\item 
		$\nu_{\rm eq}$ has a one-cut support $\Gamma[-b_1,b_1]$ which is a critical trajectory of the quadratic differential $(z^2-z_0^2)^2(z^2-b_1^2) \normalfont{\dd} z^2$. 
		\item
		The critical points $b_1$ and $z_0$ of this quadratic differential depend analytically on $\sg\in\mathcal{O}_1$.
		\ene
		
		Part II. Two-cut region. Let $\mathcal{O}_{2}$ be the open set on the complex plane $\sg$ lying to the left of the curves $\ga_5,\;\ga_6$, see Figure \ref{fig:Phase Diagram}. Then for all $\sg\in\mathcal{O}_2$ 
		\bge
		\item
		The equilibrium measure $\nu_{\rm eq}=\nu_{\rm eq}(\sg)$ is regular.
		\item
		$\nu_{\rm eq}$ has a two-cut support $\Ga[-b_2,-a_2]\cup \Ga[a_2,b_2]$ where the support cuts are critical trajectories of the quadratic differential $z^2(z^2-a_2^2)(z^2-b_2^2) \normalfont{\dd} z^2$.
		\item
		The critical points $a_2$ and $b_2$ of this quadratic differential depend analytically on $\sg\in\mathcal{O}_2$.
		\ene
		
		Part III. Three-cut region.  Let $\mathcal{O}_3$ be the open set on the complex plane $\sg$ consisting of two connected components, $\mathcal{O}_3=\mathcal{O}_{31}\cup\mathcal{O}_{32}$, lying to the left of the curves $\ga_1,\;\ga_3,\;\ga_5$ and $\ga_2,\;\ga_4,\;\ga_6$, see Figure \ref{fig:Phase Diagram}. Then for all $\sg\in\mathcal{O}_3$ 
		\bge
		\item
		The equilibrium measure $\nu_{\rm eq}=\nu_{\rm eq}(\sg)$ is regular.
		\item
		$\nu_{\rm eq}$ has a three-cut support $\Ga[-c_3,-b_3]\cup\Ga[-a_3,a_3]\cup\Ga[a_3,b_3]$, where the support cuts are critical  trajectories of the quadratic differential $(z^2-a_3^2)(z^2-b_3^2)(z^2-c_3^2)\normalfont{\dd}z^2$. 
		\ene
	\end{theorem}
	
	\begin{remark} \normalfont
		In fact $\Re a_{3}, \Im a_{3}, \Re b_{3}, \Im b_{3}, \Re c_3$ and $\Im c_3$ are real-analytic functions of $\Re \sigma$ and $\Im \sigma$ for all $\sigma \in \mathcal{O}_{3}$. In \cite{BBGMT2}, in the more general context where the external field is of even degree $2p$, $p \in \N$, among other things we establish the real-analyticity of the real and imaginary parts of the end-points for all $q$-cut regimes, $1 \leq q \leq 2p-1$, with respect to the real and imaginary parts of the complex parameters in the external field.
	\end{remark}

	\subsection{Topological Expansion of the Free Energy and Combinatorics of Four-valent Graphs}
	
	Our third main result in this paper concerns the topological expansion of the free energy, 
	\begin{equation}\label{int25}
	F_{NN}(\sg)=\frac{1}{N^2} \ln\frac{Z_{NN}(\sg)}{{\mathcal Z}_{NN}^{\rm GUE}}\,. 
	\end{equation}
	
	The existence of the $\di \frac{1}{N^2}$ expansion of the free energy for general real potentials
	\[ \mathscr{V}(z) = \frac{z^2}{2} + \sum^{\nu}_{k=1} u_k z^k, \]
	is proven in \cite{ErcolaniMcLaughlin}, where $u_k \in \R$ are such that the corresponding partition function exists. The analogous result for the complex cubic potential  $\mathscr{V}(z) = \frac{z^2}{2} + u z^3$ was proven in \cite{BD}, and in this work we extend this result for the complex quartic potential \eqref{int2}, or equivalently for \eqref{int13}.
	
	\begin{theorem}\label{thm topexp}
		{\it For all $\sg$ in the one-cut region $\mathcal{O}_1$, the free energy $F_{NN}(\sg)$ admits the topological expansion,
			\begin{equation}\label{int26}
			F_{NN}(\sg)\sim \sum_{g=0}^\infty \frac{f_{2g}(\sg)}{N^{2g}}\,,
			\end{equation}}
		and the functions $f_{2g}(\sigma)$ are analytic in $\sigma$ for all $\sigma \in \mathcal{O}_{1}$.     
	\end{theorem}
	
	\begin{remark}
		\normalfont Due to \eqref{int15b}, we also have  \begin{equation}\label{int10111}
		\mathscr{F}_{NN}(u) \sim \sum_{g=0}^M \frac{\mathcal{f}_{2g}(u)}{N^{2g}}\,.
		\end{equation}
		In \S \ref{Sec Free Energy} we show that the coefficients $\mathcal{f}_{2g}(u)$ are analytic functions of $u$ in a neighborhood of the origin independent of $g$, more precisely in a disk of radius $\frac{1}{12}$.
	\end{remark}

	As mentioned before, \eqref{int10111} is referred to as the topological expansion of the partition function. Roughly, the quest for models of quantum gravity led to the 2-dimensional reduction in which large but finite collections of different geometries on Riemann surfaces are considered, and one seeks a natural probability measure on these geometries.   F. David \cite{David85} and V. Kazakov \cite{Kaz85} first introduced such random surfaces discretized using polygons to define models of two-dimensional quantum gravity, making use of the connection between graphical enumeration and integrals over large matrices discovered by t'Hooft \cite{tHooft}. 
	To understand the probability measure, one needs to know how many of these geometries there are, and the problem in enumerative geometry that emerges is to count the number of graphs that can be embedded into a Riemann surface, according to the genus of the surface and the number of vertices of different valences.  As discovered in the subsequent works \cite{BIPZ}, 
	\cite{Bes79} and \cite{BIZ},  the topological expansion above should be an expansion of generating functions, in which $\mathcal{f}_{2g}(u)$ is a combinatorial generating function whose $j$-th coefficient yields the number of labelled graphs with $j$ vertices of valence 4, that can be embedded in a Riemann surface of genus $g$.  Yet another connection was discovered by Witten \cite{Witten}, to the intersection theory of the moduli space of Riemann surfaces, where intersection numbers can also be computed using matrix integrals.

	Since the emergence of rigorous mathematical analysis of the partition function by Riemann-Hilbert methods in \cite{ErcolaniMcLaughlin} and in \cite{BleherIts2005}, there have followed works aimed at extracting explicit information about the generating functions and about the important combinatorial coefficients.  For example, in \cite{Ercolani-McLaughlin-Pierce} the authors initiated an investigation of the generating functions in the topological expansion in the case that all vertices were of a fixed, even, valence.  They made use of both the Toda equations and the string equations, and provided a description of structural properties of the generating functions in terms of inversion of certain differential operators. They extracted some explicit information for enumeration of maps on surfaces of genus $0$, $1$, and $2$, along with recursive definitions for higher genus.  (Explicit representation of the generating function means a complete solution of the combinatorial problem for each genus and maps of a fixed valence type.)  Later, Ercolani \cite{ErcolaniCaustics, ErcolaniConservationLaws} continued this research, analyzing a hierarchy of partial differential equations coming from the Toda lattice equations (and the asymptotic expansion of the partition function) and derived semi-explicit characterizations of the $\mathcal{f}_{2g}(u)$ as rational functions of other auxiliary functions.
	
	As already mentioned above, for the three-valent case (the cubic matrix model), the Riemann-Hilbert analysis and topological expansion were established in \cite{BD} and \cite{BDY}, where in particular the authors explicitly evaluated the combinatorial coefficients explicitly for genus $0$ and $1$.  Characterization results for the generating functions for other odd valences have recently been obtained by Ercolani and Waters \cite{ErcolaniWaters1}.
	
	An interesting difference in approach between the present work and the works \cite{ErcolaniCaustics,ErcolaniConservationLaws,Ercolani-McLaughlin-Pierce,ErcolaniWaters1} is the following:  we exploit the string equation and a (possibly new) explicit equation for the first derivative of the free energy \eqref{F'}  to obtain recursive relations, whereas in the works of Ercolani and collaborators, they use the string equation and a hierarchy of partial differential equations derived from the Toda lattice system of ordinary differential equations.  At present our equation for the first derivative of the free energy is only known for the quartic model, while the analysis of Ercolani and collaborators works for more general single valence settings.
	
	In \S \ref{Sec Free Energy} we establish a number of results concerning these generating functions.  We provide recursive relations in $g$, as well as explicit representations for $g=0,1,2,$ and $3$.  The representations \eqref{number sphere}, \eqref{number torus}, and \eqref{N j+1 2} respectively for $g=0, 1$, and $2$ agree with the classical paper of Bessis, Itzykson, and Zuber \cite{BIZ}, while we believe the result for $g=3$ is new.  As with other representations, the recursive algorithm does yield explicit representations for any genus, but requires more effort as the genus increases. To that end, let us highlight the following result regarding enumeration of graphs. 
	
	\begin{theorem}
		Let $\mathscr{N}_j(g)$ be the number of connected labeled 4-valent graphs with $j$ vertices which are realizable on a closed Riemann surface of genus $g$ , but not realizable on Riemann surfaces of lower genus. For the Riemann surface of genus three we have
		\begin{equation}\label{N j+4 3 12345}
		\begin{split}
		\mathscr{N}_{j+4}(3) & = \frac{16 \cdot 48^j\left(j+3\right)!}{3(j)!} \times \\ & \left( \frac{2741}{10}(j+5)!   - \frac{291}{10}j (j+4)! - \frac{2741}{1260}  \frac{ (2j+9)!}{ 4^j (j+4)!} - \frac{292j(2j+7)!}{315 \cdot 4^j (j+3)!} \right), \quad j \in \N,
		\end{split}
		\end{equation}
		where $\mathscr{N}_{1}(3)=\mathscr{N}_{2}(3)=\mathscr{N}_{3}(3)=\mathscr{N}_{4}(3)=0$, that is to say all connected labeled 4-valent graphs with $1,2,3,$ or $4$ vertices can be realized on the sphere, torus, or the two-holed torus.
	\end{theorem}

	We also highlight a result that describes the asymptotic behavior of the number of four-valent graphs on a Riemann surface of arbitrary genus $g$, as the number of vertices grows to infinity.  The following result is basically a corollary of Theorem \ref{thm F 2g j}.
	
	\begin{theorem}\label{thm combinatorics}
		The asymptotics of the number of connected labeled $4$-valent graphs on a Riemann surface of genus $g \in \N \cup \{0\}$, as the number of vertices tends to infinity, is given by \begin{equation}\label{N j g large j asymptotics1}
		\mathscr{N}_j(g) = \mathcal{K}_g 48^j j! j^{\frac{5g-7}{2}}\left(1+O(j^{-1/2}) \right), \qquad j \to \infty, 
		\end{equation}
		where the constants $\mathcal{K}_g$ are the same as the ones in Theorem \ref{thm F 2g j}:
		\begin{equation}\label{K g in terms of Cg1}
		\mathcal{K}_g = \begin{cases}
		\di \frac{12^{\frac{5g-1}{2}}}{\left(\frac{5g-5}{2}\right)!}\frac{1}{5g-3} \mathcal{C}_{2g}, & g =2k+1, \\[10pt]
		\di \frac{ 12^{\frac{5g-1}{2}} 2^{5g-4}}{\sqrt{\pi}}\frac{\left(\frac{5g-4}{2}\right)!}{\left(5g-3\right)!}\mathcal{C}_{2g}, & g =2k,
		\end{cases} \qquad k \in \N,
		\end{equation}
		while $\mathcal{K}_0=2^{-1}\pi^{-1/2}$, and $\mathcal{K}_1=24^{-1}$, where the constants $\mathcal{C}_{2g}$ can be found recursively from the following relations
		\begin{equation}\label{C2g recurrence1}
		\mathcal{C}_{2g} = \frac{1}{2^33^{\frac{1}{2}}} \sum^{g-1}_{\ell=1}  \mathcal{C}_{2g-2\ell} \mathcal{C}_{2\ell} + \frac{(5g-6)(5g-4)}{2^83^{\frac{7}{2}}  }\mathcal{C}_{2g-2}, \qquad  \mathcal{C}_0 = -2^23^{\frac{1}{2}},  \qquad g \in \N .
		\end{equation}
	\end{theorem}
	\begin{remark}\normalfont
		It is worth noticing that the constants $\mathcal C_{2g}$ also arise in the asymptotic expansion of the one-parameter family of the Boutroux tronqu\'ee solutions
		$u(\tau)=u(\tau;\al)$ to the Painlev\'e I equation 
		\[
		u''(\tau)=6u^2(\tau)+\tau
		\]
		as $\tau\to -\infty$. Namely, as $\tau\to-\infty$, 
		\[
		u(\tau;\al)\sim \sqrt{-\frac{\tau}{6}}\,\sum_{k=0}^\infty a_k(-\tau)^{-5k/2},
		\]
		where the coefficients $a_k$ do not depend on the parameter $\al$, and they are given by the nonlinear recursion
		\begin{equation}\label{eq1}
		a_0=1,\quad a_{k+1}=\frac{25 k^2-1}{8\sqrt 6}\,a_k-\frac{1}{2}\sum_{m=1}^k a_m a_{k+1-m}
		\end{equation}
		(see, e.g., the works \cite{ BD1, Boutroux, Costin-et-al, ErcolaniCaustics, ErcolaniConservationLaws, Kapaev}). Define now the rescaled Boutroux functions
		\[
		y(t)=-2^{8/5}3^{2/5}u(-2^{9/5}3^{6/5}t).
		\]
		Then it can be checked directly, by using \eqref{eq1}, that as $t\to\infty$,
		\[
		y(t)\sim \sum_{g=0}^\infty\mathcal C_{2g}t^{(1-5g)/2},
		\]
		where the coefficients $\mathcal C_{2g}$ are given by nonlinear recursion \eqref{C2g recurrence1}. Therefore, the coefficients
		$\mathcal C_{2g}$ coincide with the coefficients of the asymptotic expansion of the rescaled Boutroux functions $y(t)$
		as $t\to\infty$.
		
		The appearance of the Boutroux tronqu\'ee solutions can be explained as follows. 
		Under the substitution $u=\sg^{-2}$ (see formula \eqref{int11}), the critical points $\sg=\pm \ii \sqrt{12}$ and $\sg=-2$ on the phase diagram, depicted 
		on Figure \ref{fig:Phase Diagram} above, are mapped to the points $u=-\frac{1}{12}$ and  $u=\frac{1}{4}$, respectively. The point $u=-\frac{1}{12}$ is closer
		to the origin than the one  $u=\frac{1}{4}\,,$ and it determines the asymptotic behavior of the Taylor coefficients of the
		functions $f_{2g}(u)$ at the origin. But as shown in the paper \cite{DK} of Duits and Kuijlaars, the double scaling limit of the
		model at the critical point $u=-\frac{1}{12}$ is described in terms of a Boutroux tronqu\'ee solution
		to the Painlev\'e I equation. It is noteworthy that the double scaling limit of the matrix
		model at the critical point $u=-\frac{1}{12}$  gives rise to the 2D continuous quantum gravity of Polyakov (see, e.g., the
		papers \cite{DiFGZJ} and \cite{Witten}, 
		and references therein). 
	\end{remark}

	\begin{remark}\normalfont
		The formula \eqref{N j g large j asymptotics1} was formulated as a conjecture in the introduction of \cite{Its-et-al} (see pages 27 through 29 of \cite{Its-et-al} for the relevant references). Here we directly quote from  \cite{Its-et-al}: \begin{quote}
			\textit{``The status of \[\mathscr{N}_j(g) \sim \mathcal{K}_g 48^j j! j^{\frac{5g-7}{2}}\]
				remains that of a conjecture. Nevertheless, the current level of development of the Riemann-Hilbert techniques, and the experience with other combinatorial problems e.g. in random permutations \cite{BDJ}, suggest that all the gaps in the above construction will be soon filled.''}
		\end{quote}
		Indeed in Theorem \ref{thm combinatorics} above we have not only established this conjecture, but furthermore we have also characterized the constants $\mathcal{K}_g$ explicitly in terms of the constants $\mathcal{C}_{2g}$. It should be mentioned that in the recent preprint \cite{ErcolaniWaters1} mentioned above, the authors provide an analogue of \eqref{N j g large j asymptotics1} and \eqref{K g in terms of Cg1} for the general single even-valence potential.  The result in the preprint is stated for even numbers of vertices (see equation A.9 of \cite{ErcolaniWaters1}), but by comparing to our result the form of the asymptotics surely holds for both even and odd numbers of vertices.   They omit the proof, but presumably it follows from a similar analysis done in the same paper for a different combinatorial problem (see Corollary 10.8 of \cite{ErcolaniWaters1}).
	\end{remark}
	\begin{remark} \normalfont
		A very interesting direction of research is to explore the precise connection between the asymptotics of the labeled graphs embedded on a Riemann surface of genus $g$  as the number of vertices go to infinty (as addressed in Theorem \ref{thm combinatorics} for the four-valent case, and in Theorem 1.4 of \cite{BD} for the three-valent case) and the asymptotics of the number of the so-called \textit{rooted maps} as the number of edges goes to infinity.
		
		Let us recall some definitions regarding the latter asymptotics from \cite{BGR}.   Let $\Sigma_g$ be the orientable surface of genus $g$. A map on $\Sigma_g$ is a graph $G$ embedded on $\Sigma_g$ such that all components of $\Sigma_g \setminus G$ are simply connected regions. These components are called faces of the map. A map is \textit{rooted} by distinguishing an edge, an end vertex of the edge and a side of the edge. Let us denote by $M_{n,g}$ the number of rooted maps on $\Sigma_g$ with $n$ edges. In \cite{BenderCanfield} Bender and Canfield showed that
		\begin{equation}\label{Mng asymp}
		M_{n,g} \sim t_g n^{\frac{5g-5}{2}}12^n \qasq n \to \infty,
		\end{equation}
		where the $t_g$ are positive constants which can be calculated recursively. One can already observe the apparent similarity between \eqref{Mng asymp} and \eqref{N j g large j asymptotics1}, which becomes even more interesting when one observes that the first three values for $t_g$ are given in \cite{BenderCanfield} by
		\[ t_0 = \frac{2}{\sqrt{\pi}} \equiv 4 \mathcal{K}_0 , \qquad t_1 = \frac{1}{24} \equiv \mathcal{K}_1 , \qandq t_2 = \frac{7}{4320\sqrt{\pi}} \equiv  \frac{1}{4}\mathcal{K}_2.  \]
		An analogous similarity can also be seen when one compares \eqref{Mng asymp} with equation 1.25 of \cite{BD} which gives the asymptotics of the number of labeled three-valent graphs as the number of vertices goes to infinity. For the asymptotics of the rooted maps see also the works \cite{BenderCanfield,BGR,Gao,Gao1,GLM} and references therein.
	\end{remark}	
	
	\subsection{Outline} The paper is organized as follows:
	
	\begin{itemize}
		\item In \S \ref{section_EM} we derive the end-point equations in the one-cut, two-cut, and three-cut regimes. These endpoint equations are algebraic in the one-cut and two-cut case, thereby allowing for explicit solutions. In the one-cut and two-cut cases we find explicit expressions for the $g$-function and the Euler-Lagrange constants.
		\item In \S \ref{Sec crit graphs z plane} we prove results about the structure of critical graphs in the $z$-plane using the theory of quadratic differentials. We also prove the openness of one-cut, two-cut, and three-cut regimes.
		\item  In \S \ref{Sec Aux QD} we use auxiliary quadratic differentials to prove the phase diagram as depicted in Figure \ref{fig:Phase Diagram}.
		\item   In \S \ref{sec rhp} we prove the topological expansion of the recurrence coefficients of the orthogonal polynomials using the Riemann-Hilbert analysis and the String equations.
		\item  In \S \ref{Sec Free Energy} we derive the Toda equations. We use the equation for $\mathcal{F}^{\prime}$ to prove the topological expansion for the free energy. As a result we extract the combinatorial information about the connected labeled $4$-valent graphs Riemann surfaces of various genera.
		\item  Finally in the Appendix \ref{Appendix Number of graphs} we provide visual illustrations of four valent graphs on the sphere and the torus with one and two vertices. We hope this helps for a deeper understanding of the combinatorial formulae \eqref{number sphere} through \eqref{N j+4 3 1}.
	\end{itemize}

	\section{Equilibrium Measure} \label{section_EM}
	
	In this section we first discuss  the equilibrium measure for a general complex polynomial $V(z)$, and then we will specify it to
	the equilibrium measure of the quartic complex polynomial 
	\eqref{int13}.
	
	\subsection{Equilibrium Measure for a General Complex Polynomial.}
	
	Let  
	\begin{equation}\label{em1}
	V(z)=\frac{z^{2p}}{2p}+\sum_{j=1}^{2p-1} \frac{v_j z^j}{j}
	\end{equation}
	be a polynomial of even degree $2p$ with the leading coefficient $\displaystyle\frac{1}{2p}$ and complex coefficients 
	$\displaystyle\frac{v_j}{j}\,,\;j=1,\ldots,2p-1$.
	We follow the work of Kuijlaars and Silva \cite{KS}, also 
	see the works \cite{Rakhmanov}, \cite{RF}, \cite{Ber}.

	For a given $\ep$, such that 
	\[
	\frac{\pi}{4p}>\ep>0,
	\]
	consider the sectors
	\begin{equation}\label{em2}
	\bga
	S^+_{\ep}&=\left\{z\in\C\;\Big |\; |\arg z|\le \frac{\pi}{4p}-\ep\right\},\\
	S^-_{\ep}&=\left\{z\in\C\;\Big |\; |\arg z-\pi|\le \frac{\pi}{4p}-\ep\right\}.
	\ena
	\end{equation}
	Observe that in these sectors,
	\begin{equation}\label{em3}
	\lim_{z\to\infty} \Re V(z)=\infty.
	\end{equation}
	Let us define a class $\mathcal T$ of {\it admissible contours} on the complex plane. By a contour we mean a {\it continuous curve} $z=z(t)$, $-\infty<t<\infty$, without self-intersections. We say that a contour $\Ga$ is admissible if
	\bge
	\item The contour $\Ga$  is a finite union of  $C^1$ Jordan arcs.
	\item There exists $\ep>0$ and $r_0>0$, such that $\Ga$ goes from $S^-_\ep$ to $S^+_\ep$ in the sense that $\forall\,r>r_0,$  
	$\exists\, t_0<t_1$ such that 
	\[
	z(t)\in S^-_\ep \setminus D_r \quad \forall\, t<t_0;\quad 
	z(t)\in S^+_\ep \setminus D_r \quad \forall\, t>t_1,
	\]
	where $D_r$ is the disk centered at the origin with radius $r$. We will assume that the contour $\Ga$ is oriented from $(-\infty)$ 
	to $(+\infty)$, where $(-\infty)$ lies in the sector $S^-_\ep$ and
	$(+\infty)$ in the sector $S^+_\ep$. The orientation defines an order
	on the contour $\Ga$.
	\ene
	An example of an admissible contour is the real line.

	Let $\Ga\in\mathcal T$ be an admissible contour and $\mathcal P(\Ga)$ the space of probability measures $\nu$ on $\Ga$ such that
	\begin{equation}\label{em4}
	\bga
	\underset{\Ga}{\int} &  |\Re V(s)|\, \dd \nu(s)<\infty.
	\ena
	\end{equation}
	Consider the following real-valued functional on $\mathcal P(\Ga)$:
	\begin{equation}\label{em5}
	I_{V,\Ga}(\nu):= \underset{\Ga\,\times\,\Ga}{\iint} \log \frac{1}{|z-s|}\,\dd \nu(z) \dd \nu(s) + \int_{\Ga} \Re V(s)\, \dd \nu(s).
	\end{equation}
	Then there exists a unique minimizer $\nu_{V,\Ga}$ of the functional 
	$I_{V,\Ga}(\nu)$, so that
	\begin{equation}\label{em6}
	\min_{\nu\in \mathcal P(\Ga)} I_{V,\Ga}(\nu)= I_{V,\Ga}(\nu_{V,\Ga}).
	\end{equation}
	See the work \cite{SaffTotik}.
	
	The probability measure $\nu_{V,\Ga}$ is called the
	{\it equilibrium measure} of the functional $I_{V,\Ga}(\nu)$.
	The support of $\nu_{V,\Ga}$ is a compact set $J_{V,\Ga}\subset \Ga$. 
	An important fact is that the equilibrium measure is uniquely determined by the {\it Euler--Lagrange variational conditions}. Namely,
	$\nu_{V,\Ga}$ is the unique probability measure $\nu$ on $\Ga$ such that
	there exists a constant $l$, a Lagrange multiplier, such that
	\begin{equation}\label{em7}
	\bga
	& U^{\nu}(z)+\frac{1}{2}\,\Re V(z)=\ell,\quad z\in \supp \nu,\\
	& U^{\nu}(z)+\frac{1}{2}\,\Re V(z)\ge \ell,\quad z\in \Ga\setminus \supp \nu,
	\ena
	\end{equation}
	where 
	\begin{equation}\label{em8}
	U^{\nu}(z)=\int_{\Ga} \log \frac {1}{|z-s|}\,\dd \nu(s) 
	\end{equation}
	is the {\it logarithmic potential} of the measure $\nu$ \cite{SaffTotik}.

	Now we maximize $I_{V}(\nu_{V,\Ga})$ over $\Ga\in \mathcal T$. The main result of the work of  Kuijlaars and Silva concerns the existence and properties of the maximizing contour  $\Ga_0\in \mathcal T$.
	They prove that the maximizing contour $\Ga_0\in\mathcal T$ exists, and
	the equilibrium measure 
	\[
	\nu_{\rm eq}=\nu_{V,\Ga_0}
	\]
	on $\Ga_0$ is supported by a set $J\subset\Ga_0$ which is a finite union 
	of {\it analytic arcs} $\Ga_0[a_k,b_k]\subset \Ga_0,\;k=1,\ldots,q$, 
	\[
	J=\bigcup_{k=1}^q \Ga_0[a_k,b_k],\quad a_1<b_1\le a_2<b_2\le \ldots\le a_q<b_q,
	\] 
	that are {\it critical trajectories}
	of a quadratic differential $Q(z)\,\dd z^2$ (see the beginning of \S \ref{Sec crit graphs z plane} for a review of definitions and basic facts about quadratic differentials), where $Q(z)$ is a polynomial of degree 
	\begin{equation}\label{em9}
	\deg Q(z)=2\deg V(z)-2=4p-2.  
	\end{equation}

	Furthermore, Kuijlaars and Silva prove that the polynomial $Q(z)$ is equal to
	\begin{equation}\label{em10}
	Q(z)=\left(-\om(z)+\frac{V'(z)}{2}\right)^2, 
	\end{equation}
	where 
	\begin{equation}\label{em11}
	\om(z)= \int_{J} \frac{ \dd \nu_{\rm eq}(s)}{z-s}
	\end{equation}
	is the resolvent of the measure $\nu_{0}$. Expanding
	\[
	\frac{1}{z-s}=\frac{1}{z}+\frac{s}{z^2}+\frac{s^2}{z^3}+\ldots,
	\]
	we obtain that as $z\to\infty$,
	\begin{equation}\label{em12}
	\om(z)=\frac{1}{z}+\frac{m_1}{z^2}+\ldots,\quad m_k=\int_J s^k
	\dd \nu_{\rm eq}(s).
	\end{equation}

	In addition, the equilibrium measure $\nu_{\rm eq}$ is absolutely continuous 
	with respect to the arc length and
	\begin{equation}\label{em13}
	\dd \nu_{\rm eq}(s)=\frac{1}{\pi \ii}\,Q_+(s)^{1/2} \dd s,
	\end{equation}
	where $Q_+(s)^{1/2}$ is the limiting value of the function
	\begin{equation}\label{em14}
	Q(z)^{1/2}=-\int_{J} \frac{ \dd \nu_{\rm eq}(s)}{z-s}+\frac{V'(z)}{2}, 
	\end{equation}
	as $z\to s\in J$ from the left-hand side of $J$ with respect to the orientation of the contour $\Ga_0$ from $(-\infty)$ to $\infty$.
	Observe that as $z\to\infty$,
	\begin{equation}\label{em15}
	Q(z)^{1/2}=-\left(\frac{1}{z}+\frac{m_1}{z^2}+\ldots\right)
	+\frac{1}{2}\left(z^{2p-1}+\sum_{j=1}^{2p-1}v_jz^{j-1}\right). 
	\end{equation}
	
	A very important result of  Kuijlaars and Silva is that
	the equilibrium measure $\nu_{\rm eq}$ is {\it unique as the max-min measure}. On the other hand, the contour $\Ga_0$ is not unique because it can be deformed outside of the support $J$ of $\nu_{\rm eq}$. 
	
	\subsection{The $g$-function.}
	We define the $g$-function as
	\begin{equation}\label{em16}
	g(z)= \int_{J} \log(z-s)\,\dd \nu_{\rm eq}(s),
	\end{equation}
	where for a fixed $s\in J$, we consider a cut of $\log(z-s)$ on the part of the curve
	$\Ga_0$ where $z<s$ with respect to the ordering on $\Ga_0$. Observe that by \eqref{em11},
	\begin{equation}\label{em17}
	g'(z)= \int_{J}\frac{\dd \nu_{\rm eq}(s)}{z-s}=\om(z),
	\end{equation}
	In addition, by \eqref{em8}, the logarithmic potential $U^{\nu_{\rm eq}}(z)$ is equal to
	\begin{equation}\label{em18}
	U^{\nu_{\rm eq}}(z)=\int_{J} \log \frac {1}{|z-s|}\,\dd \nu_{\rm eq}(s) 
	=-\Re g(z)
	\end{equation}
	hence the Euler--Lagrange variational conditions \eqref{em7} can be written as
	\begin{equation}\label{em19}
	\bga
	& -\Re g(z)+\frac{1}{2}\,\Re V(z)=\ell,\quad z\in J,\\
	& -\Re g(z)+\frac{1}{2}\,\Re V(z)\ge \ell,\quad z\in \Ga_0\setminus J.
	\ena
	\end{equation}
	
	\subsection{Regular and Singular Equilibrium Measures} An equilibrium measure $\nu_{\rm eq}$ is called {\it regular} if the following three conditions hold:
	\bge
	\item The arcs $\Ga_0[a_k,b_k],\;k=1,\ldots,q,$ of the support of $\nu_{\rm eq}$ are disjoint.
	\item The end-points $\{a_k,b_k,\;k=1,\ldots,q\}$ are simple zeros of the polynomial $Q(s)$.
	\item There is a contour $\Ga_0$ containing the support $J$ of $\nu_{\rm eq}$
	such that 
	\begin{equation}\label{em20}
	U^{\nu}(z)+\frac{1}{2}\,\Re V(z)> \ell,\quad z\in \Ga_0\setminus J.
	\end{equation}
	\ene
	An equilibrium measure $\nu_{\rm eq}$ is called {\it singular} (or
	{\it critical}) if it is not regular.
	
	\subsubsection{Regular Equilibrium Measures}  Suppose that
	an equilibrium measure $\nu_{\rm eq}$ is regular.
	Since the resolvent 
	\begin{equation}\label{em21}
	\om(z)= \int_{J} \frac{ \dd \nu_{\rm eq}(s)}{z-s}
	\end{equation}
	is analytic on $\C\setminus J$,
	it follows from equation \eqref{em10} that if the  equilibrium measure $\nu_0$ is regular then all the zeros of the
	polynomial $Q(z)$ different from the end-points $\{a_k,b_k,\;k=1,\ldots,q\}$ are of even degree, hence
	$Q(z)$ can be written as
	\begin{equation}\label{em22}
	Q(z)=\frac{1}{4}\,h(z)^2R(z),
	\end{equation}
	where $h(z)$ is a polynomial,
	\begin{equation}\label{em23}
	h(z)=\prod_{j=0}^r (z-z_j),
	\end{equation}
	with zeroes $z_0,\ldots,z_r$ different from the end-points $\{a_k,b_k,\;k=1,\ldots,q\}$,
	and
	\begin{equation}\label{em24}
	R(z)=\prod_{k=1}^q (z-a_k)(z-b_k).
	\end{equation}
	Thus,
	\begin{equation}\label{em25}
	Q(z)=\frac{1}{4}\,h(z)^2R(z)=\frac{1}{4}\prod_{j=0}^r (z-z_j)^2
	\prod_{k=1}^q (z-a_k)(z-b_k).
	\end{equation}
	By taking the square root with the plus sign, we obtain that
	\begin{equation}\label{em26}
	Q(z)^{1/2}=\frac{1}{2}\,h(z)R(z)^{1/2}
	=\frac{1}{2}\prod_{j=0}^r (z-z_j)\left[\prod_{k=1}^q (z-a_k)(z-b_k)\right]^{1/2},
	\end{equation}
	Correspondingly, equation \eqref{em13} can be rewritten as
	\begin{equation}\label{em27}
	\bga
	\dd \nu_{\rm eq}(z)&=\frac{1}{2\pi \ii}\,h(z)R_{+}(z)^{1/2} \dd z =\frac{1}{2\pi \ii}\,\prod_{j=0}^r (z-z_j)\left[\prod_{k=1}^q (z-a_k)(z-b_k)\right]_+^{1/2} \dd z.
	\ena
	\end{equation}

	Now we will apply the above results for the equilibrium measure of a general complex polynomial $V(z)$ to the quartic polynomial
	\[
	V(z)=V(z;\sg)=\frac{z^4}{4}+\frac{\sg z^2}{2}\,,\quad \sg\in \C.
	\]
	
	\subsection{Equilibrium Measures for the Quartic Polynomials
		$V(z;\sg)$}\label{SubSec EqMeas quartic}
	For the quartic polynomial in hand, equation \eqref{em10} for the 
	polynomial $Q(z)$ reads
	\begin{equation}\label{em28}
	Q(z)=\left(-\om(z)+\frac{z^3+\sg z}{2}\right)^2, \quad
	\om(z)= \int_{J} \frac{ \dd \nu_{\rm eq}(s)}{z-s}\,.
	\end{equation}
	Since the polynomial $V(z)$ is even, the uniqueness of the
	equilibrium measure $\nu_{\rm eq}$ implies that 
	\bge
	\item $\nu_{\rm eq}$ is even, $\nu_{\rm eq}(-s)=\nu_{\rm eq}(s)$.
	\item The resolvent $\om(z)$ is odd, $\om(-z)=-\om(z)$, and
	\item The polynomial $Q(z)$ is even, $Q(-z)=Q(z)$.
	\ene
	Considering $z\to\infty$, we obtain that
	\begin{equation}\label{em29}
	Q(z)=\left(\frac{z^3+\sg z}{2}
	-\frac{1}{z}-\frac{m_2}{z^3}-\ldots\right)^2=\frac{1}{4}
	\big[z^6+2\sg z^4+(\sg^2-4)z^2-4(\sg+m_2)\big].
	\end{equation}
	Since $Q(z)$ is a polynomial of degree 6, the possible number of cuts
	$q$ in formula \eqref{em25} can be $q=1,2,$ and $3$. Let us consider them in more detail.
	\subsubsection{One-Cut Equilibrium Measure}\label{subsec 2.4.1}
	When $q=1$, formula \eqref{em25} gives that
	\begin{equation}\label{em30}
	Q(z)=\frac{1}{4}\, (z-z_0)^2 (z-z_1)^2
	(z-a_1)(z-b_1).
	\end{equation}
	Since the equilibrium measure $\nu_0$ is even and the polynomial $Q(z)$ is even, we have that
	\begin{equation}\label{em31}
	-a_1=b_1,\quad -z_1= z_0,
	\end{equation}
	hence
	\begin{equation}\label{em32}
	Q(z)=\frac{1}{4}\left(z^2-z_0^2\right)^2(z^2-b_1^2).
	\end{equation}
	Equating this expression to the one \eqref{em29}, we obtain that
	\begin{equation}\label{em33}
	(z^2-c^2)^2(z^2-b^2)=z^6+2\sg z^4+(\sg^2-4)z^2-4(\sg+m_2).
	\end{equation}
	Comparing the coefficients at $z^4$ and $z^2$, we obtain the system of equations,
	\begin{equation}\label{em34}
	\left\{\;
	\bga
	&b_1^2+2z_0^2=-2\sg,\\
	&2b_1^2z_0^2+z_0^4=\sg^2-4.
	\ena
	\right.
	\end{equation}
	From the first equation we have that
	\[
	\sg=-\frac{b_1^2}{2}-z_0^2.
	\]
	Substituting this expression into the second equation and simplifying we obtain that
	\[
	b_1^2(b_1^2-4z_0^2)=16.
	\]
	Thus, we have the system of equations,
	\begin{equation}\label{em35}
	\left\{\;
	\bga
	&b_1^2+2z_0^2=-2\sg,\\
	&b_1^2(b_1^2-4z_0^2)=16.
	\ena
	\right.
	\end{equation}
	Solving it we obtain that 
	\begin{equation} \label{em36}
	b_1^2=\frac{2}{3}\left( -\sigma \pm \sqrt{12+\sigma^2} \right) \qandq z_0^2=\frac{1}{3}\left( -2\sigma \mp  \sqrt{12+\sigma^2} \right).
	\end{equation}
	As shown by Bleher and Its (see \cite{BI, BleherIts2003, BleherIts2005}), for real $\sg>-2$, the one-cut equilibrium measure persists with a real $b_1>0$. This determines the sign
	in the latter formulae,
	\begin{equation} \label{em37}
	b_1=\sqrt{\frac{2}{3}\left( -\sigma + \sqrt{12+\sigma^2} \right)} \qandq z_0=\sqrt{\frac{1}{3}\left( -2\sigma -  \sqrt{12+\sigma^2} \right)}.
	\end{equation}
	Observe that for $\sg=-2$, we have $b_1=2$ and $z_0=0$. Theorem \ref{main thm 1} tells us that the point $\sg=-2$ corresponds to a split of the cut $[-2,2]$ at $z_0=0$, when $\sg$ is decreasing from $\sg>-2$ to $\sg<-2$. As shown in \cite{BleherIts2003}, the critical behavior 
	of the quartic model at $\sg=-2$ is governed by
	the Hastings-McLeod solution to the second Painlev\'e equation PII. 
	In what follows we will show that formulae \eqref{em37} are analytically extended from the real half-line $\sg>-2$ to the whole one-cut region $\sg\in \mathcal{O}_1$ on the complex plane.

	\begin{remark}\label{Remark branches}
		\normalfont Notice that the branch cuts for $z_0(\sigma)$ and $b_1(\sigma)$ are different. Indeed for $b_1(\sigma)$, since $-\sigma + \sqrt{12+\sigma^2} \neq 0$ for all $\sigma$, there are only two branch cuts $L_{\pm}$ emanating from $\pm \ii \sqrt{12}$. However when we consider $z_0(\sigma)$, we notice that $ -2\sigma -  \sqrt{12+\sigma^2}$ does vanish for $\sigma=-2$. So, for $z_0(\sigma)$, apart from the two branch cuts emanating from $\pm \ii \sqrt{12}$ (which we chose to be $L_{\pm}$: the same as the ones for $b_1(\sigma)$) there is one more branch cut $L$ which emanates from $-2$. In this work we choose $L_{\pm} = \pm \ii \sqrt{12} - t$ and $L=-2-t$, $t>0$. We choose the branches so that for $\sigma>-2$:
		\[ z_0(\sigma)=\ii y_0 \quad \mbox{with} \quad y_0>0, \qquad \qandq \qquad b_1(\sigma)>0. \] 
		Since $z^2_0(\sigma)\in (0,\infty)$ for $\sigma\in(-\infty, -2)$, the branch cut in the $z_0^2$-plane is the positive real axis and we fix the branch of $z_0$ by fixing $0 \leq \arg(z^2_0) < 2\pi$.
	\end{remark}
	
	In the one-cut regime the $g$ function can be explicitly computed. To this end, using \eqref{em10}, \eqref{em17}, and \eqref{em32} we can write
	\begin{equation}
	g(z;\sigma) = \frac{V(z)+\ell^{(1)}_*}{2}+\frac{\eta_1(z;\sigma)}{2} , \qquad z \in \C \setminus \Ga_{\sigma}(-\infty,b_1],
	\end{equation}
	where 
	\begin{equation}\label{eta def}
	\eta_1(z;\sigma):= -\int^{z}_{b_1} (s^2-z^2_0)\sqrt{s^2-b_1^2}\dd s, \qquad z \in \C \setminus \Ga_{\sigma}(-\infty,b_1],
	\end{equation}
	in which the path of integration does not cross $\Ga_{\sigma}(-\infty,b_1(\sigma)]$. Notice that
	\begin{equation}\label{eta+ eta- on the support}
	\eta_{1,+}(z)=-\eta_{1,-}(z), \qquad z \in J_{\sigma},
	\end{equation}
	and
	\begin{equation}
	\eta_{1,+}(z)-\eta_{1,-}(z)=4\pi \ii, \qquad   z \in \Ga_{\sigma}(-\infty,-b_1).
	\end{equation}
	Using several integration by parts and trigonometric substitutions we find
	\begin{equation}\label{explicit integration one cut}
	\eta_1(z;\sigma)  = \frac{z}{8}(b_1^2+4z^2_0-2z^2)\sqrt{z^2-b_1^2} +2\log\left( \frac{z+\sqrt{z^2-b_1^2}}{b_1} \right),
	\end{equation}where we have used \eqref{em35} in simplifying the expression. Therefore we have the following explicit form of the $g$-function in the one-cut regime 
	\begin{equation}\label{g-1-cut-explicit}
	g(z;\sigma) = \frac{1}{2}\left( \frac{\sigma z^2}{2} + \frac{z^4}{4} + \ell^{(1)}_* \right)  -\frac{z}{16}(b_1^2+4\sigma+2z^2)\sqrt{z^2-b_1^2} + \log\left( \frac{z+\sqrt{z^2-b_1^2}}{b_1} \right),
	\end{equation}
	where we have also used \eqref{em35}. Here the branches must be chosen to ensure that the branch cut for $g$ is $\Ga_{\sigma}(-\infty,b_1]$.  Also the constant $\ell_*^{(1)}$ can be found using the requirement that 
	\[
	g(z;\sigma)= \log z +\mathcal O(z^{-1})
	\]
	as $z \to \infty$. Indeed,
	
	\begin{equation}\label{ell* 1 cut}
	\ell^{(1)}_* = \frac{1}{12}\left(\sigma^2-\sigma\sqrt{12+\sigma^2}\right) + \log\left( -\sigma+\sqrt{12+\sigma^2} \right)-\frac{1}{2}-\log 6.
	\end{equation}
	
	In what follows we use the notations
	\begin{equation}\label{g+mg-}
	\mathscr{G}^{(j)}_1(z;\sigma):= g_{+}(z;\sigma)- g_{-}(z;\sigma),
	\end{equation}
	and
	\begin{equation}\label{g+pg-mVml}
	\mathscr{G}^{(j)}_2(z;\sigma):= g_{+}(z;\sigma) + g_{-}(z;\sigma) - V(z;\sigma) - \ell^{(j)}_*,
	\end{equation}
	for $j=1,2$, and $3$, respectively associated with the one-cut, two-cut, and three-cut regimes. We have
	
	\noindent\begin{minipage}{.45\linewidth}
		\begin{equation}\label{g+ minus g- one cut}
		\mathscr{G}^{(1)}_1(z;\sigma) = \begin{cases}
		0, & z \in  \Ga_{\sigma}(b_1,\infty), \\
		\eta_{1,+}(z;\sigma), & z \in J_{\sigma}, \\
		2\pi \ii, & z \in \Ga_{\sigma}(-\infty,-b_1),
		\end{cases}
		\end{equation}	
	\end{minipage}	
	\begin{minipage}{.45\linewidth}
		\begin{equation}\label{g+ plus g- one cut}
		\mathscr{G}^{(1)}_2(z;\sigma) = \begin{cases}
		\eta_1(z;\sigma), & z \in \Ga_{\sigma}(b_1,\infty) \\
		0, & z \in J_{\sigma}, \\
		\eta_{1,\pm}(z;\sigma)\mp 2\pi \ii, & z \in \Ga_{\sigma}(-\infty,-b_1). 
		\end{cases}
		\end{equation}	
	\end{minipage}
	
	Note that for $z \in J_{\sigma}$, in particular we have
	\begin{equation}
	\Re(g_+(z;\sigma))+\Re(g_-(z;\sigma))-\Re V(z;\sigma)-\Re(\ell^{(1)}_*)=0.
	\end{equation}
	Also, since $\rho_V(s;\sigma) \dd s$ is a probability measure and thus real-valued on $J_{\sigma}$, we have
	\begin{equation}
	\Re(g_+(z;\sigma))+\Re(g_-(z;\sigma)) = 2\int_{J_{\sigma}}\log|z-s|\dd \nu_{\rm eq}(s).
	\end{equation}
	Comparing the last two equations with \eqref{em8} and the first member of \eqref{em7} implies that the Euler-Lagrange constant $\ell$ is given by
	\begin{equation}
	\ell\equiv\ell^{(1)}=-\frac{\Re \ell^{(1)}_*}{2},
	\end{equation}
	where $\ell_*^{(1)}$ is explicitly given by \eqref{ell* 1 cut}.
	\subsubsection{Two-Cut Equilibrium Measure}\label{Sec two cut eq meas}
	Consider now a regular equilibrium measure
	with two cuts,
	\begin{equation}\label{em38}
	J=\Ga[a_1,b_1]\cup \Ga[a_2,b_2].
	\end{equation}
	When $q=2$, formula \eqref{em25} gives that
	\begin{equation}\label{em39}
	Q(z)=\frac{1}{4}\, (z-c)^2 
	(z-a_1)(z-b_1)(z-a_2)(z-b_2).
	\end{equation}
	Since the polynomial $Q(z)$ is even, we have that $c=0$ and, in general,  we have the two cases for the end-points
	$ a_1<b_1<a_2<b_2$:
	\bge\item Either
	\begin{equation}\label{em40}
	-a_1=b_2,\quad -b_1=a_2,
	\end{equation}
	or
	\item
	\begin{equation}\label{em41}
	-a_1=b_1,\quad -a_2=b_2,
	\end{equation}
	\ene
	but we will see that the latter case is impossible, hence $Q(z)$  
	has the form
	\begin{equation}\label{em42}
	Q(z)=\frac{1}{4}\,z^2(z^2-a_2^2)(z^2-b_2^2).
	\end{equation}
	Matching this expression to  \eqref{em29}, we obtain that
	\begin{equation}\label{em43}
	z^2(z^2-a_2^2)(z^2-b_2^2)=z^6+2\sg z^4+(\sg^2-4)z^2-4(\sg+m_2),
	\end{equation}
	and equating the coefficients at $z^4$ and $z^2$ on the left and right, we obtain the system of equations,
	\begin{equation}\label{em44}
	a_2^2+b_2^2+2\sg=0,\quad (a_2^2-b_2^2)^2=16.
	\end{equation}
	Solving it, we obtain that
	\begin{equation}\label{em45}
	a_2^2=\mp 2-\sg,\quad b_2^2=\pm 2-\sg.
	\end{equation}
	For any real $\sg<-2$ we have that
	\begin{equation}\label{em46}
	a_2=\sqrt{- 2-\sg},\quad b_2=\sqrt{ 2-\sg}
	\end{equation}
	(see \cite{BI}). We will see below that the latter equations hold in the whole two-cut region $\sg\in \mathcal{O}_2$ on the complex plane. Similar to the one-cut regime, in the two-cut regime the $g$ function can also be explicitly computed. Using \eqref{em10}, \eqref{em17}, and \eqref{em42} we can write
	\begin{equation}
	g(z;\sigma) = \frac{V(z)+\ell^{(2)}_*}{2}+\frac{\eta_2(z;\sigma)}{2} , \qquad z \in \C \setminus \Ga_{\sigma}(-\infty,b_2],
	\end{equation}
	where 
	\begin{equation}\label{eta_2 def}
	\eta_2(z;\sigma) :=   -\int^z_{b_2} s\sqrt{(s^2-b_2^2)(s^2-a_2^2)} \dd s,
	\end{equation}
	in which the path of integration does not cross $\Ga_{\sigma}(-\infty,b_2]$. The latter integral can be evaluated explicitly 
	\begin{equation}
	\begin{split}
	\eta_2(z;\sigma) & =  -\frac{1}{4}(z^2 - \frac{b_2^2+a_2^2}{2})\sqrt{(z^2-b_2^2)(z^2-a_2^2)} \\ & + \frac{1}{4}\left( \frac{b_2^2-a_2^2}{2} \right)^2 \log \left[  \frac{2z^2- b_2^2-a_2^2 + 2\sqrt{(z^2-b_2^2)(z^2-a_2^2)}}{b_2^2-a_2^2} \right].
	\end{split}
	\end{equation}
	In view of \eqref{em46} this can be simplified as 
	\begin{equation}
	\begin{split}
	\eta_2(z;\sigma)& =  -\frac{1}{4}(z^2 +\sigma)\sqrt{(z^2+\sigma-2)(z^2+\sigma+2)}  \\ & +  \log \left[  \frac{z^2 +\sigma + \sqrt{(z^2+\sigma-2)(z^2+\sigma+2)}}{2} \right].
	\end{split}
	\end{equation}
	So we have the following explicit form of the $g$-function in the two-cut regime 
	\begin{equation}
	\begin{split}
	g(z;\sigma) & = \frac{1}{2}\left( \frac{\sigma z^2}{2} + \frac{z^4}{4} + \ell^{(2)}_* \right) -\frac{1}{8}(z^2 +\sigma)\sqrt{(z^2+\sigma-2)(z^2+\sigma+2)} \\ & + \frac{1}{2}\log \left[  \frac{z^2 +\sigma + \sqrt{(z^2+\sigma-2)(z^2+\sigma+2)}}{2} \right].
	\end{split}
	\end{equation}
	The constant $\ell_*^{(2)}$ can be found using the requirement that 
	\[
	g(z;\sigma)= \log z +\mathcal O(z^{-1})
	\]
	as $z \to \infty$. In this way we obtain that
	
	\begin{equation}\label{l2 star}
	\ell^{(2)}_*=\frac{\sigma^2}{4}-\frac{1}{2}.
	\end{equation}

	Recalling \eqref{g+mg-} and \eqref{g+pg-mVml} and straightforward calculations we have

	\begin{equation}
	\mathscr{G}^{(2)}_1(z;\sigma) =  \begin{cases}
	0, & z \in \Ga_{\sigma}[b_2,\infty), \\
	\eta_{2,+}(z;\sigma), & z \in \Ga_{\sigma}[a_2,b_2], \\
	\pi \ii, & z \in \Ga_{\sigma}[-a_2,a_2],  \\
	\eta_{2,+}(z;\sigma), & z \in \Ga_{\sigma}[-b_2,-a_2], \\
	2\pi\ii, & z \in \Ga_{\sigma}(-\infty,-b_2],
	\end{cases}
	\end{equation}
	and
	\begin{equation}\label{g+ plus g- 2-cut}
	\mathscr{G}^{(2)}_2(z;\sigma) = \begin{cases}
	\eta_{2}(z;\sigma), & z \in \Ga_{\sigma}[b_2,\infty), \\
	0 & z \in \Ga_{\sigma}[a_2,b_2], \\
	\eta_{2,\pm}(z;\sigma)\mp \pi \ii, & z \in \Ga_{\sigma}[-a_2,a_2], \\
	0, & z \in \Ga_{\sigma}[-b_2,-a_2], \\
	\eta_{2,\pm}(z;\sigma)\mp 2\pi \ii, & z \in \Ga_{\sigma}(-\infty,-b_2], \\
	\end{cases}
	\end{equation}	
	
	where we have used the fact that \[\int_{-a_2}^{a_2}s\sqrt{(s^2-b_2^2)(s^2-a_2^2)}=0.\]

	Notice that on the support we have
	\begin{equation}
	g_{+}(z;\sigma)+ g_{-}(z;\sigma)-V(z;\sigma) - \ell^{(2)}_*= 0.
	\end{equation}
	Taking the real part of this equation and comparing with \eqref{em7} and \eqref{em8} we find that the two-cut Euler-Lagrange constant $\ell$ is given by
	\begin{equation}
	\ell\equiv\ell^{(2)}=-\frac{\Re \ell^{(2)}_*}{2},
	\end{equation} where $\ell_*^{(2)}$ is given by \eqref{l2 star}.
	\subsubsection{Three-Cut Equilibrium Measure}
	Consider now  a regular equilibrium measure with three cuts, when
	\begin{equation}\label{em47}
	J=\Ga[a_1,b_1]\cup \Ga[a_2,b_2]\cup \Ga[a_3,b_3],\qquad a_1<b_1<a_2<b_2<a_3<b_3.
	\end{equation}
	In this case formula \eqref{em25} gives that
	\begin{equation}\label{em48}
	Q(z)=\frac{1}{4}R(z)=\frac{1}{4}\, 
	(z-a_1)(z-b_1)(z-a_2)(z-b_2)(z-a_3)(z-b_3).
	\end{equation}
	The evenness of $Q(z)$ implies that
	\begin{equation}\label{em49}
	-a_1=b_3\equiv c_3,\quad -b_1=a_3 \equiv b_3, \quad -a_2=b_2\equiv a_3,
	\end{equation}
	hence
	\begin{equation}\label{em50}
	Q(z)=\frac{1}{4}\,
	(z^2-a_3^2)(z^2-b_3^2)(z^2-c_3^2).
	\end{equation}
	Matching this equation to \eqref{em29}, we obtain that
	\begin{equation}\label{em51}
	(z^2-b_2^2)(z^2-a_3^2)(z^2-b_3^2)=z^6+2\sg z^4+(\sg^2-4)z^2-4(\sg+m_2),
	\end{equation}
	and equating the coefficients at $z^4$ and $z^2$, we obtain 
	the system of two algebraic equations with three unknowns,
	\begin{equation}\label{em52}
	\bga
	& a_3^2+b_3^2+c_3^2+2\sg=0,\\ 
	& a_3^4+b_3^4+c_3^4-2a_3^2b_3^2-2b_3^2c_3^2-2a_3^2c_3^2=16.
	\ena
	\end{equation}
	The above equations provide four real conditions to determine the six real unknowns $\Re a_3, \Im a_3, \Re b_3, \Im b_3, \Re c_3, \Im c_3$. Below we justify that the remaining two real conditions for determining the end points are given by

	\noindent\begin{minipage}{.45\linewidth}
		\begin{equation}\label{3cut gap condition}
		\Re \left( \di \int^{b_3}_{a_3} \sqrt{R(s)}  \dd s \right)=0,
		\end{equation}
	\end{minipage}	
	\begin{minipage}{.45\linewidth}
		\begin{equation}\label{b3 on the same level set as c3 21}
		\Re \left( \di \int^{c_3}_{b_3} \left( \sqrt{R(s)} \right)_+  \dd s \right)=0.
		\end{equation}
	\end{minipage}
	
	To that end, we use \eqref{em10}, \eqref{em17}, and \eqref{em50} to write the $g$-function as
	\begin{equation}
	g(z;\sigma) = \frac{V(z)+\ell^{(3)}_*}{2}+\frac{\eta_3(z;\sigma)}{2} , \qquad z \in \C \setminus \Ga_{\sigma}(-\infty,c_3],
	\end{equation}
	where 
	\begin{equation}\label{eta_3 def}
	\eta_3(z;\sigma) :=  -\int^z_{c_3} \sqrt{R(s)} \dd s =   -\int^z_{c_3} \sqrt{(s^2-a_3^2)(s^2-b_3^2)(s^2-c_3^2)} \dd s,
	\end{equation}
	in which the path of integration does not cross $\Ga_{\sigma}(-\infty,c_3]$. On the support we have
	\begin{equation}
	g_{+}(z;\sigma)+ g_{-}(z;\sigma)-V(z) -  \ell^{(3)}_*= \begin{cases}
	0, & z \in \Ga_{\sigma}[b_3,c_3] \\
	\di \int^{b_3}_{a_3} \sqrt{R(s)}  \dd s, & z \in \Ga_{\sigma}[-a_3,a_3] \\
	\di \int^{b_3}_{a_3} \sqrt{R(s)}  \dd s + \di \int^{-a_3}_{-b_3} \sqrt{R(s)} \dd s, & z \in \Ga_{\sigma}[-c_3,-b_3] \\
	\end{cases}
	\end{equation}
	Taking the real part of this equation and comparing with \eqref{em7} and \eqref{em8} yields
	\begin{equation}
	\ell\equiv\ell^{(3)}=-\frac{\Re \ell^{(3)}_*}{2},
	\end{equation}
	\begin{equation}\label{gap condition three cut 1}
	\Re \left( \di \int^{b_3}_{a_3} \sqrt{R(s)}  \dd s \right)=0,  \qandq 	\Re \left( \di \int^{-a_3}_{-b_3} \sqrt{R(s)}  \dd s \right)=0.
	\end{equation}
	Equations in \eqref{gap condition three cut 1} are the three-cut \textit{gap conditions}. Note that, due to the symmetry of $R$, if one of the above gap conditions hold, the other one holds automatically as well, so the requirement \eqref{3cut gap condition} is justified. Since the equilibrium measure \eqref{em13} is positive along the support, we have an immediate justification of the requirement \eqref{b3 on the same level set as c3 21}.
	
	\section{Critical graphs in the $z$-plane}\label{Sec crit graphs z plane}
	
	This section is devoted to characterization of the boundaries between the one-cut, two-cut and the three-cut regimes in the $\sigma$-plane using the theory of quadratic differentials. 
	
	Here we briefly recall some definitions and basic facts about quadratic differentials from \cite{Strebel}.  The critical points of a quadratic differential $Q(z)\dd z^2$ are the zeroes and poles of $Q(z)$, while all other points are called regular points of $Q(z)\dd z^2$. For any fixed $0 \leq \theta < 2\pi$ the $\theta$-arc of a quadratic differential $Q(z)\dd z^2$ is defined as the smooth curve $L_{\theta}$ along which \begin{equation}
	\arg Q(z)\dd z^2 = \theta, 
	\end{equation}
	and thus a $\theta$-arc can only contain regular points of $Q$, because at the singular points the argument is not defined. Through each regular point of a meromorphic quadratic differential passes exactly one $\theta$-arc. A maximal $\theta$-arc is called a $\theta$-trajectory. We will refer to a $\pi$-trajectory ( resp. $0$-trajectory) which is incident with a critical point as a \textit{critical trajectory} (resp. \textit{critical orthogonal trajectory}). If $b$ is a critical point of $Q(z)\dd z^2$, then the totality of the solutions to 
	\begin{equation} 
	\Re\left(\int^z_b \sqrt{Q(s)} \dd s \right) = 0,
	\end{equation}
	is referred to as the \textit{critical graph} of $\int^z_b \sqrt{Q(s)} \dd s$ (see \S 5 of \cite{Strebel}).  A critical (orthogonal) trajectory is called \textit{short} if it is incident only with finite critical points. A simple closed  \textit{geodesic polygon} with respect to a meromorphic quadratic differential $Q(z)\dd z^2$ (also referred to as a $Q$-polygon) is a Jordan curve $\Sigma$ composed of open $\theta$-arcs and their endpoints. The endpoints may be regular or critical points of $Q(z)\dd z^2$, which form the vertices of the $Q$-polygon. By a \textit{loop} we mean a geodesic polygon whose single vertex is a singular point of the associated quadratic differential. If at least one of the end points of $\Sigma$ is a singular point, we call it a \textit{singular geodesic polygon}.  Let $\Sigma$ by a $Q$-polygon, and let $\texttt{V}_{\Sigma}$ and $\texttt{Int} \Sigma$ denote its set of vertices and interior respectively. The \textit{Teichm\"uller's lemma} states that
	\begin{equation}\label{Teich Lemma}
	\#\texttt{V}_{\Sigma} - 2 = \sum_{z \in \texttt{V}_{\Sigma}} (\texttt{ord}(z)+2)\frac{\theta(z)}{2\pi} + \sum_{z \in \texttt{Int} \Sigma} \texttt{ord}(z),
	\end{equation}
	where $\theta(z)$ is the interior angle of $\Sigma$ at $z$, and $\texttt{ord}(z)$ is the order of the point $z$ with respect to the quadratic differential: it is zero for a regular point, it is $n$ ($-n$) if $z$ is a zero (pole) of order $n \in \N$ of the quadratic differential.

	\subsection{The One-cut Regime}
	
	Let us recall the definition of the function $\eta$ introduced in \S \ref{SubSec EqMeas quartic}:
	
	\begin{equation}\label{eta def1}
	\eta_1(z;\sigma):=-\int^{z}_{b_1(\sigma)} \left(s^2 - z^2_0(\sigma)\right) \sqrt{s^2- b^2_1(\sigma)} \dd s.
	\end{equation}
	We sometimes need to choose the starting point of integration to be $\pm z_0(\sigma)$, so $\eta$ as defined above may be denoted by $\eta_{b_1}$, and $\eta_{\pm z_0}$ denotes the right hand side of \eqref{eta def1} when the starting point of integration is replaced by $\pm z_0(\sigma)$ (for example see the caption of Figure \ref{fig:Conf map 1}).
	
	\begin{definition}\label{Def one cut sigma}
		\normalfont The one-cut regime $\mathcal{O}_{1}$ in the $\sigma$-plane is defined as the collection of all $\sigma \in \C$ such that
		\begin{enumerate}
			\item The critical graph $\mathscr{J}^{(1)}_{\sigma}$ of all points $z$ satisfying \begin{equation}
			\Re \left[ \eta_1(z;\sigma) \right]=0,
			\end{equation}
			contains a single Jordan arc $J_{\sigma}$ connecting $-b_1(\sigma)$ to $b_1(\sigma)$,
			\item The points $\pm z_0(\sigma)$ do not lie on $J_{\sigma}$, and
			\item There exists a complementary arc $\Ga_{\sigma}(b_1(\sigma), \infty)$ which lies entirely in \begin{equation}
			\left\{ z : \Re \left[ \eta_1(z;\sigma) \right]<0 \right\},
			\end{equation} which encompasses $(M(\sigma),\infty)$ for some $M(\sigma)>0$.
		\end{enumerate}
	\end{definition}
	
	For a fixed $\sigma$ we refer to the collection of all $z$ satisfying $  \Re[\eta_1(z;\sigma)]<0$ as the $\sigma$-\textit{stable lands} , and to the collection of all $z$ satisfying $  \Re[\eta_1(z;\sigma)]>0$ as the $\sigma$-\textit{unstable lands} (see Figure \ref{fig:stable and barren}, and the third component of Definition \ref{Def one cut sigma}).
	
	\begin{remark}
		\normalfont	 For almost all choices of branch cuts $L_{\pm}$, and $L$ (recall Remark \ref{Remark branches}) there are certain choices of $\sigma$ for which one of the components of definition \ref{Def one cut sigma} does not hold. For example, for the choice of $L_{\pm}$, and $L$ mentioned in Remark \ref{Remark branches} we give the following three examples:
		\begin{itemize}
			\item For  $\sigma=-1+1.9\ii$ the first component of Definition \ref{Def one cut sigma} does not hold, as shown in Figure \ref{fig:No Connection},
			\item For $\sigma \simeq -1+1.7795\ii$ the second component of Definition \ref{Def one cut sigma} does not hold, as shown in Figure \ref{fig:Connection about to be lost},		
			\item For $\sigma = -1.35+4\ii$ the third component of Definition \ref{Def one cut sigma} does not hold, as shown in Figure \ref{fig:stable barren -1.35 41}.
		\end{itemize}  
	\end{remark}

	\begin{lemma}\label{Lemma symmetry of J}
		The set $\mathscr{J}^{(1)}_{\sigma}$ is symmetric with respect to the origin.
	\end{lemma}
	\begin{proof}
		In view of the first part of Definition \ref{Def one cut sigma}, this simply follows from the identity
		\begin{equation}\label{eta z and eta -z}
		\eta_1(-z;\sigma)=\eta_1(z;\sigma) \pm 2\pi \ii.
		\end{equation}
	\end{proof}
	The following Lemma and Lemma \ref{continuous deformations of two cut critical graph} are particular cases of the more general Theorem which states that for a general polynomial potential of degree $p$, each one of the $q$-cut critical graphs,  $1\leq q \leq 2p-1$, deforms continuously with respect to the parameters in the potential (see Theorem 3 of \cite{BBGMT2}).   
	\begin{lemma}\label{continuous deformations of one cut critical graph}\normalfont{[Theorem 3 of \cite{BBGMT2}]}
		\textit{		The critical graph $\mathscr{J}^{(1)}_{\sigma}$ deforms continuously with respect to $\sigma$.}
	\end{lemma}
	
	Consider the one-cut quadratic differential
	\begin{equation}\label{OneCut QD}
	Q_1(z;\sigma) \dd z^2 := \left(z^2-z^2_0(\sigma)\right)^2\left(z^2-b^2_1(\sigma)\right) \dd z^2.
	\end{equation}
	By Theorem $7.1$ of \cite{Strebel}, if all four singular points $\pm b_1$ and $\pm z_0$ are distinct, there are three  $\theta$-trajectories, $0\leq \theta <2\pi$, emanating from $z=b_1(\sigma)$ and $z=-b_1(\sigma)$ each, while there are four  $\theta$-trajectories emanating from $z=z_0(\sigma)$ and $z=-z_0(\sigma)$. Two adjacent $\theta$-trajectories make an angle of $2\pi/3$ when they arrive at $z=\pm b_1(\sigma)$, while two adjacent $\theta$-trajectories make an angle of $\pi/2$ when they arrive at $z=\pm z_0(\sigma)$ (See Figure \ref{fig:Conf map 1} and its caption). The representation of the quadratic differential $Q_1(z;\sigma) \dd z^2,$ 
	near $z=\infty$ is 
	\begin{equation}\label{QD near infinity}
	\frac{1}{\Tilde{z}^4}Q_1\left(\frac{1}{\Tilde{z}};\sigma\right) \dd \Tilde{z}^2
	\end{equation}
	for $\Tilde{z}$ near zero. Therefore $z=\infty$ is a pole of order $10$ for the quadratic differential $Q_1(z;\sigma) \dd z^2$. According to Theorem $7.4$ of \cite{Strebel}, for each $0\leq \theta<2\pi$, there are $8$ directions along which $\theta$-trajectories approach $\infty$. More precisely, notice that near infinity $Q_1(z;\sigma) \dd z^2 \sim z^6 \dd z^2$, thus 
	\begin{equation}
	\eta_1(z;\sigma) =  - \int^z_{b_1} \sqrt{Q_1(s;\sigma)} \dd s \sim  -\frac{z^4}{4}, \qquad z \to \infty.
	\end{equation}
	Therefore the critical trajectories (solutions to $\Re[\eta_1(z;\sigma)]=0$) approach to infinity along the directions $\di \frac{\pi}{8}+\di \frac{k\pi}{4}$, $k=0,\cdots,7$, and orthogonal trajectories (solutions to $\Im[\eta_1(z;\sigma)]=0$) approach to infinity along the directions $\di \frac{k\pi}{4}$, $k=0,\cdots,7$.

	\begin{lemma}\label{Lemma one or two vetex polygons are ruled out}
		There are no singular finite geodesic polygons with one or two vertices associated to the quadratic differential \eqref{OneCut QD}.
	\end{lemma}
	\begin{proof}
		Suppose that such a singular finite $Q_1$-polygon exists. For this geodesic polygon, the left hand side of \eqref{Teich Lemma} is either $-1$ or zero, while the right hand side of \eqref{Teich Lemma} is certainly a positive integer. This is because such a polygon can not enclose a pole as the quadratic differential \eqref{OneCut QD} has no finite poles, and because $\texttt{ord}(\pm b_1)=1$, $\texttt{ord}(\pm z_0)=2$, $\theta(\pm b_1)\in \{ \frac{2 \pi}{3}, \frac{4 \pi}{3} \}$ and $\theta(\pm z_0)\in \{ \frac{\pi}{2}, \pi,  \frac{3 \pi}{2} \}$ and the more singular points $\mathscr{L}$ encloses, the larger the right hand side gets. Therefore \eqref{Teich Lemma} can not hold for such a polygon and this finishes the proof.
	\end{proof}

	\begin{definition}
		\normalfont	If all four singular points $\pm b_1$ and $\pm z_0$ are distinct, We denote the local critical arcs incident to $\pm b_1(\sigma)$ by $\mathcal{l}^{(\pm b_1(\sigma))}_{1}$, $\mathcal{l}^{(\pm b_1(\sigma))}_{2}$,  and $\mathcal{l}^{(\pm b_1(\sigma))}_{3}$ (labeled in counterclockwise direction), where  $\mathcal{l}^{(b_1(\sigma))}_{1}$ and $\mathcal{l}^{(-b_1(\sigma))}_{1}$ are the ones which are part of $J_{\sigma}$ (see Definition \ref{Def one cut sigma}). 
	\end{definition}
	
	In what follows in the paper, sometimes we also use the same notations for the critical trajectories incident with $\pm b_1(\sigma)$. We also usually suppress the dependence on $\sigma$ for these objects when it causes no confusion. Notice that Lemma \ref{Lemma one or two vetex polygons are ruled out} implies that the critical arcs $\mathcal{l}^{(b_1)}_{2}$ or $\mathcal{l}^{(b_1)}_{3}$ can not be connected to either $\mathcal{l}^{(-b_1)}_{2}$ or $\mathcal{l}^{(-b_1)}_{3}$.

	
	\color{black}
	
	\begin{lemma}\label{Lemma extension to infinity of critical trajectories}
		Let $\ga \in \left\{ \ell_2^{(b_1)}, \ell_3^{(b_1)}, \ell_2^{(-b_1)}, \ell_3^{(-b_1)}\right\}$. If neither $z_0(\sigma)$ or $-z_0(\sigma)$ lie on $\ga$, then $\ga$ must extend off to infinity.
	\end{lemma}
	\begin{proof}
		First assume that $z_0(\sigma)$ (and thus $-z_0(\sigma)$ due to Lemma \ref{Lemma symmetry of J}) does not lie on $\mathscr{J}^{(1)}_{\sigma}$. This means that no critical arcs emanating from $\pm z_0(\sigma)$ can be connected to the critical arcs emanating from $\pm b_1(\sigma)$. Now, due to Lemma \ref{Lemma one or two vetex polygons are ruled out} the only possibility left for the critical arcs $\mathcal{l}^{(\pm b_1)}_{2}$ and $\mathcal{l}^{(\pm b_1)}_{3}$ is that all of them must extend to infinity. Now, if $z_0(\sigma)$ does lie on $\mathscr{J}^{(1)}_{\sigma}$, then $-z_0(\sigma)$ also lies on $\mathscr{J}^{(1)}_{\sigma}$. Therefore one critical arc emanating from $z_0(\sigma)$ must be connected to one of the critical arcs emanating from $\pm b_1(\sigma)$, and one critical arc emanating from $-z_0(\sigma)$ must be connected to one of the critical arcs emanating from $\pm b_1(\sigma)$, these two connections mean that the only other possibility for the other two critical arcs (which $\pm z_0(\sigma)$ do not hit) is to extend to infinity.
	\end{proof}

	\begin{figure}[!htbp]
		\centering
		\begin{subfigure}{0.44\textwidth}
			\centering
			\includegraphics[width=\textwidth]{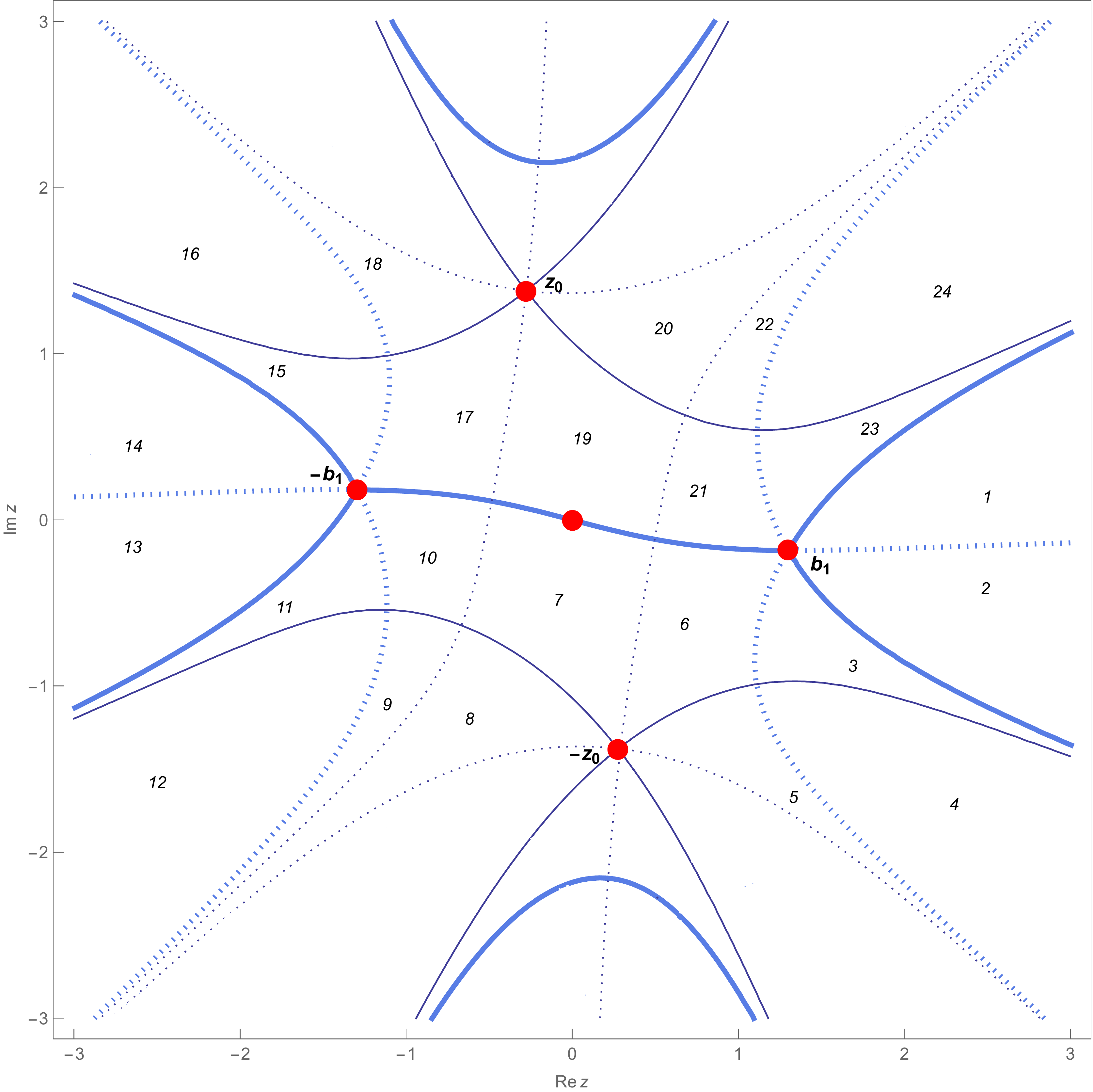}
			\caption{The critical and critical orthogonal trajectories of the one-cut quadratic differential $Q_1(z;\sigma) \dd z^2$ incident with $\pm b_1$ and $\pm z_0$. Thick lines are associated with $\Re[\eta_{1}(z;\sigma)]=0$ (solid) and $\Im[\eta_{1}(z;\sigma)]=0$ (dotted), while the thin lines are associated with $\Re[\eta_{1,z_0}(z;\sigma)]=0$ (solid) and $\Im[\eta_{1,z_0}(z;\sigma)]=0$ (dotted), where we remind that $\eta_{1,e}(z;\sigma)=\int^z_{e}\sqrt{Q_1(s;\sigma)}\dd s$, and $\eta_{1}(z;\sigma)\equiv \eta_{1,b_1}(z;\sigma)$. The red dot which is not labeled represents the origin. We have only shown the "humps" $\mathcal{L}^{(1)}_{\sigma}$ and $\mathcal{L}^{(2)}_{\sigma}$ associated with $\Re[\eta_{1,b_1}(z;\sigma)]=0$ (See Lemma \ref{lemma humps}), and have not shown the humps associated with $\Im[\eta_{1,b_1}(z;\sigma)]=0$, $\Re[\eta_{1,z_0}(z;\sigma)]=0$, and $\Im[\eta_{1,z_0}(z;\sigma)]=0$ for simplicity of the Figure.  Notice that $\eta_1(z;\sigma) \sim (z \mp b_1)^{3/2}$ as $z \to \pm b_1$, while $\eta_1(z;\sigma) \sim (z \mp z_0)^{2}$ as $z \to \pm z_0$, which determines the number of critical (critical orthogonal) trajectories incident with $\pm b_1$ and $\pm z_0$. The critical trajectories approach to infinity along the eight directions $\pi/8 + k \pi/4$, and orthogonal trajectories approach to infinity along the eight directions $k \pi/4$, $k=0,\cdots,7$. }
			\label{fig:Conf map 1}
		\end{subfigure} \hspace{0.09\textwidth}
		\begin{subfigure}{.44\textwidth}
			\centering
			\includegraphics[width=\textwidth]{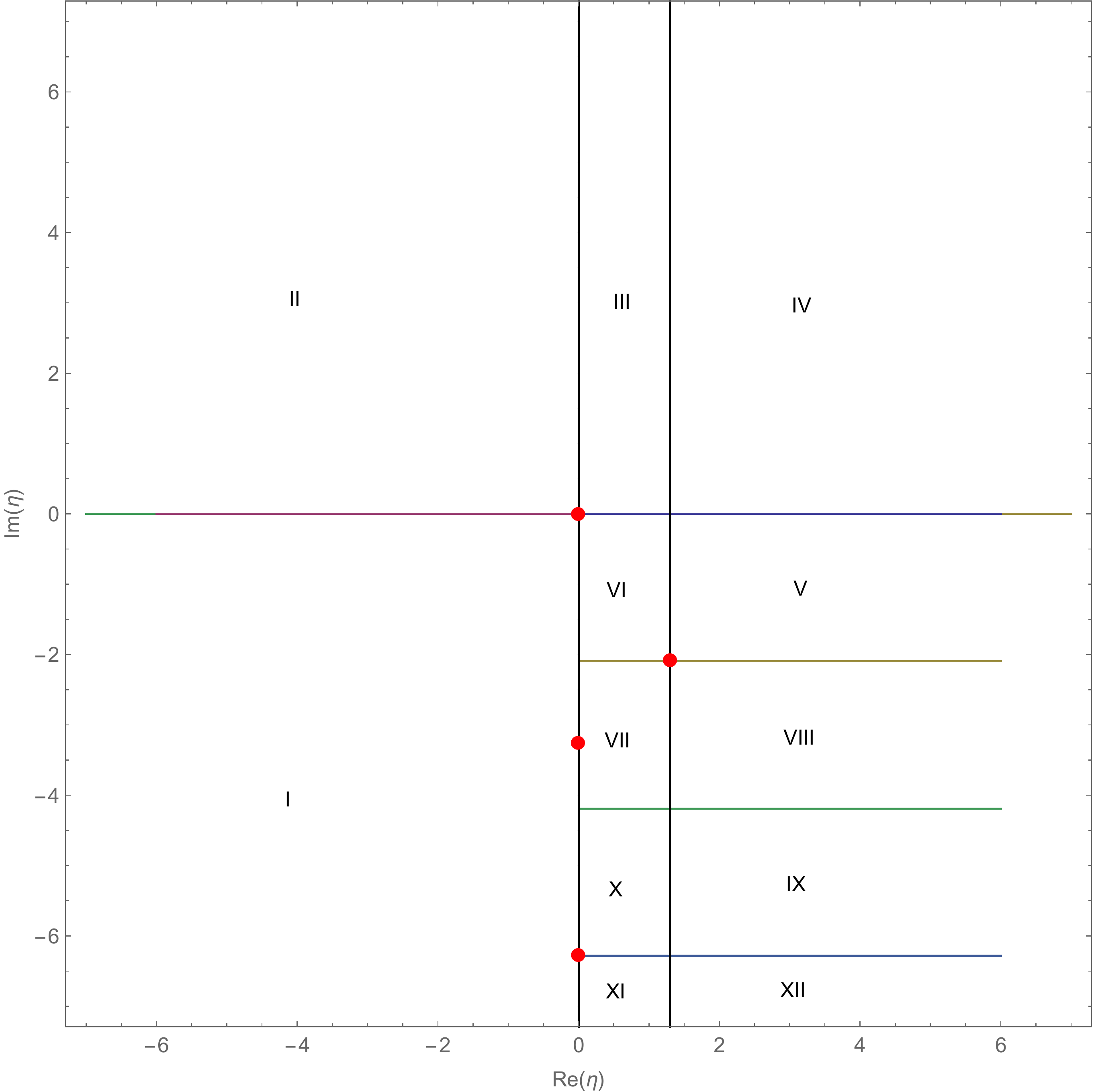}
			\caption{The image $\eta_1(\Om;\sigma)$, where $\Om$ is the union of the regions labeled by $1$ through $12$ in Figure \ref{fig:Conf map 1}. Each strip-like region labeled by a Roman numeral corresponds to the region in Figure \ref{fig:Conf map 1} labeled by the same number in Arabic numerals; for example the image of the region labeled by $6$ in Figure \ref{fig:Conf map 1} is the region labeled by VI. Notice that $\eta_1(b_1;\sigma)=0$, $\eta_{1,-}(0;\sigma)=-\pi \ii$, and $\eta_{1,-}(-b_1;\sigma)=-2\pi \ii$, while the red point in the fourth quadrant represents $\eta_1(-z_0;\sigma)$. The map $\eta$ is clearly not conformal at  $\pm b_1$ and $\pm z_0$. For example a neighborhood of $b_1$ intersected with the regions labeled by 1, 2, 3, and 6 gets mapped to a full neighborhood of the origin in the $\eta$-plane due to $\eta_1(z;\sigma) \sim (z - b_1)^{3/2}$ as $z \to b_1$, while  a neighborhood of $-z_0$ intersected with the regions labeled by 5, 6, 7, and 8 gets mapped to a full neighborhood of $\eta_1(-z_0;\sigma)$ in the $\eta$-plane due to $\eta_1(z;\sigma) \sim (z + z_0)^{2}$ as $z \to - z_0$. The image $\eta_1(\Om;\sigma)$ as depicted above shows that the regions in Figure \ref{fig:Conf map 1} labeled by 1 and 2 are stable lands, meaning that they can host the complementary contours $\Ga_{\sigma}[b_1, \infty]$ and $\Ga_{\sigma}[ -\infty, -b_1]$, respectively.}
			\label{fig:Conf map 2}    
		\end{subfigure}
		\caption{ Demonstration of the conformal mapping between the regions labeled by 1 through 12 in the $z$-plane to the $\eta_1$-plane. $\eta_1$ also maps the regions labeled by 13 to 24 to the entire plane as well. These conformal maps illustrate that the regions labeled by 1,2,13, and 14 are stable lands as shown in Figure \ref{fig:fertile barren -1 1.6}. The stable lands in Figures \ref{fig:stable barren 1 3.8} through \ref{fig:stable barren -1.35 41} can be justified similarly.}
		\label{fig:Conf Map}
	\end{figure}
	
	\begin{definition}
		\normalfont	If $ \Re \left[ \eta_1(\pm z_0(\sigma);\sigma) \right]\neq 0$, we define $\De^{(b_1)}$ (resp. $\De^{(-b_1)}$) to be the geodesic polygon with vertices $ b_1$ (resp. $-b_1$) and $\infty$, composed of $\ell_2^{( b_1)}$ and $\ell_3^{(b_1)}$ (resp. $\ell_2^{(-b_1)}$ and $\ell_3^{(-b_1)}$) with interior angle $2\pi/3$ at $b_1$ (resp. $-b_1$).
	\end{definition}
	
	The following lemma is an immediate consequence of applying the Teichm\"uller's lemma to the polygons $\De^{(\pm b_1)}$.

	\begin{lemma}\label{lemma angles at infty}
		The critical trajectories $\ell_2^{(b_1)}$ and $\ell_3^{(b_1)}$ approach to infinity along two directions $\pi/4$ apart if the geodesic polygon $\De^{(b_1)}$ does not enclose $\pm z_0$, while they approach to infinity along two directions $3\pi/4$ apart if the geodesic polygon $\De^{(b_1)}$ encloses one of $\pm z_0$. Due to symmetry the above statement is correct if we replace $b_1$ by $-b_1$. 
	\end{lemma}

	\begin{figure}[h]
		\centering
		\begin{subfigure}{0.24\textwidth}
			\centering
			\includegraphics[width=\textwidth]{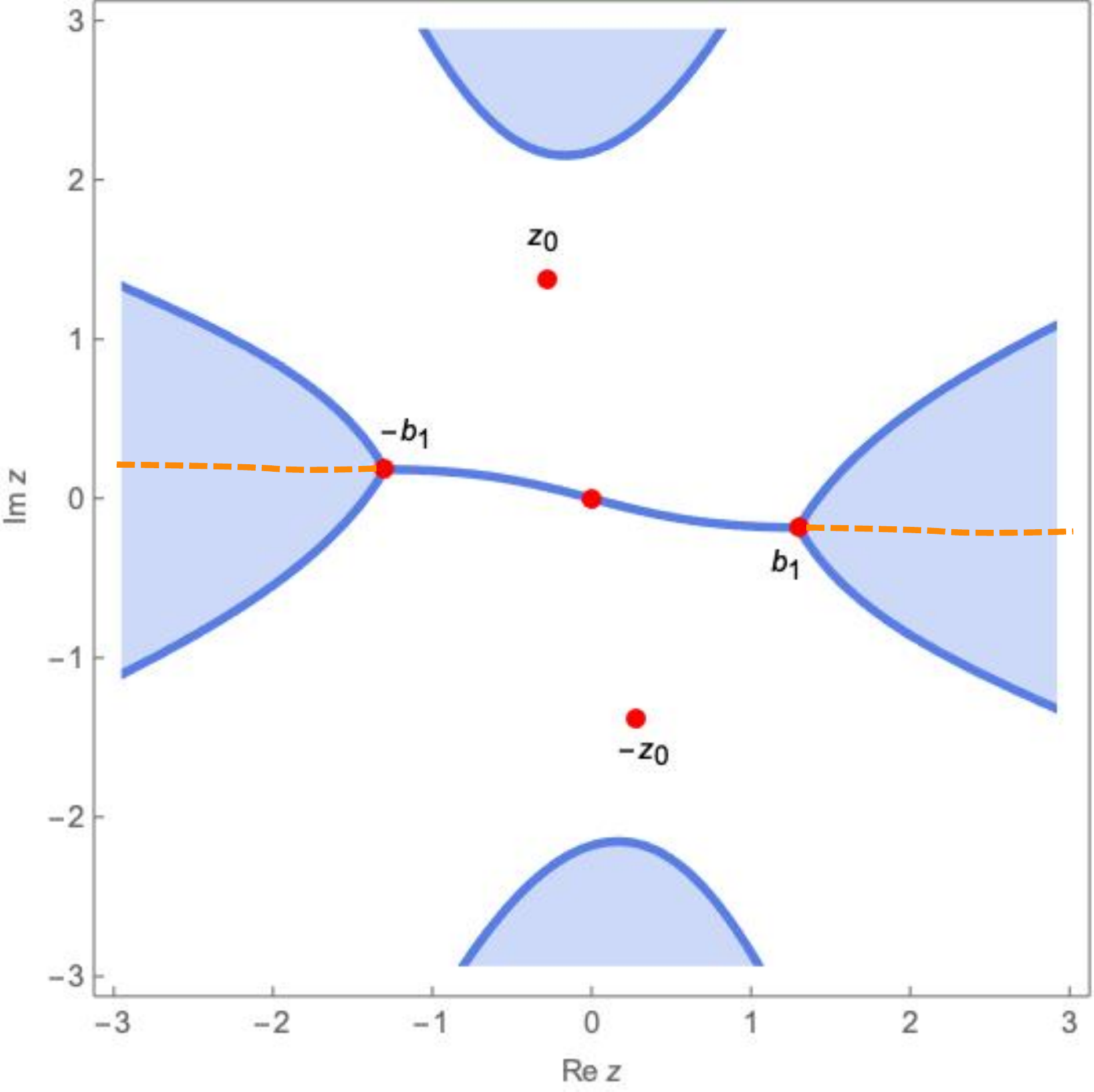}
			\caption{$\sigma=1+\ii$.}
			\label{fig:fertile barren -1 1.6}
		\end{subfigure}%
		\begin{subfigure}{.24\textwidth}
			\centering
			\includegraphics[width=\textwidth]{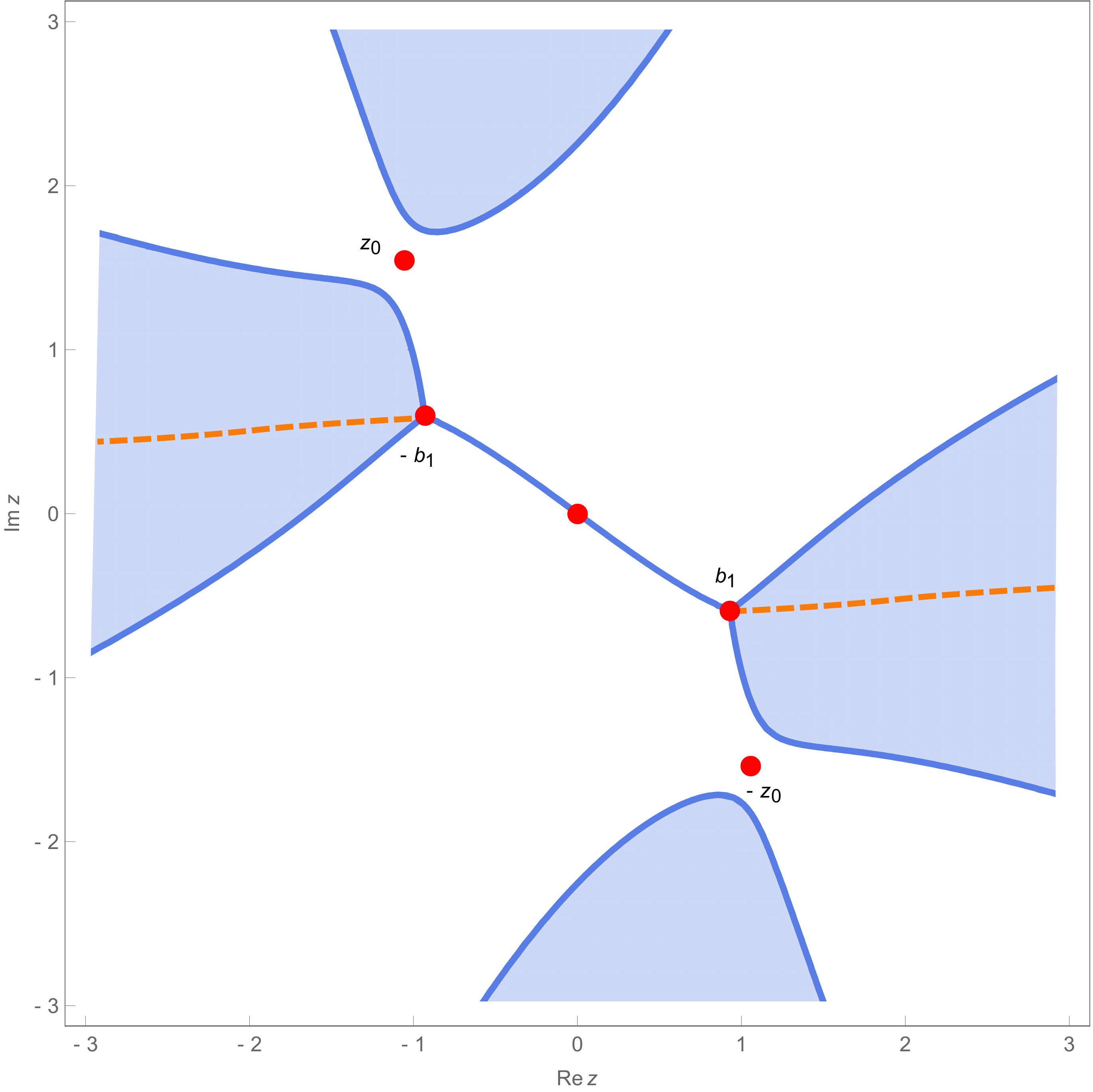}
			\caption{$\sigma=1+3.8\ii$.}
			\label{fig:stable barren 1 3.8}    
		\end{subfigure}
		\begin{subfigure}{.24\textwidth}
			\centering
			\includegraphics[width=\textwidth]{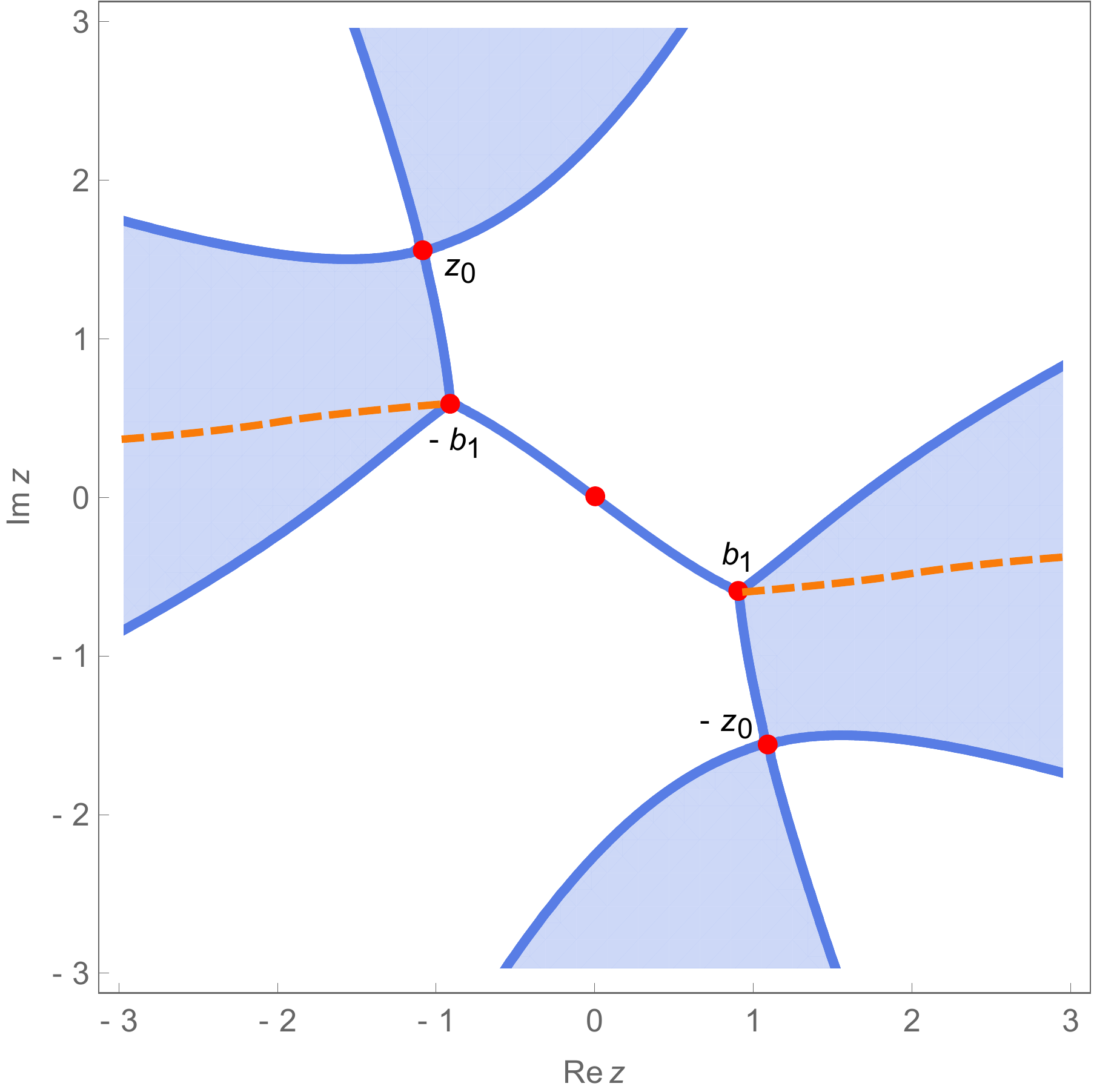}
			\caption{$\sigma_{\mbox{cr}}\simeq1+3.92\ii$.}
			\label{fig:stable barren critical}    
		\end{subfigure}
		\begin{subfigure}{.24\textwidth}
			\centering
			\includegraphics[width=\textwidth]{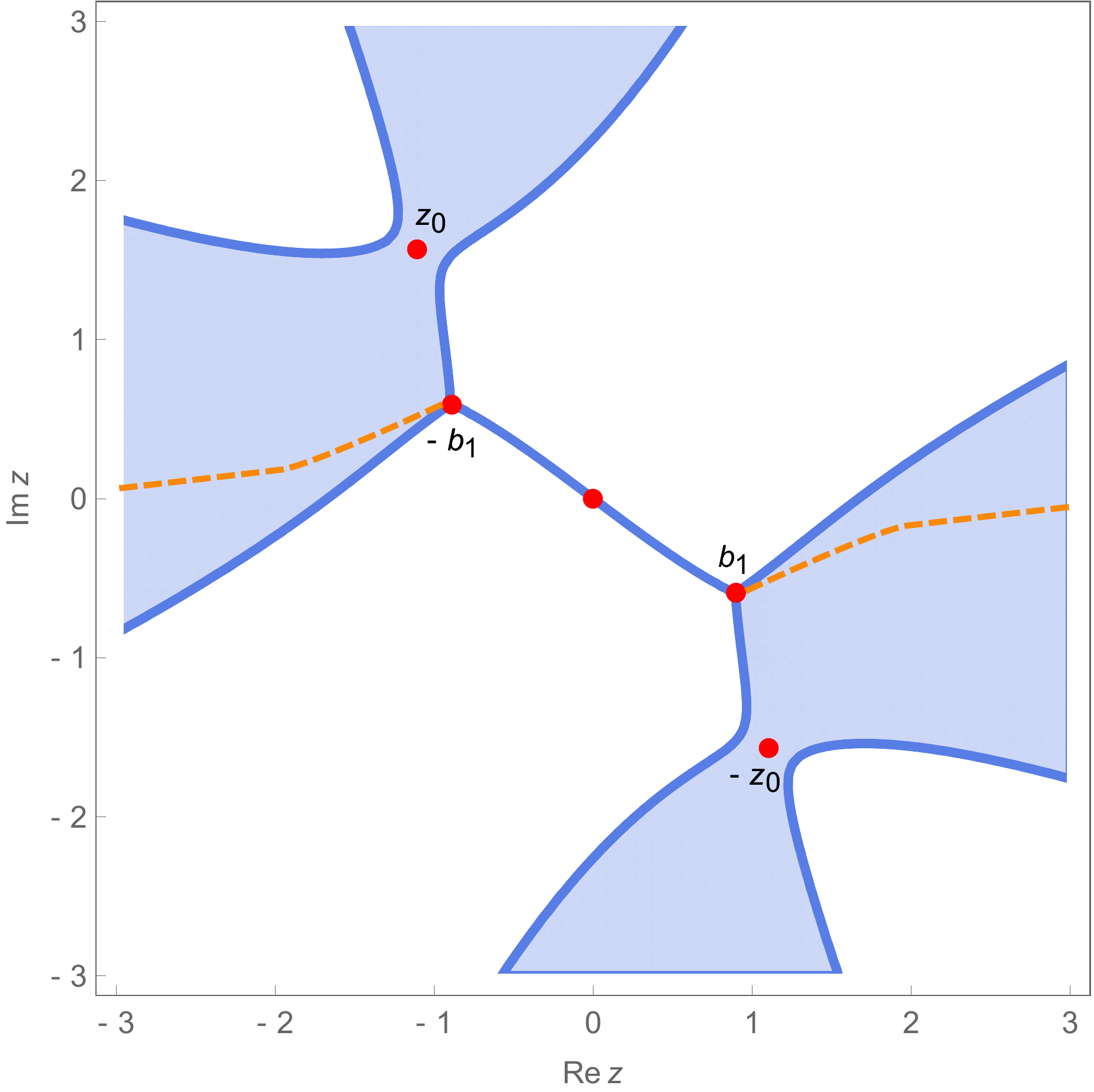}
			\caption{$\sigma=1+4\ii$.}
			\label{fig:stable barren 1 4}    
		\end{subfigure}
		
		\begin{subfigure}{0.24\textwidth}
			\centering
			\includegraphics[width=\textwidth]{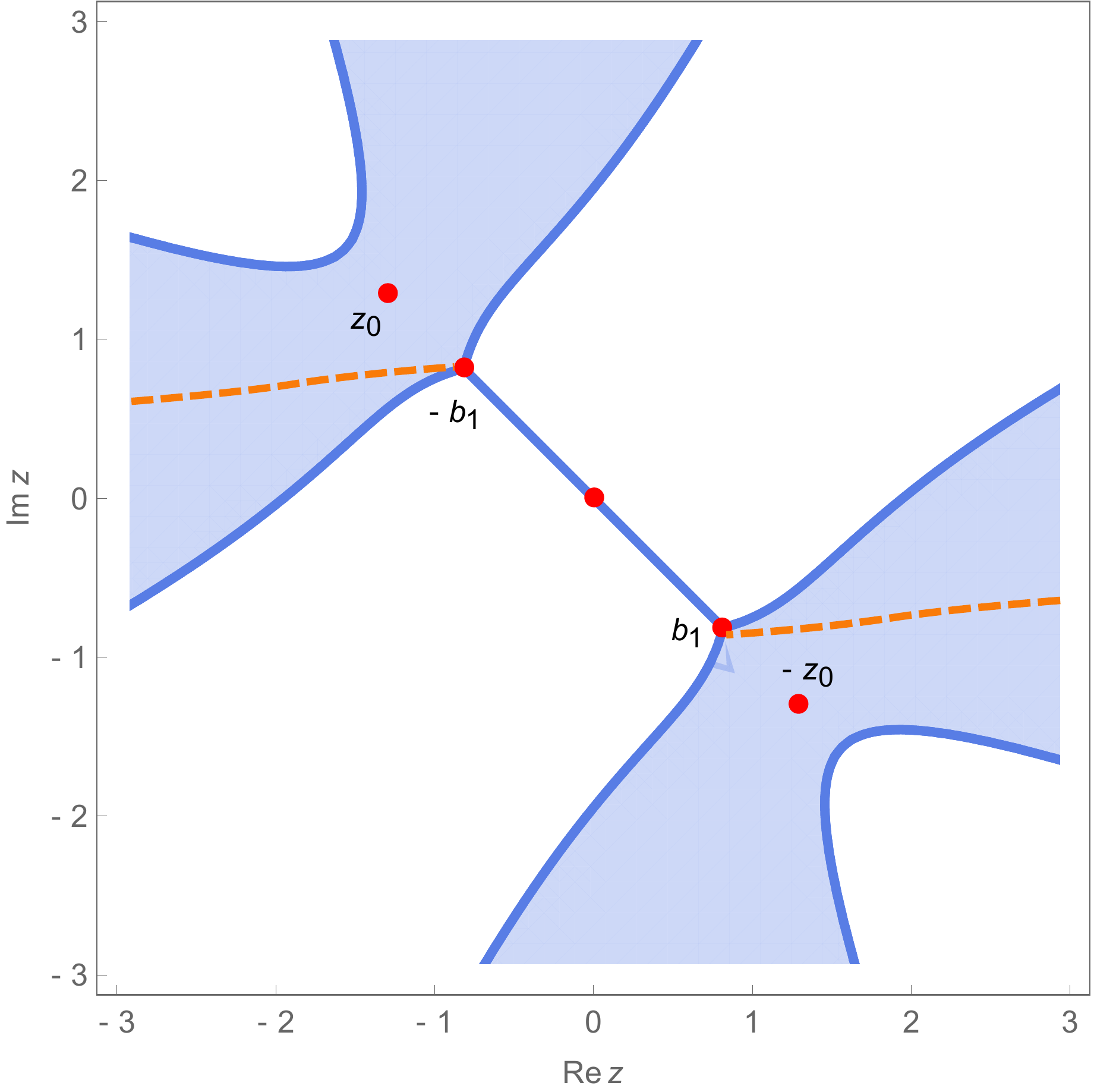}
			\caption{$\sigma=4 \ii$.}
			\label{fig:fertile barren 4}
		\end{subfigure}%
		\begin{subfigure}{.24\textwidth}
			\centering
			\includegraphics[width=\textwidth]{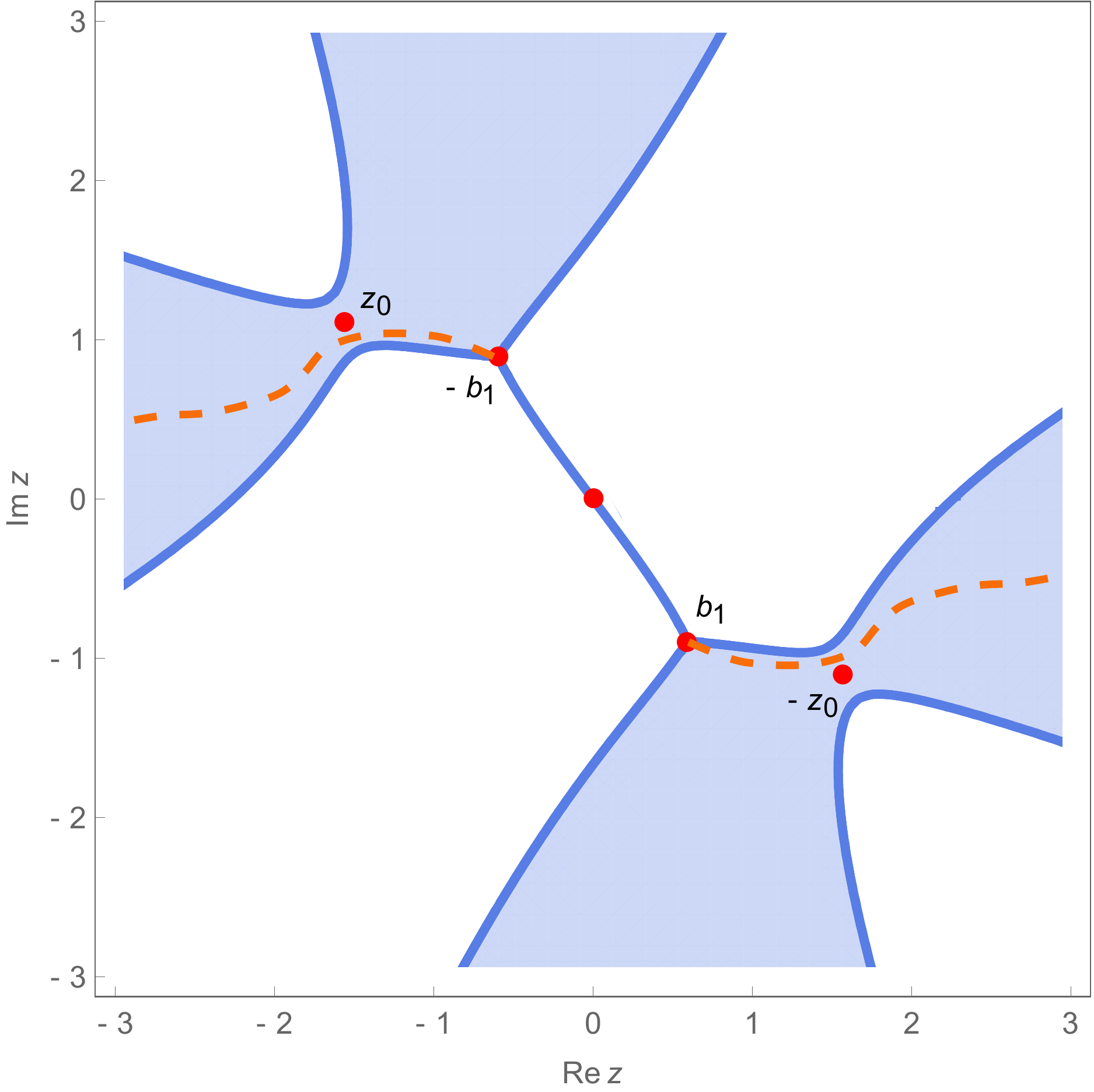}
			\caption{$\sigma=-1+4\ii$.}
			\label{fig:stable barren -1 4}    
		\end{subfigure}
		\begin{subfigure}{.24\textwidth}
			\centering
			\includegraphics[width=\textwidth]{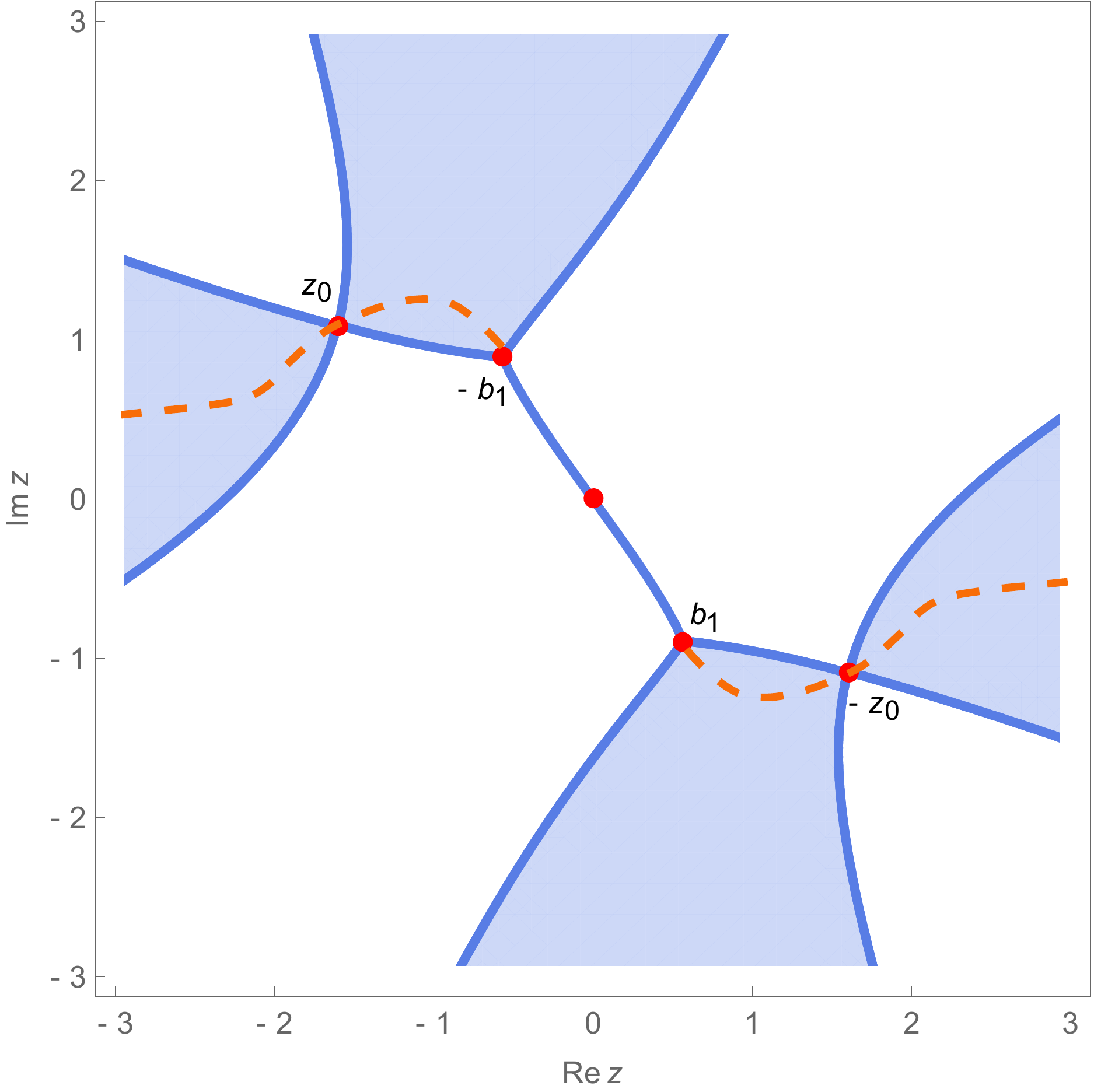}
			\caption{$\sigma_{\mbox{cr}}\simeq -1.15+4\ii$.}
			\label{fig:stable barren -1.15 4}    
		\end{subfigure}
		\begin{subfigure}{.24\textwidth}
			\centering
			\includegraphics[width=\textwidth]{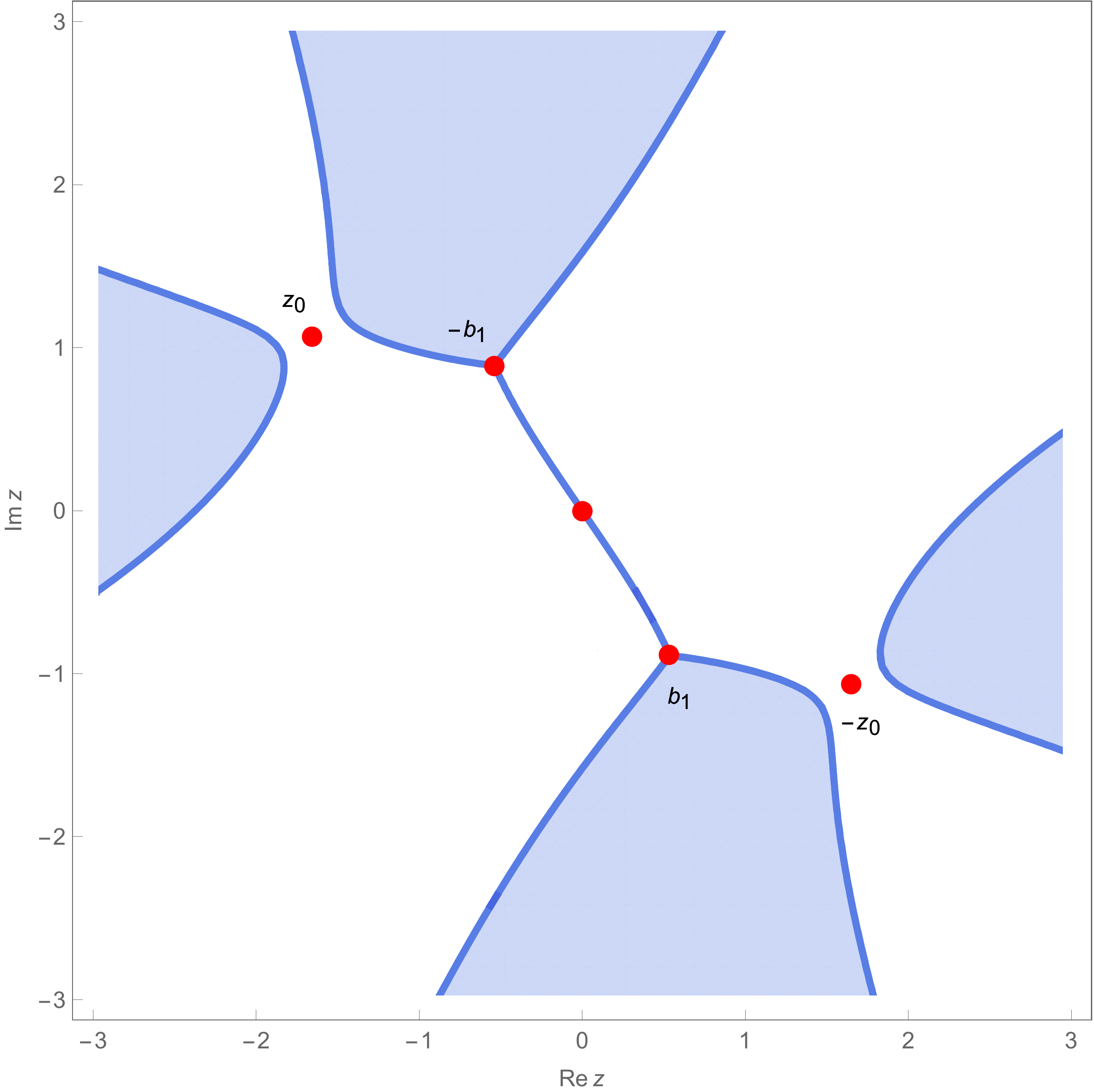}
			\caption{$\sigma=-1.35+4\ii$.}
			\label{fig:stable barren -1.35 41}    
		\end{subfigure}
		\caption{
			This sequence of figures shows allowable regions in light blue through which the contour of integration (for the orthogonal polynomials) must pass, for a varying collection of values of $\sigma$.	Regions in light blue are the $\sigma$-stable lands where $  \Re[\eta_1(z;\sigma)]<0$ and the regions in white are the $\sigma$-unstable lands where $  \Re[\eta_1(z;\sigma)]>0$. We denote the local critical arcs incident to $\pm b_1(\sigma)$ by $\mathcal{l}^{(\pm b_1)}_{1}$, $\mathcal{l}^{(\pm b_1)}_{2}$, and $\mathcal{l}^{(\pm b_1)}_{3}$ (labeled in counterclockwise direction), where  $\mathcal{l}^{(b_1)}_{1}$ and $\mathcal{l}^{(-b_1)}_{1}$ are the ones which are part of $J_{\sigma}\equiv \Ga_{\sigma}[-b_1,b_1]$. When $\pm z_0(\sigma)$ are in the $\sigma$-unstable lands, the Teichm\"uller's lemma for the geodesic polygon comprised of $\ell^{(b_1)}_2$ and $\ell^{(b_1)}_3$ with vertices at $b_1(\sigma)$ and $\infty$, necessitates that $\ell^{(b_1)}_2$ and $\ell^{(b_1)}_3$ approach infinity along two directions $\pi/4$ radians apart (see Figures (a), (b) and (h) above), while when $\pm z_0(\sigma)$ are in the $\sigma$-stable lands, the Teichm\"uller's lemma for the same geodesic polygon necessitates that $\ell^{(b_1)}_2$ and $\ell^{(b_1)}_3$ approach infinity along two directions $3\pi/4$ radians apart (see Figures (d), (e) and (f) above). In all figures above the first and the second requirements of Definition \ref{Def one cut sigma} is fulfilled. However, Figures (g) and (h) correspond to $\sigma$ values where the third requirement of Definition \ref{Def one cut sigma} is not met, while for $\sigma$ values corresponding to Figures (a) through (f), the orange dashed lines show that this requirement is fulfilled. In Figure \ref{fig:legend1} we show the location of these points with respect to the critical lines in the $\sigma$-plane.}
		\label{fig:stable and barren}
	\end{figure}

	\begin{definition}\label{Def one cut sigma minus fake transition}
		\normalfont	The subset
		$\mathcal{O}^*_{1}$ in the $\sigma$-plane is the collection of all $\sigma \in \C$ such that
		\begin{enumerate}
			\item The critical graph $\mathscr{J}^{(1)}_{\sigma}$ of all points $z$ satisfying \begin{equation}\label{level set}
			\Re \left[ \eta_1(z;\sigma) \right]=0,
			\end{equation}
			contains a single Jordan arc $J_{\sigma}$ connecting $-b_1(\sigma)$ to $b_1(\sigma)$,
			\item The points $\pm z_0(\sigma)$ do not lie on $\mathscr{J}^{(1)}_{\sigma}$, and
			\item There exists a complementary arc $\Ga_{\sigma}(b_1(\sigma), \infty)$ which lies entirely in the component of the set \begin{equation}
			\left\{ z : \Re \left[ \eta_1(z;\sigma) \right]<0 \right\},
			\end{equation} which encompasses $(M(\sigma),\infty)$ for some $M(\sigma)>0$.
		\end{enumerate}
	\end{definition}
	
	Notice that $\mathcal{O}^*_1 \subseteq \mathcal{O}_1$, since for the definition of $\mathcal{O}^*_1$ the location of $z_0$ is further restricted than what is required for the definition of $\mathcal{O}_1$. In Theorem \ref{THM Fake Transition} the significance of distinguishing these two sets will become more clear. Here in the rest of this section we will focus on proving the following Theorem:
	
	\begin{theorem}\label{O1* is open}
		The set $\mathcal{O}^*_1$ is open.
	\end{theorem}
	
	We prove this Theorem by proving several lemmas associated to the different requirements of Definition \ref{Def one cut sigma minus fake transition}. In the following three lemmas we establish some structural properties of the critical graph $\mathscr{J}^{(1)}_{\sigma}$.
	
	\begin{lemma}\label{lemma humps}
		Suppose that $\sigma \in \mathcal{O}^*_1$ and $ \Re[\eta_1(\pm z_0(\sigma);\sigma)]>0$. Then there exist two disjoint curves $\mathcal{L}^{(1)}_{\sigma}$ and $\mathcal{L}^{(2)}_{\sigma}$  as subsets of $\mathscr{J}^{(1)}_{\sigma}$, which have no intersections with $\pm b_1(\sigma)$,  $\pm z_0(\sigma)$,  $J_{\sigma}$,  $\ell_2^{(b_1)}, \ell_3^{(b_1)}, \ell_2^{(-b_1)}$, and $ \ell_3^{(-b_1)}$. Moreover, the curve $\mathcal{L}^{(1)}_{\sigma}$ approaches to infinity along the two directions $3\pi/8$ and $5\pi/8$, the curve $\mathcal{L}^{(2)}_{\sigma}$ approaches to infinity along the two directions $-3\pi/8$ and $-5\pi/8$, and the rays $\ell_2^{(b_1)}, \ell_3^{(b_1)}, \ell_2^{(-b_1)}$, and $\ell_3^{(-b_1)}$ respectively approach to infinity along the directions $-\pi/8$, $\pi/8$, $7\pi/8$, and $-7\pi/8$.
	\end{lemma}
	\begin{proof}
		Since $ \Re[\eta_1(\pm z_0(\sigma);\sigma)] \neq 0$, all four rays  $\ell_2^{(b_1)}, \ell_3^{(b_1)}, \ell_2^{(-b_1)}$, and $\ell_3^{(-b_1)}$ must extend off to infinity according to Lemma \ref{Lemma extension to infinity of critical trajectories}. By a conformal mapping argument, one can easily confirm that in a neighborhood $\mathscr{O}$ of $b_1(\sigma)$, for all $z \in \mathscr{O} \cap \De^{(b_1)}$ we have $\Re[\eta_1(z;\sigma)]<0$. Since $\sigma \in \mathcal{O}^*_1$, some ray $\Ga_{\sigma}(b_1(\sigma), \infty)$ must start from $b_1(\sigma)$ within the subset $\mathscr{O} \cap \De^{(b_1)}$ and stay within $\De^{(b_1)}$ (intersection of $\Ga_{\sigma}(b_1(\sigma), \infty)$ with boundaries of $\De^{(b_1)}$ is not possible since on $\Ga_{\sigma}(b_1(\sigma), \infty)$ we have $\Re[\eta_1(z;\sigma)]<0$ while on the boundaries of $\De^{(b_1)}$ we have $\Re[\eta_1(z;\sigma)]=0$).
		
		Now, we show that the interior of $\De^{(b_1)}$ does not contain $\pm z_0(\sigma)$. It suffices to prove that the sign of  $\Re[\eta_1(z;\sigma)]$ does not change in the interior of $\De^{(b_1)}$, because if so, then for all $z$ in the interior of $\De^{(b_1)}$ we would have $\Re[\eta_1(z;\sigma)]<0$, while it is assumed that $\Re[\eta_1(\pm z_0(\sigma);\sigma)]>0$. Notice that due to continuity, the sign of  $\Re[\eta_1(z;\sigma)]$ could only change in the interior of $\De^{(b_1)}$ if there is a curve $\mathcal{L}$ separating the regions where $\Re[\eta_1(z;\sigma)]<0$ and $\Re[\eta_1(z;\sigma)]>0$  with the following properties: $\mathcal{L}$ is a solution of $\Re[\eta_1(z;\sigma)]=0$, lies within $\De^{(b_1)}$ and not intersecting its boundaries $\ell_2^{(b_1)}$ and $\ell_3^{(b_1)}$. Being a critical trajectory, the curve $\mathcal{L}$ must go off to infinity. In the region circumscribed by $\mathcal{L}$ and the boundaries of $\De^{(b_1)}$ we have $\Re[\eta_1(z;\sigma)]<0$ so it can not contain $\pm z_0(\sigma)$. The interior of $\mathcal{L}$ (where $\Re[\eta_1(z;\sigma)]>0$) can not contain $z_0(\sigma)$ either, since if it does, $\mathcal{L}$ has to approach to infinity along two directions $3\pi/4$ radians apart by Teichm\"uller's lemma, which then means that the boundaries of $\De^{(b_1)}$ must approach to infinity along two directions $5\pi/4$ radians apart. But this is a contradiction, since the symmetry relation \eqref{eta z and eta -z} would imply that there has to be intersections between  the boundaries of $\De^{(b_1)}$ and $\De^{(-b_1)}$, which is not possible as the only singular points for the quadratic differential \eqref{OneCut QD} are $\pm b_1(\sigma)$ and $\pm z_0(\sigma)$. This finishes the proof that the interior of $\De^{(b_1)}$ does not contain $\pm z_0(\sigma)$.

		Now it is clear that $\ell_2^{(b_1)}$ and $\ell_3^{(b_1)}$ must approach to infinity along the directions $-\pi/8$ and $\pi/8$ respectively, as any other choice either: a) does not allow $\De^{(b_1)}$ to encompass $(M(\sigma),\infty)$ for some $M(\sigma)>0$, or b) violates Lemma \ref{lemma angles at infty}. By the symmetry relation \eqref{eta z and eta -z} we immediately conclude that  $\ell_2^{(-b_1)}$, and $\ell_3^{(-b_1)}$ respectively approach to infinity along the directions $7\pi/8$, and $-7\pi/8$.
		
		These rays provide four solutions at infinity. Since there are eight solutions at infinty, the other four solutions must come from two curves $\mathcal{L}^{(1)}_{\sigma}$ and $\mathcal{L}^{(2)}_{\sigma}$  each pointing towards infinity in two directions $\pi/4$ radians apart. Each of these curves do approach to infinity along \textit{two} directions as they can not be incident with $\pm z_0(\sigma)$ or $\pm b_1(\sigma)$. The curves $\mathcal{L}^{(1)}_{\sigma}$ and $\mathcal{L}^{(2)}_{\sigma}$ must be symmetric with respect to the origin due to \eqref{eta z and eta -z}. We denote the one in the upper-half plane by $\mathcal{L}^{(1)}_{\sigma}$ and the one in the lower-half plane by $\mathcal{L}^{(2)}_{\sigma}$. From what we proved earlier about  the rays $\ell_2^{(b_1)}, \ell_3^{(b_1)}, \ell_2^{(-b_1)}$, and $\ell_3^{(-b_1)}$, it is now clear that the curve $\mathcal{L}^{(1)}_{\sigma}$ approaches to infinity along the two directions $3\pi/8$ and $5\pi/8$, the curve $\mathcal{L}^{(2)}_{\sigma}$ approaches to infinity along the two directions $-3\pi/8$ and $-5\pi/8$.
	\end{proof}
	
	The following lemmas can be proven using identical arguments and thus we only state the result (see Figure \ref{fig:stable barren 1 4}).
	
	\begin{lemma}\label{lemma structure 2}
		Suppose that $\sigma \in \mathcal{O}^*_1$, $ \Re[\eta_1(\pm z_0(\sigma);\sigma)]<0$, and $\Re[z_0(\sigma)]<0$. Then there exist two disjoint curves $\mathcal{L}^{(3)}_{\sigma}$ and $\mathcal{L}^{(4)}_{\sigma}$  as subsets of $\mathscr{J}^{(1)}_{\sigma}$, which have no intersections with $\pm b_1(\sigma)$,  $\pm z_0(\sigma)$,  $J_{\sigma}$,  $\ell_2^{(b_1)}, \ell_3^{(b_1)}, \ell_2^{(-b_1)}$, and $ \ell_3^{(-b_1)}$. Moreover, the curve $\mathcal{L}^{(3)}_{\sigma}$ approaches to infinity along the two directions $5\pi/8$ and $7\pi/8$, the curve $\mathcal{L}^{(4)}_{\sigma}$ approaches to infinity along the two directions $-3\pi/8$ and $-\pi/8$, and the rays $\ell_2^{(b_1)}, \ell_3^{(b_1)}, \ell_2^{(-b_1)}$, and $\ell_3^{(-b_1)}$ respectively approach to infinity along the directions $-5\pi/8$, $\pi/8$, $3\pi/8$, and $-7\pi/8$.
	\end{lemma}
	
	\begin{lemma}\label{lemma structure 3}
		Suppose that $\sigma \in \mathcal{O}^*_1$, $ \Re[\eta_1(\pm z_0(\sigma);\sigma)]<0$, and $\Re[z_0(\sigma)]>0$. Then there exist two disjoint curves $\mathcal{L}^{(5)}_{\sigma}$ and $\mathcal{L}^{(6)}_{\sigma}$  as subsets of $\mathscr{J}^{(1)}_{\sigma}$, which have no intersections with $\pm b_1(\sigma)$,  $\pm z_0(\sigma)$,  $J_{\sigma}$,  $\ell_2^{(b_1)}, \ell_3^{(b_1)}, \ell_2^{(-b_1)}$, and $ \ell_3^{(-b_1)}$. Moreover, the curve $\mathcal{L}^{(5)}_{\sigma}$ approaches to infinity along the two directions $\pi/8$ and $3\pi/8$, the curve $\mathcal{L}^{(6)}_{\sigma}$ approaches to infinity along the two directions $-5\pi/8$ and $-7\pi/8$, and the rays $\ell_2^{(b_1)}, \ell_3^{(b_1)}, \ell_2^{(-b_1)}$, and $\ell_3^{(-b_1)}$ respectively approach to infinity along the directions $-\pi/8$, $5\pi/8$, $7\pi/8$, and $-3\pi/8$.
	\end{lemma}

	\begin{lemma}\label{humps connect at criticality}
		When $z_0(\sigma) \in \mathscr{J}^{(1)}_{\sigma}$, the components $\mathcal{L}^{(j)}_{\sigma}$ and $\mathcal{L}^{(j+1)}_{\sigma}$, $j=1,3,5$ (respectively for the curves defined in Lemmas \ref{lemma humps}, \ref{lemma structure 2}, and \ref{lemma structure 3}),  are connected to the rest of the critical graph at $\pm z_0$.
	\end{lemma}
	\begin{proof}
		This is the only possibility, as if $\mathcal{L}^{(j)}_{\sigma}$ and $\mathcal{L}^{(j+1)}_{\sigma}$ are not connected to the rest of the critical graph at $\pm z_0$, one would have too many (more than $8$) solutions of the equation $\Re[\eta_1(z;\sigma)]=0$ at $\infty$.
	\end{proof}

	\begin{lemma}\label{one cut is non-empty!}
		Any $\sigma>-2$ belongs to $\mathcal{O}^*_1$ and $\pm z_0(\sigma)$ belong to unstable lands.
	\end{lemma}
	
	\begin{proof}
		For $\sigma>-2$, we know that $b_1>0$ and $z_0=\ii y_0$, with $y_0>0$.  The local structure of the critical trajectories in a neighborhood of the critical points can be easily found by finding a ray on which $Q_1(z;\sigma) \dd z^2<0$. Locally, the other critical trajectories will be then determined based on how many critical directions are incident with the critical point. It is clear that the real interval $(-b_1,b_1)$ must be a short critical trajectory, because it is incident with $\pm b_1$,  $Q_1(z;\sigma)<0$ for all $z \in (-b_1,b_1)$, and $\dd z^2>0$ for all infinitesimal real line segments $\dd z$. Using the explicit formula \eqref{F 1 to 3} one can show that $\Re[\eta_1( z_0(\sigma);\sigma)] > 0$ for all $\sigma > -2$. Using \eqref{eta z and eta -z}, we immediately have $\Re\left[ \eta_1\left(-z_0(\sigma);\sigma\right) \right] >0$ for all $\sigma > -2$, as well. So far we have shown that all $\sigma > -2$ satisfy the first two requirements of Definition \ref{Def one cut sigma minus fake transition}. Now, we prove that the third requirement is met as well. Notice that, for fixed $\sigma>-2$, the function $\eta_1(x;\sigma)$ is real and negative for all $x>b_1(\sigma)$. This means that the complementary arc $\Ga_{\sigma}(b_1(\sigma),\infty)$ in the third requirements of Definition \ref{Def one cut sigma minus fake transition}, can be chosen as the real interval $(b_1(\sigma);\infty)$ for $\sigma>-2$.
		

	\end{proof}

	
	\begin{figure}[h]
		\centering
		\begin{subfigure}{0.31\textwidth}
			\centering
			\includegraphics[width=\textwidth]{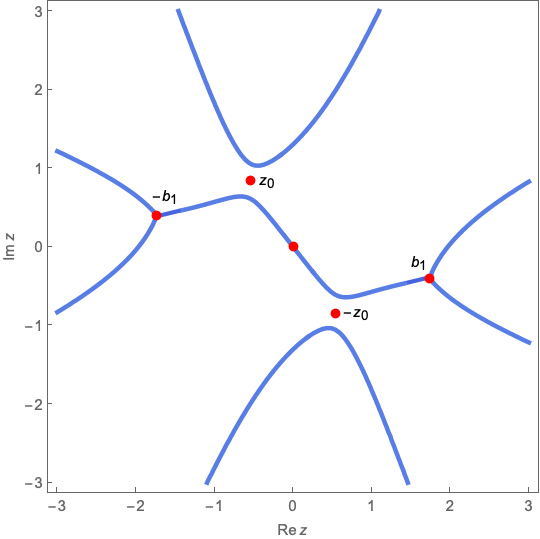}
			\caption{The critical graph $\mathscr{J}^{(1)}_{\sigma}$ when there is a connection from  $-b_1$ to $b_1$.}
			\label{fig:Connection and humps}
		\end{subfigure}%
		\hfill
		\begin{subfigure}{.31\textwidth}
			\centering
			\includegraphics[width=\textwidth]{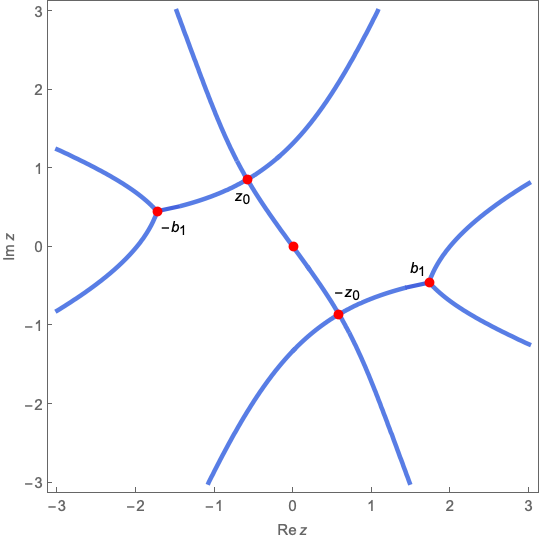}
			\caption{The critical graph $\mathscr{J}^{(1)}_{\sigma}$ at a critical value $\sigma_*$ (See Lemma \ref{lemma connectivity}).}
			\label{fig:Connection about to be lost}    
		\end{subfigure}
		\hfill
		\begin{subfigure}{.31\textwidth}
			\centering
			\includegraphics[width=\textwidth]{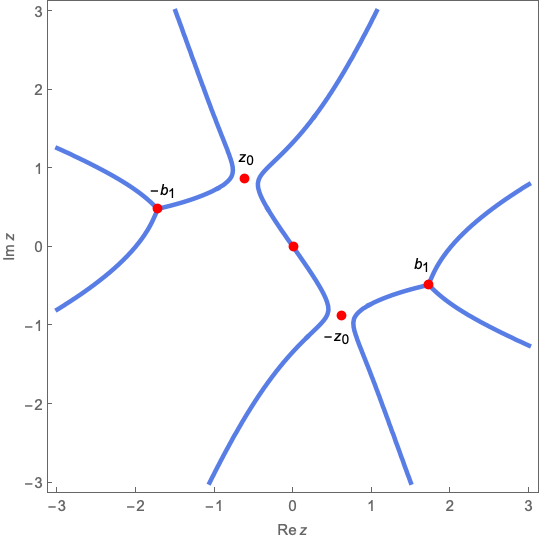}
			\caption{The critical graph $\mathscr{J}^{(1)}_{\sigma}$ when there is no connection from  $-b_1$ to $b_1$.}
			\label{fig:No Connection}    
		\end{subfigure}%
		\caption{Schematic of the continuous deformation of the critical graph $\mathscr{J}^{(1)}_{\sigma}$ ( the collection of all points $z$ satisfying $
			\Re \left[ \eta_1(z;\sigma) \right]=0$ ) from a $\sigma \in \mathcal{O}^*_1$ where there is a connection from  $-b_1$ to $b_1$, to a $\sigma \notin \mathcal{O}^*_1$ where $-z_0$ and $z_0$ both lie on $\mathscr{J}^{(1)}_{\sigma}$, and finally to a $\sigma \notin \mathcal{O}^*_1$ where there is no connection from  $-b_1$ to $b_1$.}
		\label{fig:No Connection transition}
	\end{figure}
	

	

	\begin{lemma}\label{lemma delta}
		Let $\sigma_0 \in \mathcal{O}^*_1$ and not on the branch cuts of $\Re[\eta_1(z_0(\sigma);\sigma)]$. Then there exists $\de > 0$, such that for all $\sigma$ in the $\de$-neighborhood $\{ \sigma : |\sigma - \sigma_0|<\de\}$ of $\sigma_0$, the points $\pm z_0(\sigma)$ do not lie on $\mathscr{J}^{(1)}_{\sigma}$.
		\begin{proof}
			Since $\sigma_0 \in \mathcal{O}^*_1$ we have $\Re[\eta_1(z_0(\sigma_0);\sigma_0)] \neq 0$, so without loss of generality assume that $\Re[\eta_1(z_0(\sigma_0);\sigma_0)]>0$. Since the function $\Re[\eta_1(z_0(\sigma);\sigma)]$ is continuous at $\sigma_0$, there exists $\de > 0$, such that the sign of $\Re[\eta_1(z_0(\sigma);\sigma)]$ is the same as sign of $\Re[\eta_1(z_0(\sigma_0);\sigma_0)]$ for all $\sigma$ in the $\de$-neighborhood $\{ \sigma : |\sigma - \sigma_0|<\de\}$ of $\sigma_0$. 
		\end{proof}
	\end{lemma}
	
	\begin{lemma}\label{lemma connectivity}
		Let $\sigma_0$ and $\de$ have the same meaning as in Lemma \ref{lemma delta}. For any $\hat{\sigma}$ in the $\de$-neighborhood of $\sigma_0$, there is still a connection from $-b_1(\hat{\sigma})$ to $b_1(\hat{\sigma})$ and therefore there still exist two disjoint curves $\mathcal{L}^{(1)}_{\hat{\sigma}}$ and $\mathcal{L}^{(2)}_{\hat{\sigma}}$  as subsets of $\mathscr{J}^{(1)}_{\hat{\sigma}}$ with the same description as given in Lemma \ref{lemma humps}, Lemma \ref{lemma structure 2}, or Lemma \ref{lemma structure 3}.
	\end{lemma}
	\begin{proof}
		Assume that, at $\hat{\sigma}$ there is no longer a connection from $-b_1(\hat{\sigma})$ to $b_1(\hat{\sigma})$. Therefore all six rays $\ell_1^{(b_1(\hat{\sigma}))}$,$\ell_2^{(b_1(\hat{\sigma}))},$ $ \ell_3^{(b_1(\hat{\sigma}))},$ $ \ell_1^{(-b_1(\hat{\sigma}))},$ $ \ell_2^{(-b_1(\hat{\sigma}))},$ $\ell_3^{(-b_1(\hat{\sigma}))}$ must extend to infinity since none can be connected to either $z_0(\hat{\sigma})$ or $-z_0(\hat{\sigma})$ by the choice of $\de$. The other two solutions at infinity must come from a curve  $\mathcal{L}$ symmetric with respect to the origin. Since we have four finite singular points, $\mathcal{L}$ must have one singular point of order 2 ($z_0(\hat{\sigma})$ or $-z_0(\hat{\sigma})$) and one singular point of order one ($b_1(\hat{\sigma})$ or $-b_1(\hat{\sigma})$) on one side, and the other pair of singular points on the other side. By Teichm\"uller's lemma this curve approaches to infinity along two rays $\pi$ radians apart as shown in Figure \ref{fig:No Connection}.
		
		Consider a path $\ga: [0,1] \to \{ \sigma : |\sigma - \sigma_0|<\de\}$, with $\ga(0)=\sigma_0$ and $\ga(1)=\hat{\sigma}$. Since the level sets $\mathscr{J}^{(1)}_{\sigma}$ deform in a continuous fashion with respect to $\sigma$ (for a schematic of three snapshots of this deformation see Figure \ref{fig:No Connection transition}), the above scenario requires existence of a value $\sigma_* =\ga(t_*)$ for some $0<t_*<1$ such that $z_0(\sigma_*) \in \mathscr{J}^{(1)}_{\sigma_*}$. But this is impossible by the choice of $\de$ in Lemma \ref{lemma delta}, as it would mean $\Re[\eta_1(z_0(\sigma_*),\sigma_*)]=0$.   
	\end{proof}
	
	\begin{lemma}\label{lemma good region}
		Let $\sigma_0$ and $\de$ have the same meaning as in Lemma \ref{lemma delta}. Then for all $\sigma$ in the $\de$-neighborhood of $\sigma_0$, there exists a complementary arc $\Ga_{\sigma}(b_1(\sigma), \infty)$ which lies entirely in the component of the set \begin{equation}
		\left\{ z : \Re \left[ \eta_1(z;\sigma) \right]<0 \right\},
		\end{equation} which encompasses $(M(\sigma),\infty)$ for some $M(\sigma)>0$.
	\end{lemma}
	\begin{proof}
		The structure of critical trajectories does not change unless $z_0(\sigma)$ hits the set $\mathscr{J}^{(1)}_{\sigma}$. By the choice of $\de$ and $\sigma_0$, this does not happen for any $\sigma$ in the $\de$-neighborhood of $\sigma_0$.
	\end{proof}
	Lemmas \ref{lemma delta}, \ref{lemma connectivity}, and \ref{lemma good region} are together equivalent to Theorem \ref{O1* is open}.

	\subsection{The Two-cut Regime.}
	Let us recall from \S \ref{Sec two cut eq meas} that the quadratic differential for the two cut regime is
	\begin{equation}\label{TwoCut QD}
	Q_2(z;\sigma) \dd z^2:=z^2\left(z^2-a^2_2(\sigma)\right)\left(z^2-b^2_2(\sigma)\right) \dd z^2.
	\end{equation}
	From \eqref{em46} we recall that  $a_2 \neq b_2$ and $a_2, b_2 \neq 0$ away from $\sigma= \pm 2$. Identical to the one-cut quadratic differential \eqref{OneCut QD}, we can show that the solutions to $\Re[\eta_2(z;\sigma)]$ approach to infinity along the eight directions
	$$\{ \pi/8 + k \pi/4: k=0,\cdots,7\}.$$
	
	\begin{lemma}\label{Lemma one or two vetex polygons are ruled out two cut}
		There are no singular finite geodesic polygons with one or two vertices associated to the quadratic differential \eqref{TwoCut QD}.
	\end{lemma}
	\begin{proof}
		The proof is identical to the proof of Lemma \ref{Lemma one or two vetex polygons are ruled out}.
	\end{proof}

	\begin{definition}\label{Def two cut sigma minus fake transition}
		Define the subset
		$\mathcal{O}_{2}$ in the $\sigma$-plane as the collection of all $\sigma \in \C$ such that
		\begin{enumerate}
			\item The critical graph $\mathscr{J}^{(2)}_{\sigma}$ of all points $z$ satisfying \begin{equation}\label{level set 2}
			\Re \left[ \eta_2(z;\sigma) \right]=0,
			\end{equation}
			contains a single Jordan arc connecting $-b_2(\sigma)$ to $-a_2(\sigma)$ and a single Jordan arc connecting $a_2(\sigma)$ to $b_2(\sigma)$,
			\item There exists a complementary arc $\Ga_{\sigma}(b_2(\sigma), \infty)$ which lies entirely in the component of the set \begin{equation}
			\left\{ z : \Re \left[ \eta_2(z;\sigma) \right]<0 \right\},
			\end{equation} which encompasses $(M(\sigma),\infty)$ for some $M(\sigma)>0$,
			\item There exists a complementary arc $\Ga_{\sigma}(-a_2(\sigma), a_2(\sigma))$ which lies entirely in the component of the set \begin{equation}
			\left\{ z : \Re \left[ \eta_2(z;\sigma) \right]<0 \right\}.
			\end{equation} 
		\end{enumerate}
	\end{definition}
	
	\begin{figure}[h]
		\centering
		\begin{subfigure}{0.33\textwidth}
			\centering
			\includegraphics[width=\textwidth]{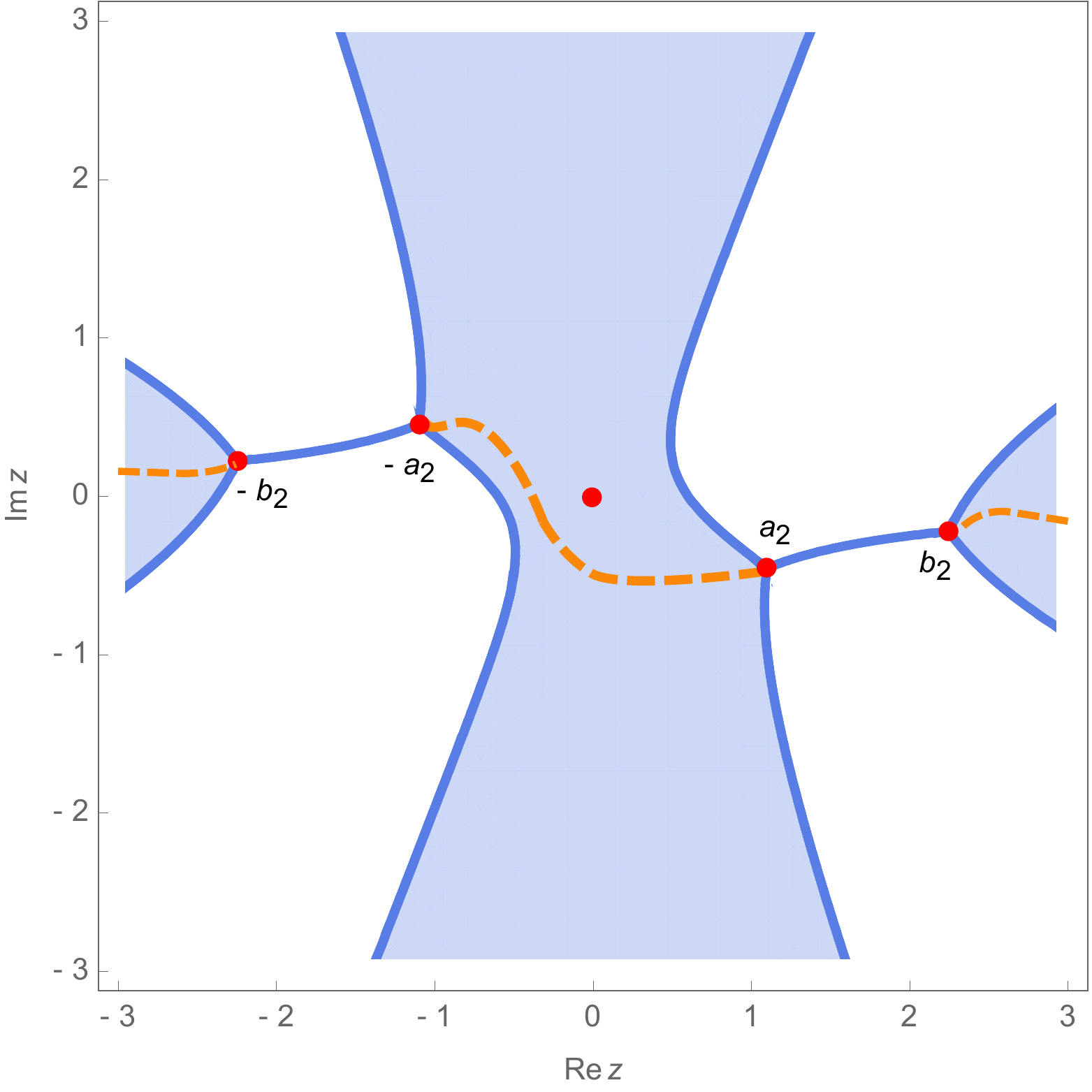}
			\caption{$\sigma=-3+\ii$.}
			\label{fig:fertile barren -3 1}
		\end{subfigure}%
		\begin{subfigure}{.33\textwidth}
			\centering
			\includegraphics[width=\textwidth]{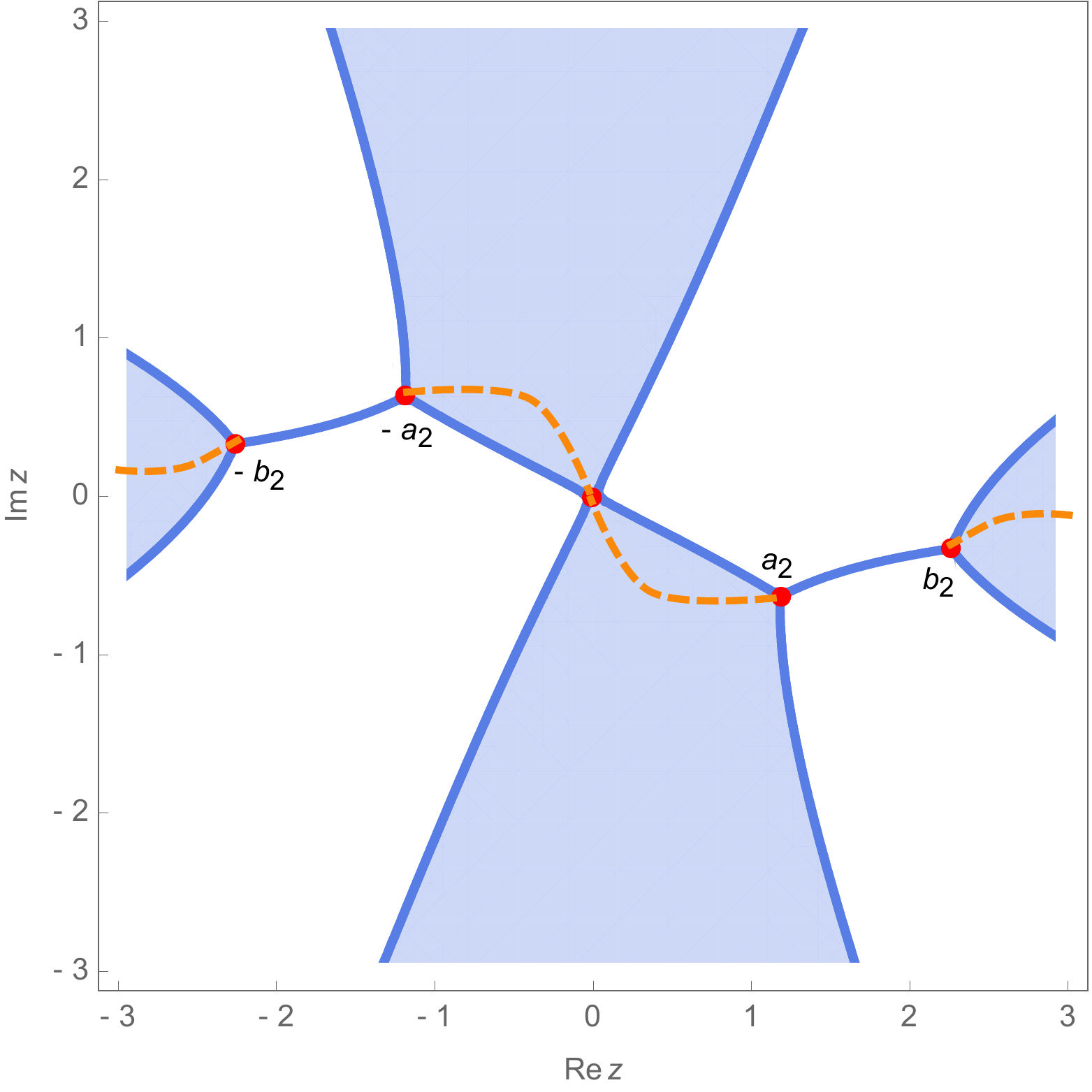}
			\caption{$\sigma_{\mbox{cr}}\simeq-3+1.5\ii$.}
			\label{fig:stable barren -3 1.5}    
		\end{subfigure}
		\begin{subfigure}{.33\textwidth}
			\centering
			\includegraphics[width=\textwidth]{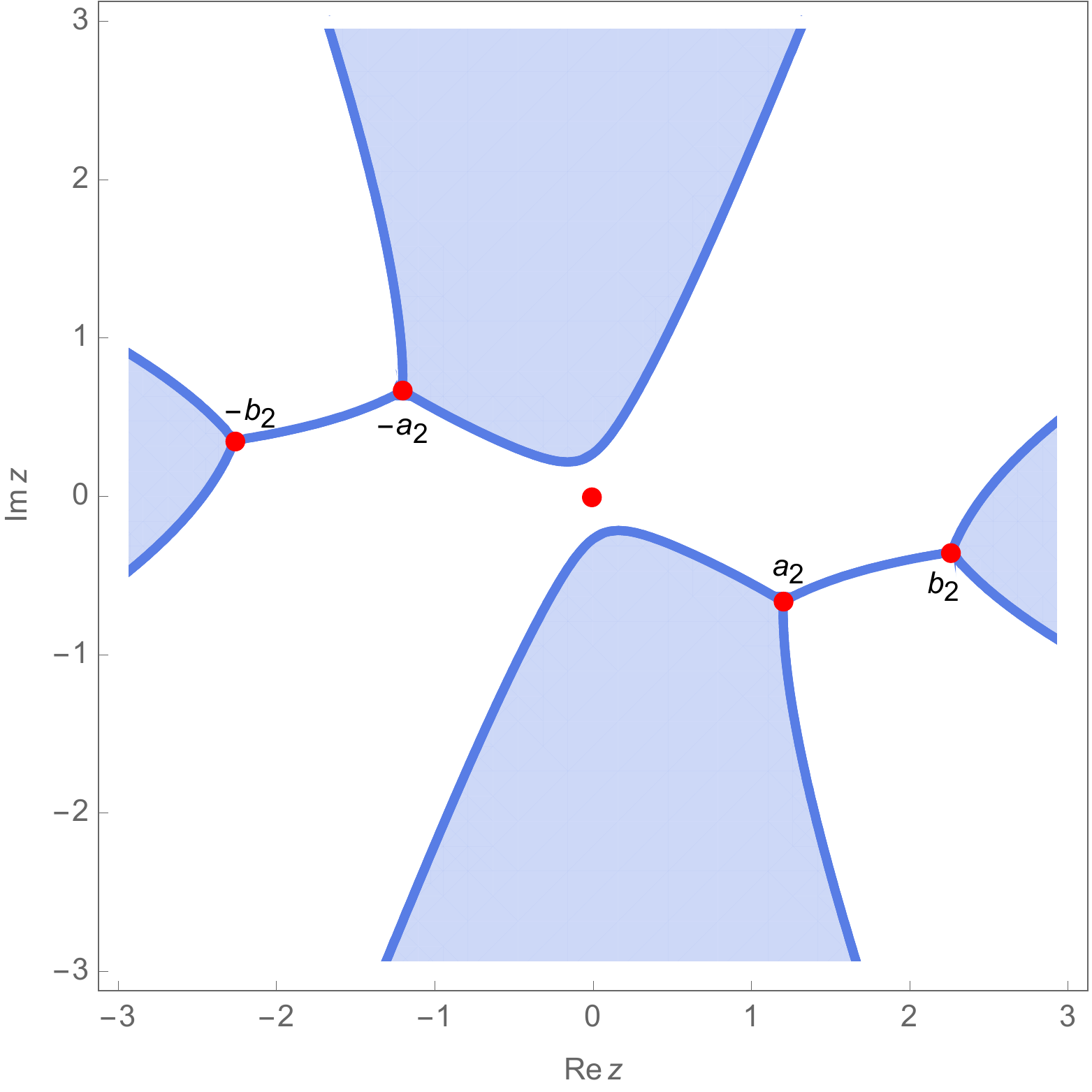}
			\caption{$\sigma=-3+1.6\ii$.}
			\label{fig:stable barren -3 1.6}    
		\end{subfigure}
		\caption{This sequence of figures shows allowable regions in light blue through which the contour of integration (for the orthogonal polynomials) must pass, for a varying collection of values of $\sigma$. Regions in light blue are the $\sigma$-stable lands where $  \Re[\eta_2(z;\sigma)]<0$ and the regions in white are the $\sigma$-unstable lands where $  \Re[\eta_2(z;\sigma)]>0$. Figure (a) corresponds to a $\sigma \in \mathcal{O}_2$ as all conditions of Definition \ref{Def two cut sigma minus fake transition} are satisfied. Figures (b) and (c) do not correspond to $\sigma \in \mathcal{O}_2$, as the third requirement of the definition \ref{Def two cut sigma minus fake transition} is not satisfied.}
		\label{fig:stable and barren 2cut}
	\end{figure}

	\begin{theorem}\label{thm 2cut is open}
		The set $\mathcal{O}_2$ as defined in Definition \ref{Def two cut sigma minus fake transition} is open.
	\end{theorem}
	
	The following Lemmas, collectively, establish the above Theorem.
	
	\begin{lemma}
		The set $\mathscr{J}^{(2)}_{\sigma}$ is symmetric with respect to the origin.
	\end{lemma}
	\begin{proof}
		In view of the first part of Definition \ref{Def two cut sigma minus fake transition}, this simply follows from the identity 
		\begin{equation}	\eta_2(-z;\sigma)=\eta_2(z;\sigma) \pm 2\pi \ii.
		\end{equation}
	\end{proof}
	
	\begin{lemma}\label{continuous deformations of two cut critical graph}\normalfont{[Theorem 3 of \cite{BBGMT2}]}
		\textit{		The critical graph $\mathscr{J}^{(2)}_{\sigma}$ deforms continuously with respect to $\sigma$.}
	\end{lemma}
	
	\begin{lemma}\label{legs connect at criticality 2cut}
		If $0 \in \mathscr{J}^{(2)}_{\sigma}$, either $\ell^{(a_2)}_2$ must connect to $\ell^{(-a_2)}_2$, or $\ell^{(a_2)}_3$ must connect to $\ell^{(-a_2)}_3$ at the origin. 
	\end{lemma}
	\begin{proof}
		This is necessary due to symmetry and to avoid too many (more than $8$) solutions of the equation $\Re[\eta_2(z;\sigma)]=0$ at $\infty$.
	\end{proof}

	\begin{corollary}\label{Cor 0 on critical graph}
		If for some $\sigma$ we have  $0 \in \mathscr{J}^{(2)}_{\sigma}$, then $\sigma \notin \mathcal{O}_2$.
	\end{corollary}
	\begin{proof}
		This is obvious now from the previous Lemma, since any path $\Ga_{\sigma}(-a_2,a_2)$ can not entirely lie in $\left\{ z : \Re \left[ \eta_2(z;\sigma) \right]<0 \right\}$ as $\mathscr{J}^{(2)}_{\sigma}$ has formed a barrier between $-a_2$ and $a_2$ (See Figure \ref{fig:stable barren -3 1.5}).
	\end{proof}
	
	\begin{lemma}\label{2cut is nonempty}
		Any $\sigma<-2$ belongs to $\mathcal{O}_2$.
	\end{lemma}
	\begin{proof}
		This can be proven in an identical way as Lemma \ref{one cut is non-empty!} and we do not provide the details here. 
	\end{proof}
	

	\begin{lemma}\label{lemma delta 2-cut}
		Let $\sigma_0 \in \mathcal{O}_2$ and not on the branch cuts of $\Re[\eta_2(0;\sigma)]$. Then there exists $\de > 0$, such that for all $\sigma$ in the $\de$-neighborhood $\{ \sigma : |\sigma - \sigma_0|<\de\}$ of $\sigma_0$, the point at the origin does not lie on $\mathscr{J}^{(2)}_{\sigma}$.
		\begin{proof}
			Since $\sigma_0 \in \mathcal{O}_2$ we have $\Re[\eta_2(0;\sigma_0)] \neq 0$. Since the function $\Re[\eta_2(0;\sigma)]$ is continuous at $\sigma_0$, there exists $\de > 0$, such that the sign of $\Re[\eta_2(0;\sigma)]$ is the same as sign of $\Re[\eta_2(0;\sigma_0)]$ for all $\sigma$ in the $\de$-neighborhood $\{ \sigma : |\sigma - \sigma_0|<\de\}$ of $\sigma_0$. 
		\end{proof}
	\end{lemma}

	The points $\pm a_2$ and $\pm b_2$ are all simple zeros of the quadratic differential. So three critical trajectories emanate from each one. We denote the local critical arcs incident to $p$ by $\ell^{(p)}_{1}$, $\ell^{(p)}_{2}$, $\ell^{(p)}_{3}$ (labeled in counterclockwise direction), where $\ell^{(p)}_{1}$ is the critical arc emanating from $p$ which makes the connection prescribed in the first requirement of Definition \ref{Def two cut sigma minus fake transition}, $p=\pm a_2, \pm b_2$.
	
	\begin{lemma}\label{where things go two cut}
		Let $\sigma \in \mathcal{O}_2$. The critical arcs $\ell_2^{(-b_2)}$, $\ell_3^{(-b_2)}$, $\ell_2^{(-a_2)}$, $\ell_3^{(-a_2)}$,$\ell_2^{(a_2)}$, $\ell_3^{(a_2)}$,$\ell_2^{(b_2)}$, and $\ell_3^{(b_2)}$, respectively approach to infinity along the directions $7\pi/8$, $-7\pi/8$, $-5\pi/8$, $5\pi/8$, $3\pi/8$, $-3\pi/8$, $-\pi/8$, and $\pi/8$. Moreover, the $\sigma$-stable and $\sigma$-unstable lands having these critical arcs as boundaries are as given in Figure \ref{fig:fertile barren -3 1}, and in particular, $\Re[\eta_2(0;\sigma)]<0$. 
	\end{lemma}
	
	\begin{proof}
		It is easy to verify that no two critical arcs from the set $\boldsymbol{L}:=\{ \ell^{(p)}_{2}, \ell^{(p)}_{3} : p=\pm a_2, \pm b_2 \}$, can be connected to one another, as it would violate Lemma \ref{Lemma one or two vetex polygons are ruled out two cut}, or would lead to geodesic polygons with more than two vertices which are not allowed by Teichm\"uller's Lemma. Moreover, no critical arc from the set $\boldsymbol{L}$ can be connected to the origin due to the third requirement of the Definition \ref{Def two cut sigma minus fake transition}.  Therefore all critical arcs from the set $\boldsymbol{L}$ must approach infinity.
		
		Notice that the point at the origin can not be enclosed by the geodesic polygon with vertices $b_2$ and $\infty$ defined by $\ell_2^{(b_2)}$, $\ell_3^{(b_2)}$. This is because,
		due to the symmetry of the critical graph with respect to the origin,  $\ell_2^{(b_2)}$ and $\ell_3^{(b_2)}$ are respectively reflections of $\ell_2^{(-b_2)}$ and $\ell_3^{(-b_2)}$ through the origin. So if the the geodesic polygon with vertices $b_2$ and $\infty$ defined by $\ell_2^{(b_2)}$, $\ell_3^{(b_2)}$ encloses the origin, so does the geodesic polygon with vertices $-b_2$ and $\infty$ defined by $\ell_2^{(-b_2)}$, $\ell_3^{(-b_2)}$. But this would mean that there is an intersection between at least one ray emanating from $b_2$ and one ray emanating from $-b_2$ at a regular point, which is impossible. For a similar reason, one can show that the endpoints $-b_2,-a_2$ and $a_2$ can not be enclosed by the geodesic polygon with vertices $b_2$ and $\infty$ defined by $\ell_2^{(b_2)}$, $\ell_3^{(b_2)}$. Now the Teichm\"uller's Lemma implies that the critical rays $\ell_2^{(b_2)}$ and $\ell_3^{(b_2)}$ must approach $\infty$ along two directions $\pi/4$ radians apart. Therefore, in order to satisfy the third requirement of the Definition \ref{Def two cut sigma minus fake transition}, $\ell_2^{(b_2)}$ and $\ell_3^{(b_2)}$ must respectively approach to $-\pi/8$, and $\pi/8$. 
		
		Now, we notice that the geodesic polygon with three vertices $a_2, b_2$, and $\infty$, comprised of $\ell^{(b_2)}_1, \ell^{(b_2)}_3$, and $\ell^{(a_2)}_2$ can not enclose the origin, because in that case, by symmetry it would enforce the geodesic polygon with three vertices $-a_2, -b_2$, and $\infty$, comprised of $\ell^{(-b_2)}_1, \ell^{(-b_2)}_3$, and $\ell^{(-a_2)}_2$ also enclose the origin, which then implies the failure of the fourth requirement of the Definition \ref{Def two cut sigma minus fake transition} (See Figure \ref{fig:stable barren -3 1.6}). The same argument shows that the geodesic polygon with three vertices $a_2, b_2$, and $\infty$, comprised of $\ell^{(b_2)}_1, \ell^{(b_2)}_2$, and $\ell^{(a_2)}_3$ can not enclose the origin as well. The Teichm\"uller's lemma for these geodesic polygons now ensures that \textbf{a)} the critical rays $\ell^{(b_2)}_3$, and $\ell^{(a_2)}_2$ must approach infinity along two directions $\pi/4$ radians apart, and \textbf{b)} the critical rays  $\ell^{(b_2)}_2$ and $\ell^{(a_2)}_3$ must approach infinity along two directions $\pi/4$ radians apart. Due to what we have already found about $\ell_2^{(b_2)}$ and $\ell_3^{(b_2)}$, we immediately conclude that  $\ell_2^{(a_2)}$ and $\ell_3^{(a_2)}$, respectively approach to infinity along the directions $3\pi/8$ and $-3\pi/8$. The angles of approach to infinity for $\ell_2^{(-b_2)}$, $\ell_3^{(-b_2)}$, $\ell_2^{(-a_2)}$, and $\ell_3^{(-a_2)}$ are found by symmetry. 
		
		As shown above, the origin does not belong to any of the three geodesic polygons having $b_2$ as a common vertex, and due to symmetry, it does not belong to any of the three geodesic polygons having $-b_2$ as a common vertex. So the origin has to belong to the geodesic polygon with three vertices $\pm a_2$ and $\infty$ (see Figure \ref{fig:fertile barren -3 1}). By a straightforward conformal mapping argument similar to the one shown in Figure \ref{fig:Conf Map}, we can show that one has the $\sigma$-stable and $\sigma$-unstable lands as shown in Figure \ref{fig:fertile barren -3 1}, which in particular implies $\Re[\eta_2(0;\sigma)]<0$.
	\end{proof}
	
	\begin{lemma}\label{lemma connectivity 2 cut}
		Let $\sigma_0$ and $\de$ have the same meaning as in Lemma \ref{lemma delta 2-cut}. For any $\sigma$ in the $\de$-neighborhood of $\sigma_0$, the first requirement of Definition \ref{Def two cut sigma minus fake transition} is still met.
	\end{lemma}
	\begin{proof}
		For the sake of arriving at a contradiction, assume that for some $\hat{\sigma}$ in the $\de$-neighborhood of $\sigma_0$ there is no connection between $-b_2(\hat{\sigma})$ to $-a_2(\hat{\sigma})$, and thus, due to symmetry, no connection between $a_2(\hat{\sigma})$ to $b_2(\hat{\sigma})$.
		By the choice of $\de$, for all $\sigma$ in the $\de$-neighborhood of $\sigma_0$, in particular for $\hat{\sigma}$, the point at the origin does not lie on $\mathscr{J}^{(2)}_{\sigma}$. Therefore $\Re[\eta_2(0;\hat{\sigma})] \neq 0$. This means that all critical arcs in the set $\{ \ell^{(p)}_{1}, \ell^{(p)}_{2}, \ell^{(p)}_{3} : p=\pm a_2, \pm b_2 \}$ must approach to infinity (again, it is easy to observe that no two critical arcs in the set  $\{ \ell^{(p)}_{1}, \ell^{(p)}_{2}, \ell^{(p)}_{3} : p=\pm a_2, \pm b_2 \}$ can be connected to one another, as it would violate the Teichm\"uller's lemma). But this would mean one has twelve solutions at $\infty$, which is a contradiction.
	\end{proof}

	\begin{lemma}\label{lemma complementary contour from b_2 to infinity}
		Let $\sigma_0$ and $\de$ have the same meaning as in Lemma \ref{lemma delta 2-cut}. For any $\sigma$ in the $\de$-neighborhood of $\sigma_0$, the third and the fourth requirements of Definition \ref{Def two cut sigma minus fake transition} are still met.
	\end{lemma}
	
	\begin{proof}
		Due to continuous deformation of $\mathscr{J}^{(2)}_{\sigma}$ with respect to $\sigma$, as $\sigma$ varies from $\sigma_0$ in the $\de$-neighborhood of $\sigma_0$, the critical trajectories $\ell^{(\pm a_2)}_2$, $\ell^{(\pm a_2)}_3$, $\ell^{(\pm b_2)}_2$, and $\ell^{(\pm b_2)}_3$ continuously deform without hitting the origin. This ensures that one has the same structure for the critical graph as shown in Figure \ref{fig:fertile barren -3 1}. Thus, the third and the fourth requirements of the Definition \ref{Def two cut sigma minus fake transition} are still met.
	\end{proof}

	Lemmas \ref{lemma delta 2-cut}, \ref{lemma connectivity 2 cut}, and \ref{lemma complementary contour from b_2 to infinity} together imply Theorem \ref{thm 2cut is open}.
	
	\subsection{The Three-cut Regime}
	
	The quadratic differential for the three-cut regime is
	\begin{equation}\label{ThreeCut QD}
	Q_3(z;\sigma) \dd z^2:=\left(z^2-a^2_3(\sigma)\right)\left(z^2-b^2_3(\sigma)\right)\left(z^2-c^2_3(\sigma)\right) \dd z^2.
	\end{equation}
	Also denote
	\begin{equation}
	\eta_3(z;\sigma) := \int_{c_3}^{z} \sqrt{\left(s^2-a^2_3(\sigma)\right)\left(s^2-b^2_3(\sigma)\right)\left(s^2-c^2_3(\sigma)\right)} \dd s
	\end{equation}
	
	Identical to the one-cut quadratic differential \eqref{OneCut QD}, we can show that the three-cut critical trajectories (solutions of $\Re[\eta_3(z;\sigma)]=0$) approach to infinity along the eight directions
	$$\{ \pi/8 + k \pi/4: k=0,\cdots,7\}.$$

	\begin{definition}\label{Def three cut sigma}
		Define the subset
		$\mathcal{O}_{3}$ in the $\sigma$-plane as the collection of all $\sigma \in \C$ such that the points $a_3(\sigma) \neq 0$, $b_3(\sigma)$, and $c_3(\sigma)$ as solutions of \eqref{em52}, \eqref{3cut gap condition}, and \eqref{b3 on the same level set as c3 21} are distinct and
		\begin{enumerate}
			\item The critical graph $\mathscr{J}^{(3)}_{\sigma}$ of all points $z$ satisfying \begin{equation}\label{level set 3}
			\Re \left[ \eta_3(z;\sigma) \right]=0,
			\end{equation}
			contains a single Jordan arc connecting $-c_3(\sigma)$ to $-b_3(\sigma)$, a single Jordan arc connecting $-a_3(\sigma)$ to $a_3(\sigma)$, and a single Jordan arc connecting $b_3(\sigma)$ to $c_3(\sigma)$.
			\item There exists a complementary arc $\Ga_{\sigma}(c_3(\sigma), \infty)$ which lies entirely in the component of the set \begin{equation}
			\left\{ z : \Re \left[ \eta_3(z;\sigma) \right]<0 \right\},
			\end{equation} which encompasses $(M(\sigma),\infty)$ for some $M(\sigma)>0$.
			\item There exists a complementary arc $\Ga_{\sigma}(a_3(\sigma), b_3(\sigma))$ which lies entirely in the component of the set \begin{equation}
			\left\{ z : \Re \left[ \eta_3(z;\sigma) \right]<0 \right\}.
			\end{equation} 
		\end{enumerate}
	\end{definition}

	\begin{theorem}
		$\mathcal{O}_3$ is an open set.
	\end{theorem}
	\begin{proof}
		Let $\sigma_0 \in \mathcal{O}_3$. For the sake of arriving at a contradiction, let us assume that there is no neighborhood of $\sigma_0$ consisting only of $\sigma \in \mathcal{O}_3$. This means that there exists a sequence $\{\sigma_k\}^{\infty}_{k=1}$ converging to $\sigma_0$, so that $\sigma_k \notin \mathcal{O}_3$. Since for all $\sigma$, the equilibrium measure and the Riemann-Hilbert contour exists and is unique (\cite{KS} uniqueness in the gaps are up to homotopy) $\sigma_k$ belongs to $\overline{\mathcal{O}_1} \cup \overline{\mathcal{O}_2}$. Therefore there is a subsequence $\{\hat{\sigma}_j\}^{\infty}_{j=1}$ of $\{\sigma_k\}^{\infty}_{k=1}$ convergent to $\sigma_0$, with  $\hat{\sigma}_j$ either all belong to $\overline{\mathcal{O}_1}$ or all belong to $\overline{\mathcal{O}_2}$. Without loss of generality, let us assume that $\hat{\sigma}_j$ all belong to $\overline{\mathcal{O}_1}$. Now consider a subsequence $\{\tilde{\sigma}_{\ell}\}^{\infty}_{\ell=1}$ of $\{\hat{\sigma}_j\}^{\infty}_{j=1}$ convergent to $\sigma_0$ so that all $\tilde{\sigma}_{\ell}$ belong to $\mathcal{O}_1$. Notice that we can always choose such a sequence, because even if there are infinitely many members of $\{\hat{\sigma}_j\}^{\infty}_{j=1}$ belonging to $\overline{\mathcal{O}_1} \setminus \mathcal{O}_1$, for each $j$ we can consider a sequence $\{ \sigma^{(j)}_m \}^{\infty}_{m=1} \subset \mathcal{O}_1$ convergent to $\hat{\sigma}_j$, and then via a diagonal process we can choose a sequence entirely in $\mathcal{O}_1$ convergent to $\sigma_0$. But since $\mathcal{O}_1$ is open, a sequence entirely in $\mathcal{O}_1$ can only converge to $\sigma_0 \notin \mathcal{O_1}$, only if $\sigma_0 \in \overline{\mathcal{O}_1} \setminus \mathcal{O}_1$. But if $\sigma_0 \in$ I $\cup$ XII we know that $a_3(\sigma_0)=b_3(\sigma_0)$, and if $\sigma_0 \in$ VII $\cup$ IX we know that $b_3(\sigma_0)=c_3(\sigma_0)$ (See Figure \ref{fig:Phase Diagram}), in either case we would have $\sigma_0 \notin \mathcal{O}_3$, which is a contradiction.
	\end{proof}
	
	\subsection{Evolution of the Critical Graphs and the Support of the Equilibrium Measure Through Phase Transitions.}

	The critical contours in Figure \ref{fig:Phase Diagram} divide the complex $\sigma$-plane into the one-cut, two-cut and three-cut regimes. We observe that the phase transition from the one-cut to the two-cut regime occurs only through the multi-critical point at $\sigma=-2$. Indeed, in the following figures one can see how the support of the equilibrium measure splits into two symmetric cuts as $\sigma$ is altered from one-cut regime through $\sigma=-2$ into the two-cut regime:
	
	\begin{figure}[h]
		\centering
		\begin{subfigure}{0.24\textwidth}
			\centering
			\includegraphics[width=\textwidth]{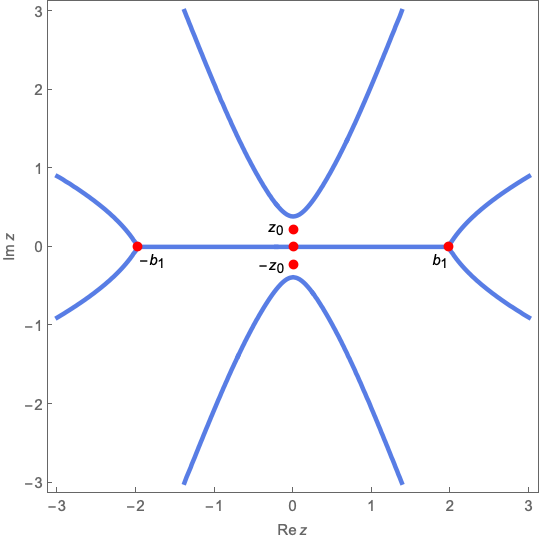}
			\caption{$\mathscr{J}^{(1)}_{\sigma_1}$ at $\sigma_1=-1.9$.}
			\label{fig1-3birth:sigma=-1.9}
		\end{subfigure}%
		\hfill
		\begin{subfigure}{.24\textwidth}
			\centering
			\includegraphics[width=\textwidth]{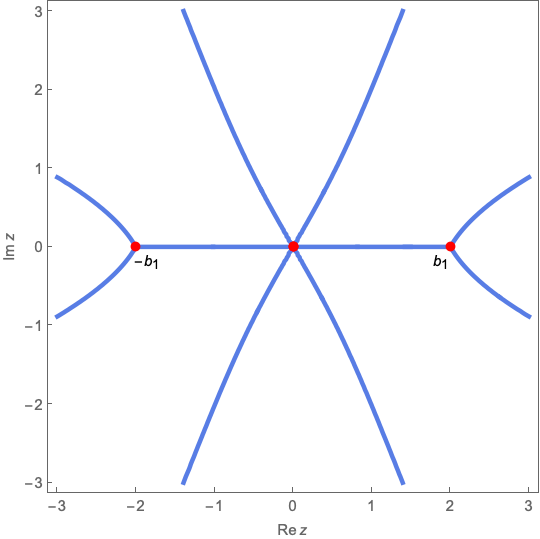}
			\caption{$\mathscr{J}^{(1)}_{\sigma_{\mbox{\tiny cr}}}$ (or $\mathscr{J}^{(2)}_{\sigma_{\mbox{\tiny cr}}}$) at $\sigma_{\mbox{\tiny cr}}=-2$}
			\label{fig1-3birth:sigma=-2}    
		\end{subfigure}
		\hfill
		\begin{subfigure}{.24\textwidth}
			\centering
			\includegraphics[width=\textwidth]{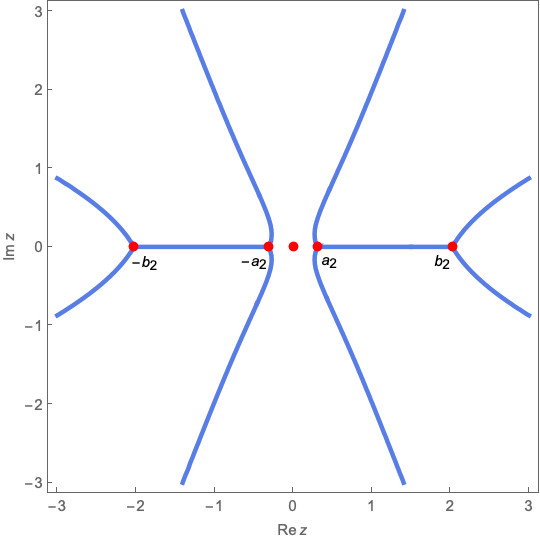}
			\caption{$\mathscr{J}^{(2)}_{\sigma_2}$ at $\sigma_2=-2.1$}
			\label{fig1-3birth:sigma=-2.1}    
		\end{subfigure}%
		\caption{Snapshots of the continuous evolution of the critical graph $\mathscr{J}^{(1)}_{\sigma_1}$ to the critical graph $\mathscr{J}^{(2)}_{\sigma_2}$ as $\sigma$ varies from $\sigma_1=-1.9$ to $\sigma_2=-2.1$ through the multi-critical point $\sigma_{\mbox{\tiny cr}}=-2$ (locate the $\sigma$-values in Figure \ref{fig:Phase Diagram}). At the critical value, just before the split, the point $z_0$ gets trapped at the origin between different portions of the critical graph.}
		\label{fig:1-2birth}
	\end{figure}
	
	Figures \ref{fig:1-3splitt}, and \ref{fig:2-3birth} below show how the critical graph continuously evolves (see Lemmas \ref{continuous deformations of one cut critical graph} and \ref{continuous deformations of two cut critical graph}) as $\sigma$ changes from a non-critical real value to a critical value.
	
	\begin{figure}[h]
		\centering
		\begin{subfigure}{0.24\textwidth}
			\centering
			\includegraphics[width=\textwidth]{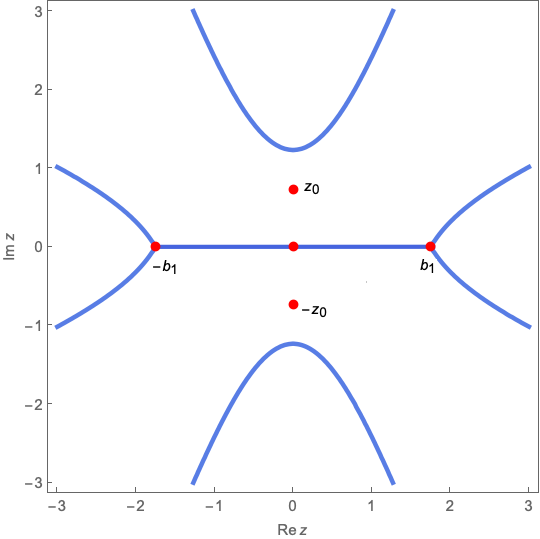}
			\caption{$\sigma=-1$}
			\label{fig1-3birth:sigma=-1}
		\end{subfigure}%
		\begin{subfigure}{.24\textwidth}
			\centering
			\includegraphics[width=\textwidth]{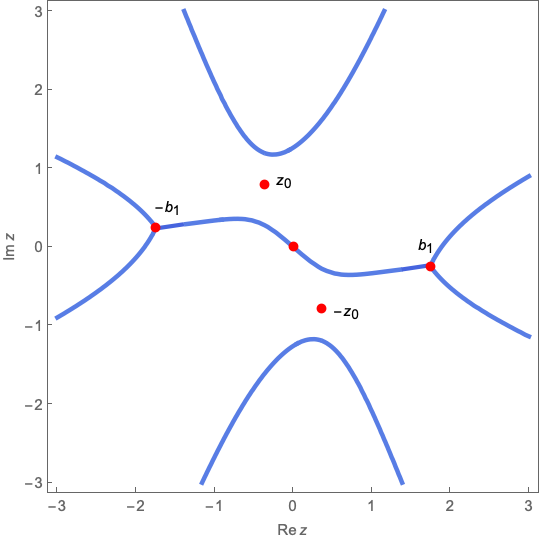}
			\caption{$\sigma=-1+\ii$}
			\label{fig1-3birth:sigma=-1+i}    
		\end{subfigure}
		\begin{subfigure}{.24\textwidth}
			\centering
			\includegraphics[width=\textwidth]{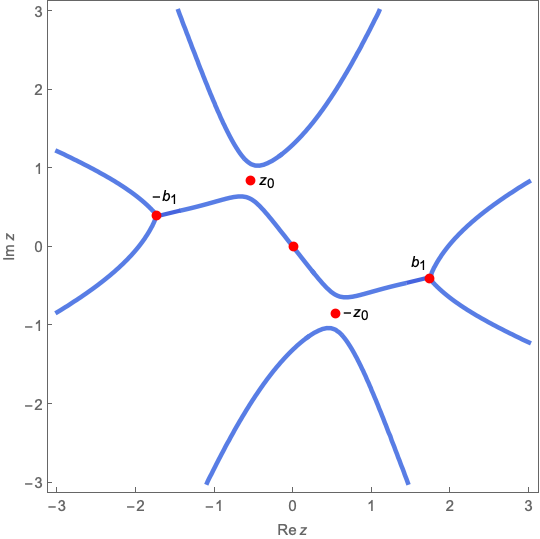}
			\caption{$\sigma=-1+1.6\ii$}
			\label{fig1-3birth:sigma=-1+1.6i}    
		\end{subfigure}%
		\begin{subfigure}{.24\textwidth}
			\centering
			\includegraphics[width=\textwidth]{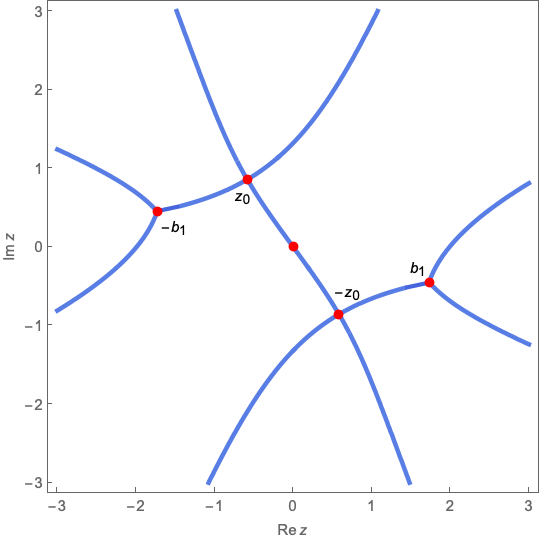}
			\caption{$\sigma_{\mbox{\tiny cr}}\simeq-1+1.7795\ii$}
			\label{fig1-3birth:sigma=-1+1.779i}    
		\end{subfigure}
		\caption{Snapshots of the continuous evolution of the critical graph $\mathscr{J}^{(1)}_{\sigma}$ as $\sigma$ changes from $-1$ in the vertical direction up to the critical value $\sigma_{\mbox{\tiny cr}}\simeq-1+1.7795\ii$ (locate the $\sigma$-values in Figure \ref{fig:Phase Diagram}). At the critical value, just before the split of the support of the equilibrium measure, the point $z_0$ gets trapped between various portions of the critical graph.}
		\label{fig:1-3splitt}
	\end{figure}
	
	\begin{figure}[h]
		\centering
		\begin{subfigure}{0.24\textwidth}
			\centering
			\includegraphics[width=\textwidth]{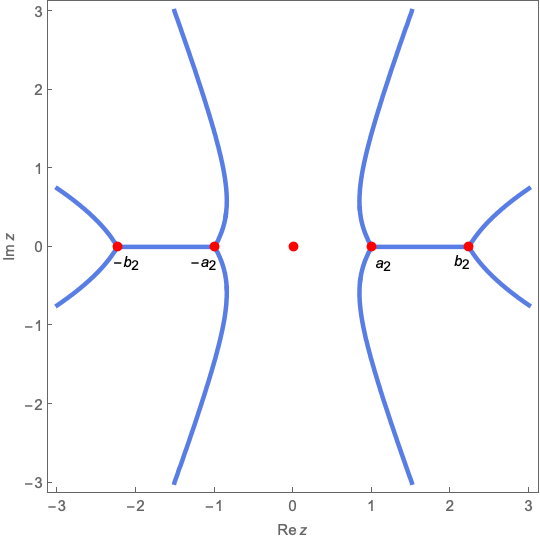}
			\caption{$\sigma=-3$}
			\label{fig1-3birth:sigma=-3}
		\end{subfigure}%
		\begin{subfigure}{.24\textwidth}
			\centering
			\includegraphics[width=\textwidth]{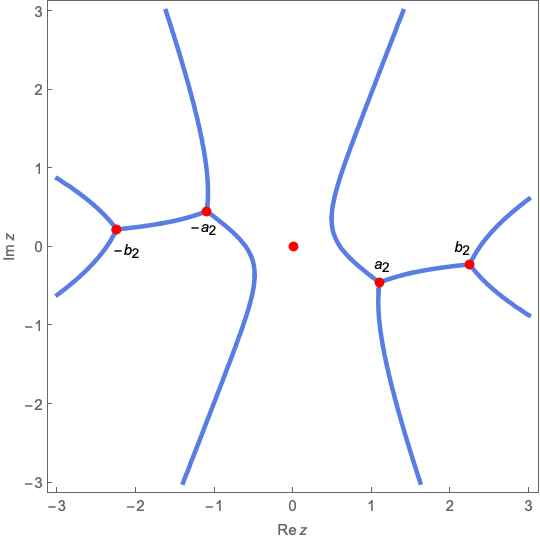}
			\caption{$\sigma=-3+\ii$}
			\label{fig1-3birth:sigma=-3+i}    
		\end{subfigure}
		\begin{subfigure}{.24\textwidth}
			\centering
			\includegraphics[width=\textwidth]{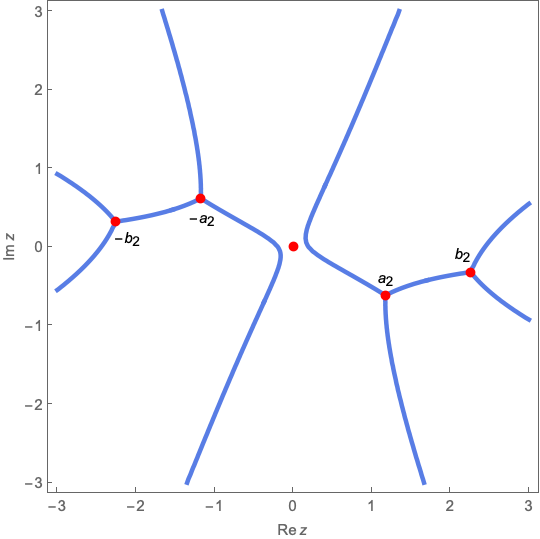}
			\caption{$\sigma=-3+1.45\ii$}
			\label{fig1-3birth:sigma=-3+1.45i}    
		\end{subfigure}%
		\begin{subfigure}{.24\textwidth}
			\centering
			\includegraphics[width=\textwidth]{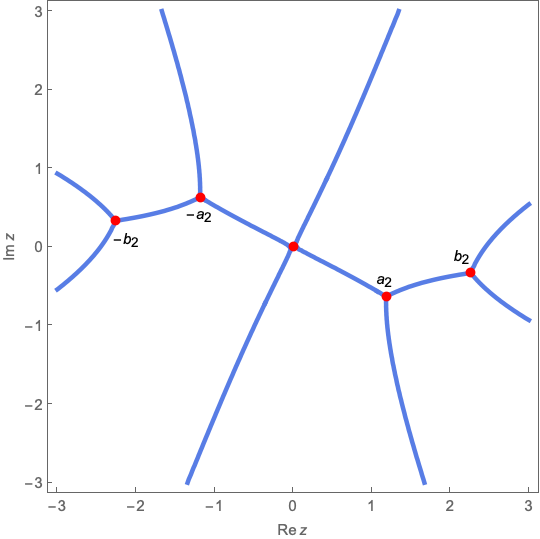}
			\caption{$\sigma_{\mbox{\tiny cr}}\simeq-3+1.5025\ii$}
			\label{fig1-3birth:sigma=-3+1.5i}    
		\end{subfigure}
		\caption{Snapshots of the continuous evolution of the critical graph $\mathscr{J}^{(2)}_{\sigma}$ as $\sigma$ changes from $-3$ in the vertical direction up to the critical value $\sigma_{\mbox{\tiny cr}}\simeq-3+1.5025\ii$ (locate the $\sigma$-values in Figure \ref{fig:Phase Diagram}). At the critical value, just before the birth of a cut at the origin, the origin gets trapped between different portions of the critical graph.}
		\label{fig:2-3birth}
	\end{figure}
	
	
	In Figures \ref{fig:1-3splitSEQM} through \ref{fig:2-3birthSEQM} we show how the support of the equilibrium measure evolves when it is altered from pre-critical one-cut or two-cut values to post-critical three-cut ones.
	
	\begin{figure}[h]
		\centering
		\begin{subfigure}{0.19\textwidth}
			\centering
			\includegraphics[width=\textwidth]{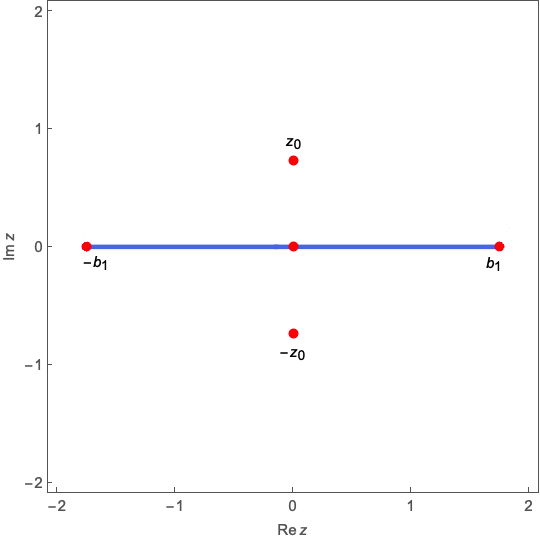}
			\caption{\tiny $\sigma=-1$}
			\label{fig1-3splitSEQM:sigma=-1}
		\end{subfigure}%
		\begin{subfigure}{.19\textwidth}
			\centering
			\includegraphics[width=\textwidth]{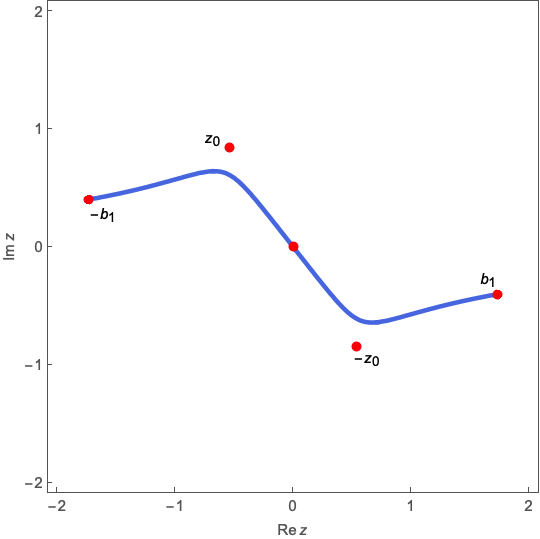}
			\caption{\tiny $\sigma=-1+1.6\ii$}
			\label{fig1-3split:sigmaSEQM=-1+1.6i}    
		\end{subfigure}
		\begin{subfigure}{.19\textwidth}
			\centering
			\includegraphics[width=\textwidth]{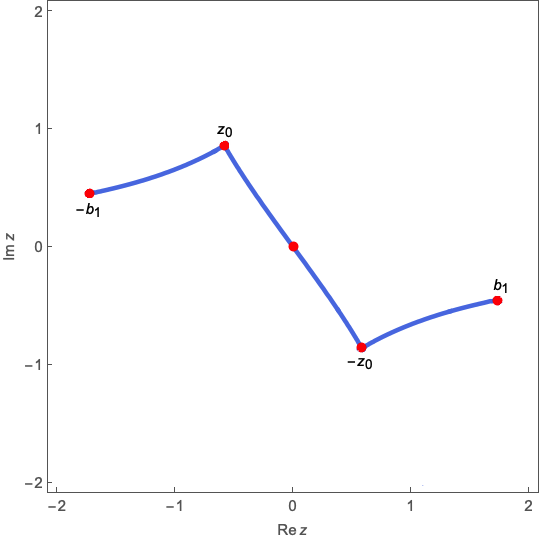}
			\caption{\tiny $\sigma_{\mbox{\tiny cr}}\simeq-1+1.7795\ii$}
			\label{fig1-3split:sigmaSEQM=-1+1.78i}    
		\end{subfigure}%
		\begin{subfigure}{.19\textwidth}
			\centering
			\includegraphics[width=\textwidth]{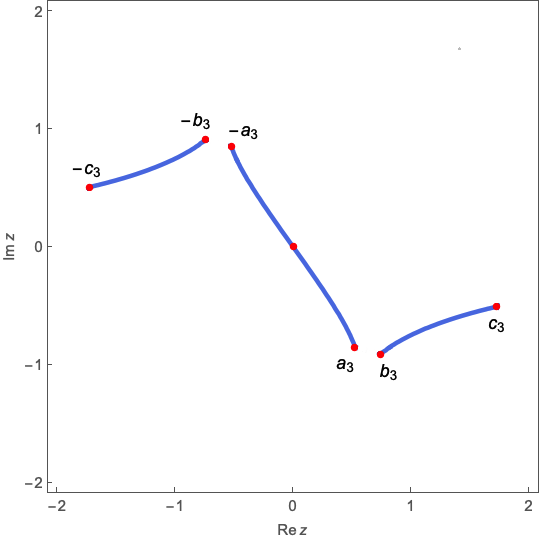}
			\caption{\tiny $\sigma=-1+2\ii$}
			\label{fig1-3split:sigmaSEQM=-1+2i}    
		\end{subfigure}
		\begin{subfigure}{.19\textwidth}
			\centering
			\includegraphics[width=\textwidth]{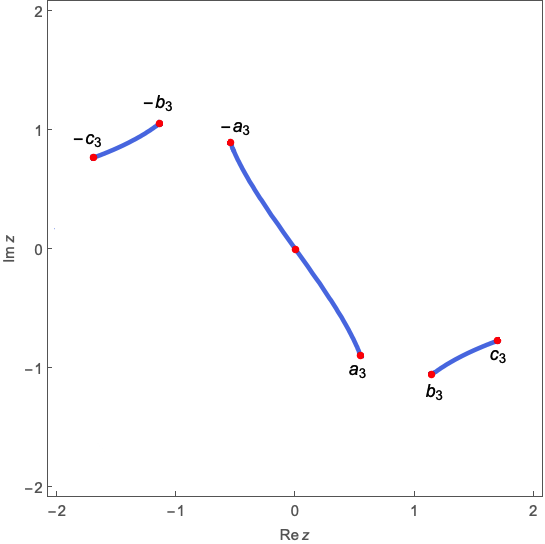}
			\caption{\tiny $\sigma=-1+3\ii$}
			\label{fig1-3split:sigmaSEQM=-1+3i} 
		\end{subfigure}
		\caption{Snapshots of the continuous evolution of the support of the equilibrium measure in transition from the one-cut into the three-cut regime: the support of the equilibrium measure is at the onset of splitting into three symmetric cuts with respect to the origin at a critical value $\sigma_{\mbox{\tiny cr}} \in \ga_1$ , See Figure \ref{fig:Phase Diagram}.}
		\label{fig:1-3splitSEQM}
	\end{figure}
	
	\begin{figure}[h]
		\centering
		\begin{subfigure}{0.19\textwidth}
			\centering
			\includegraphics[width=\textwidth]{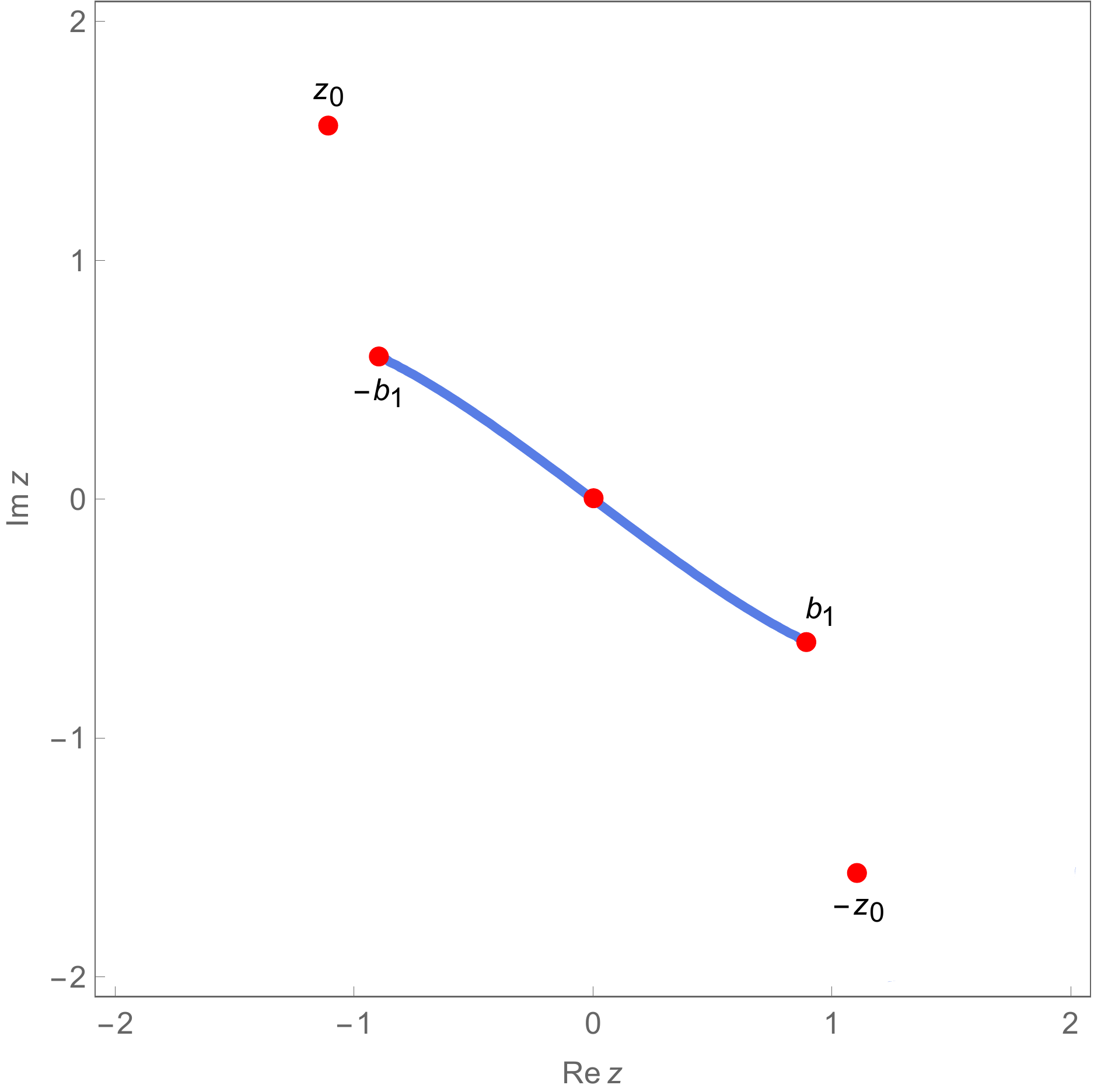}
			\caption{\tiny $\sigma=1+4\ii$}
			\label{fig1-3birthSEQM:sigma=1+4i}
		\end{subfigure}%
		\begin{subfigure}{.19\textwidth}
			\centering
			\includegraphics[width=\textwidth]{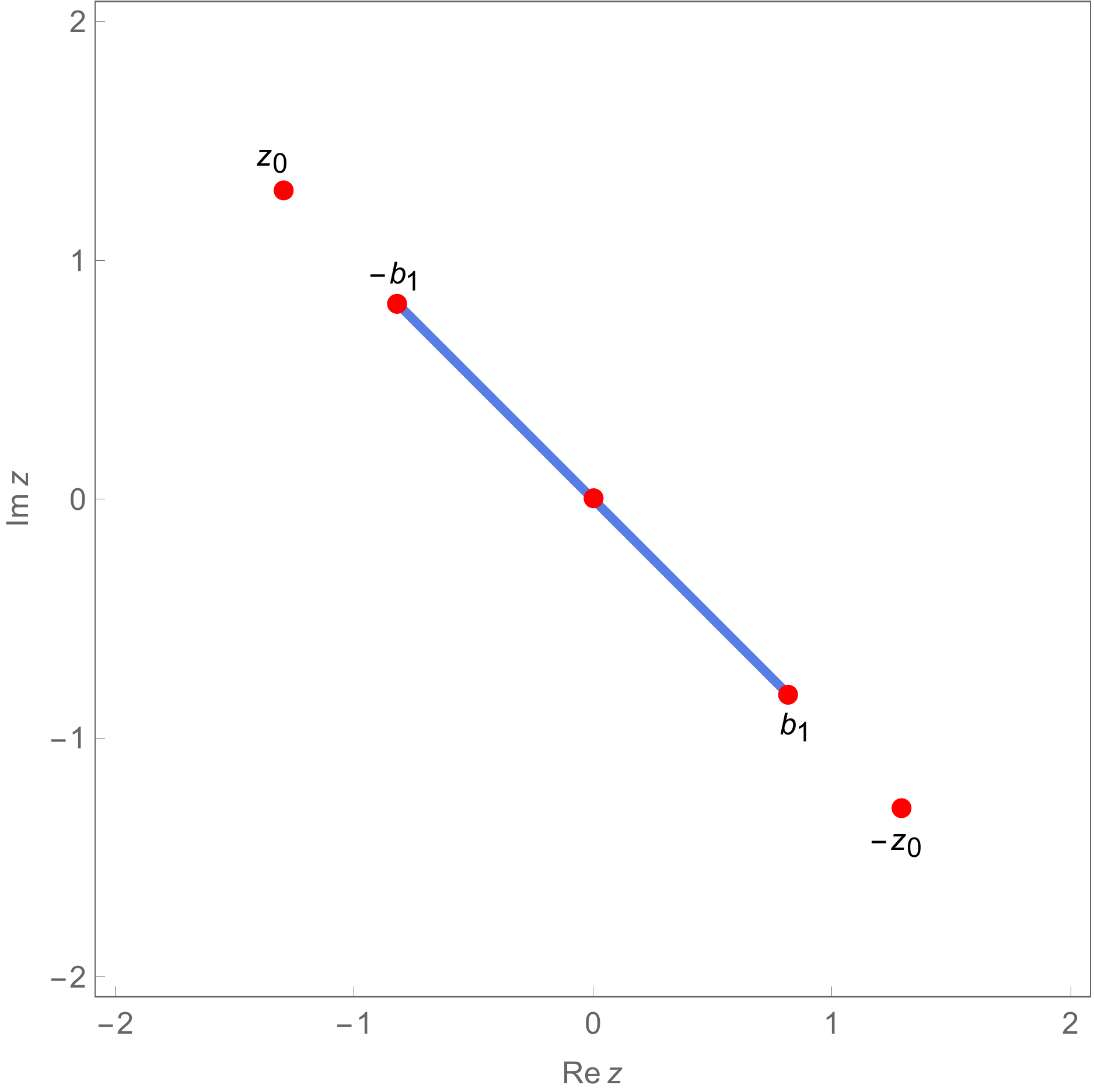}
			\caption{\tiny $\sigma=4\ii$}
			\label{fig1-3birth:sigmaSEQM=4i}    
		\end{subfigure}
		\begin{subfigure}{.19\textwidth}
			\centering
			\includegraphics[width=\textwidth]{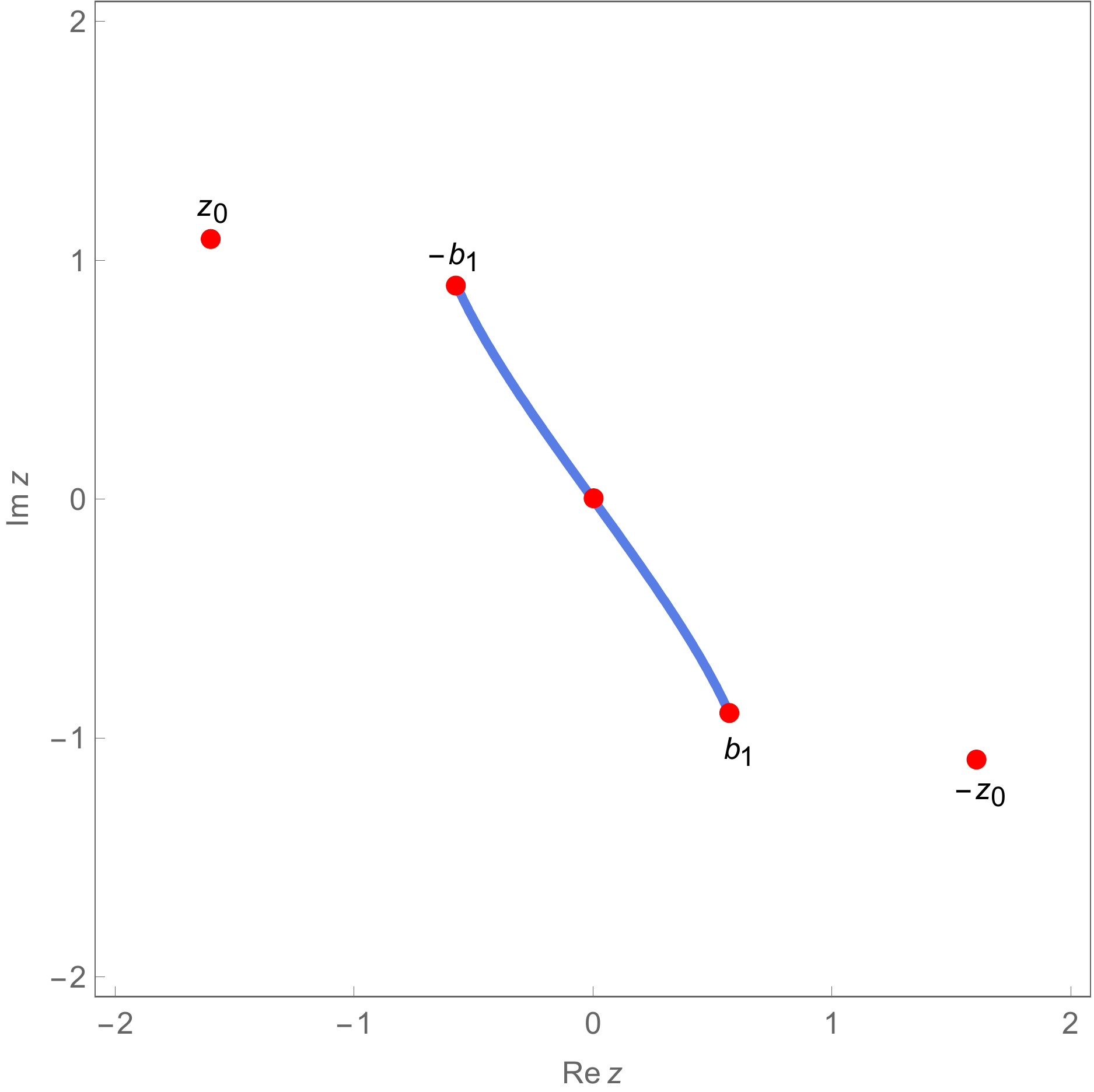}
			\caption{\tiny $\sigma_{\mbox{\tiny cr}}\simeq-1.15+4\ii$}
			\label{fig1-3birth:sigmaSEQM=-1.15+4i}    
		\end{subfigure}%
		\begin{subfigure}{.19\textwidth}
			\centering
			\includegraphics[width=\textwidth]{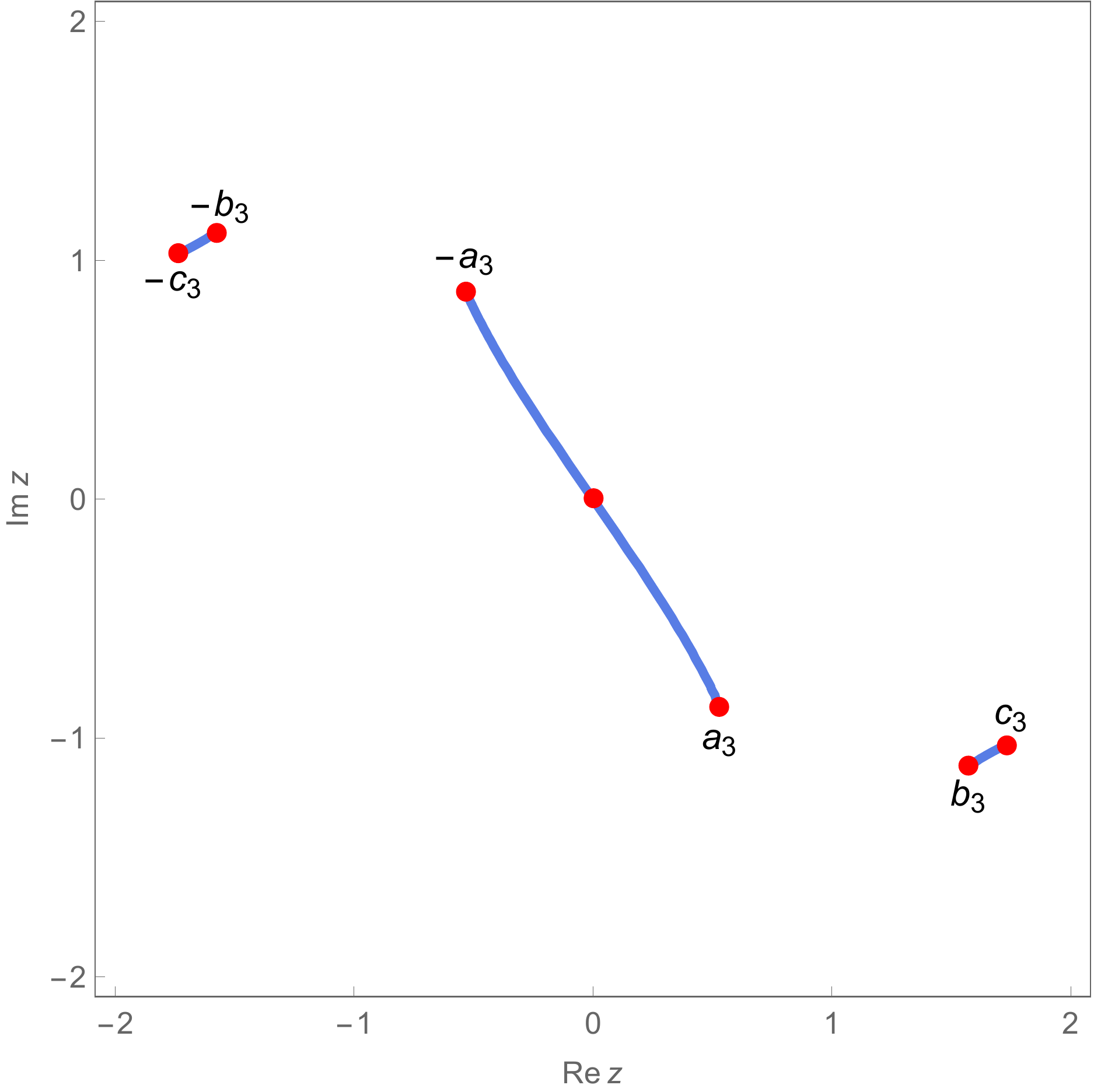}
			\caption{\tiny $\sigma=-1.35+4\ii$}
			\label{fig1-3birth:sigmaSEQM=-1.35+4i}    
		\end{subfigure}
		\begin{subfigure}{.19\textwidth}
			\centering
			\includegraphics[width=\textwidth]{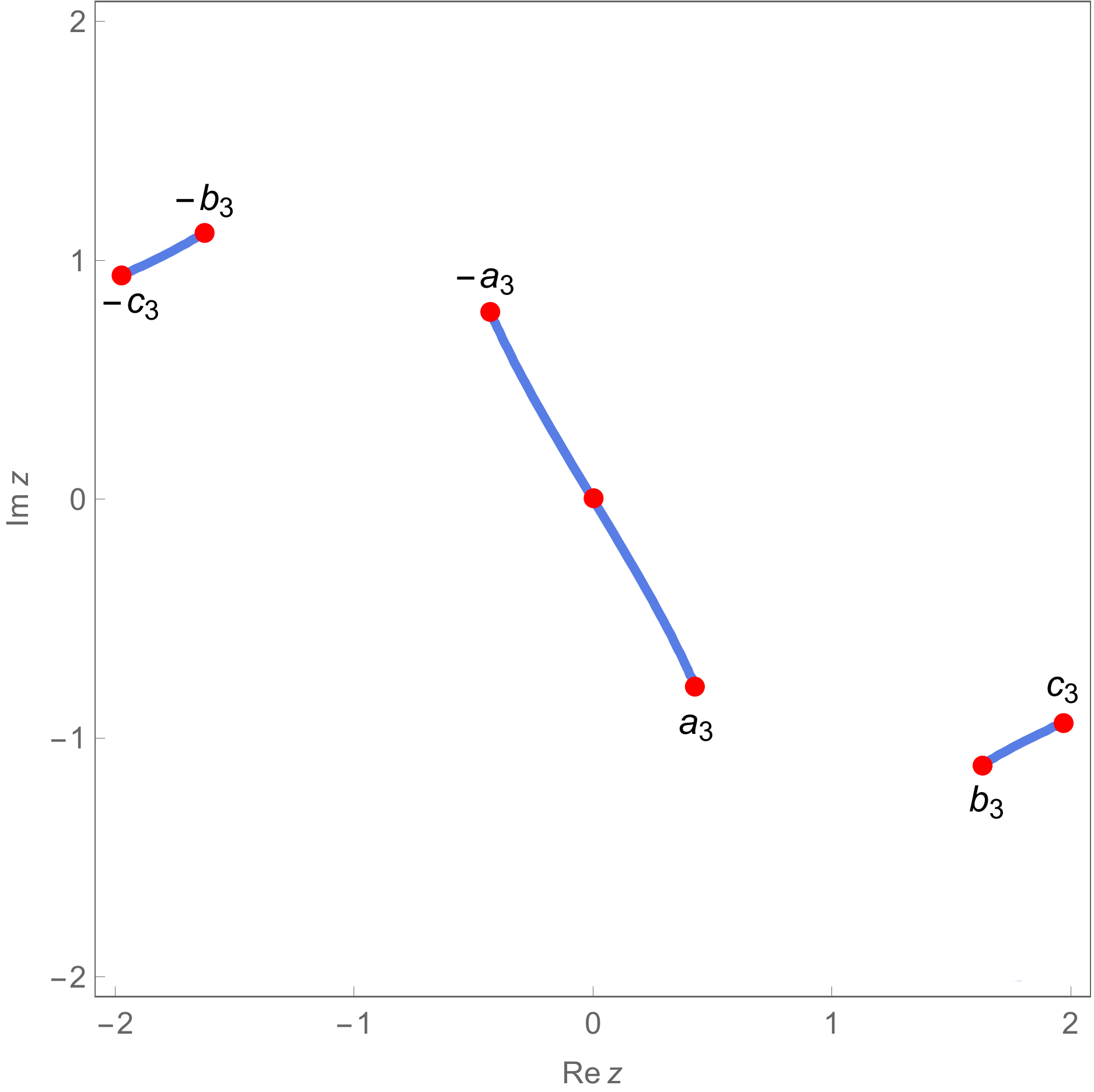}
			\caption{\tiny $\sigma=-2+4\ii$}
			\label{fig1-3birth:sigmaSEQM=-2+4i} 
		\end{subfigure}
		\caption{Snapshots of the continuous evolution of the support of the equilibrium measure in transition from the one-cut into the three-cut regime via birth of two symmetric cuts at $\pm z_0(\sigma_{\mbox{cr}})$, for some $\sigma_{\mbox{cr}} \in \ga_3$, See Figure \ref{fig:Phase Diagram}. Also, see Figure \ref{fig:legend1} to see the respective location of points with respect to the critical contours in the $\sigma$-plane.}
		\label{fig:1-3birthSEQM}
	\end{figure}

	\begin{figure}[h]
		\centering
		\begin{subfigure}{0.19\textwidth}
			\centering
			\includegraphics[width=\textwidth]{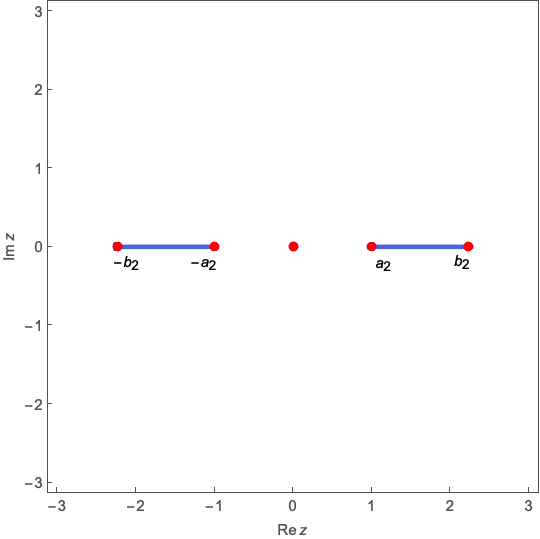}
			\caption{\tiny $\sigma=-3$}
			\label{fig2-3birthSEQM:sigma=-3}
		\end{subfigure}%
		\begin{subfigure}{.19\textwidth}
			\centering
			\includegraphics[width=\textwidth]{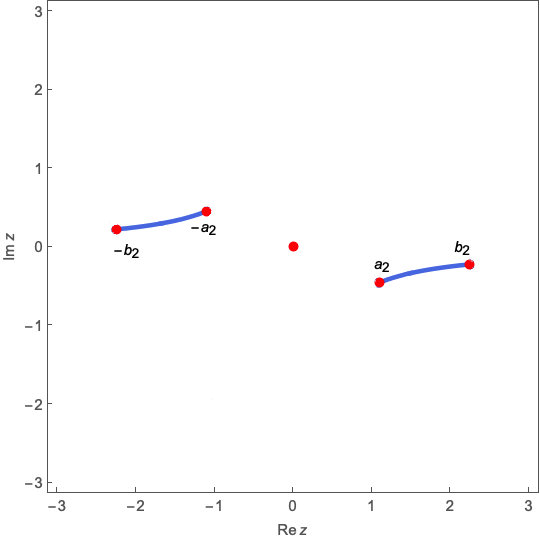}
			\caption{\tiny $\sigma=-3+\ii$}
			\label{fig2-3birthSEQM:sigma=-3+i}    
		\end{subfigure}
		\begin{subfigure}{.19\textwidth}
			\centering
			\includegraphics[width=\textwidth]{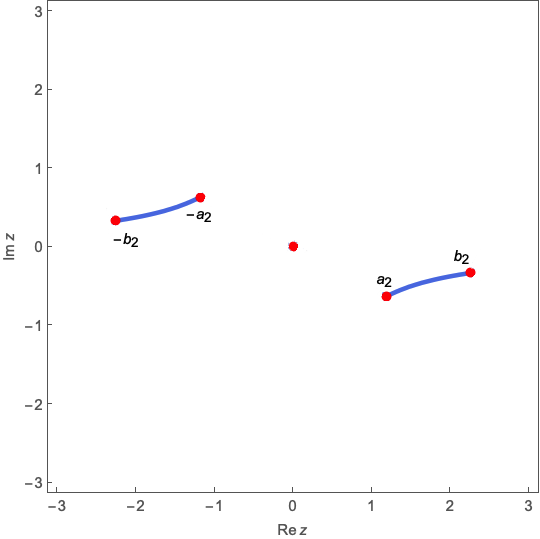}
			\caption{\tiny $\sigma_{\mbox{\tiny cr}}=-3+1.5025\ii$}
			\label{fig2-3birthSEQM:sigma=-3+1.5i}    
		\end{subfigure}%
		\begin{subfigure}{.19\textwidth}
			\centering
			\includegraphics[width=\textwidth]{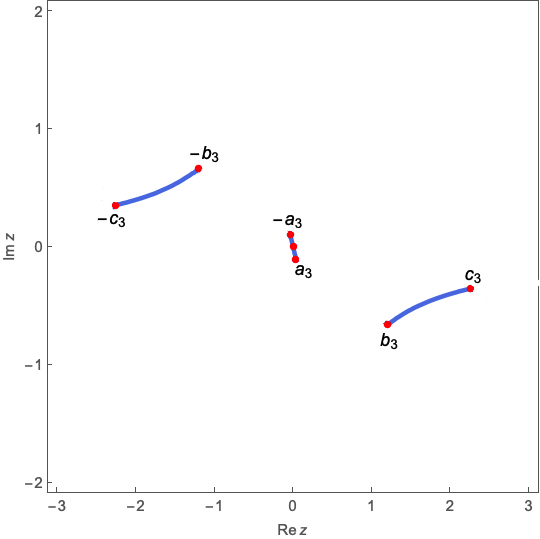}
			\caption{\tiny $\sigma=-3+1.6\ii$}
			\label{fig2-3birthSEQM:sigma=-3+1.6i}    
		\end{subfigure}
		\begin{subfigure}{.19\textwidth}
			\centering
			\includegraphics[width=\textwidth]{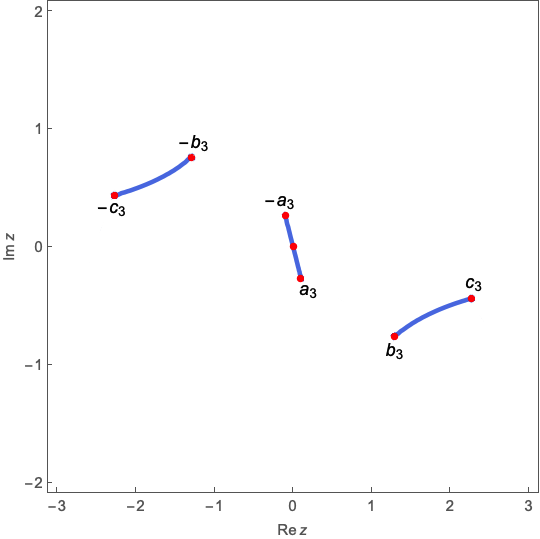}
			\caption{\tiny $\sigma=-3+2\ii$}
			\label{fig2-3birthSEQM:sigma=-3+2i}    
		\end{subfigure}
		\caption{Snapshots of the continuous evolution of the support of the equilibrium measure in transition from the two-cut into the three-cut regime: at a critical value $\sigma_{\mbox{\tiny cr}} \in \ga_5$ (see Figure \ref{fig:Phase Diagram}) a cut is about to be born at the origin yielding a system of three symmetric cuts with respect to the origin.}
		\label{fig:2-3birthSEQM}
	\end{figure}

	\section{Phase Diagram in the $\sigma$-plane and Auxiliary Quadratic Differentials}\label{Sec Aux QD} Similar to the approach taken in \cite{BDY}, to analytically describe the transitions from the one-cut to the three-cut regime and from the two-cut to the three-cut regime we can use the critical trajectories of the associated \textit{auxiliary} quadratic differentials. 
	
	\subsection{One-cut to Three-cut Transition.}\label{subsec 1 to 3 aux QD} In this subsection we search for an analytic description for the values of $\sigma$ such that $z_0(\sigma) \in \mathscr{J}^{(1)}_{\sigma}$, that is $\Re[\Psi(\sigma)]=0$ where
	
	\begin{equation}\label{F 1 to 3}
	\begin{split}
	\Psi(\sigma)& := \eta_1(z_0(\sigma);\sigma) = -\frac{\sigma}{4}\sqrt{\frac{1}{3}\left( -2\sigma -  \sqrt{12+\sigma^2} \right)} \sqrt{-\sqrt{12+\sigma^2}} \\ &  +2\log\left( \frac{\sqrt{\frac{1}{3}\left( -2\sigma -  \sqrt{12+\sigma^2} \right)}+  \sqrt{-\sqrt{12+\sigma^2}}}{\sqrt{\frac{2}{3}\left( -\sigma +  \sqrt{12+\sigma^2} \right)}} \right).
	\end{split}
	\end{equation}
	If we compute
	\begin{equation}\label{1 to 3 QD trial}
	\left[\frac{\dd \Psi}{\dd \sigma}\right]^2=\frac{1}{12}\left(12 + \sigma^2 + 2\sigma \sqrt{12 + \sigma^2}\right),
	\end{equation}
	we do not obtain a meromorphic quadratic differential, which is the preferred object to deal with (as opposed to what we had in \eqref{2 to 3 QD}). However, if we express $\sigma$ and $z_0(\sigma)$ in terms of $b_1 \equiv b_1(\sigma)$ via \eqref{em35}, then a direct calculation shows that in the variable $b_1$ we do obtain a meromorphic quadratic differential:
	\begin{equation}
	\Xi(b_1):=\left[\frac{\dd \Psi}{\dd b_1}  \right]^2 = \frac{(16-b_1^4)(16+3b_1^4)^3}{256b_1^{10}}.
	\end{equation}
	We can make things a bit simpler, as in the variable $\be:=b_1^2$ we arrive at:
	\begin{equation}
	\Xi(\be)= \left[\frac{\dd \Psi}{\dd \be}  \right]^2 = \frac{(16-\be^2)(16+3\be^2)^3}{1024\be^{6}}.
	\end{equation}
	Thus we can express $\Psi$ as \begin{equation}\label{Psi}
	\Psi(\be) = \int^{\be}_{-\frac{4 \ii}{\sqrt{3}}} \sqrt{\Xi(s) } \dd s. 
	\end{equation}
	The initial point of integration is chosen to be $\be=-4 \ii / \sqrt{3}$ as this corresponds to $b_1 = z_0$ (and $\sigma = i \sqrt{12}$) where $\Psi=0$. Therefore, the preimage (under the map $\sigma \mapsto \be$) of the critical trajectories of the quadratic differential $\Xi(\be) \dd \be^2$ \textit{includes}, and as described further below \textit{not} equal to, the set $\{\sigma \ : \ \Re[\eta_1(z_0(\sigma);\sigma)]=0\}$. So one is naturally directed to study the critical trajectories of the auxiliary quadratic differential $\Xi(\be) \dd \be^2$. Notice that it has two simple zeros at $\pm 4$, two zeros of order three at $\pm 4 \ii / \sqrt{3}$, a pole of order six at zero and a pole of order six at infinity (recall \eqref{QD near infinity}). Therefore three critical trajectories emanate from $\be=4$ and $\be=-4$ each, while five critical trajectories emanate from $\be=4 \ii / \sqrt{3}$ and $\be=-4 \ii / \sqrt{3}$ each (Theorem $7.1$ of \cite{Strebel}). Also, there are 4 critical trajectories incident with $\be=0$ and $\be=\infty$ (Theorem $7.4$ of \cite{Strebel}). The local structure of the critical trajectories in a neighborhood of the critical points can be easily found by finding a ray on which $\Xi(\be)\dd \be^2<0$. Locally, the other critical trajectories will be then determined based on how many critical directions are incident with the critical point. For example it is simple to check that $\Xi(\be)\dd \be^2<0$ when $\be = \ep i + 4 \ii / \sqrt{3}$, $\ep>0$. The other four critical directions at $4 \ii / \sqrt{3}$ are now determined by forming equal angles $2\pi/5$ between adjacent critical directions.  Similar analysis gives the local structure in the neighborhood of other critical points. At infinity $\Xi(\be) \dd \be^2 \sim - \frac{27}{1024}\be^2 \dd \be^2 $, thus the integral of its square root behaves like $\frac{3 \ii \sqrt{3}}{64}\be^2$ and thus the four solutions to $\Re[\Psi(\be)]=0$ near infinity respectively have asymptotic angles $0$, $\pi/2$, $\pi$, and $3\pi/2$. Using this for solutions near infinity, and having already determined the local critical structure near finite critical points, the only global structure (connection of critical trajectories) consistent with the Teichm{\"u}ller's lemma is shown in Figure \ref{fig:Crit Traj be plane}. A calculation shows that $\Psi(\be)$ as defined in \eqref{Psi} differs from $\Psi(-\be)$ and from $\Psi(\overline{\be})$ by additive purely imaginary quantities. This explains the symmetry with respect to the origin and the real axis in Figure \ref{fig:Crit Traj be plane}.
	\begin{figure}[b]
		\centering
		\includegraphics[scale=0.35]{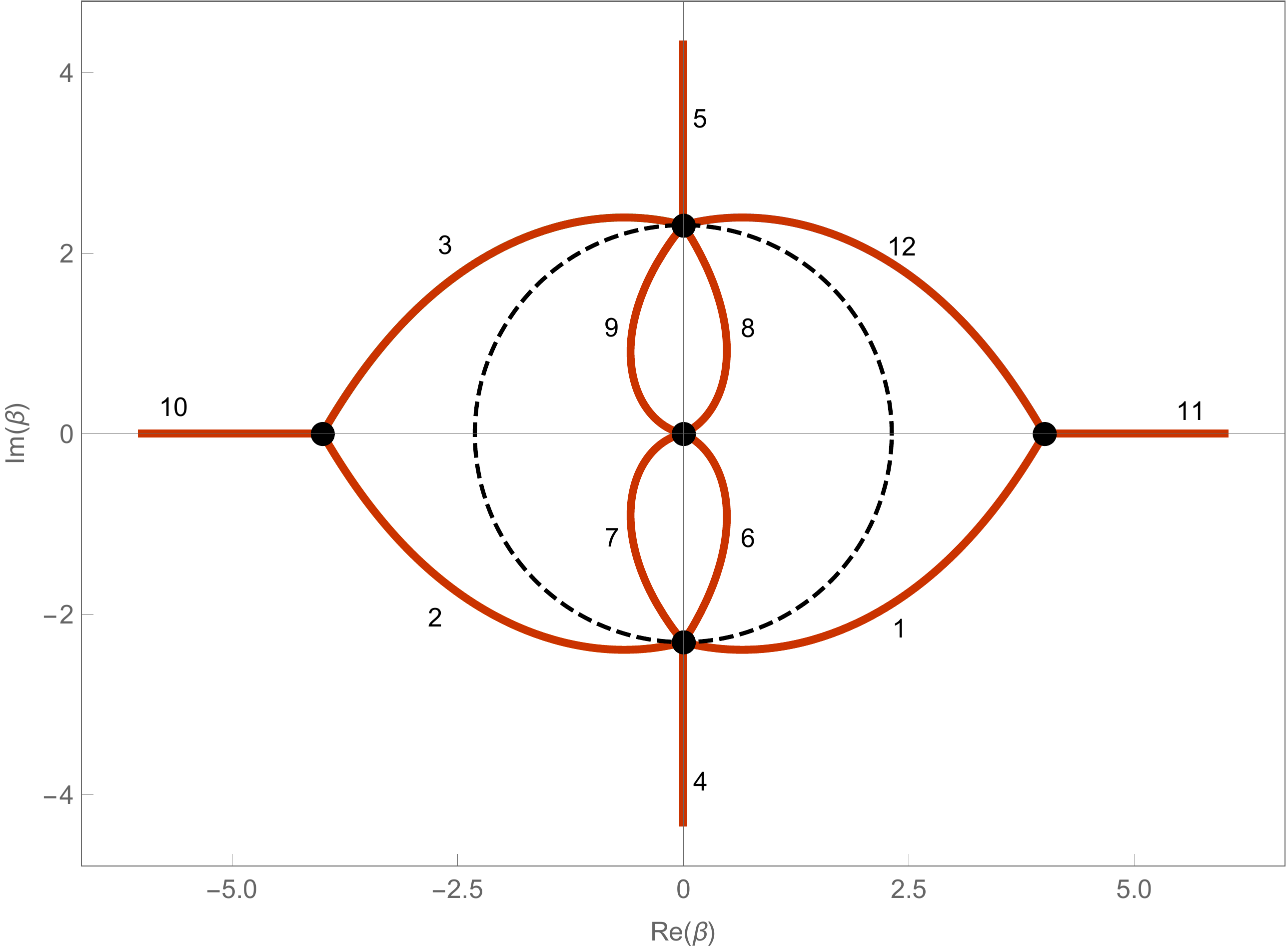}
		\caption{The red lines show the critical graph $\mathcal{T}$ of the auxiliary quadratic differential $ 2^{-10}\be^{-6} (16-\be^2)(16+3\be^2)^3 \dd \be^2 \equiv \Xi(\be)\dd \be^2$ in the $\be$-plane. The black dots show the critical points of $\Xi(\be)\dd \be^2$: simple zeros at $\pm 4$, zeros of order three at $\pm 4 \ii / \sqrt{3}$, pole of order six at zero. The actual critical graph in the $\sigma$-plane corresponding to the transition form the 1-cut regime to the 3-cut regime is a subset of the image of $\mathcal{T}$ under the Joukowsky map $\sigma=-\frac{3}{4}\be+\frac{4}{\be}$ (see Figure \ref{fig:Crit Traj 1to3 sigma plane}) which maps both the interior and the exterior of the circle of radius $4 / \sqrt{3}$ onto the complement of the imaginary line segment in the $\sigma$-plane connecting $-\ii \sqrt{12}$ to $\ii \sqrt{12}$.}
		\label{fig:Crit Traj be plane}
	\end{figure}
	
	From \eqref{em36} we can simply express $z^2_0$ and $\sigma$ in terms of $\beta$ as
	
	\noindent\begin{minipage}{.5\linewidth}
		\begin{alignat}{2}
		&z^2_0 &&= \frac{\be}{4}+\frac{4}{\be}, \label{z0 beta}
		\end{alignat}	
	\end{minipage}	
	\begin{minipage}{.5\linewidth}
		\begin{alignat}{2}
		&\sigma &&= -\frac{3}{4}\be+\frac{4}{\be}. \label{sigma beta} 
		\end{alignat}	
	\end{minipage}
	We observe that the map from the $\be$-plane to the $\sigma$-plane is a Joukowski map which maps both the interior and the exterior of the circle of radius $4 / \sqrt{3}$ onto the complement of the imaginary line segment in the $\sigma$-plane connecting $-\ii \sqrt{12}$ to $\ii \sqrt{12}$. Therefore the image $\widehat{\Sigma}$ of the critical trajectories of $\Xi(\be)\dd \be^2$ in the $\be$-plane under the Joukowsky map $\be \mapsto \sigma$ provides all the candidates for the  1-cut to 3-cut phase transition in the $\sigma$-plane. 
	\begin{figure}[h]
		\centering
		\includegraphics[scale=0.35]{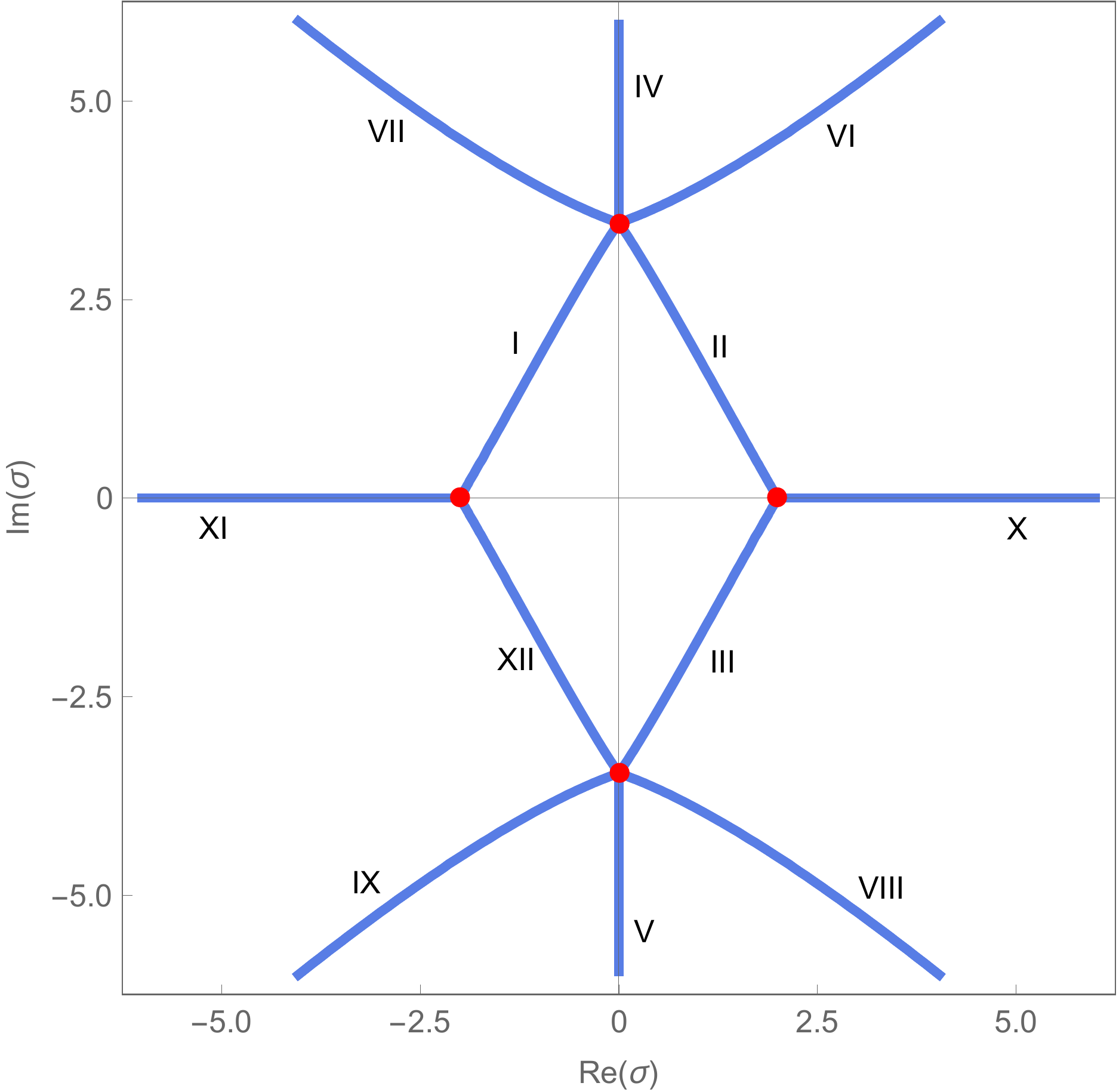}
		\caption{The image $\widehat{\Sigma}$ of the critical graph $\mathcal{T}$ of $2^{-10}\be^{-6} (16-\be^2)(16+3\be^2)^3 \dd \be^2 \equiv \Xi(\be)\dd \be^2$ under the Joukowsky map $\be \mapsto \sigma = -\frac{3}{4}\be+\frac{4}{\be}$. The red dots at $\pm2$ and $\pm \ii \sqrt{12}$ are the images of the critical points of $\Xi(\be)\dd \be^2$. The components I, II, III, IV, V, X, XI, and XII are the images of the parts of $\mathcal{T}$ in the \textit{exterior} of the circle of radius $4 / \sqrt{3}$, while the components VI, VII, VIII, and IX  are the images of the parts of $\mathcal{T}$ in the \textit{interior} of the circle of radius $4 / \sqrt{3}$ (see Figure \ref{fig:Crit Traj be plane}).}
		\label{fig:Crit Traj 1to3 sigma plane}
	\end{figure}
	
	Inverting the Joukowsky map we obtain \begin{equation}\label{be + be -}
	\be^{(\pm)}(\sigma) = \frac{2}{3} \left(-\sigma \pm \sqrt{12+\sigma^2}\right).
	\end{equation} 
	We choose the branch cuts for the square root to be the two rays connecting $ \ii \sqrt{12}$ to $-\infty + \ii \sqrt{12}$ and $- \ii \sqrt{12}$ to $-\infty - \ii \sqrt{12}$, and we fix the branch according to $\arg(\sigma-\ii\sqrt{12})=0$ for $\sigma=x+\ii\sqrt{12}$, and $\arg(\sigma+\ii\sqrt{12})=0$ for $\sigma=x-\ii\sqrt{12}$, $x>0$.  However, recalling \eqref{em37}, our one-cut computations are based on $\be^{(+)}$, not $\be^{(-)}$. Therefore, among the twelve components of $\widehat{\Sigma}$, the actual candidates for 
	1-cut to 3-cut phase transition in the $\sigma$-plane are those which get mapped by $\be^{(+)}$ to the critical trajectories of $\Xi(\be)\dd \be^2$ in the $\be$-plane. By straight-forward calculations we observe that $\be^{(+)}$ does \textit{not} map the components of $\widehat{\Sigma}$  labeled by II, III, IV, V, and X to $\mathcal{T}$, while it does map the components of $\widehat{\Sigma}$ labeled by I, XII, VI, VII, VIII, IX, and XI respectively to the components of $\mathcal{T}$  labeled by $1$, $12$, $6$, $7$, $8$, $9$, and $11$ ( Actually it can be checked that $\be^{(-)}$ maps the components of $\widehat{\Sigma}$ labeled by II, III, IV, V, and X respectively to the components of $\mathcal{T}$  labeled by $2$, $3$, $4$, $5$, and $10$ \label{footnote be -}). This means that the only places in the $\sigma$-plane at which 1-cut to 3-cut phase transition could happen are the components of $\widehat{\Sigma}$ labeled by I, XII, VI, VII, VIII, IX, and XI, see Figure \ref{fig:Crit Traj 1to3 sigma plane: obvious stuff removed}.
	
	\begin{figure}[h]
		\centering
		\begin{subfigure}{0.46\textwidth}
			\includegraphics[scale=0.3]{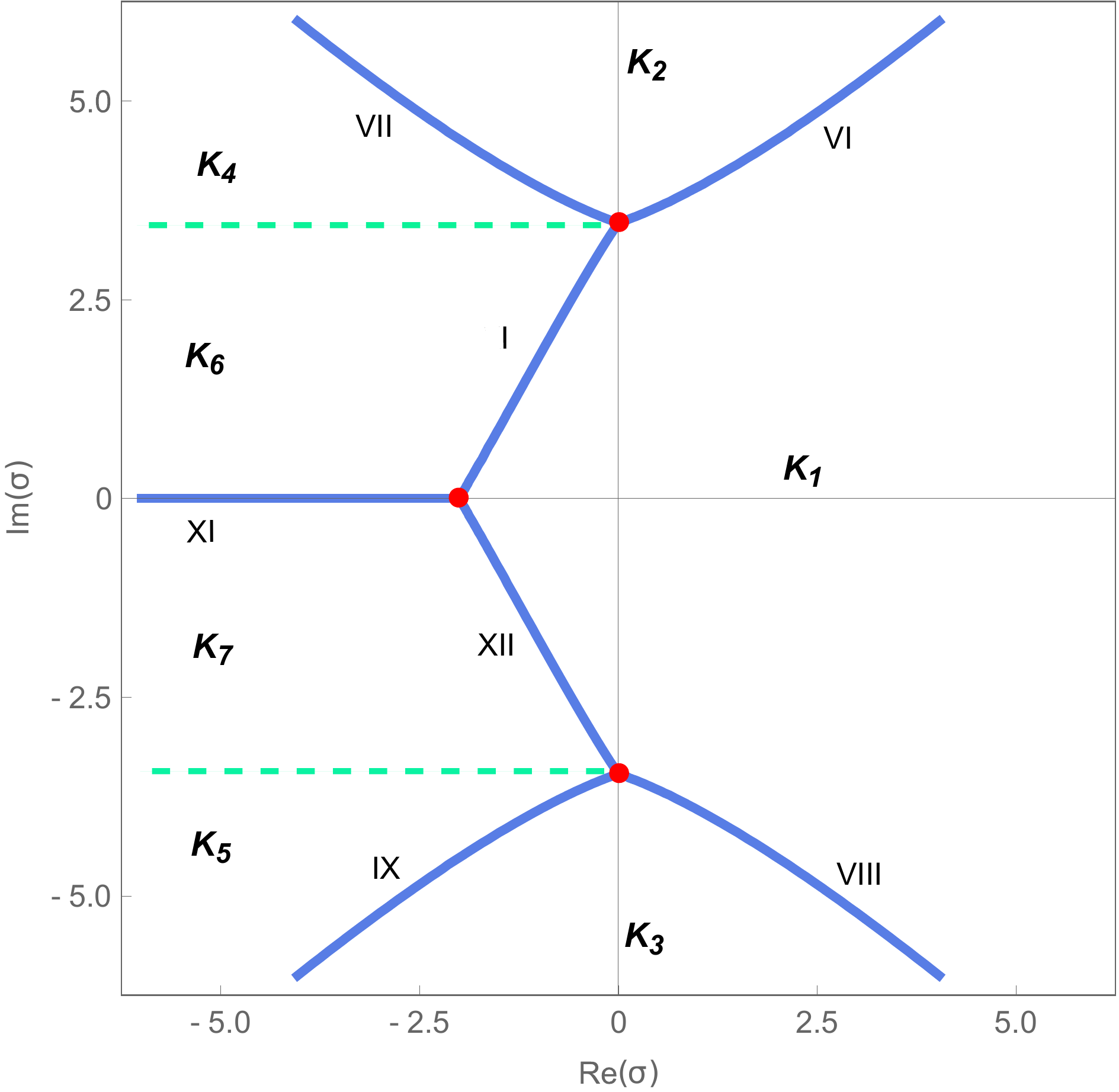}
			\caption{The set $\widehat{\Sigma}$ shown in Figure \ref{fig:Crit Traj 1to3 sigma plane} \textit{without} its components labeled by II, III, IV, V, and X. Notice that $\be^{(+)}$ maps the components of $\widehat{\Sigma}$ labeled by I, XII, VI, VII, VIII, IX, and XI respectively to the components of $\mathcal{T}$  labeled by $1$,$12$,$6$,$7$,$8$,$9$, and $11$, while it does \textit{not} map the components of $\widehat{\Sigma}$  labeled by II, III, IV, V, and X to $\mathcal{T}$ (See Figures \ref{fig:Crit Traj be plane} and \ref{fig:Crit Traj 1to3 sigma plane}). We use the convention that the multicritical points at $\sigma=-2$ and $\sigma=\pm \ii \sqrt{12}$ belong to the critical lines incident to them, for instance $-2,\ii \sqrt{12} \in$ I. The green dashed lines represent the branch cuts $L_{\pm}$ (see Remark \ref{Remark branches}). To rigorously understand the boundaries of the one-cut region we study the sign of $\Re[\eta_1(\pm z_0(\sigma);\sigma)]$ in the infinite regions $\boldsymbol{K}_i$, $i=1,\cdots,5$.}
			\label{fig:Crit Traj 1to3 sigma plane: obvious stuff removed}
		\end{subfigure}
		\hfill
		\begin{subfigure}{0.46\textwidth}
			\includegraphics[scale=0.3]{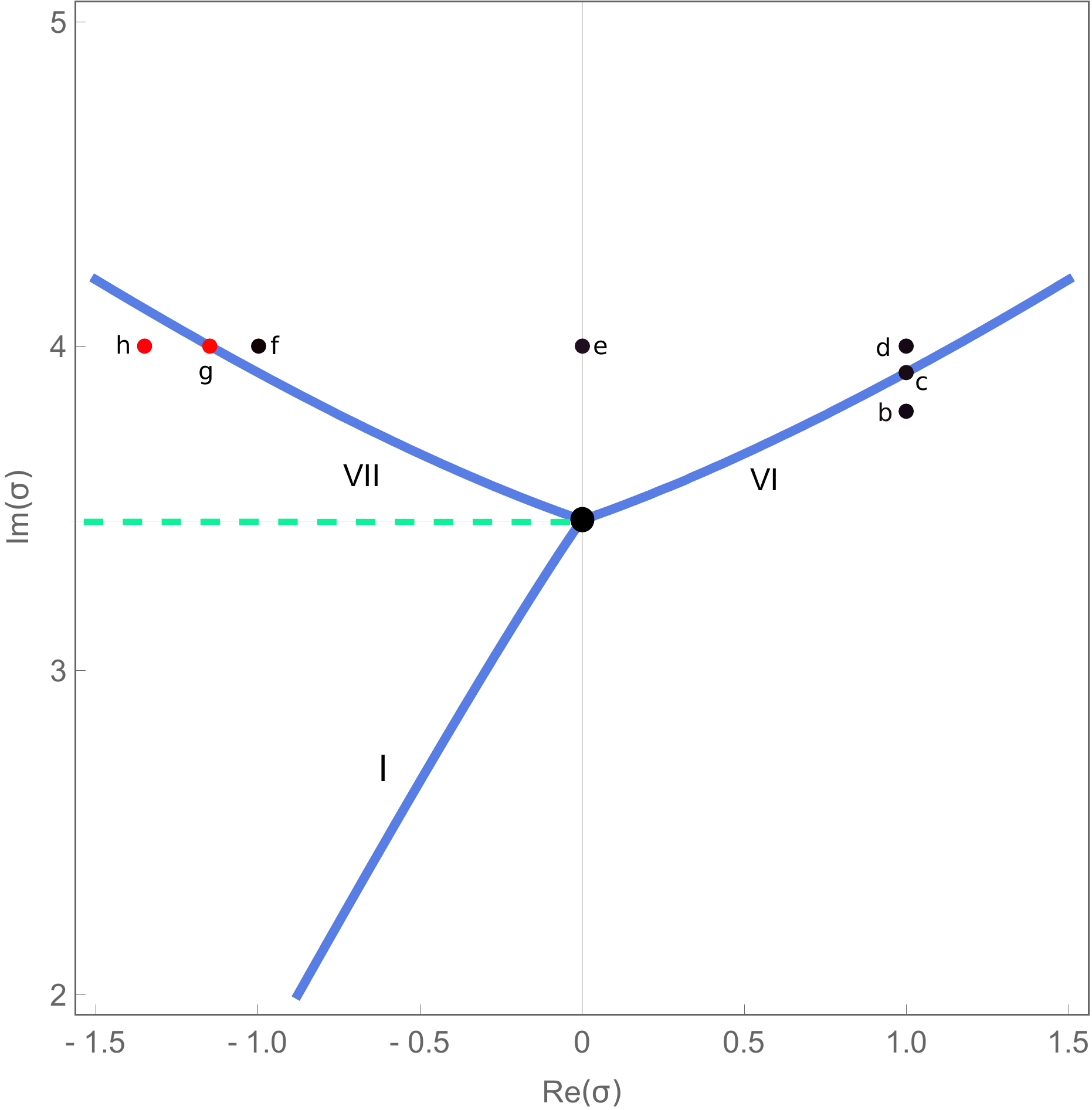}
			\caption{An enlargement of Figure \ref{fig:Crit Traj 1to3 sigma plane: obvious stuff removed} around the point $\ii \sqrt{12}$ which shows the relative location of points B through H considered in Figure \ref{fig:stable and barren} with respect to the lines VI and VII. We recall that $b=1+3.8\ii$, $c\simeq1+3.9187\ii$, $d=1+4\ii$, $e=4\ii$, $f=-1+4\ii$, $g\simeq-1.15+4\ii$, $h=-1.35+4\ii$. As shown in Figure \ref{fig:stable and barren}, at the parameter values $c$ and $g$, the points $\pm z_0(\sigma)$ lie on the critical graph. At the parameter values $b$ and $h$, the points $\pm z_0(\sigma)$ belong to the unstable lands, while at the parameter values $d$, $e$, and $f$, the points $\pm z_0(\sigma)$ belong to stable lands. The alignment of the points $\pm b_1$ and $\pm z_0$ along the bisector of the second and fourth quadrants for $\sigma=\ii y$, with $y>\sqrt{12}$, is in particular shown for point $e$ in Figure \ref{fig:fertile barren 4}.}
			\label{fig:legend1}
		\end{subfigure}
		\label{fig: remaining 1to3}
		\caption{Part (a) shows the remaining candidates (all values of $\sigma$ on the blue curves) for the one-cut to three-cut transitions. Part (b) shows an enlarged picture in a neighborhood of the tri-critical point $\ii \sqrt{12}$.}
	\end{figure}	
	We will later show that for another reason the 1-cut to 3-cut phase transition could not happen along the component labeled by XI, and for yet another reason it can not happen along the components labeled by VI and VIII.
	
	
	
	

	\begin{lemma}\label{Lemma who hits who 1}
		Let $\sigma \in$ \normalfont{I} $\cup$ VI $\cup$ XII $\cup$ VIII \textit{\ and different from $-2$ and $\pm \ii \sqrt{12}$. Then one of the following three possibilities holds:}
		\begin{itemize}
			\item[a)]  $\pm z_0(\sigma) \in J_{\sigma}$
			\item[b)]    $z_0(\sigma) \in \ell_2^{(-b_1(\sigma))}$ and $-z_0(\sigma) \in \ell_2^{(b_1(\sigma))}$
			\item[c)]  $z_0(\sigma) \in \ell_3^{(b_1(\sigma))}$ and $-z_0(\sigma) \in \ell_3^{(-b_1(\sigma))}$.
		\end{itemize}
		\begin{proof}
			This follows from continuous deformations of $z_0(\sigma)$ and $\mathscr{J}^{(1)}_{\sigma}$ with respect to $\sigma$. So we start from some $\sigma_0>-2$ where we know the structure $\mathscr{J}^{(1)}_{\sigma}$, and $\pm z_0(\sigma_0)= \pm \ii y_0$, for some $y_0>0$ (See, e.g. Figure \ref{fig1-3birth:sigma=-1} for $\sigma_0=-1$). If we continuously deform $z_0(\sigma)$ and $\mathscr{J}^{(1)}_{\sigma}$ starting from $\sigma_0$ it is clear that $z_0(\sigma_1)$ can only hit $J_{\sigma_1}$, $\ell_2^{(-b_1(\sigma_1))}$, or $\ell_3^{(b_1(\sigma_1))}$. The three possibilities in the statement of the Lemma now follow from Lemma \ref{Lemma symmetry of J}, more precisely: i) if $z \in \ell_2^{(-b_1(\sigma))}$ then  $-z \in \ell_2^{(b_1(\sigma))}$, ii) if $z \in \ell_3^{(-b_1(\sigma))}$ then  $-z \in \ell_3^{(b_1(\sigma))}$, and iii) if  $z \in J_{\sigma}$ then $-z \in J_{\sigma}$. 
		\end{proof}	 
	\end{lemma}

	\begin{theorem}\label{sigma in I}
		For $\sigma \in \mbox{\normalfont I} \cup \mbox{\normalfont XII}$, it holds that $\pm z_0(\sigma) \in J_{\sigma}$.
	\end{theorem}
	\begin{proof}
		We only prove the theorem for $\sigma \in$ I, as the theorem for $\sigma \in$ XII can be proven identically. Notice that for $\sigma=-2 \in$ I (see the caption of Figure \ref{fig:Crit Traj 1to3 sigma plane: obvious stuff removed}), we indeed have $\pm z_0(-2)=0 \in J_{-2} = [-2,2]$. Obviously for all $\sigma \in$ I, we need to have $z_0(\sigma)$ either belong to $J_{\sigma}$ or to $\mathscr{J}_{\sigma} \setminus J_{\sigma}$. For the sake of arriving at a contradiction, let us assume that for some $\sigma_1 \in$ I,  $z_0(\sigma_1)$ belongs to  $\mathscr{J}_{\sigma} \setminus J_{\sigma}$. Due to continuity in deformations of $z_0(\sigma)$ and   $\mathscr{J}_{\sigma}^{(1)}$, there has to be some intermediate $\sigma_0 \in$ I between $\sigma=-2$ and $\sigma_1$ so that $z_0(\sigma_0)$ simultaneously belongs to $J_{\sigma_0}$ and $\mathscr{J}_{\sigma_0} \setminus J_{\sigma_0}$. But this would lead to a geodesic polygon with two vertices at one of $\pm b_1(\sigma_0)$ and $z_0(\sigma_0)$ which is impossible due to Lemma \ref{Lemma one or two vetex polygons are ruled out}. One gets the same contradiction using the identical argument for $-z_0(\sigma_1)$.
	\end{proof}
	
	\begin{lemma}\label{Lemma K1 K2 K3}
		In $\boldsymbol{K}_1$ we have $\Re[\eta_1(\pm z_0(\sigma);\sigma)]>0$ while in $\boldsymbol{K}_2 \cup \boldsymbol{K}_3$ we have $\Re[\eta_1(\pm z_0(\sigma);\sigma)]<0$.
	\end{lemma}
	
	\begin{proof}
		Consider the region $\boldsymbol{K}_1$ shown in Figure \ref{fig:Crit Traj 1to3 sigma plane: obvious stuff removed}. under the map $\be^{(+)}$ (see \eqref{be + be -}) this region is mapped to the region bounded by the components 1, 12, 6, and 8 shown in Figure \ref{fig:Crit Traj be plane}, which we denote by $K_1$. Also, the region $\boldsymbol{K}_2$ gets mapped by  $\be^{(+)}$ to the interior of the components 6 and 7 of Figure \ref{fig:Crit Traj be plane} , which we denote by $K_2$, while the region $\boldsymbol{K}_3$ is mapped by  $\be^{(+)}$  to the interior of the components 8 and 9 of Figure \ref{fig:Crit Traj be plane}, which we denote by $K_3$.

		Consider the conformal map \eqref{Psi} restricted to $K_1 \cup K_2$. It is straightforward to see that $\Psi$ maps $K_1 \cup K_2$ to the entire $\Psi$-plane, where $K_1$ is either mapped to the right-half or the left-half plane. Indeed, $\Psi$ maps $K_1$ to the right half plane and $K_2$ to the left half plane. To see this it is enough to find a single point $\be_0 \in K_1$ and show that $\Re[\Psi(\be_0)]>0$.  Recall that the set $\{\sigma: \sigma>-2\}$ is inside $\boldsymbol{K}_1$, and its image $(0,4) \subset K_1$. From Lemma  \ref{one cut is non-empty!} we know that $\Re[\Psi(\be)]>0$ for all $\be \in (0,4)$ and thus for all $\be \in K_1$, and consequently $\Re[\Psi(\be)]<0$ for all $\be \in K_2$. Consequently, $\Re\left[ \Psi\left(\sigma\right) \right] >0$ for all $\sigma \in \boldsymbol{K}_1$ and $\Re\left[ \Psi\left(\sigma\right) \right] <0$ for all $\sigma \in \boldsymbol{K}_2$. 
		
		Similarly, by considering the conformal map 
		\begin{equation*}\label{Psi22}
		\Psi_*(\be) = \int^{\be}_{\frac{4 \ii}{\sqrt{3}}} \sqrt{\Xi(s) } \dd s,
		\end{equation*}
		restricted to $K_1 \cup K_3$, we can show that $\Re\left[ \Psi\left(\sigma\right) \right] <0$ for all $\sigma \in \boldsymbol{K}_3$. So we have justified that when $\sigma$ passes from $\boldsymbol{K}_1$ to $\boldsymbol{K}_2$  (resp. $\boldsymbol{K}_3$) through VI (resp. VIII), the function $\Re\left[ \Psi\left(\sigma\right) \right] \equiv \Re\left[ \eta_1\left(z_0(\sigma);\sigma\right)\right]$ changes sign from positive to negative, that is $z_0(\sigma)$ moves from an unstable land to a stable land. We have the same conclusion for $-z_0(\sigma)$ due to \eqref{eta z and eta -z}.
	\end{proof}

	\begin{theorem}\label{Lemma Who Hits Who}
		It holds that
		
		\begin{itemize}
			\item $-z_0(\sigma) \in \ell_2^{(b_1)}$ and $z_0(\sigma) \in \ell_2^{(-b_1)}$ for $\sigma \in \mbox{\normalfont VI} $, 
			\item $z_0(\sigma) \in \ell_3^{(b_1)}$ and $-z_0(\sigma) \in \ell_3^{(-b_1)}$, for $\sigma \in \mbox{\normalfont VIII}$,
			\item $-z_0(\sigma) \in \ell_3^{(b_1)}$ and $z_0(\sigma) \in \ell_3^{(-b_1)}$, for $\sigma \in \mbox{\normalfont VII},$ and
			\item $z_0(\sigma) \in \ell_2^{(b_1)}$ and $-z_0(\sigma) \in \ell_2^{(-b_1)}$, for $\sigma \in \mbox{\normalfont IX} $.
		\end{itemize}

	\end{theorem}
	
	\begin{proof}
		We only prove this for $\sigma \in$ VI and $\sigma \in$ VII, as the proof for $\sigma \in$ VIII and $\sigma \in$ IX can be done identically.	We first show this locally in an $\ep$-neighborhood $D_{\ep}(\ii \sqrt{12})$ of $\ii \sqrt{12}$, for small enough $\ep>0$. Notice that as $\sigma$ approaches to $\ii \sqrt{12}$, $\mp z_0(\sigma)$ approaches to $\pm b_1(\sigma)$. So we consider the asymptotics of  $\eta_1(-z_0(\sigma);\sigma)$ as $\sigma$ approaches to $\ii \sqrt{12}$. We indeed find that the order of vanishing is $5/4$: 
		\begin{equation}\label{asymp eta -z0}
		\eta_1(-z_0(\sigma);\sigma) = \frac{4\sqrt{2}}{5}3^{-1/8}e^{3 \pi \ii/8} \left(\sigma - \ii \sqrt{12}\right)^{5/4} \left(1+O\left(\left(\sigma - \ii \sqrt{12}\right)^{1/4}\right)\right). 
		\end{equation}
		From the properties of the auxiliary quadratic differential we know that the local angle between the components labeled by $1$ and $6$ in Figure \ref{fig:Crit Traj be plane} is $2 \pi/5$. The map \eqref{sigma beta} is not conformal at $\be=-4 \ii/ \sqrt{3}$, indeed
		\[ \sigma(-\frac{4\ii}{\sqrt{3}})=\ii \sqrt{12}, \qquad \sigma^{\prime}(-\frac{4\ii}{\sqrt{3}})=0, \qquad \sigma^{\prime \prime}(-\frac{4\ii}{\sqrt{3}})=-\frac{3\sqrt{3}}{8}\ii. \]
		This means that the local angle at $\ii \sqrt{12}$ between the images I and VI (see Figure \ref{fig:Crit Traj 1to3 sigma plane: obvious stuff removed}) of $1$ and $6$ is $4 \pi/5$. This analysis also gives us the local angles $\theta_1, \theta_6, \theta_{7}$ respectively of components I, VI, and VII made with the ray $x+\ii \sqrt{12}$, $x>0$: \[ \theta_1 = \frac{-7\pi}{10}, \qquad \theta_6 = \frac{\pi}{10}, \qquad \theta_7 = \frac{9\pi}{10}, \qquad \mbox{where} \qquad  \sigma - \ii \sqrt{12} = \rho e^{\ii \theta}, \qquad -\pi<\theta<\pi.  \]
		We can now notice that
		\begin{itemize}
			\item If $\sigma \in$ I, the leading order approximation of $\eta_1(-z_0(\sigma);\sigma)$ given by \eqref{asymp eta -z0},  is purely imaginary $\ii y_1(\rho)$, with  $y_1(\rho)<0$,as $3\pi/8+5\theta_1/4=-\pi/2$, 
			\item If $\sigma \in$ VI,  the leading order approximation of $\eta_1(-z_0(\sigma);\sigma)$ given by \eqref{asymp eta -z0},  is purely imaginary $\ii y_6(\rho)$, with  $y_6(\rho)>0$, as $3\pi/8+5\theta_6/4=\pi/2$,  and 
			\item If $\sigma \in$ VII,  the leading order approximation of $\eta_1(-z_0(\sigma);\sigma)$ given by \eqref{asymp eta -z0},  is purely imaginary $\ii y_7(\rho)$, with  $y_7(\rho)<0$, as $3\pi/8+5\theta_7/4=3\pi/2$. 	
		\end{itemize}
		
		This means that as $\sigma$ approaches to VI, from $\boldsymbol{K}_1 \cap D_{\ep}(\ii \sqrt{12})$, $-z_0(\sigma)$ must approach $\ell^{(b_1(\sigma))}_2$ from the right (We orient $\ell^{(b_1(\sigma))}_2$,  $\ell^{(b_1(\sigma))}_3$, $\ell^{(-b_1(\sigma))}_2$, and  $\ell^{(-b_1(\sigma))}_3$ in the outward direction as they emanate from $\pm b_1(\sigma)$), where we know that $\Im(\eta(z;\sigma))>0$ (See Figures \ref{fig:Conf map 1} and \ref{fig:Conf map 2}), and as $\sigma$ approaches to VII, from $\boldsymbol{K}_2 \cap D_{\ep}(\ii \sqrt{12})$, $-z_0(\sigma)$ must approach from the right to $\ell^{(b_1(\sigma))}_3$ where we know that $\Im(\eta(z;\sigma))<0$ by the identical conformal mapping arguments used to draw the Figures \ref{fig:Conf map 1} and \ref{fig:Conf map 2}. Notice in the latter case, $-z_0(\sigma)$ can not approach the support $J_{\sigma}$ where we also know $\Im(\eta(z;\sigma))<0$. This is because it has to do so via the unstable lands, while in $\boldsymbol{K}_2$ we know that $\Re[\eta(-z_0(\sigma);\sigma)]<0$. The symmetry implies that as $\sigma$ approaches to VI, from $\boldsymbol{K}_1\cap D_{\ep}(\ii \sqrt{12})$, $z_0(\sigma)$ must approach $\ell^{(-b_1(\sigma))}_2$ from the right, and as $\sigma$ approaches to VII, from $\boldsymbol{K}_2 \cap D_{\ep}(\ii \sqrt{12})$, $z_0(\sigma)$ must approach from the right to $\ell^{(-b_1(\sigma))}_3$.
		
		Now we extend this local result to the entirety of VI and VII using the same argument presented in Theorem \ref{sigma in I}.

	\end{proof}

	Due to Lemma \ref{Lemma Who Hits Who}, the values of  $\sigma \in \mbox{\normalfont VI} \cup \mbox{\normalfont VIII}$ do not belong to $\mathcal{O}^*_1$. However, in the following Theorem we show that they do belong to the larger set  $\mathcal{O}_1$ (recall the Definitions \ref{Def one cut sigma} and \ref{Def one cut sigma minus fake transition}).
	\begin{theorem}\label{THM Fake Transition}
		All $\sigma \in \mbox{\normalfont VI} \cup \mbox{\normalfont VIII}$ belong to $\mathcal{O}_1$. 
	\end{theorem}
	\begin{proof}
		By Theorem \ref{Lemma Who Hits Who}, at $\sigma_{*} \in$ VI, we have $-z_0(\sigma_*) \in \ell_2^{(b_1)}$ and $z_0(\sigma_*) \in \ell_2^{(-b_1)}$. By Lemma \ref{humps connect at criticality}, there are no disconnected components for the critical graph and one has four critical trajectories incident at right angles at both $\pm z_0(\sigma_*)$. Among these four critical trajectories, two must come from the two legs $\mathcal{L}^{(\sigma_*)}_{2,l}$ and $\mathcal{L}^{(\sigma_*)}_{2,r}$ of $\mathcal{L}^{(\sigma_*)}_2$ which make a $\pi/2$ angle with each other at $-z_0(\sigma_*)$ and approach to infinity respectively along the directions $-5 \pi/8$ and $-3 \pi/8$, while the other two must come from $\ell_2^{(b_1)}$ folding at a $\pi/2$ angle into a \textit{short} critical trajectory $\hat{\ell}^{(b_1(\sigma_*))}_2$ (connecting $b_1(\sigma_*)$ to $-z_0(\sigma_*)$) and another component $\check{\ell}^{(b_1(\sigma_*))}_2$ connecting $z_0(\sigma_*)$ to infinity along the angle $-\pi/8$.	Notice that $\ell_3^{(b_1)}$ must still approach to infinity at the angle $\pi/8$ when $\sigma_* \in$ VI, due to continuity of deformations and that it has not been hit by $\pm z_0(\sigma_*)$. Thus, when $\sigma_* \in$ VI, one still has the region $\Om^{(\sigma_*)}_1$ which encompasses $(M(\sigma_*),\infty)$ for some $M(\sigma_*)>0$ (See Figure \ref{fig:stable barren critical}). 
		
		Notice that there is still a single connection from  $-b_1(\sigma_*)$ to $b_1(\sigma_*)$ to avoid having too many solutions at $\infty$. This proves that any $\sigma \in$ VI belongs to $\mathcal{O}_1$. An identical argument shows that  any $\sigma \in$ VIII belongs to $\mathcal{O}_1$.
	\end{proof}
	
	\begin{theorem}\label{THMK2K3}
		The regions $\boldsymbol{K}_2$ and $\boldsymbol{K}_3$ shown in Figure \ref{fig:Crit Traj 1to3 sigma plane: obvious stuff removed}  both belong to $\mathcal{O}^*_1$.
	\end{theorem}
	\begin{proof}
		We only prove this for $ \boldsymbol{K}_2$ as the proof for $\boldsymbol{K}_3$ is identical.	As $\sigma$ moves from $\sigma_* \in$ VI to some  $\sigma_1 \in \boldsymbol{K}_2$, $\pm z_0(\sigma_1)$ must lie in $\sigma_1$-stable lands by Lemma \ref{Lemma K1 K2 K3}. Recalling Figure \ref{fig:stable barren critical} for $\sigma_* \in$ VI, this is only possible if at the onset of the entrance of $z_0$ into the stable lands, $\mathcal{L}^{(\sigma_*)}_{2,l}$ and $\hat{\ell}^{(b_1(\sigma_*))}_2$ form the \textit{new} $\ell^{(b_1)}_2$ and $\mathcal{L}^{(\sigma_*)}_{2,r}$ and $\check{\ell}^{(b_1(\sigma_*))}_2$ form the \textit{new} hump $\mathcal{L}^{(4)}_{\sigma}$ (these notations are introduced partially in the proof of Theorem \ref{THM Fake Transition} and in the statement of Theorem \ref{lemma structure 2}), as the other possibility where $\hat{\ell}^{(b_1(\sigma_*))}_2$ and $\check{\ell}^{(b_1(\sigma_*))}_2$ form the new  $\ell^{(b_1)}_2$ is not allowed by the Teichm\"uller's lemma regardless of which stable land $z_0$ enters. This means that one indeed has Figure \ref{fig:stable barren 1 4} once $\sigma$ moves from $\sigma_* \in$ VI to some  $\sigma_1 \in \boldsymbol{K}_2$. Since at $\sigma_1$, $\ell^{(b_1)}_3$ still approaches to $\infty$ along the $\pi/8$ direction, $\ell^{(b_1)}_2$ approaches to $\infty$ along the $-5\pi/8$ direction. At $\sigma_1$ one indeed has a contour $\Ga_{\sigma_1}(b_1(\sigma_1), \infty)$ entirely in the stable lands which encompasses $(M(\sigma_1),\infty)$ for some $M(\sigma_1)>0$ (See the orange dashed lines in Figure \ref{fig:stable barren 1 4}). One always has this connection to infinty as long as $z_0(\sigma)$ does not hit $\ell^{(b_1)}_3$ which could block this access to the positive real axis (See Figures \ref{fig:fertile barren 4} and \ref{fig:stable barren -1 4}). But for all $\sigma \in \boldsymbol{K}_2$ this does not happen which finishes the proof that $\boldsymbol{K}_2 \subset \mathcal{O}^*_1$.
	\end{proof}

	\begin{theorem}\label{THM VII IX}
		The lines VII and IX form part of the boundaries of the one-cut region. More precisely, for all $\sigma \in$ VII $\cup$ IX, and all $\sigma \in \boldsymbol{K}_4 \cup \boldsymbol{K}_5$ the one-cut definition does not hold.  
	\end{theorem}
	\begin{proof}
		We only provide the proof for $\sigma \in $ VII $\cup \boldsymbol{K}_4$, as the proof for $\sigma \in $ IX $\cup \boldsymbol{K}_5$ is exactly identical.
		Recall that by $\boldsymbol{K}_4$, we denote the infinite region in the $\sigma$-plane bounded by $L_+$ and VII (see Remark \ref{Remark branches} and Figure \ref{fig:Crit Traj 1to3 sigma plane: obvious stuff removed}). At $\sigma_* \in$ VII, by Theorem \ref{Lemma Who Hits Who}, we know that $\pm z_0(\sigma_*) \in \ell^{(\mp b_1(\sigma_*))}_3$, and by a similar reasoning to that provided in the proof of Theorem \ref{THM Fake Transition}, we know that the structure of the stable and unstable lands are as depicted in Figure \ref{fig:stable barren -1.15 4}. This shows that no $\sigma \in$ VII belongs to $\mathcal{O}_1$ since the third requirement of Definition \ref{Def one cut sigma} can not be met. 
		
		We denote the four critical trajectories incident with $-z_0(\sigma_*)$, as
		
		\begin{itemize}
			\item  $\mathcal{L}^{(\sigma_*)}_{4,l}$ and $\mathcal{L}^{(\sigma_*)}_{4,r}$ obtained from folding of $\mathcal{L}^{(\sigma_0)}_{4}$ in two perpendicular components at $-z_0(\sigma_*)$, in the limiting process $\boldsymbol{K}_2 \ni \sigma_0 \to \sigma_* \in$ VII (for the notation $\mathcal{L}_4$ recall Lemma \ref{lemma structure 2}),
			\item  the short critical trajectory $\hat{\ell}^{(b_1(\sigma_*))}_3$ (connecting $b_1(\sigma_*)$ to $-z_0(\sigma_*)$), and
			$\check{\ell}^{(b_1(\sigma_*))}_3$ connecting $z_0(\sigma_*)$ to infinity along the angle $\pi/8$, which are obtained from folding of $\ell^{(b_1(\sigma_0))}_3$ in two perpendicular components at $-z_0(\sigma_*)$, in the limiting process $\boldsymbol{K}_2 \ni \sigma_0 \to \sigma_* \in$ VII.
		\end{itemize}
		By a similar argument to that shown in the proof of Lemma \ref{Lemma K1 K2 K3}, we can show that $\Re[\eta_1(\pm z_0(\sigma);\sigma)]>0$ for all $\sigma \in \boldsymbol{K}_4$.  In other words, as $\sigma$ moves from $\sigma_* \in$ VII to some  $\sigma_1 \in \boldsymbol{K}_4$, $\pm z_0(\sigma_1)$ must lie in $\sigma_1$-unstable lands. This is only possible if at the onset of the entrance of $z_0$ into the unstable lands, $\mathcal{L}^{(\sigma_*)}_{4,l}$ and $\hat{\ell}^{(b_1(\sigma_*))}_3$ together form the \textit{new} $\ell^{(b_1)}_3$ and $\mathcal{L}^{(\sigma_*)}_{4,r}$ and $\check{\ell}^{(b_1(\sigma_*))}_3$ together form the \textit{new} hump which provides the necessary solutions at $\infty$. Notice that the other possibilities lead to contradiction with Teichm\"uller's lemma, regardless of which unstable land $z_0$ enters.  This means that one indeed has the Figure \ref{fig:stable barren -1.35 41}, which proves that all $\sigma \in \boldsymbol{K}_4$ can not belong to $\mathcal{O}_1$ as the third requirement of Definition \ref{Def one cut sigma} can not be met.

	\end{proof}

	\begin{theorem}\label{THM I XII}
		The lines I and XII form part of the boundaries of the one-cut region. More precisely, for all $\sigma \in$ I $\cup$ XII, and all $\sigma \in \boldsymbol{K}_6 \cup \boldsymbol{K}_7$ the one-cut definition does not hold.  
	\end{theorem}
	\begin{proof}
		We only provide the proof for $\sigma \in $ I $\cup \boldsymbol{K}_6$, as the proof for $\sigma \in $ XII $\cup \boldsymbol{K}_7$ is exactly identical.	First notice that on I the second requirement of Definition \ref{Def one cut sigma} can not be met due to Theorem \ref{sigma in I}. 
		
		Now we show that any point $\sigma \in \boldsymbol{K}_6$ does not belong to $\mathcal{O}^*_1$. For the sake of arriving at a contradiction, assume that there exists a  $\sigma_1 \in \boldsymbol{K}_6$ which belongs to $\mathcal{O}^*_1$. We can now deform the branch-cut $L_+$ (see Remark \ref{Remark branches}) so that the point $\sigma_1$ and the component VII lie on the same side of $L_+$. But since $\mathcal{O}^*_1$ is open, this means that all $\sigma$ bounded by the deformed branch cut $L_+$, and the component VII must belong to $\mathcal{O}^*_1$. This contradicts Theorem \ref{THM VII IX}.
	\end{proof}
	
	Theorems \ref{THM Fake Transition}, \ref{THMK2K3},  \ref{THM VII IX}, \ref{THM I XII} and the fact that $\boldsymbol{K}_1 \subset \mathcal{O}^*_1$ can be formulated as the following Theorem.
	
	\begin{theorem}\label{THM ONE CUT}
		The one-cut region is the region labeled so in Figure \ref{fig:Phase Diagram}.
	\end{theorem}
	This characterization, immediately implies the openness on the one-cut region.
	\begin{corollary}\label{O1 is open}
		The set $\mathcal{O}_1$ is open.
	\end{corollary}

	\subsection{Two-cut to Three-cut Transition.}
	Due to the symmetry with respect to the origin, the transition from the two-cut regime to the three-cut regime could only occur through birth of a cut at the origin. Define \begin{equation}\label{Function : phase transition 2 to 3}
	\Phi(\sigma) := \eta_2(0;\sigma) = -\frac{\sigma \sqrt{\sigma^2-4}}{4}   + 
	\log\left[ \frac{ \sigma + \sqrt{\sigma^2-4}}{2}  \right].
	\end{equation} The values of $\sigma$ for which such a transition takes place are those at which the real part of \eqref{Function : phase transition 2 to 3} vanishes. A calculation shows
	\begin{equation}\label{2 to 3 QD}
	\Upsilon(\sigma):= \left[\frac{\dd \Phi}{\dd \sigma}\right]^2=\frac{1}{4}(\sigma^2-4).
	\end{equation}
	Thus the auxiliary quadratic differential associated to the transition from the two-cut regime to the three-cut regime is $\Upsilon(\sigma) \dd \sigma^2$, since we can express $\Phi(\sigma)$ as
	\begin{equation}\label{Phiii}
	\Phi(\sigma) = \int_{2}^{\sigma} \sqrt{\Upsilon(s)} \dd s,
	\end{equation}
	where we have chosen the lower bound of integration as such since at $\Phi(2)=0$. The auxiliary quadratic differential $\Upsilon(\sigma) \dd \sigma^2$ has two simple zeros at $\pm 2$ and a pole of order $6$ at infinity (recall \eqref{QD near infinity}). Therefore three critical trajectories emanate from $\sigma=2$ and $\sigma=-2$ each, and four critical trajectories are incident with infinity. At infinity $\Upsilon(\sigma) \dd \sigma^2 \sim  \frac{1}{4}\sigma^2 \dd \sigma^2 $, thus the integral of its square root behaves like $\frac{1}{4}\sigma^2$ and thus the four solutions to $\Re[\Phi(\sigma)]=0$ near infinity respectively have asymptotic angles $\pi/4$, $3\pi/4$, $5\pi/4$, and $7\pi/4$. Since $\Upsilon(\sigma) \dd \sigma^2<0$ at any point in $(-2,2)$, we can immediately determine the local structure of critical trajectories at $\pm2$. A calculation shows that $\Phi(\sigma)$ as defined in \eqref{Phiii} differs from $\Phi(-\sigma)$ and from $\Phi(\bar{\sigma})$ by additive purely imaginary quantities which imposes a symmetry with respect to the origin and with respect to the real axis in the critical graph of $\Upsilon(\sigma) \dd \sigma^2$, which also means that one has a symmetry with respect to the imaginary axis as well. This symmetry ensures that the geodesic polygon with vertices $\sigma=2$ and $\infty$ must \textit{entirely lie in the right half-plane}, because if the polygon were to hit the imaginary axis at say $\ii y_*$, it would make $\sigma= \ii y_*$ (and also $\sigma= -\ii y_*$) a non-regular point of the quadratic differential $\Upsilon(\sigma) \dd \sigma^2$, which is a contradiction (recall that through each regular point of a meromorphic quadratic differential passes exactly one $\theta$-arc). Based on what we discussed above and what we know about the asymptotic angles at infinity, the critical graph shown in Figure \ref{fig:Crit Traj 2to3 sigma plane} is indeed correct.
	\begin{figure}[ht]
		\centering
		\includegraphics[scale=0.35]{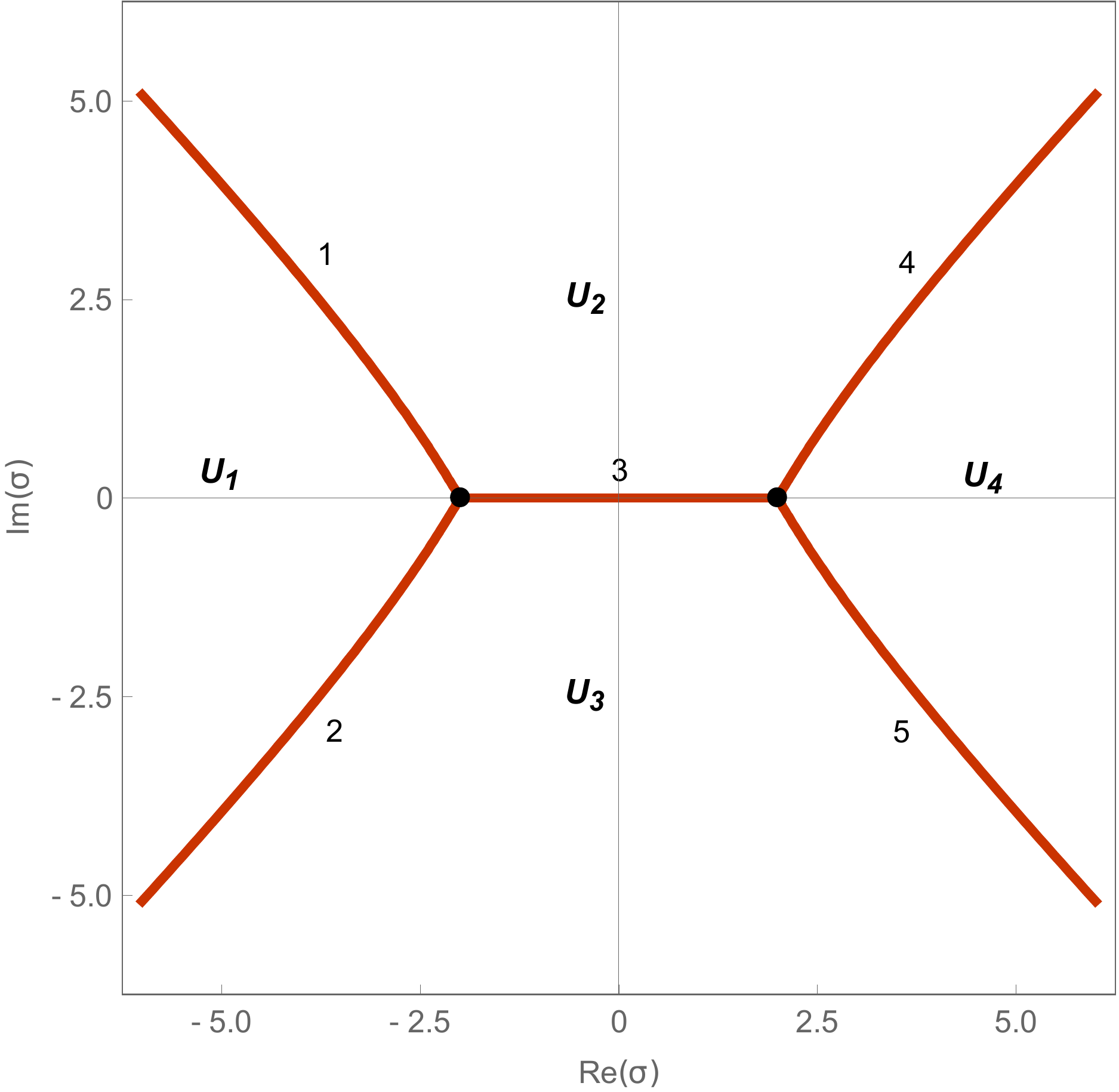}
		\caption{The critical graph $\widecheck{\Sigma}$ of the auxiliary quadratic differential $ \frac{1}{4}(\sigma^2-4) \dd \sigma^2 \equiv \Upsilon(\sigma) \dd \sigma^2$ whose components are the candidates for the two-cut to the three-cut transition. The black dots show the critical points of $\Upsilon(\sigma) \dd \sigma^2$ which are simple zeros at $\pm 2$.}
		\label{fig:Crit Traj 2to3 sigma plane}
	\end{figure}
	We can prove the following Lemma similarly as we proved Lemma \ref{Lemma K1 K2 K3}, thus we do not provide the details.
	\begin{lemma}\label{Lemma U1 U2 U3}
		In $\boldsymbol{U}_1$ we have $\Re[\eta_2(0;\sigma)]<0$ while in $\boldsymbol{U}_2 \cup \boldsymbol{U}_3$ we have $\Re[\eta_2(0;\sigma)]>0$.
	\end{lemma}
	
	\begin{theorem}\label{THM 1 2}
		The two-cut regime is the region labeled so in Figure \ref{fig:Phase Diagram}.
	\end{theorem}
	
	\begin{proof}
		Firstly, recalling Lemma \ref{2cut is nonempty} and Theorem \ref{thm 2cut is open} we can show that $\boldsymbol{U}_1 \subset \mathcal{O}_2$.	
		By Corollary \ref{Cor 0 on critical graph}, none of the lines labeled by $\boldsymbol{1}$ to $\boldsymbol{5}$ in Figure \ref{fig:Crit Traj 2to3 sigma plane} belong to $\mathcal{O}_2$. Now, we show that no $\sigma \in \boldsymbol{U}_4$ can belong to $\mathcal{O}_2$. Indeed as described above, the region $\boldsymbol{U}_4$ must lie entirely in the right half plane which itself belongs to the one-cut region by Theorem \ref{THM ONE CUT}, so $\boldsymbol{U}_4 \subset \mathcal{O}_1$. Due to the uniqueness of the support of the equilibrium measure  no $\sigma \in \boldsymbol{U}_4$ can belong to $\mathcal{O}_2$.  Recalling Lemma \ref{legs connect at criticality 2cut}, we know that for $\sigma_0 \in \boldsymbol{1}$, either $\ell^{(a_2)}_2$ must connect to $\ell^{(-a_2)}_2$, or $\ell^{(a_2)}_3$ must connect to $\ell^{(-a_2)}_3$ at the origin (see, e.g. Figure \ref{fig:stable barren -3 1.5}).   Similar to the argument provided in Theorem \ref{THM VII IX} we can show that the only possible transition for the critical graph of $Q_2(z;\sigma) \dd z^2$ as $\boldsymbol{1} \ni \sigma_0 \to \sigma_1 \in \boldsymbol{U}_2$ is the one shown from Figure \ref{fig:stable barren -3 1.5} to Figure \ref{fig:stable barren -3 1.6}. This makes sure that no point in $\boldsymbol{U}_2$ could belong to $\mathcal{O}_2$ (recall Theorem \ref{thm 2cut is open}). Identically one can show that no point in $\boldsymbol{U}_3$ could belong to $\mathcal{O}_2$.   
	\end{proof}
	
	\begin{remark}
		\normalfont	The line labeled by VII (resp. IX) in Figure \ref{fig:Crit Traj 1to3 sigma plane: obvious stuff removed} has no intersection with the line labeled by $\boldsymbol{1}$ (resp. $\boldsymbol{2}$) in Figure \ref{fig:Crit Traj 2to3 sigma plane}. This is obvious due to the uniqueness of the support of the equilibrium measure. Indeed, if there was an intersection point $\sigma$, then it would simultaneously be a degenerate one-cut and a degenerate two-cut $\sigma$. Another contradiction would be that  if there was an intersection point $\sigma$, then we would have had a way to continuously make a transition from a two-cut $\sigma_2$ to a one-cut $\sigma_1$ through some $ \boldsymbol{1} \ni \sigma_0 \neq -2$. But we know that the only way for such a transition is through degeneration of the gap from $-a_2$ to $a_2$, that is $\pm a_2 \to 0$ which is only possible if $\sigma \to -2$ (recall \eqref{em46}). 
	\end{remark}
	
	\begin{theorem}
		The three-cut regime is the region labeled so in Figure \ref{fig:Phase Diagram}.
	\end{theorem}
	\begin{proof}
		By Theorems \ref{THM ONE CUT} and \ref{THM 1 2} we have already proven that	in the region labeled as the "three-cut regime" in Figure \ref{fig:Phase Diagram} the one-cut and two-cut requirements are not satisfied. Since for all $\sigma$, the equilibrium measure and the Riemann-Hilbert contour exists and is unique (see \cite{KS}, where uniqueness of the contour outside of the support of the equilibrium measure is up to homotopy) we conclude that all sigma in that region must necessarily satisfy the requirements of Definition \ref{Def three cut sigma}.
	\end{proof}
	
	\section{The Riemann-Hilbert Problem in the One-cut Regime, String Equations and Topological Expansion of the Recurrence Coefficients}\label{sec rhp}
	
	In this section we follow \cite{BleherIts2005}. For simplicity of notation, let us use $b$ instead of $b_1$ in this section. Assume that $\sigma$ belongs to the one-cut regime and let us consider the set of monic orthogonal polynomials $P_n(z;N)$, $\deg P_n(z;N) = n$, satisfying
	\begin{equation}\label{orthogonality}
	\int_{\Ga_{\sigma}} P_n(s;N)P_k(s;N)e^{-NV(s)}\dd s=h_{n}(N)\de_{n,k}, \qquad k=0,1,\cdots,n,
	\end{equation}
	where we suppress the dependence on $\sigma$ in all of the quantities and functions. Since the potential $V$ is even, these polynomials satisfy the following recurrence relation
	\begin{equation}
	zP_n(z;N)=P_{n+1}(z;N)+\ga^2_{n}(N)P_{n-1}(z;N).
	\end{equation}
	Corresponding to this system of orthogonal polynomials one has the following Riemann-Hilbert problem \cite{FIK}
	\begin{itemize}
		\item \textbf{RH-$\boldsymbol Y$1} $\qquad Y(z;n,N)$ is holomorphic in $\C \setminus \Ga_{\sigma}$.
		\item \textbf{RH-$\boldsymbol Y$2} $\qquad Y_+(z;n,N)=Y_-(z;n,N)J_Y(z;N), \quad z \in \Ga_{\sigma}, \qquad$ 
		where 
		\begin{equation}
		J_Y(z;N) = \begin{pmatrix}1 & w(z;N) \\ 0 & 1 \end{pmatrix}, \qquad \qquad w(z;N):=e^{-NV(z)}.
		\end{equation}
		\item \textbf{RH-$\boldsymbol Y$3} $\qquad Y(z;n,N) = (I + O(z^{-1}))\begin{pmatrix} z^n & 0 \\ 0 & z^{-n} \end{pmatrix}, \qasq z \to \infty$.
	\end{itemize}
	The representation of the solution of this Riemann-Hilbert problem in terms of OPs is due to Fokas, Its and Kitaev \cite{FIK} and is given by
	\begin{equation}
	\di \begin{pmatrix}
	P_n(z;N) & \mathscr{C}[P_nw](z;N) \\
	-\di \frac{2\pi \ii}{h_{n-1}(N)} P_{n-1}(z;N) & -\di \frac{2\pi \ii}{h_{n-1}(N)}\mathscr{C}[P_{n-1}w](z;N)
	\end{pmatrix},
	\end{equation}
	where $\mathscr{C}[f]$ is the Cauchy transform of the function $f$ with respect to the contour $\Ga_{\sigma}$. Using the three-term recurrence relation and the orthogonality conditions one can easily observe that 
	\begin{equation}\label{ga to h}
	\ga_{n}^2(N) = \frac{h_{n}(N)}{h_{n-1}(N)}.
	\end{equation}
	For a parameter $\varkappa > 0$, define
	\begin{equation}
	V_{\varkappa}(z;\sigma)\equiv V_{\varkappa}:= V/\varkappa.
	\end{equation} Doing the one-cut endpoint calculations for $V_{\varkappa}$, similar to those done in \S \ref{subsec 2.4.1}, we find \begin{equation}
	\rho_V(z;\sigma,\varkappa) = \frac{1}{2\pi \ii \varkappa}(z^2-z^2_0(\varkappa;\sigma))\left(\sqrt{z^2-b^2(\varkappa;\sigma)}\right)_+,
	\end{equation}
	and
	\begin{equation}\label{b and z_0 one cut s}
	b(\varkappa;\sigma)=\sqrt{\frac{2}{3}\left( -\sigma + \sqrt{12\varkappa+\sigma^2} \right)} \qandq z_0(\varkappa;\sigma)=\sqrt{\frac{1}{3}\left( -2\sigma -  \sqrt{12\varkappa+\sigma^2} \right)}.
	\end{equation}
	The multi-critical points are obtained when $z_0=b$ and when $z_0=0$, these possibilities respectively correspond to \begin{equation}\label{multicritical kappa}
	12\varkappa+\sigma^2 = 0 , \qandq -2\sigma -  \sqrt{12\varkappa+\sigma^2} =0.
	\end{equation} 
	Therefore, the multi-critical points are \begin{equation}\label{multicritical kappa 1}
	\sigma = \pm \ii \sqrt{12\varkappa}, \qandq \sigma = -2\sqrt{\varkappa},
	\end{equation}
	and we have a similar picture as Figure \ref{fig:Phase Diagram} for the critical lines corresponding to transition from the one-cut to the three-cut regime for $V_{\varkappa}$. We also have the following formulae for the corresponding $g$-function
	\begin{equation}
	g(z;\sigma,\varkappa)=\frac{1}{\varkappa}g(z;\sigma,1), 
	\end{equation}and the Euler-Lagrange constant
	\begin{equation}
	\ell^{(1)}_*(\sigma,\varkappa)=\frac{1}{\varkappa}\ell^{(1)}_*(\sigma,1),
	\end{equation}
	where $g(z;\sigma,1)$, and $\ell^{(1)}_*(\sigma,1)$ are respectively given by \eqref{g-1-cut-explicit} and \eqref{ell* 1 cut}. 
	
	\smallskip
	
	Note that for $\sigma$ in the one-cut regime the expressions in \eqref{b and z_0 one cut s} are locally analytic in $\sigma$ and $\varkappa$. In particular, for any fixed $\sigma$, there exists $\ep(\sigma)>0$ such that $b(\varkappa;\sigma)$ and $z_0(\varkappa;\sigma)$ are analytic in $\varkappa$ in the interval \begin{equation}\label{1 pm ep}
	1-\ep(\sigma)<\varkappa<1+\ep(\sigma), \qquad 0<\ep(\sigma)<1.
	\end{equation}
	In particular, let \begin{equation}
	\varkappa=\frac{n}{N}.
	\end{equation}For this choice of $\varkappa$, we have
	\begin{equation}
	NV(z)=nV_\varkappa(z).
	\end{equation}
	
	Note that the orthogonal polynomials (and hence their norms) with respect to $e^{-NV(z)}$ and $e^{-nV_{\varkappa}(z)}$ are identical as they are built by bordered Hankel determinants out of moments of the identical weight functions.
	The Riemann-Hilbert problem corresponding to $e^{-nV_{\varkappa}(z)}$ would be exactly similar to \textbf{RH-$\boldsymbol Y$1} through \textbf{RH-$\boldsymbol Y$3}, except that one should make replacements $V \mapsto V_\varkappa$ and $N \mapsto n$ in \textbf{RH-$\boldsymbol Y$2}.
	As a result of the Riemann-Hilbert analysis for orthogonal polynomials on the line with respect to the weight $e^{-nV_{\varkappa}(z)}$ we obtain a $1/N$ asymptotic expansion for $\ga^2_n$:
	\begin{equation}\label{ga^2 large N}
	R_n(\varkappa;\sigma):=\ga^2_n(\varkappa;\sigma) \sim \sum^{\infty}_{j=0} \frac{r_j(\varkappa;\sigma)}{N^j}, \qasq N \to \infty,
	\end{equation}
	where
	\begin{equation}\label{f0}
	r_0(\varkappa;\sigma)   \equiv  \frac{b^2(\varkappa;\sigma)}{4}.
	\end{equation}
	The Riemann-Hilbert analysis is standard and we do not provide the details here. For one-cut real potentials see e.g. \cite{Charlier, CharlierGharakhloo}, and for complex one-cut potentials see e.g. \cite{BDY, BertolaTovbis2015,KuijlaarsMcLaughlin2001,KuijlaarsMcLaughlin2004}.
	
	Now, recall (see e.g. \cite{BL}) the string equations 
	
	\begin{equation}
	\ga_n[V'(\mathcal{Q})]_{n,n-1}=\frac{n}{N}, \qandq [V'(\mathcal{Q})]_{n,n}=0,
	\end{equation}
	where
	\begin{equation}
	\mathcal{Q}= \begin{pmatrix}
	0 & \ga_1 & 0 & 0 & \cdots \\
	\ga_1 & 0 & \ga_2 & 0 & \cdots \\
	0 & \ga_2 & 0 & \ga_3 & \cdots \\
	\vdots & \ddots & \ddots & \ddots & \cdots\\
	\end{pmatrix}.
	\end{equation}
	The relevant quantities reduce to:
	
	\begin{equation}
	[\mathcal{Q}]_{n,n-1}=\ga_n, \qquad  [\mathcal{Q}]_{n,n}=0, \qquad     [\mathcal{Q}^3]_{n,n-1}=\ga_n\ga^2_{n-1}+\ga_n^3+\ga_n\ga^2_{n+1}, \qquad  [\mathcal{Q}^3]_{n,n}=0.
	\end{equation}
	Therefore the second string equation is automatically satisfied and the first one can be written as 
	\begin{equation}\label{string}
	\ga_n^2(\varkappa;\sigma)\left(\sigma+\ga_n^2(\varkappa;\sigma)+\ga_{n-1}^2(\varkappa;\sigma)+\ga_{n+1}^2(\varkappa;\sigma)\right) = \varkappa,
	\end{equation}or
	\begin{equation}
	R_n(\varkappa;\sigma)\left(\sigma+R_{n-1}(\varkappa;\sigma)+R_n(\varkappa;\sigma)+R_{n+1}(\varkappa;\sigma)\right) = \varkappa.
	\end{equation}
	Note that 
	\begin{equation}\label{ga^2 n pm 1 large N}
	R_{n\pm1}(\varkappa;\sigma)  \sim \sum^{\infty}_{j=0} \frac{r_j(\varkappa \pm N^{-1},\sigma)}{N^j}, \qasq N \to \infty.
	\end{equation}
	Evaluation of Taylor expansions of $r_j$, centered at $\varkappa=n/N$, at $\varkappa \pm N^{-1}$  yields
	\begin{equation}\label{ga^2 n p 1 large N}
	R_{n+1}(\varkappa;\sigma) \sim \sum^{\infty}_{j=0} \frac{1}{N^j} \sum^{\infty}_{\ell=0} \frac{r^{(\ell)}_j(\varkappa;\sigma)}{\ell ! N^{\ell}}, \qasq N \to \infty.
	\end{equation} and 
	\begin{equation}\label{ga^2 n m 1 large N}
	R_{n-1}(\varkappa;\sigma) \sim \sum^{\infty}_{j=0} \frac{1}{N^j} \sum^{\infty}_{\ell=0} \frac{(-1)^{\ell}r^{(\ell)}_j(\varkappa;\sigma)}{\ell ! N^{\ell}}, \qasq N \to \infty.
	\end{equation}
	Direct computation gives the following asymptotic expansion for the left hand side of \eqref{string} in inverse powers of $N$:
	\begin{equation}\label{string hat}
	R_n(\varkappa;\sigma)\left(\sigma+R_{n-1}(\varkappa;\sigma)+R_n(\varkappa;\sigma)+R_{n+1}(\varkappa;\sigma)\right) \sim \sum^{\infty}_{j=0} \frac{\hat{r}_j(\varkappa;\sigma)}{N^j},
	\end{equation}
	where
	\begin{equation}\label{hat f}
	\hat{r}_j(\varkappa;\sigma) = \sigma r_j(\varkappa;\sigma)+\sum^j_{\ell=0} r_{\ell}(\varkappa;\sigma) \left[ 3 r_{j-\ell}(\varkappa;\sigma)+\sum^{j-\ell-1}_{m=0} (1+(-1)^{j-\ell-m}) \frac{r^{(j-\ell-m)}_m(\varkappa;\sigma)}{(j-\ell-m)!} \right]
	\end{equation}
	So the string equation \eqref{string} can be written as \begin{equation}\label{string list}
	\hat{r}_0(\varkappa;\sigma)=\varkappa, \qandq \hat{r}_j(\varkappa;\sigma)=0, \qquad j \in \N.
	\end{equation}
	\begin{lemma}
		$\ga^2_n(\sigma)$ has a power series expansion in in inverse powers of $N^2$.
	\end{lemma}
	\begin{proof}
		We prove that $r_{2k-1}(\varkappa;\sigma)=0$, $k \in \N$, by induction. For $k=1$ it can be seen as follows. Using \eqref{hat f} we write \eqref{string list} for $j=1$ to get
		\begin{equation}
		\hat{r}_{1}(\varkappa;\sigma) = r_1(\varkappa;\sigma) (\sigma + 6r_0(\varkappa;\sigma)) = 0.
		\end{equation}
		Using \eqref{b and z_0 one cut s} and \eqref{f0} we can show
		\begin{equation}
		\sigma+ 6r_0(\varkappa;\sigma) = \sqrt{12\varkappa+\sigma^2},
		\end{equation}
		which is nonzero as we are away from the multicritical points $\pm \ii \sqrt{12\varkappa}$ (see \eqref{multicritical kappa}, \eqref{multicritical kappa 1} and above). Therefore \begin{equation}
		r_1(\varkappa;\sigma) = 0.
		\end{equation}
		Assume that all $r_{2k-1}(\varkappa;\sigma)=0$, for $k \in 2,3,\cdots, k_0$ (induction hypothesis). So we can update the definition \eqref{ga^2 large N} of $\ga_n$ to be
		\begin{equation}\label{ga^2 large N 1}
		R_n(\varkappa;\sigma) \sim \sum^{k_0-1}_{j=0} \frac{r_{2j}(\varkappa;\sigma)}{N^{2j}} + \sum^{\infty}_{j=2k_0} \frac{r_j(\varkappa;\sigma)}{N^j}, \qasq N \to \infty,
		\end{equation}and thus we have analogues of equations \eqref{ga^2 n pm 1 large N}, \eqref{ga^2 n p 1 large N}, \eqref{ga^2 n m 1 large N}, \eqref{string hat} and \eqref{hat f}. Now we show that $r_{2k_0+1}(\varkappa;\sigma)=0$.  Using \eqref{hat f} we write \eqref{string list} for $j=2k_0+1$. Note that the functions $r_{2k-1}$, $ 1 \leq k \leq k_0 $ and their derivatives should be disregarded as being zero, due to the update in the definition of $\ga_n$ by the induction hypothesis. Therefore we get 
		\begin{equation}
		\hat{r}_{2k_0+1}(\varkappa;\sigma) =  r_{2k_0+1}(\varkappa;\sigma) (\sigma + 6r_0(\varkappa;\sigma)) = 0,
		\end{equation}
		which implies that $r_{2k_0+1}(\varkappa;\sigma) = 0$, and thus
		\begin{equation}\label{Asymp ga_n^2}
		R_n(\varkappa;\sigma) \sim \sum^{\infty}_{j=0} \frac{r_{2j}(\varkappa;\sigma)}{N^{2j}}.
		\end{equation}
	\end{proof}
	Here it is worthwhile to provide explicit formulae for some $r_{2k}$. We recall that $r_0$ is given by \eqref{f0}.  Using \eqref{string}, \eqref{hat f} and \eqref{string list} we can show
	\begin{equation}
	r_2(\varkappa;\sigma) = - \frac{r_0(\varkappa;\sigma)r^{''}_0(\varkappa;\sigma)}{\sigma+6r_{0}(\varkappa;\sigma)},
	\end{equation}
	and
	\begin{equation}
	r_4(\varkappa;\sigma) = - \frac{\frac{1}{12}r_0(\varkappa;\sigma)r^{(4)}_0(\varkappa;\sigma)+3r^2_2(\varkappa;\sigma)+r_2(\varkappa;\sigma)r^{''}_0(\varkappa;\sigma)+r_0(\varkappa;\sigma)r^{''}_2(\varkappa;\sigma)}{\sigma+6r_{0}(\varkappa;\sigma)}.
	\end{equation}
	We can actually find the following recursive formula for all $r_{2j}$. Indeed,
	\begin{equation}\label{f2j}
	r_{2j}(\varkappa;\sigma) = - \frac{\La_{2j}(\varkappa;\sigma)}{\sigma+6r_{0}(\varkappa;\sigma)}, \qquad j \in \N,
	\end{equation}
	where
	\begin{equation}\label{LA2j}
	\Lambda_{2j}(\varkappa;\sigma):=3\sum^{j-1}_{\ell=1}r_{2\ell}(\varkappa;\sigma)r_{2j-2\ell}(\varkappa;\sigma)+2\sum^{j-1}_{\ell=0}r_{2\ell}(\varkappa;\sigma)\sum^{j-\ell-1}_{m=0}\frac{r^{(2j-2\ell-2m)}_{2m}(\varkappa;\sigma)}{(2j-2\ell-2m)!}.
	\end{equation}
	
	\section{Topological Expansion of the Free Energy}\label{Sec Free Energy}
	
	\begin{proposition}
		For $\sigma$ in the one-cut regime we have
		\begin{equation}\label{F''}
		\frac{\partial^2 F}{\partial \sigma^2} = \frac{N^2}{4n^2}R_n(R_{n-1}+R_{n+1}),
		\end{equation}and
		\begin{equation}\label{F'}
		\frac{\partial F}{\partial \sigma} = - \frac{N^2}{2n^2}R_n(\frac{n}{N}+R_{n-1}R_{n+1}).
		\end{equation}
	\end{proposition}
	\begin{proof}
		By the Heine's identity for Hankel determinants and noting that $h_k=D_{k+1}/D_k$, where $D_k$ is the $k\times k$ Hankel determinant generated by the weight $e^{-NV(z;\sigma)}$, we have
		\begin{equation}\label{e1}
		Z_{nN}(\sigma)=n!\prod_{k=0}^{n-1} h_k,
		\end{equation}
		hence
		\begin{equation}\label{e2}
		F\equiv F_{nN}(\sigma) =\frac{\ln n!}{n^2} +\frac{1}{n^2}\sum_{k=0}^{n-1} \ln h_k.
		\end{equation}
		Differentiating equation \eqref{orthogonality} with $j=k$, we obtain that
		\begin{equation}\label{e3}
		\begin{aligned}
		\frac{\partial h_k}{\partial \sg}&=-\frac{N}{2} \int_{\Ga} z^2 P_k^2(z) e^{-NV(z)} \,\dd z\\
		&=-\frac{N}{2} \int_{\Ga} [P_{k+1}(z)+\ga_k^2P_{k-1}(z)]^2 e^{-NV(z)} \,\dd z\\
		&=-\frac{N}{2}\,\left(h_{k+1}+\ga_k^4h_{k-1}\right)
		=-\frac{Nh_k}{2}\,\left(\ga^2_{k+1} +\ga_k^2\right),
		\end{aligned}
		\end{equation}
		hence
		\begin{equation}\label{e4}
		\frac{\partial \ln h_k}{\partial \sg}=-\frac{N}{2}\,\left(\ga^2_{k+1} +\ga_k^2\right),
		\end{equation}
		and by \eqref{e2},
		\begin{equation}\label{e5}
		\frac{\partial F}{\partial \sg}=-\frac{N}{2n^2}\sum_{k=0}^{n-1} \left(\ga^2_{k+1}+\ga_k^2\right).
		\end{equation}
		It is convenient to introduce also the $\psi$-functions,
		\begin{equation}\label{d9}
		\psi_k(z):=\frac{1}{\sqrt{h_k}}P_k(z)e^{-NV(z)/2}.
		\end{equation}
		They satisfy the orthogonality conditions,
		\begin{equation}\label{d10}
		\int_{\Ga} \psi_j(z)\psi_k(z) \,dz=\de_{jk},
		\end{equation}
		and the recurrence relation,
		\begin{equation}\label{d11}
		z\psi_k(z)=\ga_{k+1}\psi_{k+1}(z)+\ga_k\psi_{k-1}(z).
		\end{equation}
		Define the vector function
		\begin{equation}\label{e10}
		\vec\Psi_k(z)=
		\begin{pmatrix}
		\psi_k(z)\\
		\psi_{k-1}(z)
		\end{pmatrix},\quad k\ge 1.
		\end{equation}
		Then $\vec\Psi_k(z)$ satisfies the ODE
		\begin{equation}\label{e11}
		\vec\Psi_k'(z)=NA_k(z)\vec\Psi_k(z),
		\end{equation}
		where
		\begin{equation}\label{e12}
		A_k(z)=
		\begin{pmatrix}\di
		-\frac{V'(z)}{2} -\ga_k u_k(z) & \ga_k v_k(z) \\ 
		-\ga_k v_{k-1}(z) & \di \frac{V'(z)}{2} +\ga_k u_k(z)
		\end{pmatrix},
		\end{equation}
		and
		\begin{equation}\label{e13}
		u_k(z)=[W(\mathcal{Q},z)]_{k,k-1},\quad v_k(z)=[W(\mathcal{Q},z)]_{kk},
		\end{equation}
		where
		\begin{equation}\label{e14}
		W(\mathcal{Q},z)=\frac{V'(\mathcal{Q})-V'(z)}{\mathcal{Q}-z}\,.
		\end{equation}
		(see equation (1.5.2) in \cite{BL})
		For the quartic polynomial \eqref{int13}, we obtain from \eqref{e14} that
		\begin{equation}\label{e15}
		W(\mathcal{Q},z)=\mathcal{Q}^2+ \mathcal{Q}z+z^2+\sg,
		\end{equation}
		hence
		\begin{equation}\label{e16}
		u_k(z)=\ga_k z,\quad v_k(z)=\ga_k^2+\ga_{k+1}^2+z^2+\sg.
		\end{equation}
		Substituting these formulae in \eqref{e12}, we obtain that
		\begin{equation}\label{e17}
		\begin{aligned}
		&\frac{1}{N}\,\psi_k'(z)=-\left(\frac{z^3+\sg z}{2}+\ga_k^2z\right)\psi_k(z)+\ga_k(\ga_k^2+\ga_{k+1}^2+z^2+\sg)\psi_{k-1}(z),\\
		&\frac{1}{N}\,\psi_{k-1}'(z)=-\ga_k(\ga_{k-1}^2+\ga_{k}^2+z^2+\sg)\psi_k(z) +\left(\frac{z^3+\sg z}{2}+\ga_k^2z\right)\psi_{k-1}(z).
		\end{aligned}
		\end{equation}
		Differentiating \eqref{ga to h} with respect to $\sigma$, and using \eqref{e4} yields
		\begin{equation}
		\frac{\partial \ga_k}{\partial \sg}
		=\frac{N\ga_k}{4}\,(\ga_{k-1}^2-\ga_{k+1}^2),
		\end{equation}
		hence
		\begin{equation}\label{e19}
		\frac{\partial \ga_k^2}{\partial \sg}
		=\frac{N\ga_k^2}{2}\,(\ga_{k-1}^2-\ga_{k+1}^2).
		\end{equation}
		Differentiating equation \eqref{e5}, we obtain that
		\begin{equation}\label{e20}
		\begin{aligned}
		\frac{\partial^2 F}{\partial \sg^2}
		&=-\frac{N^2}{2n^2}\sum_{k=0}^{n-1}
		\left(\frac{\partial\ga_{k+1}^2}{\partial\sg}
		+\frac{\partial\ga_{k}^2}{\partial\sg}\right)\\
		&=-\frac{N^2}{4n^2}\sum_{k=0}^{n-1}
		\left[\ga_{k+1}^2(\ga_{k}^2-\ga_{k+2}^2)
		+\ga_{k}^2(\ga_{k-1}^2-\ga_{k+1}^2)\right]\\
		&=\frac{N^2}{4n^2}\sum_{k=0}^{n-1} (I_{k+1}-I_k),
		\end{aligned}
		\end{equation}
		where
		\begin{equation}\label{e21}
		I_k=\ga_{k}^2(\ga_{k-1}^2+\ga_{k+1}^2),\quad I_0=0.
		\end{equation}
		Observe that the last sum in \eqref{e20} is telescopic and $I_0=0$, hence
		\begin{equation}\label{e22}
		\begin{aligned}
		\frac{\partial^2 F}{\partial \sg^2}
		=\frac{N^2}{4n^2}\,\ga_n^2(\ga_{n-1}^2+\ga_{n+1}^2),
		\end{aligned}
		\end{equation}
		(cf. equation (1.4.21) in \cite{BL}) which is exactly \eqref{F''} recalling the notation introduced in \eqref{ga^2 large N}.
		
		\smallskip
		
		From equation \eqref{int15a} we have that
		\begin{equation}\label{dt1}
		\frac{\partial F}{\partial \sg}=\frac{1}{n^2Z}\,\frac{\partial Z}{\partial \sg}\,,
		\end{equation}
		and from \eqref{int14} we have
		\begin{equation}\label{dt2}
		\frac{\partial Z}{\partial \sg}=-\frac{N}{2} \int_\Gamma\cdots\int_\Gamma
		\left(\sum_{k=1}^n z_k^2\right)\prod_{1\leq j<k\leq n}(z_j-z_k)^2
		\prod_{k=1}^n e^{-N V(z_k)}\mathrm dz_1\cdots \mathrm dz_n,
		\end{equation}
		hence
		\begin{equation}\label{dt3}
		\frac{\partial F}{\partial \sg}=-\frac{N}{2n^2}\,
		\big\langle \sum_{k=1}^n z_k^2\big\rangle\,,
		\end{equation}
		where
		\begin{equation}\label{dt4}
		\begin{aligned}
		&\big\langle f(z_1,\ldots,z_n)\big\rangle =\frac{\di\int_\Gamma\cdots\int_\Gamma
			f(z_1,\ldots,z_n)\prod_{1\leq j<k\leq n}(z_j-z_k)^2
			\prod_{k=1}^n e^{-N V(z_k)}\mathrm dz_1\cdots \mathrm dz_n,}
		{\di\int_\Gamma\cdots\int_\Gamma
			\prod_{1\leq j<k\leq n}(z_j-z_k)^2
			\prod_{k=1}^n e^{-N V(z_k)}\mathrm dz_1\cdots \mathrm dz_n,}\,.
		\end{aligned}
		\end{equation}
		By the permutation symmetry,
		\begin{equation}\label{dt5}
		\big\langle \sum_{k=1}^n z_k^2\big\rangle=n\big\langle z_1^2\big\rangle
		=n\int_\Gamma z^2 p(z)dz,
		\end{equation}
		where
		\begin{equation}\label{dt6}
		p(z)=\frac{1}{Z}\,
		\int_\Gamma\cdots\int_\Gamma
		\prod_{1\leq j<k\leq n}(z_j-z_k)^2
		\prod_{k=1}^n e^{-N V(z_k)}\mathrm dz_2\cdots \mathrm dz_n\bigg|_{z_1=z},
		\end{equation}
		is the one-point density function. Thus, from \eqref{dt3} we obtain that
		\begin{equation}\label{dt7}
		\frac{\partial F}{\partial \sg}=-\frac{N}{2n}\,
		\int_\Gamma z^2p(z)dz\,.
		\end{equation}The one-point density function $p(z)$ can be expressed in terms of
		the orthogonal polynomials $P_k(z)$ as
		\begin{equation}\label{dt8}
		p(z)=\frac{1}{n}\, \sum_{k=0}^{n-1}\psi_k(z)^2,
		\end{equation}
		where $\psi_k(z)$ are defined in \eqref{d9} 
		(see, e.g., formula (1.2.24) in \cite{BL}). By the Christoffel-Darboux formula, equation \eqref{dt8} can be reduced to
		\begin{equation}\label{dt9}
		p(z)=\frac{\ga_n}{n}\, \left[\psi_n'(z)\psi_{n-1}(z)-\psi_{n-1}'(z)\psi_n(z)\right].
		\end{equation}
		(cf. formula (1.3.5) in \cite{BL}). By equations \eqref{e17},
		\begin{equation}\label{dt10}
		\begin{aligned}
		&\frac{1}{N}\,\psi_n'(z)=
		-\left(\frac{z^3+\sg z}{2}+\ga_n^2z\right)\psi_n(z)+\ga_n(\ga_n^2+\ga_{n+1}^2+z^2+\sg)\psi_{n-1}(z),\\
		&\frac{1}{N}\,\psi_{n-1}'(z)=-\ga_n(\ga_{n-1}^2+\ga_{n}^2+z^2
		+\sg)\psi_n(z) +\left(\frac{z^3+\sg z}{2}+\ga_n^2z\right)\psi_{n-1}(z).
		\end{aligned}
		\end{equation}
		Substituting these formulae into \eqref{dt9}, we obtain that
		\begin{equation}\label{dt11}
		\begin{aligned}
		p(z)&=\frac{ N}{n}\, \Big[\ga_n^2(\ga_{n-1}^2+\ga_{n}^2+z^2+\sg)\psi_n^2(z)\\
		&-\ga_n(z^3+\sg z+2\ga_n^2z)\psi_n(z)\psi_{n-1}(z)\\
		&+\ga_n^2(\ga_n^2+\ga_{n+1}^2+z^2+\sg)\psi^2_{n-1}(z)\Big].
		\end{aligned}
		\end{equation}
		Substituting this furthermore in \eqref{dt7}, we obtain that
		\begin{equation}\label{dt12}
		\begin{aligned}
		\frac{\partial F}{\partial \sg}=-\frac{ N^2}{2n^2}\,
		&\bigg[\int_\Gamma \ga_n^2(\ga_{n-1}^2+\ga_{n}^2+z^2+\sg)z^2\psi_n^2(z)\,dz\\
		&-\int_\Gamma  \ga_n(z^3
		+\sg z+2\ga_n^2z)z^2\psi_{n-1}(z)\psi_{n}(z)\,dz\\
		&+\int_\Gamma \ga_n^2(\ga_{n}^2+\ga_{n+1}^2+z^2
		+\sg)z^2\psi^2_{n-1}(z),dz
		\bigg].
		\end{aligned}
		\end{equation}
		Iterating three term recurrence relation \eqref{d11}, we obtain that
		\begin{equation}\label{dt13}
		\begin{aligned}
		z^2\psi_n(z)&=\ga_{n+1}\ga_{n+2}\psi_{n+2}(z)
		+(\ga_{n}^2+\ga_{n+1}^2)\psi_n(z)
		+\ga_{n-1}\ga_{n}\psi_{n-2}(z),\\
		z^2\psi_{n-1}(z)&=\ga_{n}\ga_{n+1}\psi_{n+1}(z)
		+(\ga_{n-1}^2+\ga_{n}^2)\psi_{n-1}(z)
		+\ga_{n-2}\ga_{n-1}\psi_{n-3}(z),\\
		z^3\psi_n(z)&=\ga_{n+1}\ga_{n+2}\ga_{n+3}\psi_{n+3}(z)
		+\ga_{n+1}(\ga_{n}^2+\ga_{n+1}^2+\ga_{n+2}^2)\psi_{n+1}(z)\\
		&+\ga_{n}(\ga_{n-1}^2+\ga_{n}^2+\ga_{n+1}^2)\psi_{n-1}(z)
		+\ga_{n-2}\ga_{n-1}\ga_{n}\psi_{n-3}(z),
		\end{aligned}
		\end{equation}
		hence
		\begin{equation}\label{dt14}
		\begin{aligned}
		&\int_\Gamma z^2\,\psi_n^2(z)\,dz=
		\ga_{n}^2+\ga_{n+1}^2,\\
		&\int_\Gamma z^4\,\psi_n^2(z)\,dz=
		\ga_{n+1}^2\ga_{n+2}^2+(\ga_{n}^2+\ga_{n+1}^2)^2
		+\ga_{n-1}^2\ga_{n}^2,\\
		&\int_\Gamma z^3\psi_{n-1}(z)\psi_{n}(z)\,dz
		=\ga_{n}(\ga_{n-1}^2+\ga_{n}^2+\ga_{n+1}^2)\\
		&\int_\Gamma z^5\psi_{n-1}(z)\psi_{n}(z)\,dz
		=\ga_n\ga_{n+1}^2(\ga_{n}^2+\ga_{n+1}^2+\ga_{n+2}^2)\\
		&+\ga_{n}(\ga_{n-1}^2+\ga_{n}^2+\ga_{n+1}^2)
		(\ga_{n-1}^2+\ga_{n}^2)
		+\ga_{n-2}^2\ga_{n-1}^2\ga_n\,.
		\end{aligned}
		\end{equation}
		This gives that
		\begin{equation}\label{dt15}
		\begin{aligned}
		A&:=\int_\Gamma \ga_n^2(\ga_{n-1}^2
		+\ga_{n}^2+z^2+\sg)z^2\psi_n^2(z)\,dz\\
		&=\ga_n^2(\ga_{n-1}^2+\ga_n^2+\sg)(\ga_{n}^2+\ga_{n+1}^2)\\
		&+\ga_n^2\big[\ga_{n+1}^2\ga_{n+2}^2+(\ga_{n}^2+\ga_{n+1}^2)^2
		+\ga_{n-1}^2\ga_{n}^2\big],\\
		B&:=\int_\Gamma  \ga_n(z^3+\sg z+2\ga_n^2z)z^2\psi_n(z)\psi_{n-1}(z)\,dz\\
		&=\ga_n^2\ga_{n+1}^2(\ga_{n}^2+\ga_{n+1}^2+\ga_{n+2}^2)\\
		&+\ga_{n}^2(\ga_{n-1}^2+\ga_{n}^2+\ga_{n+1}^2)
		(\ga_{n-1}^2+\ga_{n}^2)
		+\ga_{n-2}^2\ga_{n-1}^2\ga_n^2\\
		&+(\sg+2\ga_n^2)\ga_{n}^2(\ga_{n-1}^2+\ga_{n}^2+\ga_{n+1}^2),\\
		C&:=\int_\Gamma \ga_n^2(\ga_{n}^2+\ga_{n+1}^2+z^2+\sg)z^2\psi^2_{n-1}(z)\,dz\\
		&=\ga_n^2(\ga_{n}^2+\ga_{n+1}^2+\sg)(\ga_{n-1}^2+\ga_{n}^2)\\
		&+\ga_n^2\big[\ga_{n}^2\ga_{n+1}^2+(\ga_{n-1}^2+\ga_{n}^2)^2
		+\ga_{n-2}^2\ga_{n-1}^2\big],
		\end{aligned}
		\end{equation}
		and by \eqref{dt12},
		\begin{equation}\label{dt16}
		\frac{\partial F}{\partial \sg}=
		-\frac{ N^2}{2n^2}\,(A-B+C).
		\end{equation}
		The expression $(A-B+C)$ turns out to be remarkably simple:
		\begin{equation}\label{dt17}
		A-B+C=\ga_n^2\big[\ga_n^2(\sg+\ga_{n-1}^2+\ga_n^2+\ga_{n+1}^2)
		+\ga_{n-1}^2\ga_{n+1}^2\big].
		\end{equation}
		By string equation \eqref{string},
		\begin{equation}\label{dt18}
		\ga_n^2(\sg+\ga_{n-1}^2+\ga_n^2+\ga_{n+1}^2)=\frac{n}{N}\,,
		\end{equation}
		hence
		\begin{equation}\label{dt19}
		A-B+C=\ga_n^2\left( \frac{n}{N}+ \ga_{n-1}^2\ga_{n+1}^2            \right)
		\end{equation}
		and
		\begin{equation}\label{dt20}
		\frac{\partial F}{\partial \sg}=
		-\frac{ N^2}{2n^2}\,\ga_n^2\left( \frac{n}{N}+ \ga_{n-1}^2\ga_{n+1}^2            \right).
		\end{equation}
		Recalling the notation introduced in \eqref{ga^2 large N} we have arrived at \eqref{F'}.
	\end{proof}
	A straightforward calculation shows that $R_{n-1}R_{n+1}$ (and thus $\di \frac{\partial F}{\partial \sigma}$ according to \eqref{F'} and \eqref{Asymp ga_n^2}) have power series expansion in inverse powers of $N^2$. Indeed, from \eqref{Asymp ga_n^2} we have
	\begin{equation}\label{ga^2 n p 1 large N 1}
	R_{n+1}(\varkappa;\sigma) \sim \sum^{\infty}_{j=0} \frac{1}{N^{2j}} \sum^{\infty}_{\ell=0} \frac{r^{(\ell)}_{2j}(\varkappa;\sigma)}{\ell ! N^{\ell}}, \qasq N \to \infty.
	\end{equation} and 
	\begin{equation}\label{ga^2 n m 1 large N 1}
	R_{n-1}(\varkappa;\sigma) \sim \sum^{\infty}_{j=0} \frac{1}{N^{2j}} \sum^{\infty}_{\ell=0} \frac{(-1)^{\ell}r^{(\ell)}_{2j}(\varkappa;\sigma)}{\ell ! N^{\ell}}, \qasq N \to \infty.
	\end{equation}
	Note that \eqref{ga^2 n p 1 large N 1} and \eqref{ga^2 n m 1 large N 1} can be written as 
	
	\begin{equation}
	R_{n+1}(\varkappa;\sigma) \sim \sum^{\infty}_{m=0} \widehat{A}_m(\varkappa;\sigma) N^{-m}, \qandq R_{n-1}(\varkappa;\sigma) \sim \sum^{\infty}_{m=0} \widehat{B}_m(\varkappa;\sigma) N^{-m},
	\end{equation}
	where
	\begin{equation}\label{A B hat}
	\widehat{A}_m(\varkappa;\sigma)= \di \sum_{\substack{2j+\ell=m \\ j,\ell \in \N\cup\{0\} }} \frac{r^{(\ell)}_{2j}(\varkappa;\sigma)}{\ell !}, \qandq     \widehat{B}_m(\varkappa;\sigma)= \di \sum_{\substack{2j+\ell=m \\ j,\ell \in \N\cup\{0\} }} \frac{(-1)^{\ell}r^{(\ell)}_{2j}(\varkappa;\sigma)}{\ell !}
	\end{equation}
	In particular, note that 
	\begin{equation}\label{ABAB hat}
	\widehat{A}_{2k}(\varkappa;\sigma)=\widehat{B}_{2k}(\varkappa;\sigma), \qandq   \widehat{A}_{2k-1}(\varkappa;\sigma)=-\widehat{B}_{2k-1}(\varkappa;\sigma), \qquad k \in \N.
	\end{equation}
	Therefore
	\begin{equation}
	R_{n-1}(\varkappa;\sigma)R_{n+1}(\varkappa;\sigma) \sim \sum^{\infty}_{j=0}C_j(\varkappa;\sigma)N^{-j},
	\end{equation}
	where
	\begin{equation}\label{Cj}
	\begin{split}
	C_j(\varkappa;\sigma) & = \sum_{\substack{m'+k=j \\ m',k \in \N\cup\{0\} }} \widehat{A}_{m'}(\varkappa;\sigma)\widehat{B}_{k}(\varkappa;\sigma) \\ &
	= \sum_{\substack{2m+k=j \\ m,k \in \N\cup\{0\} }} \widehat{A}_{2m}(\varkappa;\sigma)\widehat{B}_{k}(\varkappa;\sigma) + \sum_{\substack{2m+1+k=j \\ m,k \in \N\cup\{0\} }} \widehat{A}_{2m+1}(\varkappa;\sigma)\widehat{B}_{k}(\varkappa;\sigma)  \\ &
	= \sum_{\substack{2m+k=j \\ m,k \in \N\cup\{0\} }} \widehat{B}_{2m}(\varkappa;\sigma)\widehat{B}_{k}(\varkappa;\sigma) - \sum_{\substack{2m+1+k=j \\ m,k \in \N\cup\{0\} }} \widehat{B}_{2m+1}(\varkappa;\sigma)\widehat{B}_{k}(\varkappa;\sigma),
	\end{split}
	\end{equation}
	where we have used \eqref{ABAB hat}. Now we show that $C_j=0$ for odd $j$. Let $j=2M+1$, thus $m$ runs from $0$ to $M$. Then \eqref{Cj} can be written as 
	\begin{equation}
	\begin{split}
	C_{2M+1}(\varkappa;\sigma) & = \sum^M_{m=0}\widehat{B}_{2m}(\varkappa;\sigma)\widehat{B}_{2(M-m)+1}(\varkappa;\sigma)-\sum^M_{m=0}\widehat{B}_{2m+1}(\varkappa;\sigma)\widehat{B}_{2(M-m)}(\varkappa;\sigma) \\ & = \sum^M_{m=0}\widehat{B}_{2m}(\varkappa;\sigma)\widehat{B}_{2(M-m)+1}(\varkappa;\sigma)- \sum^M_{\ell=0}\widehat{B}_{2\ell}(\varkappa;\sigma)\widehat{B}_{2(M-\ell)+1}(\varkappa;\sigma) = 0,
	\end{split}
	\end{equation}
	where in the second summation we have used $\ell \equiv M-m$. Therefore we have
	\begin{equation}
	R_{n-1}(\varkappa;\sigma)R_{n+1}(\varkappa;\sigma) \sim \sum^{\infty}_{j=0}\frac{C_{2j}(\varkappa;\sigma)}{N^{2j}}, 
	\end{equation}
	where\begin{equation}\label{C2M}
	C_{2M}(\varkappa;\sigma) = \sum^M_{m=0}\widehat{A}_{2m}(\varkappa;\sigma)\widehat{A}_{2(M-m)}(\varkappa;\sigma)-\sum^M_{m=1}\widehat{A}_{2m-1}(\varkappa;\sigma)\widehat{A}_{2(M-m)+1}(\varkappa;\sigma).
	\end{equation}
	Put
	\begin{equation}\label{D to C}
	D_0(\varkappa;\sigma)\equiv C_0(\varkappa;\sigma) + \varkappa, \qandq D_{2j}(\varkappa;\sigma)\equiv C_{2j}(\varkappa;\sigma), \qquad j \in \N.
	\end{equation}
	So we can write 
	\begin{equation}\label{F' series}
	\frac{\partial F}{\partial \sigma} = - \frac{1}{2\varkappa^2}R_n(\varkappa;\sigma)(\varkappa+R_{n-1}(\varkappa;\sigma)R_{n+1}(\varkappa;\sigma)) \sim  \sum^{\infty}_{g=0} \frac{E_{2g}(\varkappa;\sigma)}{N^{2g}},
	\end{equation}
	where
	\begin{equation}\label{E2j}
	E_{2g}(\varkappa;\sigma) = - \frac{1}{2\varkappa^2} \sum_{k=0}^{g}r_{2k}(\varkappa;\sigma)D_{2g-2k}(\varkappa;\sigma).
	\end{equation}
	Integrating \eqref{F' series} we get \begin{equation}\label{F}
	F(\varkappa;\sigma) \sim \sum^{\infty}_{g=0} \frac{f_{2g}(\varkappa;\sigma)}{N^{2g}}. 
	\end{equation}
	
	From recurrence equations \eqref{f2j}, \eqref{LA2j} with initial data \eqref{f0}, \eqref{b and z_0 one cut s} we obtain the analyticity of the coefficients $ r_{2j}(\sigma)\equiv r_{2j}(1;\sigma)$ with respect to $ \sigma\in \mathcal{O}_1$. Then from equations \eqref{A B hat}, \eqref{C2M} , and \eqref{D to C} we obtain the analyticity of the coefficients $D_{2j}(\sigma)\equiv D_{2j}(1;\sigma), \sigma\in O_1$. Finally, from equation \eqref{E2j} we get the analyticity of the coefficients $E_{2g}(\sigma)\equiv E_{2g}(1;\sigma)$. Now, in view of \eqref{F' series}, \eqref{F}, Remark \ref{Remark branches}, and Theorem \ref{THM ONE CUT}, this implies the analyticity of the coefficients $f_{2g}(\sigma)\equiv f_{2g}(1;\sigma)$ with respect to $\sigma\in \mathcal{O}_1$. We have thus proven Theorem \ref{thm topexp}.
	
	
	Because of \eqref{F} and \eqref{int15b}, $\mathscr{F}(\varkappa;u)$ also has an asymptotic expansion in inverse powers of $N^2$:
	
	\begin{equation}\label{Fu N^-2 expansion}
	\mathscr{F}(\varkappa;u) \sim \sum^{\infty}_{g=0} \frac{\mathcal{f}_{2g}(\varkappa;u)}{N^{2g}}.
	\end{equation}

	\subsection{Derivation of $\mathscr{F}'(u)$} In this subsection we suppress the dependence of objects on $n$ and $N$ as these parameters are fixed. Let us rewrite the equation \eqref{int13}, \eqref{int15}, and \eqref{int15b} as 
	\begin{equation}
	V(u^{\frac{1}{4}}\ze;u^{-\frac{1}{2}})=\mathscr{V}(\ze,u),
	\end{equation}
	
	\begin{equation}
	Z(u^{-\frac{1}{2}})=u^{\frac{n^2}{4}}\mathcal{Z}(u),
	\end{equation}
	and
	\begin{equation}\label{FFu}
	F(u^{-\frac{1}{2}})=\frac{\ln u}{4}+\mathscr{F}(u).
	\end{equation}
	We consider also monic orthogonal polynomials $\mathcal{P}_k(\ze;u)=\ze^k+\cdots$ such that
	\begin{equation}\label{Pu}
	\int_{\R} \mathcal{P}_j(\ze;u)\mathcal{P}_k(\ze;u)e^{-N\mathscr{V}(\ze;u)}\dd \ze = \mathcal{h}_k(u)\de_{jk}.
	\end{equation}
	In \eqref{orthogonality}, make the change of variables $z = u^{\frac{1}{4}}\ze$, and recalling that $\sigma=u^{-\frac{1}{2}}$ we get
	\begin{equation}\label{PP}
	\int_{\R} P_j(u^{\frac{1}{4}}\ze;u^{-\frac{1}{2}})P_k(u^{\frac{1}{4}}\ze;u^{-\frac{1}{2}})e^{-NV(u^{\frac{1}{4}}\ze;u^{-\frac{1}{2}})}u^{\frac{1}{4}} \dd \ze =h_k(u^{-\frac{1}{2}})\de_{jk}.
	\end{equation}
	Note that deformation of the integration contour back to the real line is possible by the Cauchy theorem. We can write \eqref{PP} as 
	\begin{equation}\label{PPP}
	\int_{\R} \left[u^{-\frac{j}{4}}P_j(u^{\frac{1}{4}}\ze;u^{-\frac{1}{2}})\right] \left[u^{-\frac{k}{4}}P_k(u^{\frac{1}{4}}\ze;u^{-\frac{1}{2}})\right]e^{-N\mathscr{V}(\ze,u)} \dd \ze =u^{-\frac{k}{2}-\frac{1}{4}}h_k(u^{-\frac{1}{2}})\de_{jk}.
	\end{equation}
	Comparing with \eqref{Pu} yields
	\begin{equation}
	\mathcal{P}_k(\ze;u) = u^{-\frac{k}{4}}P_k(u^{\frac{1}{4}}\ze;u^{-\frac{1}{2}}), 
	\end{equation}
	and
	\begin{equation}
	\mathcal{h}_k(u) = u^{-\frac{k}{2}-\frac{1}{4}}h_k(u^{-\frac{1}{2}}).
	\end{equation}
	We have the three-term recurrence relation
	\begin{equation}
	\ze \mathcal{P}_k(\ze;u)=\mathcal{P}_{k+1}(\ze;u)+\mathscr{R}_k(u)\mathcal{P}_{k-1}(\ze;u).
	\end{equation}
	Like \eqref{ga to h} we have
	\begin{equation}\label{Ru to Rsigma}
	\mathscr{R}_k(u) = \frac{\mathcal{h}_k(u)}{\mathcal{h}_{k-1}(u)} = u^{-\frac{1}{2}}R_k(u^{-\frac{1}{2}}),
	\end{equation}
	where $R_k(\sigma)\equiv \ga_k^2(\sigma)$. The string equation \eqref{string} can be written as 
	\begin{equation}
	R_n(u^{-\frac{1}{2}}) \left[ u^{-\frac{1}{2}} + R_{n-1}(u^{-\frac{1}{2}}) + R_n(u^{-\frac{1}{2}})+R_{n+1}(u^{-\frac{1}{2}}) \right] = \varkappa.
	\end{equation}
	Therefore
	\begin{equation}
	u^{\frac{1}{2}}\mathscr{R}_n(u) \left[ u^{-\frac{1}{2}} + u^{\frac{1}{2}}\mathscr{R}_{n-1}(u)+u^{\frac{1}{2}}\mathscr{R}_n(u)+u^{\frac{1}{2}}\mathscr{R}_{n+1}(u)  \right]= \varkappa,
	\end{equation} or
	\begin{equation}\label{string eq u}
	\mathscr{R}_n(u) \left[ 1 + u\mathscr{R}_{n-1}(u)+u\mathscr{R}_n(u)+u\mathscr{R}_{n+1}(u)  \right]= \varkappa.
	\end{equation}
	This is the string equation for $\mathscr{R}_n(u)$ (which is the analogue of equation 4.33 in \cite{BIZ}). 
	
	\begin{remark}\normalfont
		Note that the orthogonal polynomial objects, like $P$, $\mathcal{P}$, $h$, $\mathscr{h},R $ and $\mathscr{R}$ are functions of $n$ and $N$, or equivalently $n$ and $\varkappa \equiv n/N$.This explains the notations used below. 
	\end{remark}
	
	We have the topological expansion of the recurrence coefficient $\mathscr{R}_n(\varkappa;u)$,
	\begin{equation}\label{R n u}
	\mathscr{R}_n(\varkappa;u) \sim \sum^{\infty}_{g=0} \frac{\mathcal{r}_{2g}(\varkappa;u)}{N^{2g}},
	\end{equation}
	where the coefficients $\mathcal{r}_{2g}$ are analytic functions of $u$ and $\varkappa$ at the point $u=0$, $\varkappa=1$. Note that
	
	\begin{equation}\label{Rnu n pm 1 large N}
	\mathscr{R}_{n\pm1}(\varkappa;u) \sim \sum^{\infty}_{g=0} \frac{\mathcal{r}_{2g}(\varkappa \pm N^{-1};u)}{N^{2g}}, \qasq N \to \infty.
	\end{equation}
	Evaluation of Taylor expansions of $\mathcal{r}_{2g}$, centered at $\varkappa=n/N$, at $\varkappa \pm N^{-1}$  yields
	\begin{equation}\label{R u n p 1 large N}
	\mathscr{R}_{n+1}(\varkappa;u) \sim \sum^{\infty}_{g=0} \frac{1}{N^{2g}} \sum^{\infty}_{\ell=0} \frac{\mathcal{r}^{(\ell)}_{2g}(\varkappa;u)}{\ell ! N^{\ell}}, \qasq N \to \infty,
	\end{equation} and 
	\begin{equation}\label{R u n m 1 large N}
	\mathscr{R}_{n-1}(\varkappa;u) \sim \sum^{\infty}_{g=0} \frac{1}{N^{2g}} \sum^{\infty}_{\ell=0} \frac{(-1)^{\ell}\mathcal{r}^{(\ell)}_{2g}(\varkappa;u)}{\ell ! N^{\ell}}, \qasq N \to \infty.
	\end{equation}
	Now we can write the large-$N$ series expansion for the left-hand side of the string equation, indeed
	\begin{equation}
	\mathscr{R}_n(u) \left[ 1 + u\mathscr{R}_{n-1}(u)+u\mathscr{R}_n(u)+u\mathscr{R}_{n+1}(u)  \right] \sim \sum^{\infty}_{g=0} \frac{\hat{\mathcal{r}}_{2g}(\varkappa;u)}{N^{2g}}
	\end{equation}
	where
	\begin{equation}
	\hat{\mathcal{r}}_{0}(\varkappa;u) = \mathcal{r}_{0}(\varkappa;u)\left( 1+3u\mathcal{r}_{0}(\varkappa;u) \right) 
	\end{equation}
	and
	\begin{equation}
	\hat{\mathcal{r}}_{2g}(\varkappa;u) = \sum^g_{\ell=0}  \mathcal{r}_{2g-2\ell}(\varkappa;u) \left( 3u\mathcal{r}_{2\ell}(\varkappa;u)+2u \sum^{\ell-1}_{k=0} \frac{\mathcal{r}^{(2\ell-2k)}_{2k}(\varkappa;u)}{(2\ell-2k)!} \right), \qquad g \in \N.
	\end{equation}
	So, the string equations \eqref{string eq u} can be written as
	\begin{equation}
	\mathcal{r}_{0}(\varkappa;u)\left( 1+3u\mathcal{r}_{0}(\varkappa;u) \right) = \varkappa,
	\end{equation}
	and for $g \in \N$ 
	\begin{equation}
	(1+3u\mathcal{r}_{0}(\varkappa;u))\mathcal{r}_{2g}(\varkappa;u)+ \sum^g_{\ell=1}  \mathcal{r}_{2g-2\ell}(\varkappa;u) \left( 3u\mathcal{r}_{2\ell}(\varkappa;u)+2u \sum^{\ell-1}_{k=0} \frac{\mathcal{r}^{(2\ell-2k)}_{2k}(\varkappa;u)}{(2\ell-2k)!} \right) = 0.
	\end{equation}
	
	\begin{equation}
	\begin{split}
	(1+3u\mathcal{r}_{0}(\varkappa;u))\mathcal{r}_{2g}(\varkappa;u) & + 3u\sum^g_{\ell=1}  \mathcal{r}_{2g-2\ell}(\varkappa;u) \mathcal{r}_{2\ell}(\varkappa;u) \\ & + 2u\sum^g_{\ell=1}  \mathcal{r}_{2g-2\ell}(\varkappa;u)\sum^{\ell-1}_{k=0} \frac{\mathcal{r}^{(2\ell-2k)}_{2k}(\varkappa;u)}{(2\ell-2k)!} = 0. 
	\end{split}
	\end{equation}

	\begin{equation}
	\begin{split}
	(1+6u\mathcal{r}_{0}(\varkappa;u))\mathcal{r}_{2g}(\varkappa;u) & + 3u\sum^{g-1}_{\ell=1}  \mathcal{r}_{2g-2\ell}(\varkappa;u) \mathcal{r}_{2\ell}(\varkappa;u) \\ & + 2u\sum^g_{\ell=1}  \mathcal{r}_{2g-2\ell}(\varkappa;u)\sum^{\ell-1}_{k=0} \frac{\mathcal{r}^{(2\ell-2k)}_{2k}(\varkappa;u)}{(2\ell-2k)!} = 0.
	\end{split}
	\end{equation}
	Therefore we can explicitly find $\mathcal{r}_{2g}$ recursively from
	\begin{equation}\label{r0}
	\mathcal{r}_{0}(\varkappa;u) = \frac{-1 + \sqrt{1+12\varkappa u}}{6u},
	\end{equation}
	and
	\begin{equation}\label{r2g}
	\mathcal{r}_{2g}(\varkappa;u) =-u \frac{\mathscr{A}_{2g}(\varkappa;u)}{\sqrt{1+12\varkappa u}}, \qquad g \in \N,
	\end{equation}
	where
	\begin{equation}\label{AAA}
	\mathscr{A}_{2g}(\varkappa;u):=  3\sum^{g-1}_{\ell=1}  \mathcal{r}_{2g-2\ell}(\varkappa;u) \mathcal{r}_{2\ell}(\varkappa;u) + 2\sum^g_{\ell=1}  \mathcal{r}_{2g-2\ell}(\varkappa;u)\sum^{\ell-1}_{k=0} \frac{\mathcal{r}^{(2\ell-2k)}_{2k}(\varkappa;u)}{(2\ell-2k)!}.
	\end{equation}
	Indeed, the first few $\mathcal{r}_{2g}$'s  are given by 
	\begin{equation}
	\mathcal{r}_{2}(\varkappa;u) = \frac{u\left(-1+\sqrt{1+12 \varkappa u}\right)}{(1+12 \varkappa u)^2},
	\end{equation}
	\begin{equation}
	\mathcal{r}_{4}(\varkappa;u) = \frac{63u^3\left(-3-8\varkappa u+3\sqrt{1+12 \varkappa u}\right)}{(1+12 \varkappa u)^{9/2}},
	\end{equation}
	
	\begin{equation}
	\mathcal{r}_{6}(\varkappa;u) = \frac{54u^5\left(-2699-12788\varkappa u+(2699+444\varkappa u)\sqrt{1+12\varkappa u}\right)}{(1+12 \varkappa u)^{7}},
	\end{equation}
	
	\begin{equation} \begin{split}
	& \mathcal{r}_{8}(\varkappa;u) = \frac{27u^7}{(1+12 \varkappa u)^{19/2}} \\ & \times \left(13386672\varkappa^2u^2-58115796\varkappa u - 9348347+(7280964\varkappa u + 9348347)\sqrt{1+12\varkappa u}\right) .
	\end{split}
	\end{equation}
	
	In fact we can prove the following lemma for any $g \in \N$.
	
	\begin{lemma}\label{main lemma}
		For any $g \in \N$, we can write
		
		\begin{equation}\label{r2g C2g}
		\mathcal{r}_{2g}(\varkappa;u) = \mathcal{C}_{2g}(\varkappa)\left(u+\frac{1}{12\varkappa}\right)^{\frac{1-5g}{2}}  \left(\mathcal{q}_{g}(\varkappa;u)+\sqrt{u+\frac{1}{12\varkappa}}\widehat{\mathcal{q}}_{g}(\varkappa;u)\right),
		\end{equation}
		where $\mathcal{q}_{g}$, and $\widehat{\mathcal{q}}_{g}$ are polynomials with $\mathcal{q}_{g}(\varkappa;- \frac{1}{12\varkappa})=1$, and the constants $\mathcal{C}_{2g}(\varkappa)$ can be found recursively from the following relations
		\begin{equation}\label{C2g recurrence}
		\mathcal{C}_{2g}(\varkappa) = \frac{1}{2^33^{\frac{1}{2}}\varkappa^{\frac{1}{2}}} \sum^{g-1}_{\ell=1}  \mathcal{C}_{2g-2\ell}(\varkappa) \mathcal{C}_{2\ell}(\varkappa) + \frac{(5g-6)(5g-4)}{2^83^{\frac{7}{2}} \varkappa^{\frac{9}{2}} }\mathcal{C}_{2g-2}(\varkappa), \qquad g \in \N ,
		\end{equation}
		with 
		\begin{equation}\label{C0}
		\mathcal{C}_0(\varkappa) = -2^23^{\frac{1}{2}}\varkappa^{\frac{3}{2}},
		\end{equation}
		appearing in the series expansion of $\mathcal{r}_0$ near $-\frac{1}{12\varkappa}$:
		\[ \mathcal{r}_0(\varkappa;u)= \left(2\varkappa+\mathcal{C}_0(\varkappa)\sqrt{u+\frac{1}{12\varkappa}} \right) \left( 1 +  \sum^{\infty}_{j=1} (12\varkappa)^j \left( u + \frac{1}{12\varkappa} \right)^j\right). \]
	\end{lemma}
	\begin{proof}
		The only contributing terms to the leading singular behavior of $\mathcal{r}_{2g}$ are the first sum in \eqref{AAA} and the single term corresponding to $\ell=g$ and $k=g-1$ in the second sum in \eqref{AAA}.
	\end{proof}

	\begin{lemma}\label{main lemma'}
		For any $g,\ell \in \N$, we have
		\begin{equation}
		\begin{split}
		\mathcal{r}^{(\ell)}_{2g}(\varkappa;u) & \equiv \frac{\partial^{\ell}}{\partial \varkappa^{\ell}} \mathcal{r}_{2g}(\varkappa;u) = \frac{\mathcal{C}_{2g}(\varkappa)}{24^{\ell}\varkappa^{2\ell}}\prod_{j=1}^{\ell}(5g+2j-3) \left(u+\frac{1}{12\varkappa}\right)^{\frac{1-5g-2\ell}{2}} \\ & \left(\mathcal{q}_{g,\ell}(\varkappa;u)+\sqrt{u+\frac{1}{12\varkappa}}\widehat{\mathcal{q}}_{g,\ell}(\varkappa;u)\right),		
		\end{split}
		\end{equation}
		where $\mathcal{q}_{g,\ell}$, and $\widehat{\mathcal{q}}_{g,\ell}$ are polynomials, with $\mathcal{q}_{g,\ell}(\varkappa;-\frac{1}{12\varkappa})=1$.
	\end{lemma}
	
	As shown in \cite{BD}, the generating function for the constants $\mathcal{C}_{2g} \equiv \mathcal{C}_{2g}(1)$ satisfies the Painlev\'{e} I differential equation. Indeed, if one defines
	$$\mathcal{u}(\tau):=-2^{-\frac{8}{5}}3^{-\frac{2}{5}} \mathcal{y}(-2^{-\frac{9}{5}}3^{-\frac{6}{5}}\tau),$$
	where
	\begin{equation}
	\mathcal{y}(t) := \sum^{\infty}_{g=0}  \mathcal{C}_{2g} t^{\frac{1-5g}{2}},
	\end{equation}
	then in view of \eqref{C2g recurrence}, one can check that $\mathcal{u}$ satisfies the standard form of Painlv\'{e} I:
	\begin{equation}
	\mathcal{u}^{\prime \prime}(\tau) = 6 \mathcal{u}^2(\tau)+\tau.  
	\end{equation}
	
	\subsection{Large $N$ Expansion of $\mathscr{F}'(u)$}
	
	From \eqref{FFu} we have
	\begin{equation}\label{Fu'}
	\mathscr{F}^{'}(u) = \left( -\frac{1}{2} u^{-\frac{3}{2}} \right)F^{'}(\sigma) \bigg|_{\sigma=u^{-\frac{1}{2}}} - \frac{1}{4u}.
	\end{equation}
	Now, from \eqref{F'} we have
	\begin{equation}
	F^{'}(\sigma) \bigg|_{\sigma=u^{-\frac{1}{2}}}  = - \frac{1}{2\varkappa^2}R_n(u^{-\frac{1}{2}})\left(\varkappa+R_{n-1}(u^{-\frac{1}{2}})R_{n+1}(u^{-\frac{1}{2}})\right).
	\end{equation}
	Therefore
	\begin{equation}\label{F' R n u}
	\mathscr{F}^{'}(u) = \frac{1}{4\varkappa^2}\mathscr{R}_n(u) \left(\frac{\varkappa}{u}+\mathscr{R}_{n-1}(u)\mathscr{R}_{n+1}(u)\right)- \frac{1}{4u}.
	\end{equation}
	From \eqref{Fu'} and \eqref{F' series} we have the following asymptotic expansion for $\mathscr{F}^{'}(u)$ in inverse powers of $N^2$:
	
	\begin{equation}\label{F' u asymp}
	\mathscr{F}^{'}(u) \sim \sum^{\infty}_{g=0} \frac{\mathscr{E}_{2g}(\varkappa;u)}{N^{2g}}.
	\end{equation}
	We can find $\mathscr{E}_{2g}$ by substituting \eqref{R n u} into \eqref{F' R n u}. Indeed, we find
	
	\begin{equation}
	\mathscr{R}_{n+1}(\varkappa;u) \sim \sum^{\infty}_{m=0} A_m(\varkappa;u) N^{-m}, \qandq \mathscr{R}_{n-1}(\varkappa;u) \sim \sum^{\infty}_{m=0} B_m(\varkappa;u) N^{-m},
	\end{equation}
	where
	\begin{equation}\label{A B}
	A_m(\varkappa;u)= \di \sum_{\substack{2j+\ell=m \\ j,\ell \in \N\cup\{0\} }} \frac{\mathcal{r}^{(\ell)}_{2j}(\varkappa;u)}{\ell !}, \qandq     B_m(\varkappa;u)= \di \sum_{\substack{2j+\ell=m \\ j,\ell \in \N\cup\{0\} }} \frac{(-1)^{\ell}\mathcal{r}^{(\ell)}_{2j}(\varkappa;u)}{\ell !}
	\end{equation}
	In particular, note that 
	\begin{equation}\label{ABAB}
	A_{2k}(\varkappa;u)=B_{2k}(\varkappa;u), \qandq   A_{2k-1}(\varkappa;u)=-B_{2k-1}(\varkappa;u), \qquad k \in \N.
	\end{equation}
	Therefore
	\begin{equation}
	\mathscr{R}_{n-1}(\varkappa;u)\mathscr{R}_{n+1}(\varkappa;u) \sim \sum^{\infty}_{j=0}C_j(\varkappa;u)N^{-j}, 
	\end{equation}
	where
	\begin{equation}\label{C g}
	\begin{split}
	C_j(\varkappa;u) & = \sum_{\substack{m'+k=j \\ m',k \in \N\cup\{0\} }} A_{m'}(\varkappa;u)B_{k}(\varkappa;u) \\ &
	= \sum_{\substack{2m+k=j \\ m,k \in \N\cup\{0\} }} A_{2m}(\varkappa;u)B_{k}(\varkappa;u) + \sum_{\substack{2m+1+k=j \\ m,k \in \N\cup\{0\} }} A_{2m+1}(\varkappa;u)B_{k}(\varkappa;u)  \\ &
	= \sum_{\substack{2m+k=j \\ m,k \in \N\cup\{0\} }} B_{2m}(\varkappa;u)B_{k}(\varkappa;u) - \sum_{\substack{2m+1+k=j \\ m,k \in \N\cup\{0\} }} B_{2m+1}(\varkappa;u)B_{k}(\varkappa;u),
	\end{split}
	\end{equation}
	where we have used \eqref{ABAB}. Now we show that $C_j=0$ for odd $j$. Let $j=2M+1$, thus $m$ runs from $0$ to $M$. Then \eqref{C g} can be written as 
	\begin{equation}
	\begin{split}
	C_{2M+1}(\varkappa;u) & = \sum^M_{m=0}B_{2m}(\varkappa;u)B_{2(M-m)+1}(\varkappa;u)-\sum^M_{m=0}B_{2m+1}(\varkappa;u)B_{2(M-m)}(\varkappa;u) \\ & = \sum^M_{m=0}B_{2m}(\varkappa;u)B_{2(M-m)+1}(\varkappa;u)- \sum^M_{\ell=0}B_{2\ell}(\varkappa;u)B_{2(M-\ell)+1}(\varkappa;u) = 0,
	\end{split}
	\end{equation}
	where in the second summation we have used $\ell \equiv M-m$. Therefore we have
	\begin{equation}
	\mathscr{R}_{n-1}(\varkappa;u)\mathscr{R}_{n+1}(\varkappa;u) \sim \sum^{\infty}_{k=0}\frac{C_{2k}(\varkappa;u)}{N^{2k}}, 
	\end{equation}
	where\begin{equation}\label{C2g}
	C_{2k}(\varkappa;u) = \sum^k_{m=0}A_{2m}(\varkappa;u)A_{2(k-m)}(\varkappa;u)-\sum^k_{m=1}A_{2m-1}(\varkappa;u)A_{2(k-m)+1}(\varkappa;u).
	\end{equation}
	From \eqref{R n u}, \eqref{F' R n u}, \eqref{F' u asymp}, and \eqref{C2g} we find
	
	\begin{equation}\label{EEEEE0}
	\mathscr{E}_{0}(\varkappa;u)=\frac{1}{4\varkappa^2} \left( \mathcal{r}^2_{0}(\varkappa;u) + \frac{\varkappa}{u} \right)\mathcal{r}_{0}(\varkappa;u) - \frac{1}{4u},
	\end{equation}
	and
	\begin{equation}\label{E2gu}
	\mathscr{E}_{2g}(\varkappa;u)= \frac{1}{4\varkappa^2} \left[ \left( \mathcal{r}^2_{0}(\varkappa;u) + \frac{\varkappa}{u} \right)\mathcal{r}_{2g}(\varkappa;u) +\sum_{k=1}^{g}C_{2k}(\varkappa;u)\mathcal{r}_{2g-2k}(\varkappa;u) \right]
	\end{equation}
	
	\begin{lemma}\label{lemma 6.5.}
		For any $g \in \N$, we can write
		
		\begin{equation}\label{E2g C2g}
		\mathscr{E}_{2g}(\varkappa;u) = \frac{2^43^2 \varkappa}{3-5g}\mathcal{C}_{2g}(\varkappa)\left(u+\frac{1}{12\varkappa}\right)^{\frac{3-5g}{2}}  \left(\mathscr{P}_{g}(\varkappa;u)+\sqrt{u+\frac{1}{12\varkappa}}\widehat{\mathscr{P}}_{g}(\varkappa;u)\right),
		\end{equation}
		where $\mathscr{P}_{1}$ and $\widehat{\mathscr{P}}_{1}$ are Taylor series centered at zero with radius of convergence $1/12\varkappa$, while $\mathscr{P}_{g}$, and $\widehat{\mathscr{P}}_{g}$ are polynomials for $g \in \N \setminus \{1\}$. For all $g\in \N$ we have $\mathscr{P}_{g}(\varkappa;- \frac{1}{12\varkappa})=1$.
	\end{lemma}
	\begin{proof}
		Differentiating \eqref{int15b}  gives
		\begin{equation}
		\frac{\partial^2 F}{\partial \sigma^2} =4u^3\frac{\partial^2 \mathscr{F}}{\partial u^2}  + 6u^2\frac{\partial \mathscr{F}}{\partial u}+ \frac{u}{2},
		\end{equation}
		where we recall that $\sigma=u^{-1/2}$. Now recalling \eqref{F''} and \eqref{Ru to Rsigma} we obtain
		\begin{equation}
		4u^2\frac{\partial^2 \mathscr{F}}{\partial u^2} + 6u\frac{\partial \mathscr{F}}{\partial u}+ \frac{1}{2}= \frac{1}{4\varkappa^2} \mathscr{R}_{n}(\varkappa;u)\left(\mathscr{R}_{n-1}(\varkappa;u)+\mathscr{R}_{n+1}(\varkappa;u)\right)
		\end{equation}
		Using \eqref{R n u}, \eqref{R u n p 1 large N} and \eqref{R u n m 1 large N} we find  
		\begin{equation}
		4u^2\frac{\partial^2 \mathscr{F}}{\partial u^2} + 6u\frac{\partial \mathscr{F}}{\partial u}+ \frac{1}{2}\sim \frac{1}{2\varkappa^2} \sum^{\infty}_{g=0} \frac{\mathcal{A}_{2g}(\varkappa;u)}{N^{2g}},
		\end{equation}
		where 
		\begin{equation}
		\mathcal{A}_{2g}(\varkappa;u) = \sum^g_{k=0} \mathcal{r}_{2g-2k}(\varkappa;u)\sum^k_{\ell=0} \frac{\mathcal{r}_{2k-2\ell}^{(2\ell)}(\varkappa;u)}{(2\ell)!}
		\end{equation}
		So we have
		\begin{equation}\label{Diff Eq F0}
		4u^2 \mathcal{f}^{\prime\prime}_0+ 6u\mathcal{f}^{\prime}_0+ \frac{1}{2} = \frac{1}{2\varkappa^2}\mathcal{r}^2_{0}(\varkappa;u),
		\end{equation}
		and
		\begin{equation}\label{Diff Eq F2g}
		4u^2 \mathcal{f}^{\prime\prime}_{2g}+ 6u\mathcal{f}^{\prime}_{2g} = \frac{1}{2\varkappa^2} \mathcal{A}_{2g}(\varkappa;u).
		\end{equation}
		Note that solving the differential equation \eqref{Diff Eq F0} for $\mathcal{f}^{\prime}_0 \equiv \mathscr{E}_0$ yields the expected expression \eqref{EEEE0} which was directly found from \eqref{EEEEE0}. From Lemma \ref{main lemma} and \eqref{E2gu} we know that
		
		\begin{equation}
		\mathcal{f}^{\prime}_{2g}(\varkappa;u) \equiv \mathscr{E}_{2g}(\varkappa;u)=\widehat{\mathcal{C}}_{2g}(\varkappa)\left(u+\frac{1}{12\varkappa}\right)^{\al(g)} \left( \mathcal{p}(\varkappa;u) + \sqrt{u+\frac{1}{12\varkappa}} \widehat{\mathcal{p}}(\varkappa;u) \right),
		\end{equation}
		for some polynomials $\mathcal{p}$ and $\widehat{\mathcal{p}}$, with $\mathcal{p}(\varkappa;-\frac{1}{12\varkappa})=1$. Our goal is to determine $\widehat{\mathcal{C}}_{2g}(\varkappa)$ and $\al(g)$ through the equality \eqref{Diff Eq F2g}. We obviously have 
		\begin{equation}
		\mathcal{f}^{\prime \prime}_{2g}(\varkappa;u) =\al(g)\widehat{\mathcal{C}}_{2g}(\varkappa)\left(u+\frac{1}{12\varkappa}\right)^{\al(g)-1} \left( \mathcal{t}(\varkappa;u) + \sqrt{u+\frac{1}{12\varkappa}} \widehat{\mathcal{t}}(\varkappa;u) \right),
		\end{equation}
		for some polynomials $\mathcal{t}$ and $\widehat{\mathcal{t}}$, with $\mathcal{t}(\varkappa;-\frac{1}{12\varkappa})=1$. Therefore \begin{equation}\label{AA2g 1}
		\frac{1}{2\varkappa^2} \mathcal{A}_{2g}(\varkappa;u) = \frac{1}{36\varkappa^2} \al(g)\widehat{\mathcal{C}}_{2g}(\varkappa)\left(u+\frac{1}{12\varkappa}\right)^{\al(g)-1} \left( \mathscr{P}(\varkappa;u) + \sqrt{u+\frac{1}{12\varkappa}} \widehat{\mathscr{P}}(\varkappa;u) \right),
		\end{equation}
		where $\mathscr{P}$ and $\widehat{\mathscr{P}}$ are polynomials, with $\mathscr{P}(\varkappa;-\frac{1}{12\varkappa})=1$. Let us rewrite $\mathcal{A}_{2g}$ as 
		\begin{equation}
		\begin{split}
		\mathcal{A}_{2g}(\varkappa;u) & = 2\mathcal{r}_{2g}(\varkappa;u)\mathcal{r}_{0}(\varkappa;u) + \mathcal{r}_{0}(\varkappa;u) \sum^g_{\ell=1} \frac{\mathcal{r}_{2g-2\ell}^{(2\ell)}(\varkappa;u)}{(2\ell)!}  \\ & \sum^{g-1}_{k=1} \mathcal{r}_{2g-2k}(\varkappa;u)\mathcal{r}_{2k}(\varkappa;u) + \sum^{g-1}_{k=1} \mathcal{r}_{2g-2k}(\varkappa;u)\sum^k_{\ell=1} \frac{\mathcal{r}_{2k-2\ell}^{(2\ell)}(\varkappa;u)}{(2\ell)!}
		\end{split}
		\end{equation}
		From Lemma \ref{main lemma'} we have \begin{equation}
		\mathcal{r}_{2k-2\ell}^{(2\ell)}(\varkappa;u) = c_{k,\ell}(\varkappa) \left(u+\frac{1}{12\varkappa}\right)^{\frac{1-5k+\ell}{2}}\left(1+O\left(\sqrt{u+\frac{1}{12\varkappa}}\right)\right), \qquad k,\ell \in \N.
		\end{equation}
		Therefore
		\begin{equation}
		\mathcal{r}_{2g-2k}(\varkappa;u)\mathcal{r}_{2k-2\ell}^{(2\ell)}(\varkappa;u) = c_{k,\ell}(\varkappa) \mathcal{C}_{2g}(\varkappa) \left(u+\frac{1}{12\varkappa}\right)^{\frac{2-5g+\ell}{2}}\left(1+O\left(\sqrt{u+\frac{1}{12\varkappa}}\right)\right),
		\end{equation} for $k,\ell \in \N$.
		Also from Lemma \ref{main lemma} we have
		\begin{equation}
		\mathcal{r}_{2g-2k}(\varkappa;u)\mathcal{r}_{2k}(\varkappa;u) = \mathcal{C}_{2g}(\varkappa)\mathcal{C}_{2g-2k}(\varkappa)\left(u+\frac{1}{12\varkappa}\right)^{\frac{2-5g}{2}}\left(1+O\left(\sqrt{u+\frac{1}{12\varkappa}}\right)\right),
		\end{equation}
		for $k\in\{1,\cdots,g-1\}$. In view of the equations above, \eqref{r0} and Lemma \ref{main lemma} we conclude that the most singular term in $\mathcal{A}_{2g}$ is in fact $2\mathcal{r}_{2g}(\varkappa;u)\mathcal{r}_{0}(\varkappa;u)$, Therefore
		\begin{equation}\label{AA2g 2}
		\mathcal{A}_{2g}(\varkappa;u) = 4\varkappa \mathcal{C}_{2g}(\varkappa)\left(u+\frac{1}{12\varkappa}\right)^{\frac{1-5g}{2}}  \left(\mathscr{Q}_{g}(\varkappa;u)+\sqrt{u+\frac{1}{12\varkappa}}\widehat{\mathscr{Q}}_{g}(\varkappa;u)\right),
		\end{equation}
		where $\mathscr{Q}$ and $\widehat{\mathscr{Q}}$ are polynomials, with $\mathscr{Q}(\varkappa;-\frac{1}{12\varkappa})=1$. Comparing \eqref{AA2g 1} and \eqref{AA2g 2} we find that $\mathscr{Q} = \mathscr{P}$, $\widehat{\mathscr{Q}}=\widehat{\mathscr{P}}$ and moreover \begin{equation}
		\al(g) = \frac{3-5g}{2},
		\end{equation}
		and thus
		\begin{equation}
		\widehat{\mathcal{C}}_{2g}(\varkappa) = \frac{2^43^2 \varkappa}{3-5g}\mathcal{C}_{2g}(\varkappa).
		\end{equation}
	\end{proof}
	We would like to point out that the branching singularity described by Lemma \ref{lemma 6.5.} has been observed in the physical literature, e.g. see \cite{DiFGZJ}.

	\subsection{Number of Four-valent Graphs on a Compact Riemann Surface of Genus $g$}
	In what follows we will use $\varkappa=1$, and simpler notations $\mathscr{E}_{2g}(1;u) \equiv \mathscr{E}_{2g}(u)$, $g \in \N \cup \{0\}$. From \eqref{E2gu} the first few $\mathscr{E}_{2g}$'s are explicitly by 
	\begin{equation}\label{EEEE0}
	\mathscr{E}_{0}(u) = \frac{1}{216u^3} \left[-1-18u+(1+12u)^{3/2} \right]-\frac{1}{4u}, 
	\end{equation}
	
	\begin{equation}\label{EEEE2}
	\mathscr{E}_{2}(u) = \frac{(1+12u)^{-1}}{24u} \left[ 1-\sqrt{1+12u} \right],
	\end{equation}
	
	\begin{equation}
	\mathscr{E}_{4}(u) = \frac{7u(1+12u)^{-7/2}}{4} \left[ 1 - \frac{1}{14} (1+12u) -\frac{13}{14}\sqrt{1+12u}
	\right], 
	\end{equation}
	
	\begin{equation}
	\begin{split}
	\mathscr{E}_{6}(u) & = \frac{2450u^3(1+12u)^{-6}}{4} \times \\ & \left[ 1+ \frac{291}{2450}(1+12u) + \left(-\frac{3033}{2450}+\frac{146}{1225}(1+12u)\right)\sqrt{1+12u} \right],
	\end{split}
	\end{equation}

	We can write power series expansions for $\mathscr{E}_{0}(u)$, $\mathscr{E}_{2}(u)$, $\mathscr{E}_{4}(u)$, and $\mathscr{E}_{6}(u)$ we will find:
	
	\begin{equation}
	\mathscr{E}_{0}(u) =  \sum^{\infty}_{j=0} \frac{(-1)^{j+1} 3^{j+1}(2j+1)!}{j!(j+3)!}u^j,
	\end{equation}
	and
	\begin{equation}
	\mathscr{E}_{2}(u) = \frac{1}{2} \sum^{\infty}_{j=0} (-1)^{j+1}12^{j} \left[ 1 - \frac{ (2j+2)!}{4^{j+1}((j+1)!)^2} \right]u^j,
	\end{equation}
	
	\begin{equation}
	\mathscr{E}_{4}(u) = \frac{1}{16} \sum^{\infty}_{j=0} (-1)^{j}12^{j} \left[ \frac{ (2j+5)!(28j+65)}{30 \cdot 4^{j}j!(j+2)!(2j+5)} - 13(j+2)(j+1) \right]u^{j+1},
	\end{equation}
	and
	\begin{equation}
	\begin{split}
	& \mathscr{E}_{6}(u) =  \frac{1}{4} \sum^{\infty}_{j=0}  \frac{(-1)^{j}12^{j}}{j!} \times \\ & \left\{ \frac{32892}{j+1} \left[ \frac{(2j+11)!}{60480 \cdot 4^j(j+5)!} -\frac{(j+6)!}{5!} \right] + \frac{ 291 (j+5)!}{10} +  \frac{ 73 (2j+9)!}{315 \cdot 4^j (j+4)!}  \right\}u^{j+4}. 
	\end{split}
	\end{equation}
	Integrating these from $0$ to $u$, and comparing to \eqref{Fu N^-2 expansion} we obtain

	\begin{equation}\label{Fu0}
	\mathcal{f}_{0}(u) =  \sum^{\infty}_{j=1} \frac{(-1)^{j} 3^{j}(2j-1)!}{(j)!(j+2)!}u^{j},
	\end{equation}
	and
	\begin{equation}\label{Fu2}
	\mathcal{f}_{2}(u) = \frac{1}{24} \sum^{\infty}_{j=1}  \frac{(-1)^{j}12^{j}}{j} \left[ 1 - \frac{ (2j)!}{4^{j}(j!)^2} \right]u^{j},
	\end{equation}
	
	\begin{equation}\label{Fu4}
	\mathcal{f}_{4}(u) =  \sum^{\infty}_{j=3}  \frac{(-1)^{j}12^{j}}{2304j} \left[\frac{8 (2j)!(28j+9)}{15 \cdot 4^{j}(j-2)! j!} -13j(j-1) \right]u^{j},
	\end{equation}
	and
	\[     \mathcal{f}_{6}(u) = \sum^{\infty}_{j=1} \mathcal{f}^{(j+4)}_{6} u^{j+4},\]
	where
	\begin{equation}\label{Fu6}
	\begin{split}
	& \mathcal{f}^{(j+4)}_{6}  = \frac{1}{48}  \frac{(-1)^{j}12^{j}}{(j-1)!(j+4)}\times \\ &  \Bigg\{ \frac{32892}{j} \left[\frac{(j+5)!}{5!} -\frac{(2j+9)!}{15120 \cdot 4^{j}(j+4)!}  \right]  - \frac{ 291 (j+4)!}{10} -  \frac{ 292 (2j+7)!}{315 \cdot 4^{j} (j+3)!} \Bigg\} .
	\end{split}
	\end{equation}
	
	\begin{remark}\normalfont
		The formulae \eqref{Fu0}, \eqref{Fu2}, and \eqref{Fu4} respectively reconfirm the formulae $(5.15)$, $(7.21)$, and $(7.33)$ of \cite{BIZ}. 
	\end{remark}

	The Wick's theorem and the Feynman graph representation (see, e.g. \cite{Zvonkin}) give that \begin{equation}\label{Fu to 4-valent graphs}
	\mathcal{f}_{2g}(u) = \sum^{\infty}_{j=1} \frac{(-1)^j\mathscr{N}_j(g)u^j}{j! 4^j},
	\end{equation}
	where $\mathscr{N}_j(g)$ is the number of connected labeled 4-valent graphs with $j$ vertices which are realizable on a closed Riemann surface of genus $g$ which are not realizable on Riemann surfaces of lower genus. Comparing \eqref{Fu to 4-valent graphs} with the formulae \eqref{Fu0} through \eqref{Fu6} we can compute $\mathscr{N}_j(g)$. For the sphere we find
	
	\begin{equation}\label{number sphere}
	\mathscr{N}_j(0) = \frac{12^j\left( 2j-1 \right)!}{(j+2)!}, \qquad j \in \N.
	\end{equation}
	For the torus we obtain
	\begin{equation}\label{number torus}
	\mathscr{N}_j(1) = \frac{12^j\left( 4^j(j!)^2-(2j)! \right)}{24j(j!)}, \qquad j \in \N.
	\end{equation}
	For the Riemann surface of genus two we get
	\begin{equation}\label{N j+1 2}
	\mathscr{N}_{j+1}(2) = \frac{12^j\left( 2j+2 \right)!(28j+37)}{360(j+1)(j-1)!}-13j(j+1)j!48^{j-1}, \qquad j \in \N,
	\end{equation}
	where $\mathscr{N}_{1}(2)=0$, which is clear as all three labeled 4-valent graphs with one vertex are realizable on the sphere and the torus. Also notice that the above formula yields $\mathscr{N}_{2}(2)=0$, which is consistent with the fact that all 96 labeled 4-valent graphs with two vertices are already realizable on the sphere and the torus (see Appendix \ref{Appendix Number of graphs}). Finally for the Riemann surface of genus three we arrive at
	\begin{equation}\label{N j+4 3 1}
	\begin{split}
	\mathscr{N}_{j+4}(3) & = \frac{16 \cdot 48^j\left(j+3\right)!}{3(j)!} \\ & \left( \frac{2741}{10}(j+5)!   - \frac{291}{10}j (j+4)! - \frac{2741}{1260}  \frac{ (2j+9)!}{ 4^j (j+4)!} - \frac{292j(2j+7)!}{315 \cdot 4^j (j+3)!} \right),
	\end{split}
	\end{equation}
	for $j \in \N,$ where $\mathscr{N}_{1}(3)=\mathscr{N}_{2}(3)=\mathscr{N}_{3}(3)=\mathscr{N}_{4}(3)=0$. 
	\begin{remark}\normalfont
		We emphasize that, with increasing effort, similar analysis allows for computing any $\mathscr{N}_{j}(g)$, $j,g \in \N$.
	\end{remark}
	
	We have the following asymptotic formulae for $\mathscr{N}_{j}(g)$, as $j \to \infty$ for $g=0,1,2,3$:
	
	\begin{equation}\label{Nj0 asymp}
	\mathscr{N}_{j}(0) =\frac{1}{\sqrt{2}j^3} \left(\frac{48j}{e}\right)^j \left( 1 - \frac{73}{24j} + \frac{8209}{1152j^2}-  \frac{6341837}{414720j^3} + O(j^{-4}) \right),
	\end{equation}
	
	\begin{equation}\label{Nj1 asymp}
	\mathscr{N}_{j}(1) =\frac{\sqrt{2\pi}}{24\sqrt{j
	}} \left(\frac{48j}{e}\right)^j \left( 1 - \frac{1}{\sqrt{\pi}\sqrt{j}} + \frac{1}{12j}+  \frac{1}{24\sqrt{\pi}j^{3/2}} + O(j^{-2}) \right),
	\end{equation}
	
	\begin{equation}\label{Nj2 asymp}
	\mathscr{N}_{j}(2) =\frac{7\sqrt{2} j^2}{1080} \left(\frac{48j}{e}\right)^j \left( 1 - \frac{195\sqrt{\pi}}{224\sqrt{j}} - \frac{121}{168j}+  \frac{715\sqrt{\pi}}{896j^{3/2}} + O(j^{-2}) \right),
	\end{equation}
	
	\begin{equation}\label{Nj3 asymp}
	\mathscr{N}_{j}(3) =\frac{245\sqrt{2\pi} j^{9/2}}{995328} \left(\frac{48j}{e}\right)^j \left( 1 - \frac{43136}{8575\sqrt{\pi}\sqrt{j}} - \frac{12709}{2940j}+  \frac{30928}{1029\sqrt{\pi}j^{3/2}} + O(j^{-2}) \right),
	\end{equation}
	Denote
	\begin{equation}\label{F2g j and N g j}
	\mathcal{f}^{(j)}_{2g} := \frac{(-1)^j\mathscr{N}_j(g)}{j! 4^j}, \qquad \mbox{thus} \qquad 
	\mathcal{f}_{2g}(u) = \sum^{\infty}_{j=1} \mathcal{f}^{(j)}_{2g} u^j.
	\end{equation}

	\begin{theorem}\label{thm F 2g j}
		For all $g \in \N \cup \{0\}$, as $j \to \infty$ we have
		\begin{equation}\label{F2gj large j}
		\mathcal{f}^{(j)}_{2g} = \frac{\mathcal{K}_g}{u^{j}_c}j^{\frac{5g-7}{2}} \left( 1 + O(j^{-1/2}) \right),
		\end{equation}
		where $u_c=-\di \frac{1}{12}$. The constants $\mathcal{K}_g$ are explicitly given in terms of the constants $\mathcal{C}_{2g}$ by 
		\begin{equation}\label{K g in terms of Cg}
		\mathcal{K}_g = \begin{cases}
		\di \frac{12^{\frac{5g-1}{2}}}{\left(\frac{5g-5}{2}\right)!}\frac{1}{5g-3} \mathcal{C}_{2g}, & g =2k+1, \\[10pt]
		\di \frac{ 12^{\frac{5g-1}{2}} 2^{5g-4}}{\sqrt{\pi}}\frac{\left(\frac{5g-4}{2}\right)!}{\left(5g-3\right)!}\mathcal{C}_{2g}, & g =2k,
		\end{cases} \qquad k \in \N,
		\end{equation}
		while $\mathcal{K}_0=2^{-1}\pi^{-1/2}$, and $\mathcal{K}_1=24^{-1}$.
	\end{theorem}
	
	\begin{proof}
		For $g=0$ and $g=1$ the expression \eqref{F2gj large j} with the quantities $\mathcal{K}_0=2^{-1}\pi^{-1/2}$ and $\mathcal{K}_1=24^{-1}$ can be immediately found from \eqref{Nj0 asymp}, \eqref{Nj1 asymp}, \eqref{F2g j and N g j}, and the Stirling formula. For $g \in \N \setminus \{1\}$, by \eqref{E2g C2g} we have that \begin{equation}\label{E2g series}
		\mathscr{E}_{2g}(u) =  \sum^{\mathscr{m}_g}_{\ell=0} \mathscr{C}_{\ell,g} \left(u+\frac{1}{12}\right)^{\frac{3-5g+\ell}{2}},  \qquad \mathscr{C}_{0,g} = \frac{2^43^2}{3-5g}\mathcal{C}_{2g}, 
		\end{equation}
		for some $\mathscr{m}_g \in \N$ and some constants $\mathscr{C}_{\ell,g}$. For these finite sum of powers of $u-u_c$, we can use the generalized binomial theorem to write their Taylor series, and thus the Taylor series of $\mathscr{E}_{2g}$, centered at zero with radius of convergence $\frac{1}{12}$.  Therefore
		\begin{equation}\label{E2gu Taylor}
		\mathscr{E}_{2g}(u) = \sum^{\infty}_{j=0} \mathscr{E}^{(j)}_{2g} u^j, \qquad |u| \leq \frac{1}{12},
		\end{equation}
		where
		\begin{equation}
		\mathscr{E}^{(j)}_{2g} =12^j \sum^{\mathscr{m}_g}_{\ell=0}\mathscr{C}_{\ell,g} 12^{\frac{5g-3-\ell}{2}} \binom{\frac{3-5g+\ell}{2}}{j}.
		\end{equation}
		By integrating \eqref{E2gu Taylor} we find
		\begin{equation}
		\mathcal{f}_{2g}(u) = \sum^{\infty}_{j=1} \mathcal{f}^{(j)}_{2g} u^j,
		\end{equation}where 
		\begin{equation}\label{F 2g j}
		\mathcal{f}^{(j)}_{2g} =\frac{12^{j-1}}{j} \sum^{\mathscr{m}_g}_{\ell=0}\mathscr{C}_{\ell,g} 12^{\frac{5g-3-\ell}{2}} \binom{\frac{3-5g+\ell}{2}}{j-1}.
		\end{equation}
		Notice that for $\ell \in \R$ and as $j \to \infty$, one has $\binom{r+\ell}{j}/\binom{r}{j} = c(r,\ell) j^{-\ell}\left(1+O(
		j^{-1})\right)$. Therefore, considering the large $j$ asymptotics, the main contribution in \eqref{F 2g j} comes from $\ell=0$ and we thus have
		\begin{equation}\label{F 2g j main contribution}
		\mathcal{f}^{(j)}_{2g} =  \frac{12^{j+\frac{5g-5}{2}}}{j} \mathscr{C}_{0,g} \binom{\frac{3-5g}{2}}{j-1}\left( 1 + O\left(\frac{1}{\sqrt{j}}\right) \right).
		\end{equation} 
		By a straightforward calculation we find
		\begin{equation}
		\di \binom{\frac{3-5g}{2}}{j-1} = \begin{cases}
		\di \frac{(-1)^{j-1} \left(j+\frac{5g-7}{2} \right)!}{\left(\frac{5g-5}{2} \right)!(j-1)!}, & g=2k+1, \\[20pt]
		\di \frac{(-1)^{j-1} \left(\frac{5g-4}{2} \right)!\left(2j+5g-7 \right)!}{2^{2j-3}(5g-4)!(j-1)!\left(j+\frac{5g-8}{2} \right)!}, & g=2k,
		\end{cases} \qquad k \in \N.
		\end{equation}
		Therefore by \eqref{F 2g j main contribution} we have

		\begin{equation}
		\mathcal{f}^{(j)}_{2g} =  12^{j+\frac{5g-5}{2}} \mathscr{C}_{0,g}\begin{cases}
		\di \frac{(-1)^{j-1} \left(j+\frac{5g-7}{2} \right)!}{\left(\frac{5g-5}{2} \right)!j!}\left( 1 + O\left(\frac{1}{\sqrt{j}}\right) \right), & g=2k+1, \\[20pt]
		\di \frac{(-1)^{j-1} \left(\frac{5g-4}{2} \right)!\left(2j+5g-7 \right)!}{2^{2j-3}(5g-4)!j!\left(j+\frac{5g-8}{2} \right)!}\left( 1 + O\left(\frac{1}{\sqrt{j}}\right) \right), & g=2k,
		\end{cases}
		\end{equation} 
		for $ k \in \N.$ Now, the formulae \eqref{F2gj large j} and \eqref{K g in terms of Cg} immediately follow by applying the Stirling's formula and using the second member of \eqref{E2g series}.
	\end{proof}
	
	\begin{corollary}
		The asymptotics of the number of connected labeled $4$-valent graphs on a Riemann surface of genus $g$, as the number of vertices tends to infinity, is given by \begin{equation}\label{N j g large j asymptotics}
		\mathscr{N}_j(g) = \mathcal{K}_g 48^j j! j^{\frac{5g-7}{2}}\left(1+O(j^{-1/2}) \right), \qquad j \to \infty, 
		\end{equation}
		where the constants $\mathcal{K}_g$ are the same as the ones in Theorem \ref{thm F 2g j}.
	\end{corollary}
	\begin{remark}\normalfont
		One shall check that the description \eqref{K g in terms of Cg}-\eqref{N j g large j asymptotics} is in agreement with the asymptotic expressions \eqref{Nj2 asymp} and \eqref{Nj3 asymp} obtained from the explicit formulae \eqref{N j+1 2} and \eqref{N j+4 3 1} for $\mathscr{N}_j(2)$ and $\mathscr{N}_j(3)$. Checking this agreement requires knowing the values of $ \mathcal{C}_4$ and $\mathcal{C}_6$, which can be recursively found from \eqref{C2g recurrence} and \eqref{C0}. We find
		\begin{equation}
		\mathcal{C}_4=\frac{7^2}{2^{15}3^{13/2}}, \qandq \mathcal{C}_6=\frac{5^27^2}{2^{21}3^{10}}.
		\end{equation}
		Using these, from \eqref{K g in terms of Cg} we obtain \begin{equation}
		\mathcal{K_2}= \frac{7}{1080\sqrt{\pi}} \qandq \mathcal{K_3}= \frac{245}{995328},
		\end{equation}
		which together with \eqref{N j g large j asymptotics} are in agreement with \eqref{Nj2 asymp} and \eqref{Nj3 asymp} respectively.
	\end{remark}

\section{Appendix: Number of Labeled Connected Four-valent Graphs \\ With One or Two Vertices on the Sphere and the Torus}\label{Appendix Number of graphs}

In this appendix we would like show some illustrations for a deeper understanding of \eqref{number sphere} and \eqref{number torus}. We specifically do this for graphs with one and two vertices where the number of graphs are not too large, which allows for a complete discussion here. To this end, notice that \eqref{number sphere} and \eqref{number torus} give \begin{equation}\label{numbers j=1 and j=2}
\mathscr{N}_1(0)=2, \qquad 	\mathscr{N}_1(1)=1, \qquad	\mathscr{N}_2(0)=36, \qquad	\mathscr{N}_2(1)=60.
\end{equation}
It is easy to verify the first two members of \eqref{numbers j=1 and j=2}. Consider a labeled 4-valent graph with a single vertex $v$. label the edges emanating from $v$ in a counterclockwise way by $e_1, e_2, e_3,$ and $e_4$. For the simplicity of the Figures \ref{fig: g=0,1 j=2, e1->e_2}, \ref{fig: g=0,1 j=2, e1->e_6}, \ref{fig: g=1 j=2, e1->e_2}, \ref{fig: g=1 j=2, e1->e_3}, and \ref{fig: g=1 j=2, e1->e_6} an edge $e_k$ will be simply denoted by $k$ on the graphs. It is clear that there are two ways to make a desired graph on the sphere in which one connects the adjacent edges ($\mathscr{N}_1(0)=2$). The graph in which the opposite edges are connected can not be realized on the sphere, but can be realized on the torus($\mathscr{N}_1(1)=1$). 

Now we justify the third member of \eqref{numbers j=1 and j=2}. Label the vertices by $v_1$ and $v_2$. Label the edges emanating from $v_1$ in a counterclockwise way by $e_1, e_2, e_3, e_4$, and label the edges emanating from $v_2$ in a counterclockwise way by $e_5,e_6,e_7,e_8$. When $e_j$ connects to $e_k$ we use the notation $e_j \leftrightarrow e_k$.  

We exhaust all possibilities for which 4-valent  connected labeled graphs with two vertices can be realized on the sphere. Let us start with $e_1$. This edge can be connected to any other labeled edge except for $e_3$, because obviously if these edges are connected then either $e_2$ or $e_4$ would have no destination. Now we show that $e_1$ can be connected to $e_2$ or $e_4$ in eight distinct graphs, while it can be connected to either $e_5, e_6, e_7,$ or $e_8$ in five distinct graphs, which confirms that $\mathscr{N}_2(0)=2\cdot8+4\cdot5=36$. Figure \ref{fig: g=0,1 j=2, e1->e_2} shows all eight connected labeled 4-valent graphs with two vertices on the sphere with a connection between $e_1$ and $e_2$, while Figure \ref{fig: g=0,1 j=2, e1->e_6} shows all five connected labeled 4-valent graphs with two vertices on the sphere with a connection between $e_1$ and $e_6$.

\begin{figure}[h]
	\centering
	\begin{subfigure}{0.2\textwidth}
		\centering
		\includegraphics[width=\textwidth]{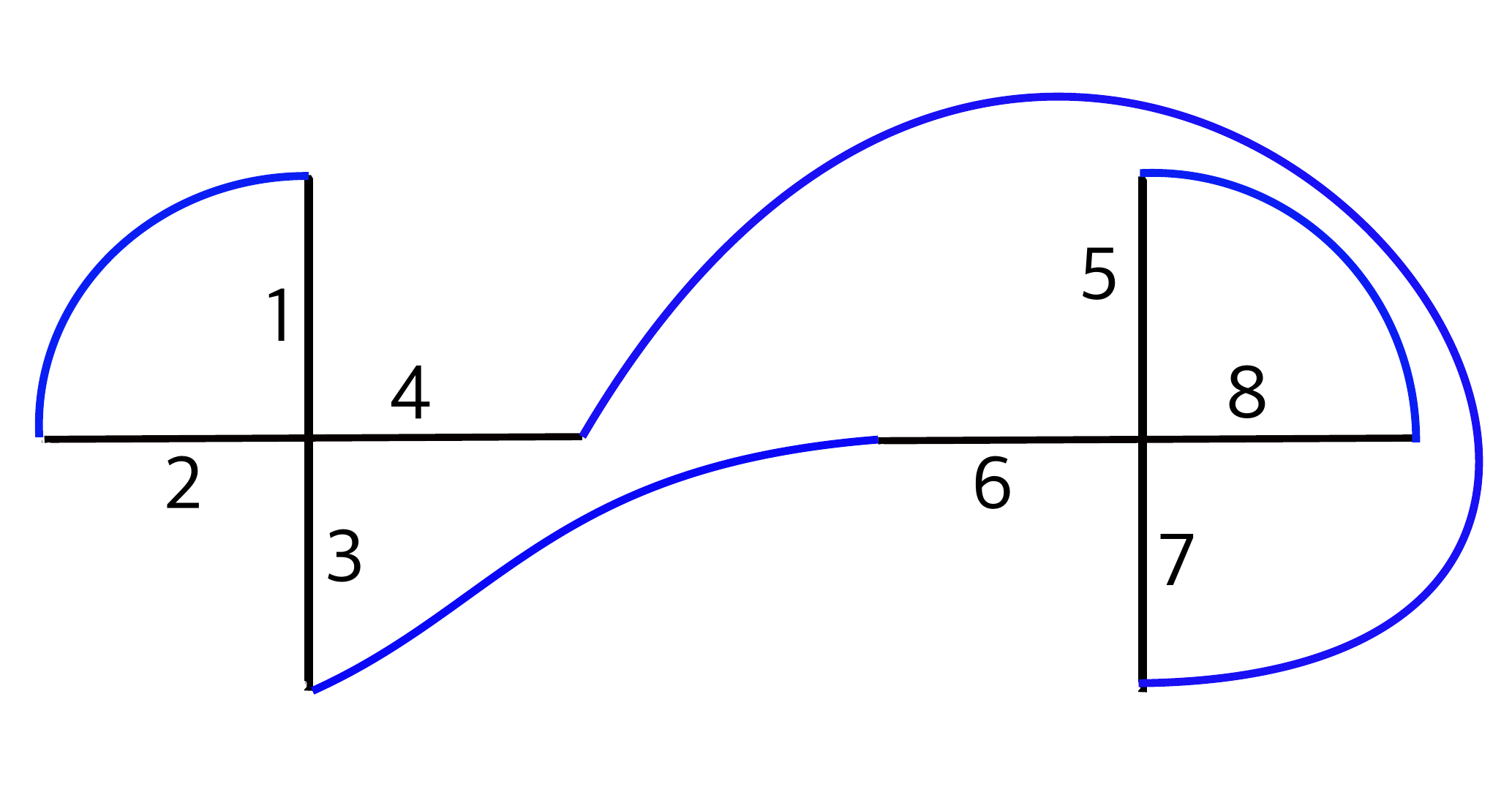}
	\end{subfigure} \hspace{0.3cm}
	\begin{subfigure}{0.2\textwidth}
		\centering
		\includegraphics[width=\textwidth]{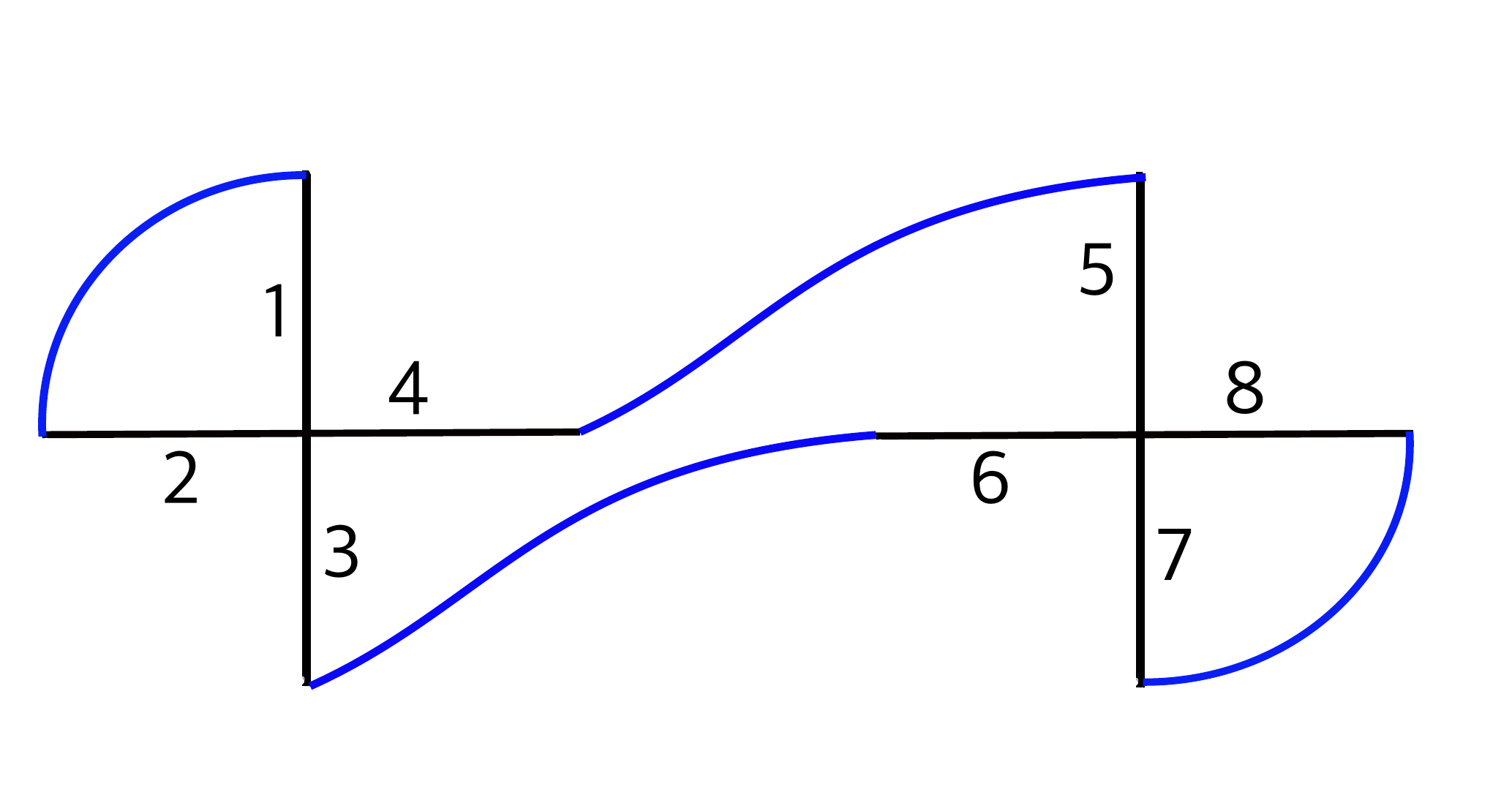}
	\end{subfigure} \hspace{0.3cm}
	\begin{subfigure}{0.2\textwidth}
		\centering
		\includegraphics[width=\textwidth]{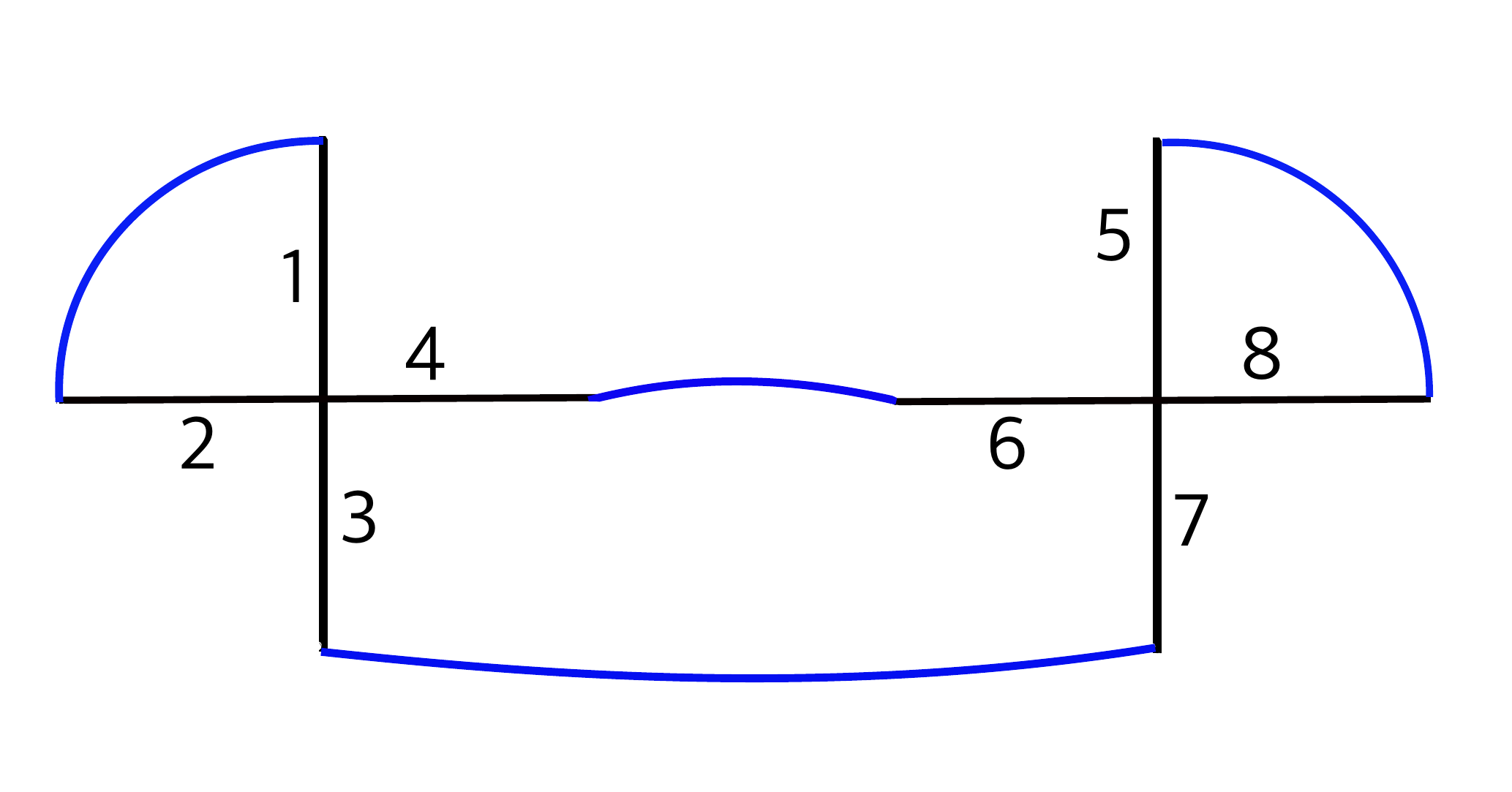}
	\end{subfigure} \hspace{0.3cm}
	\begin{subfigure}{0.2\textwidth}
		\centering
		\includegraphics[width=\textwidth]{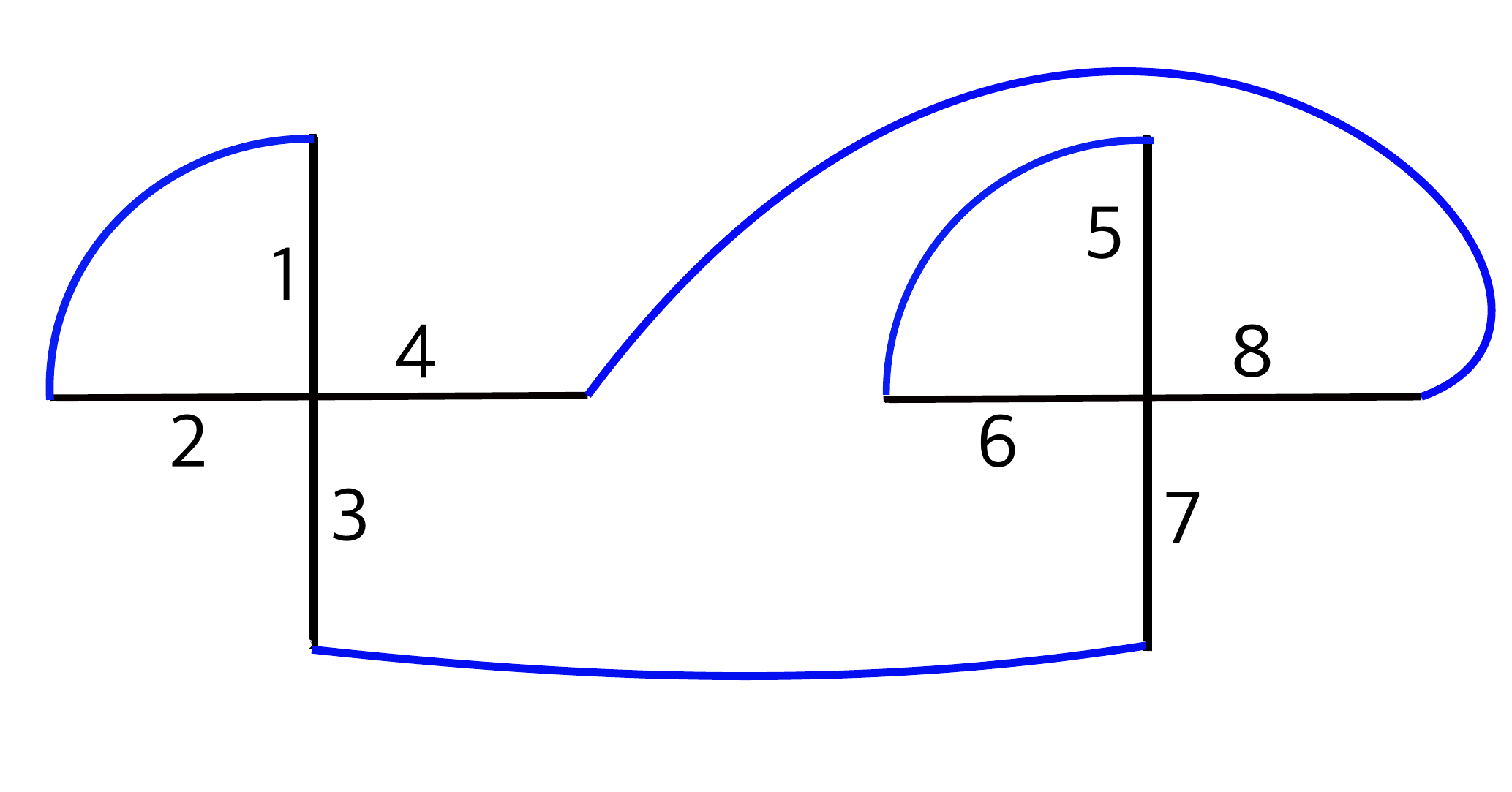}
	\end{subfigure}%
	
	\medskip
	
	\centering
	\begin{subfigure}{0.2\textwidth}
		\centering
		\includegraphics[width=\textwidth]{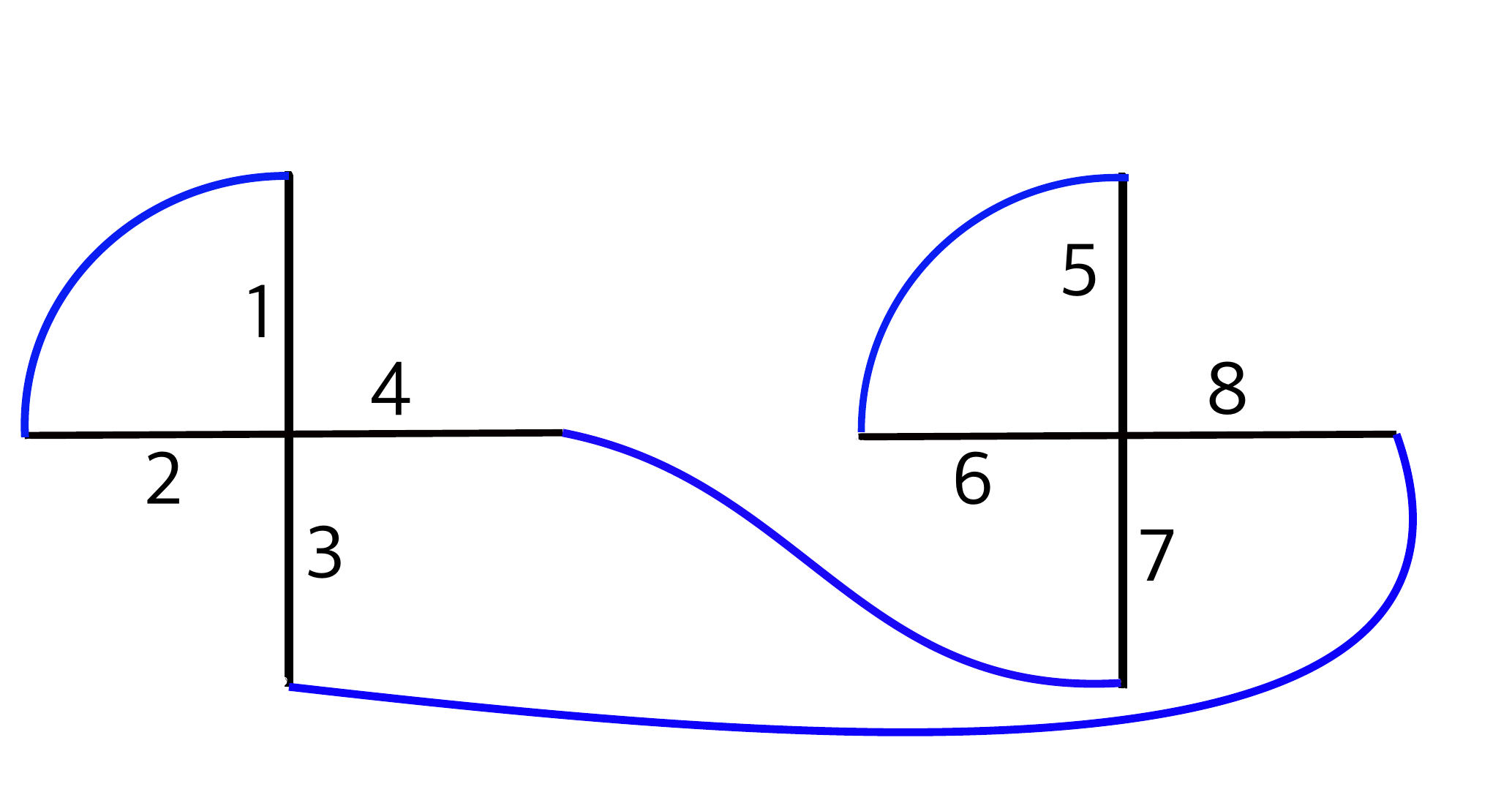}
	\end{subfigure} \hspace{0.3cm}
	\begin{subfigure}{0.2\textwidth}
		\centering
		\includegraphics[width=\textwidth]{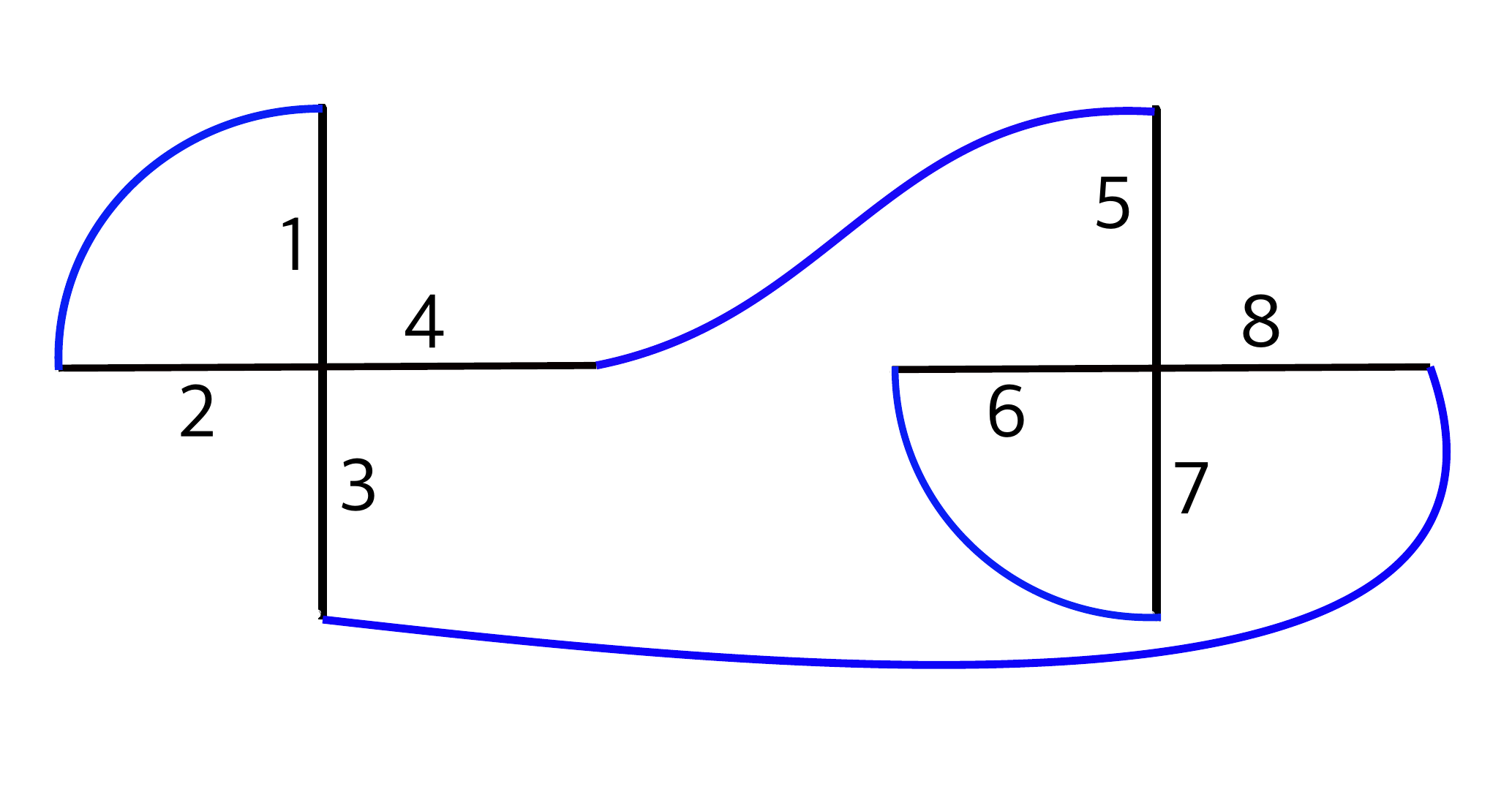}
	\end{subfigure} \hspace{0.3cm}
	\begin{subfigure}{0.2\textwidth}
		\centering
		\includegraphics[width=\textwidth]{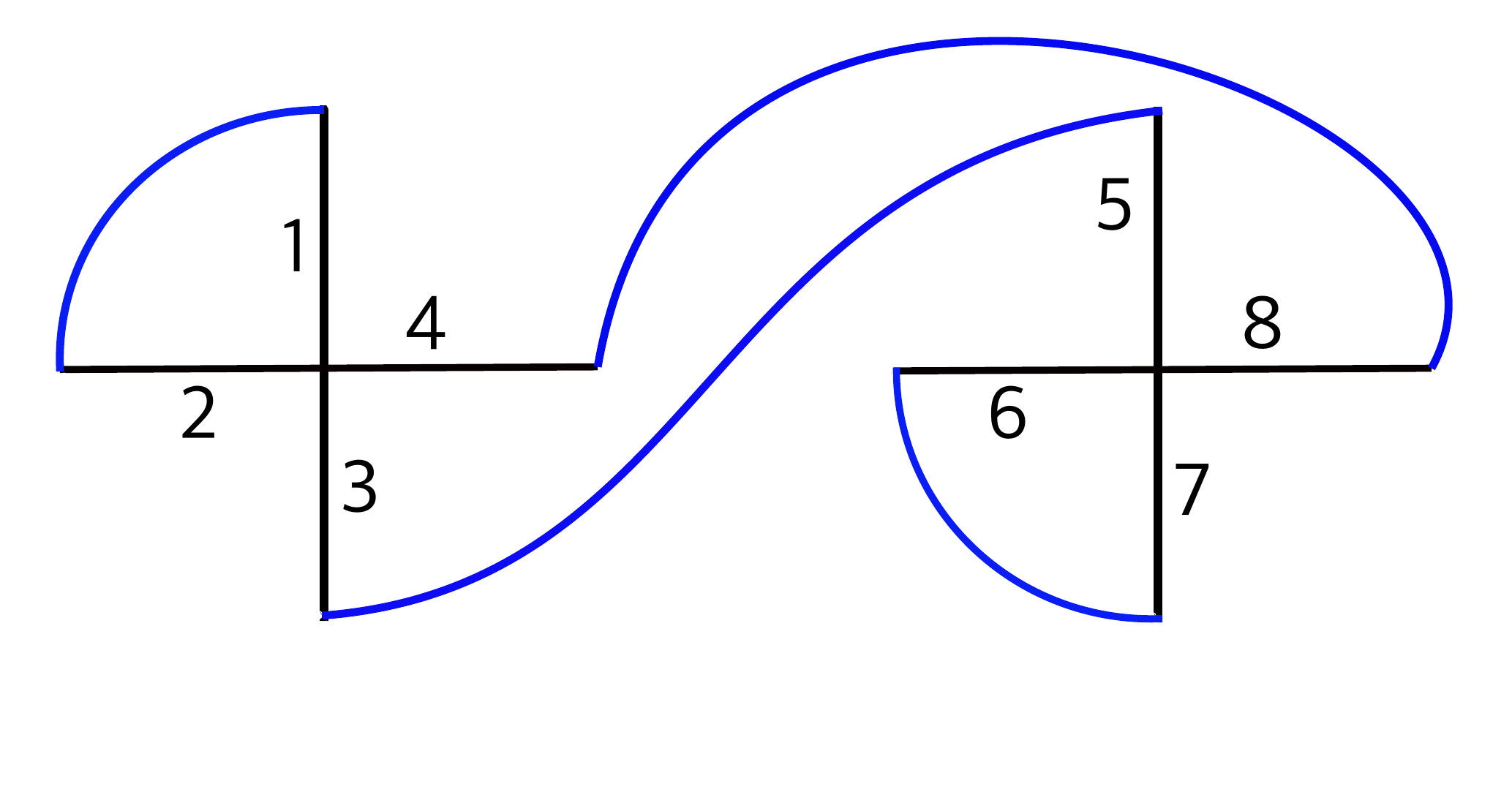}
	\end{subfigure} \hspace{0.3cm}
	\begin{subfigure}{0.2\textwidth}
		\centering
		\includegraphics[width=\textwidth]{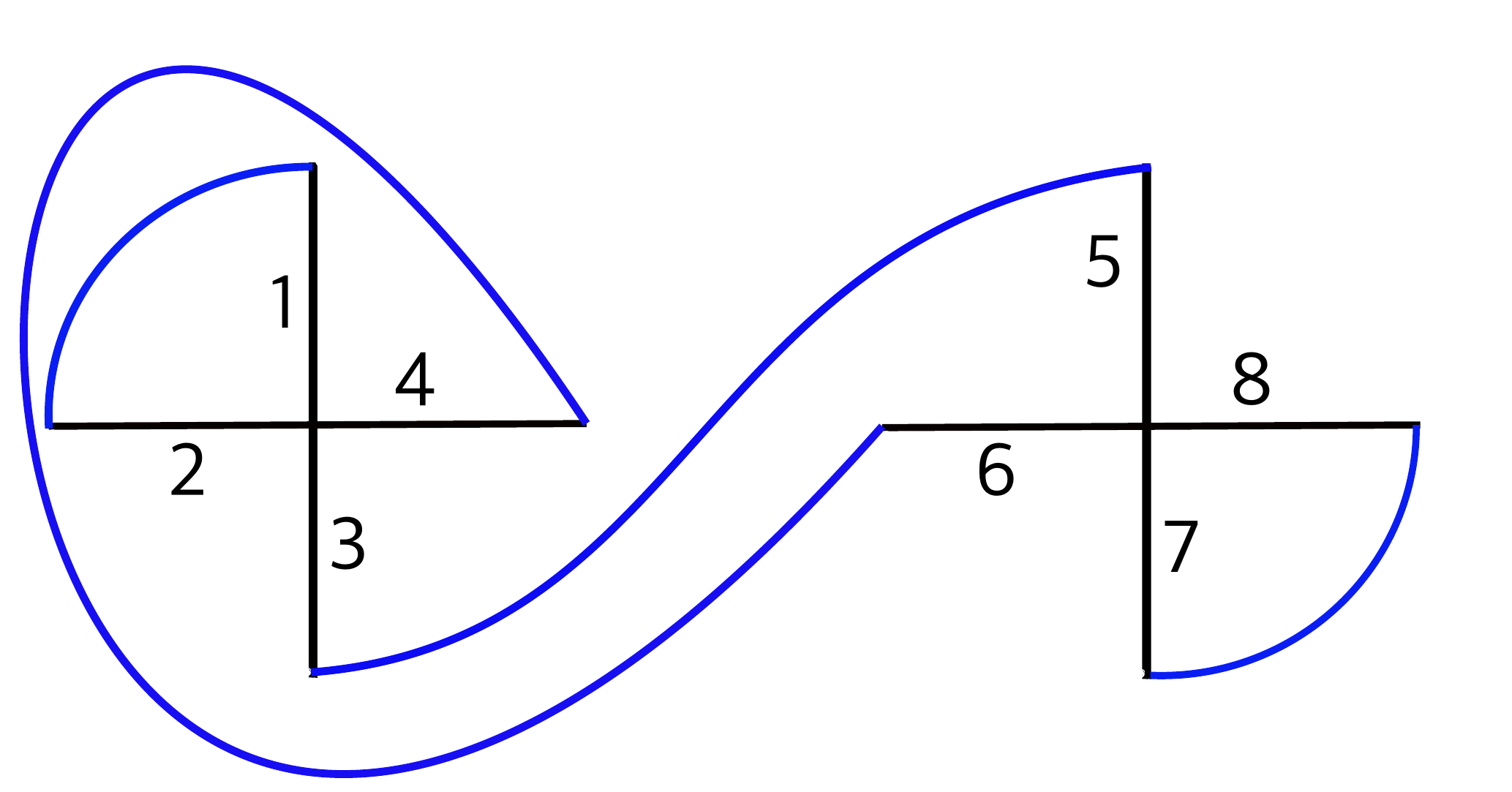}
	\end{subfigure}%
	\caption{All eight labeled connected 4-valent graphs with two vertices, where $e_1$ connects to $e_2$ and realizable on the sphere. Identically, for the case where $e_1$ connects to $e_4$, there are also eight distinct graphs. For the simplicity of the Figures \ref{fig: g=0,1 j=2, e1->e_2}, \ref{fig: g=0,1 j=2, e1->e_6}, \ref{fig: g=1 j=2, e1->e_2}, \ref{fig: g=1 j=2, e1->e_3}, and \ref{fig: g=1 j=2, e1->e_6} an edge $e_k$ will be simply denoted by $k$ on the graphs.}
	\label{fig: g=0,1 j=2, e1->e_2}
\end{figure}

\begin{figure}[h]
	\centering
	\begin{subfigure}{0.2\textwidth}
		\centering
		\includegraphics[width=\textwidth]{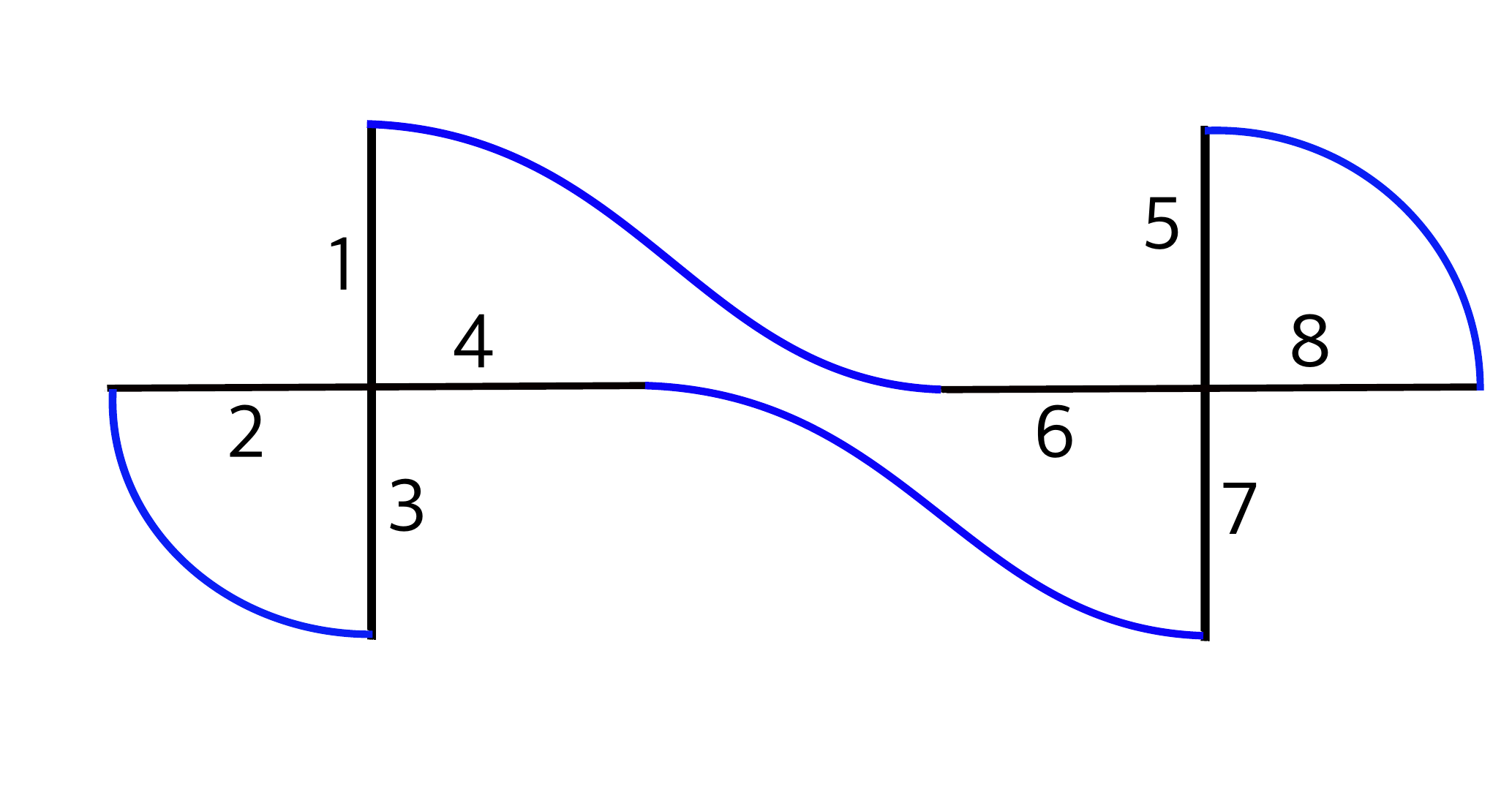}
	\end{subfigure} \hspace{0.5cm}
	\begin{subfigure}{0.2\textwidth}
		\centering
		\includegraphics[width=\textwidth]{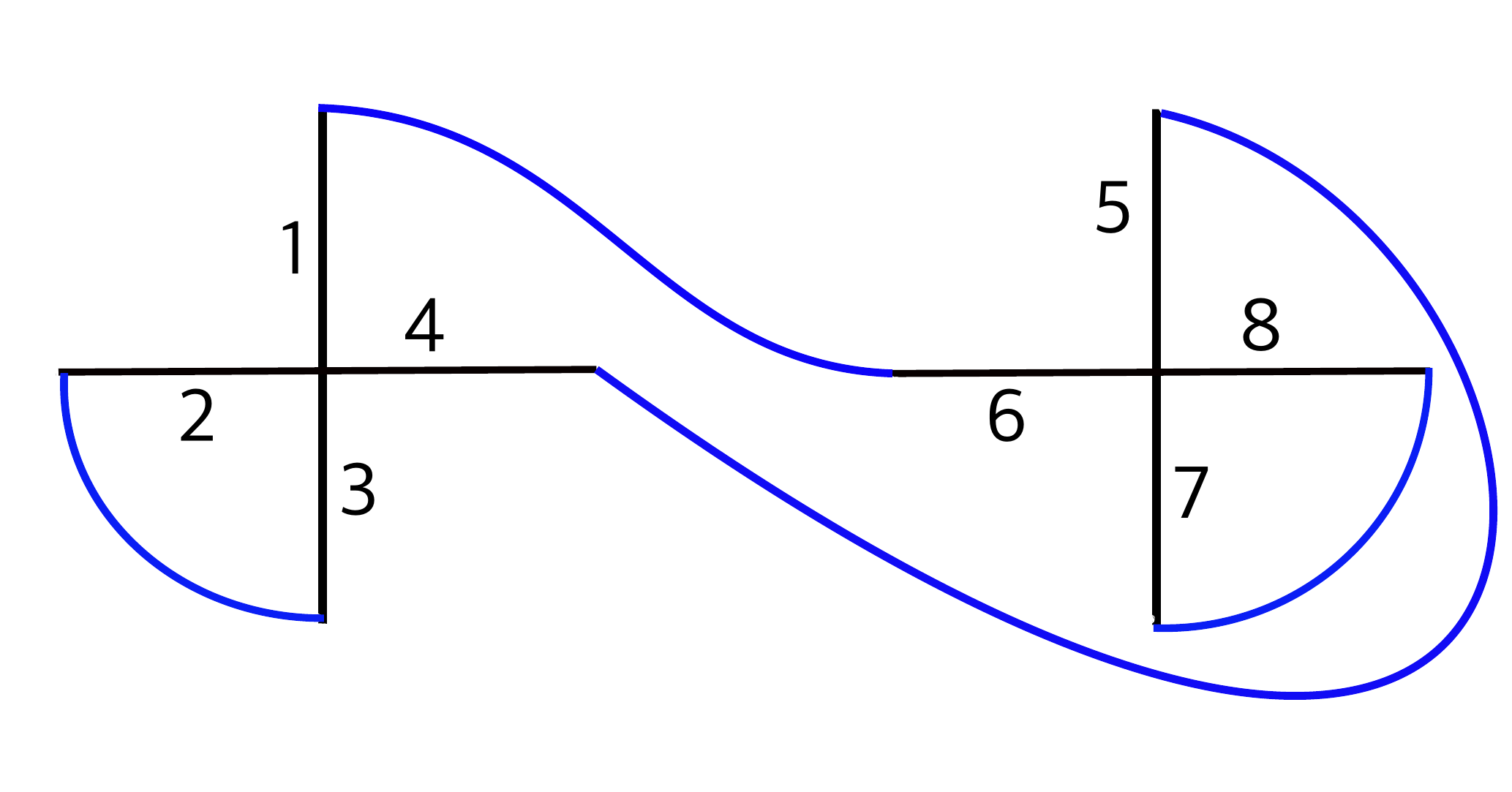}
	\end{subfigure} \hspace{0.5cm}
	\begin{subfigure}{0.2\textwidth}
		\centering
		\includegraphics[width=\textwidth]{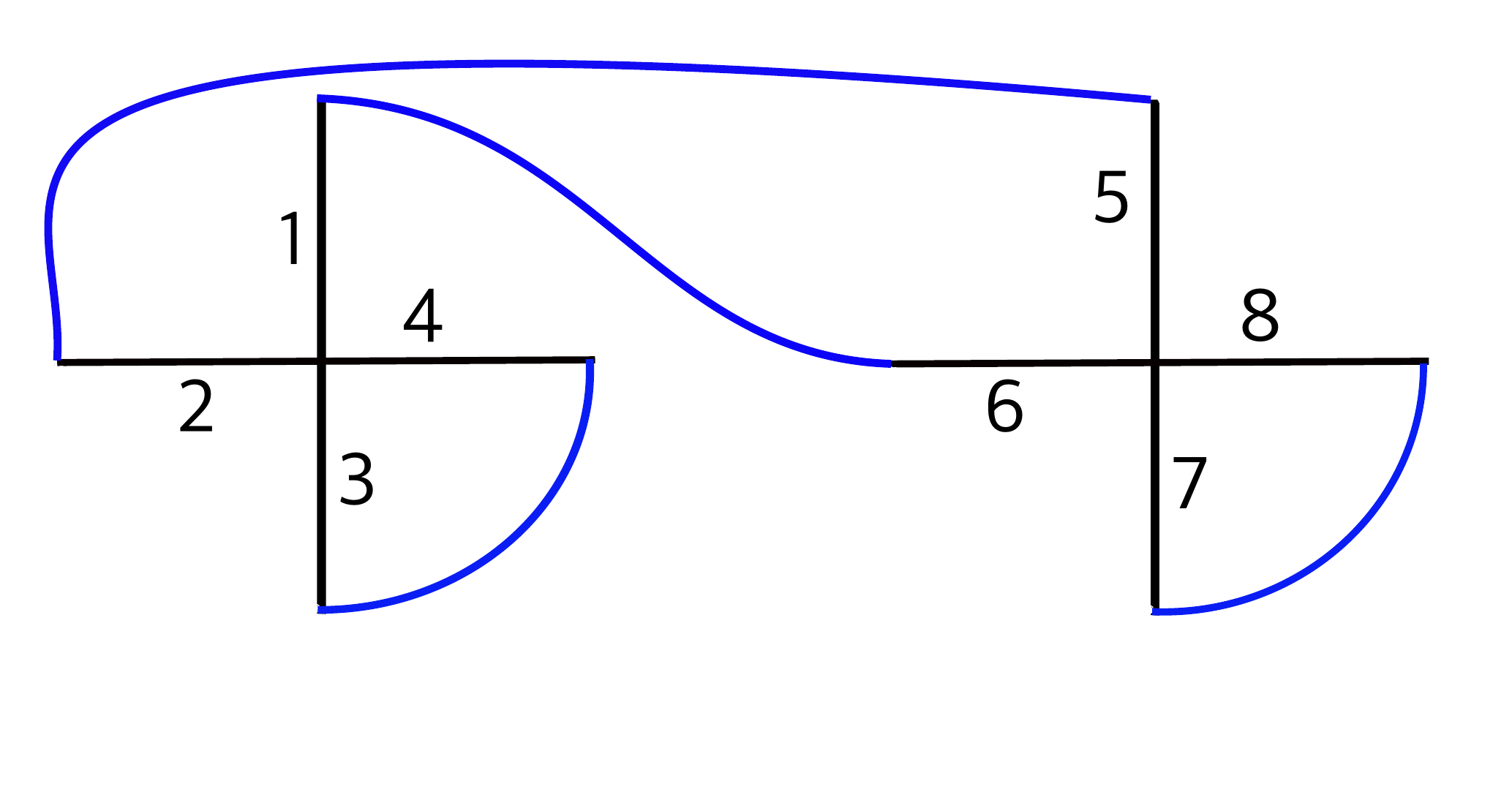}
	\end{subfigure}
	
	\medskip
	
	\centering
	\begin{subfigure}{0.2\textwidth}
		\centering
		\includegraphics[width=\textwidth]{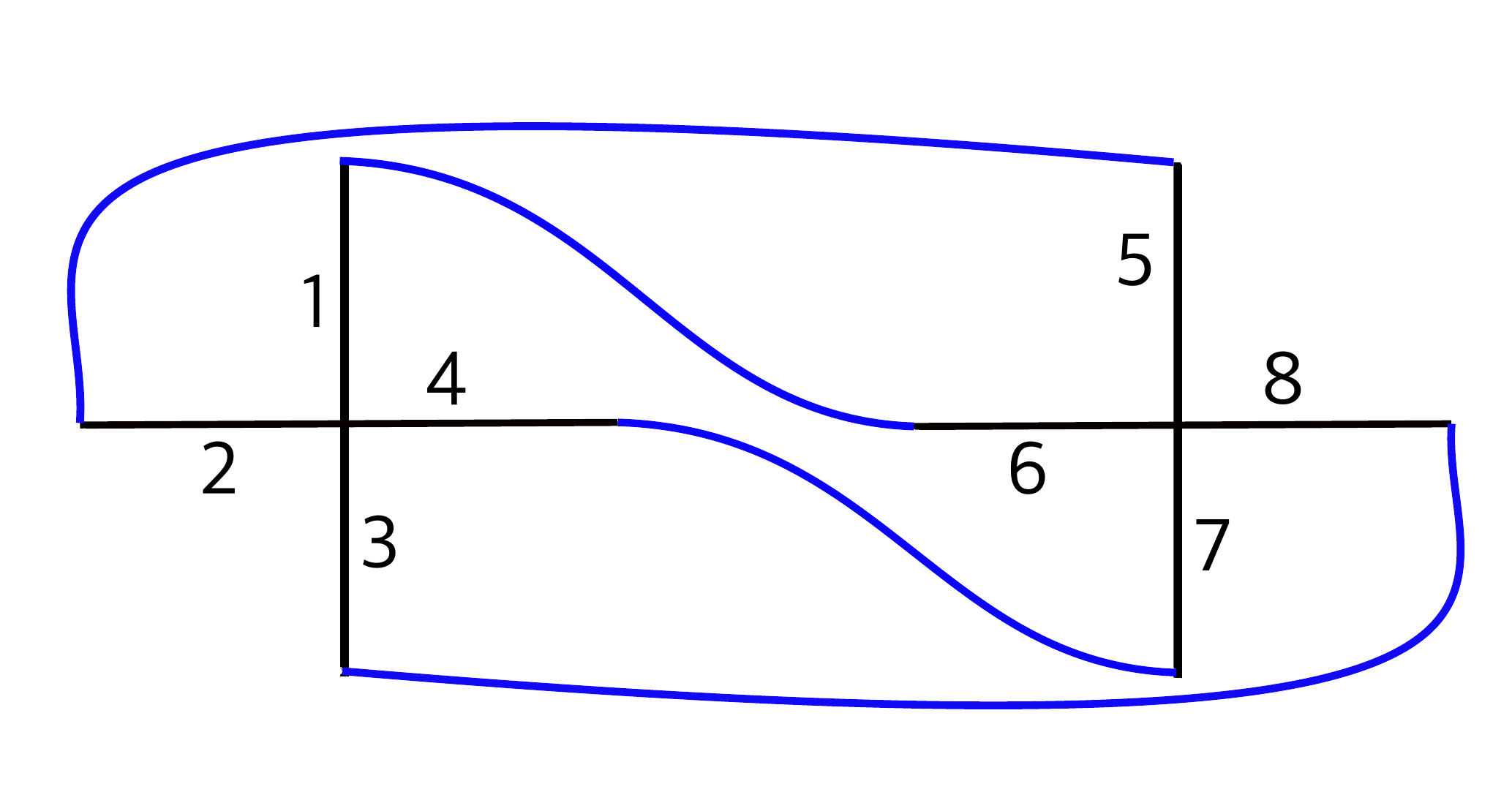}
	\end{subfigure} \hspace{0.5cm}
	\begin{subfigure}{0.2\textwidth}
		\centering
		\includegraphics[width=\textwidth]{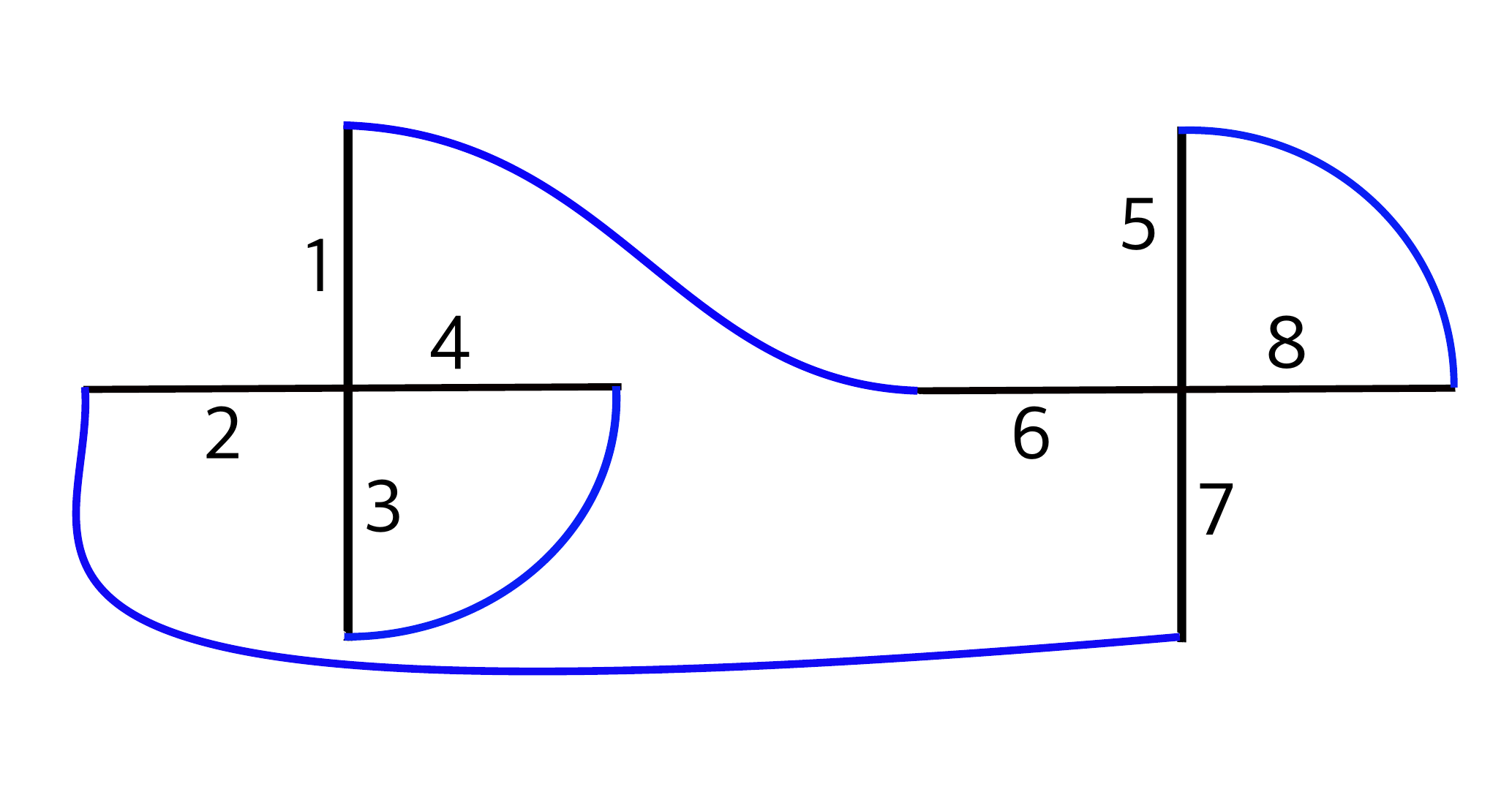}
	\end{subfigure} 
	\caption{All five labeled connected 4-valent graphs with two vertices, where $e_1$ connects to $e_6$ and realizable on the sphere. Identically, for each of the cases where $e_1$ connects to $e_5$, $e_7$, or $e_8$, there are also five distinct graphs. this Figure together with Figure \ref{fig: g=0,1 j=2, e1->e_2} confirm that $\mathscr{N}_2(0)=2\cdot8+4\cdot5=36$.}
	\label{fig: g=0,1 j=2, e1->e_6}
\end{figure}

Now we try to justify that $\mathscr{N}_2(1)=60$. Notice that there are three distinct graphs if one enforces two pairwise connections. Since we have already realized two graphs with $e_1 \leftrightarrow e_2$ and $e_3 \leftrightarrow e_6$ on the sphere (see the first two graphs in Figure \ref{fig: g=0,1 j=2, e1->e_2}) there is only one remaining graph with  $e_1 \leftrightarrow e_2$ and $e_3 \leftrightarrow e_6$ to be realized on the torus. Similarly, there is one graph left to be realized on the torus with the specifications: 1) $e_1 \leftrightarrow e_2$ and $e_3 \leftrightarrow e_7$ , 2) $e_1 \leftrightarrow e_2$ and $e_3 \leftrightarrow e_8$ , and 3) $e_1 \leftrightarrow e_2$ and $e_3 \leftrightarrow e_5$. Figure \ref{fig: g=1 j=2, e1->e_2} shows all graphs with  $e_1 \leftrightarrow e_2$ which can be realized on the torus but not on the sphere.

\begin{figure}[!htp]
	\centering
	\begin{subfigure}{0.24\textwidth}
		\centering
		\includegraphics[width=\textwidth]{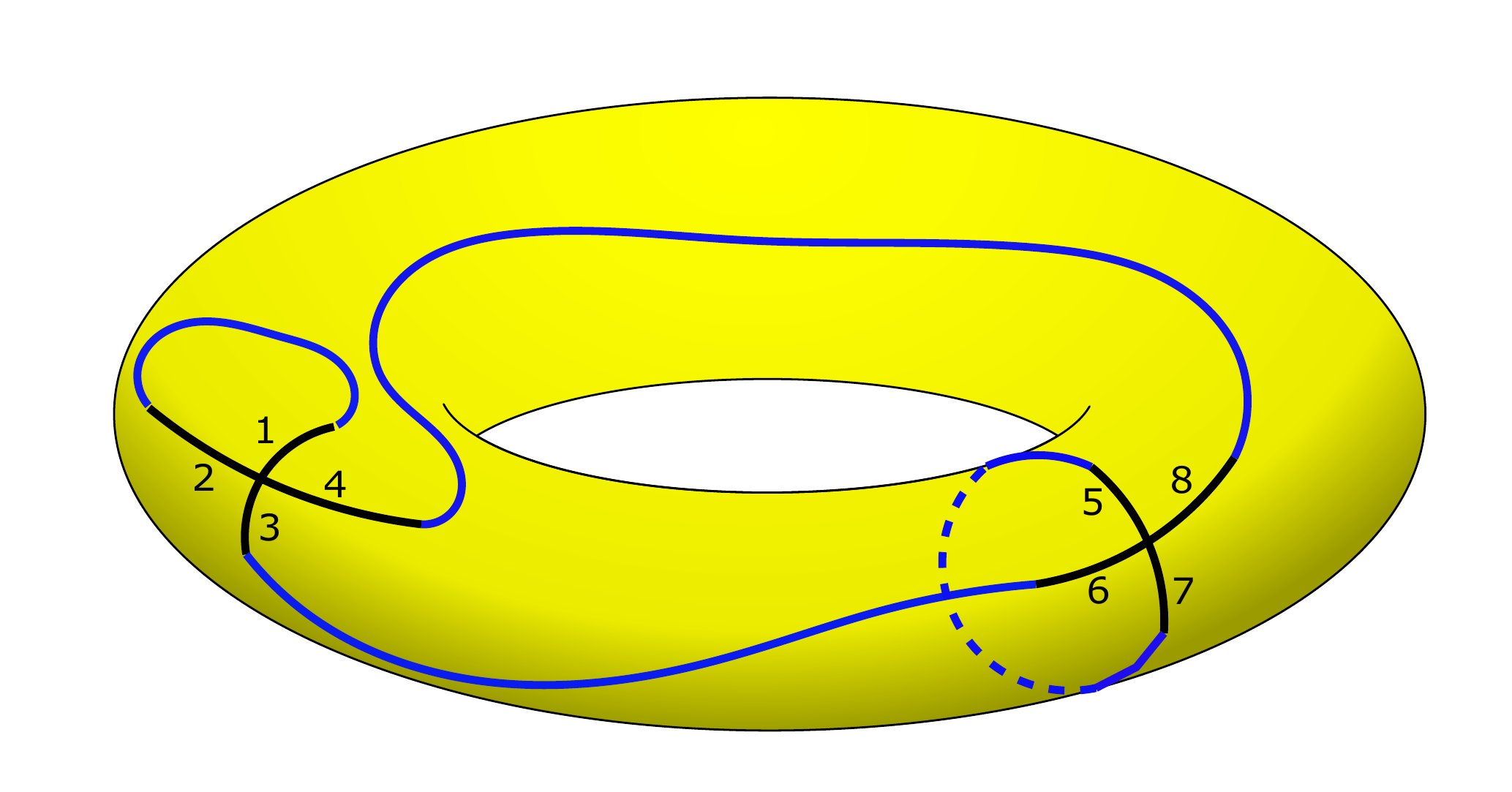}
	\end{subfigure} \hfil
	\begin{subfigure}{0.24\textwidth}
		\centering
		\includegraphics[width=\textwidth]{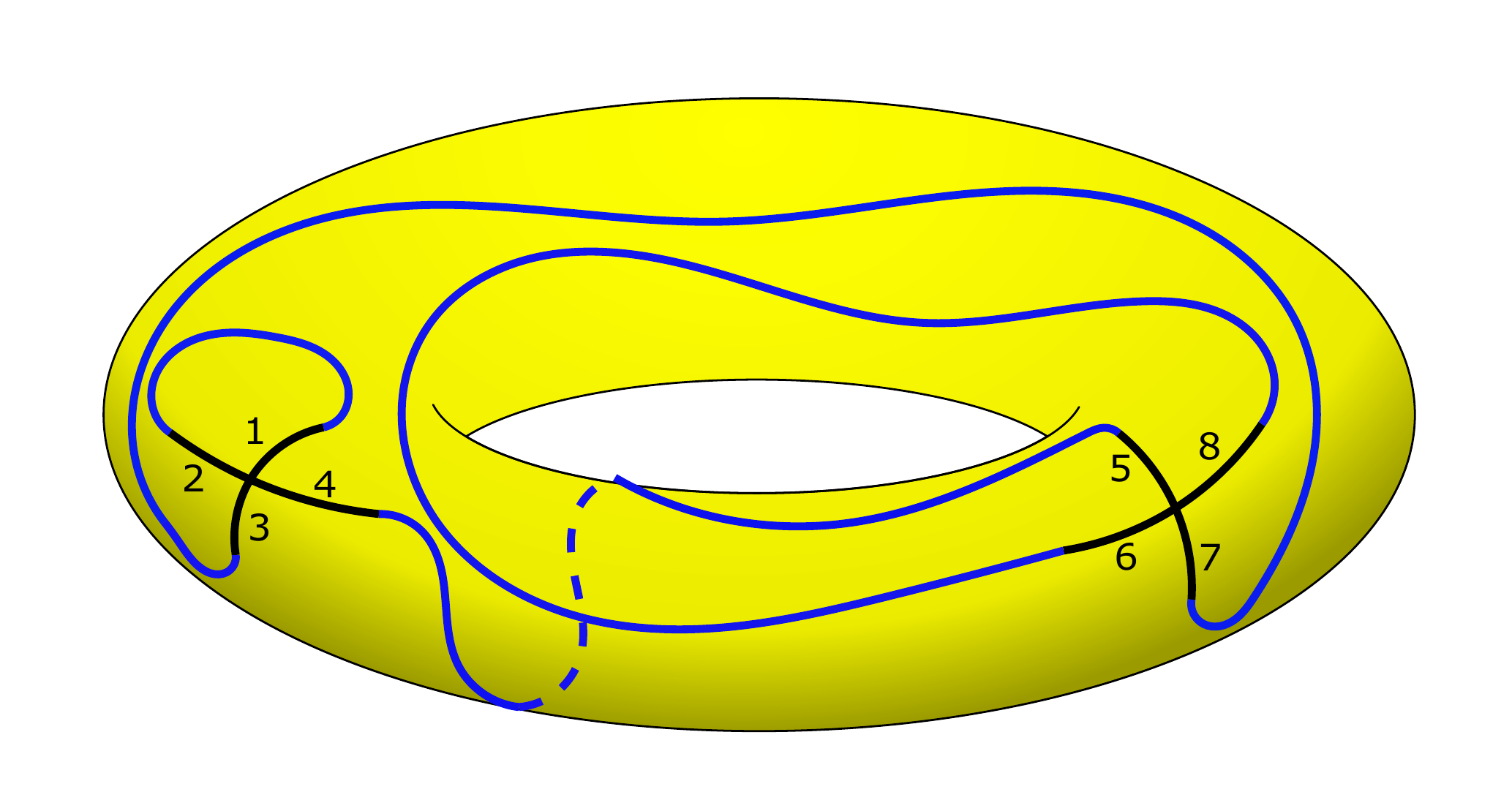}
	\end{subfigure} \hfil
	\begin{subfigure}{0.24\textwidth}
		\centering
		\includegraphics[width=\textwidth]{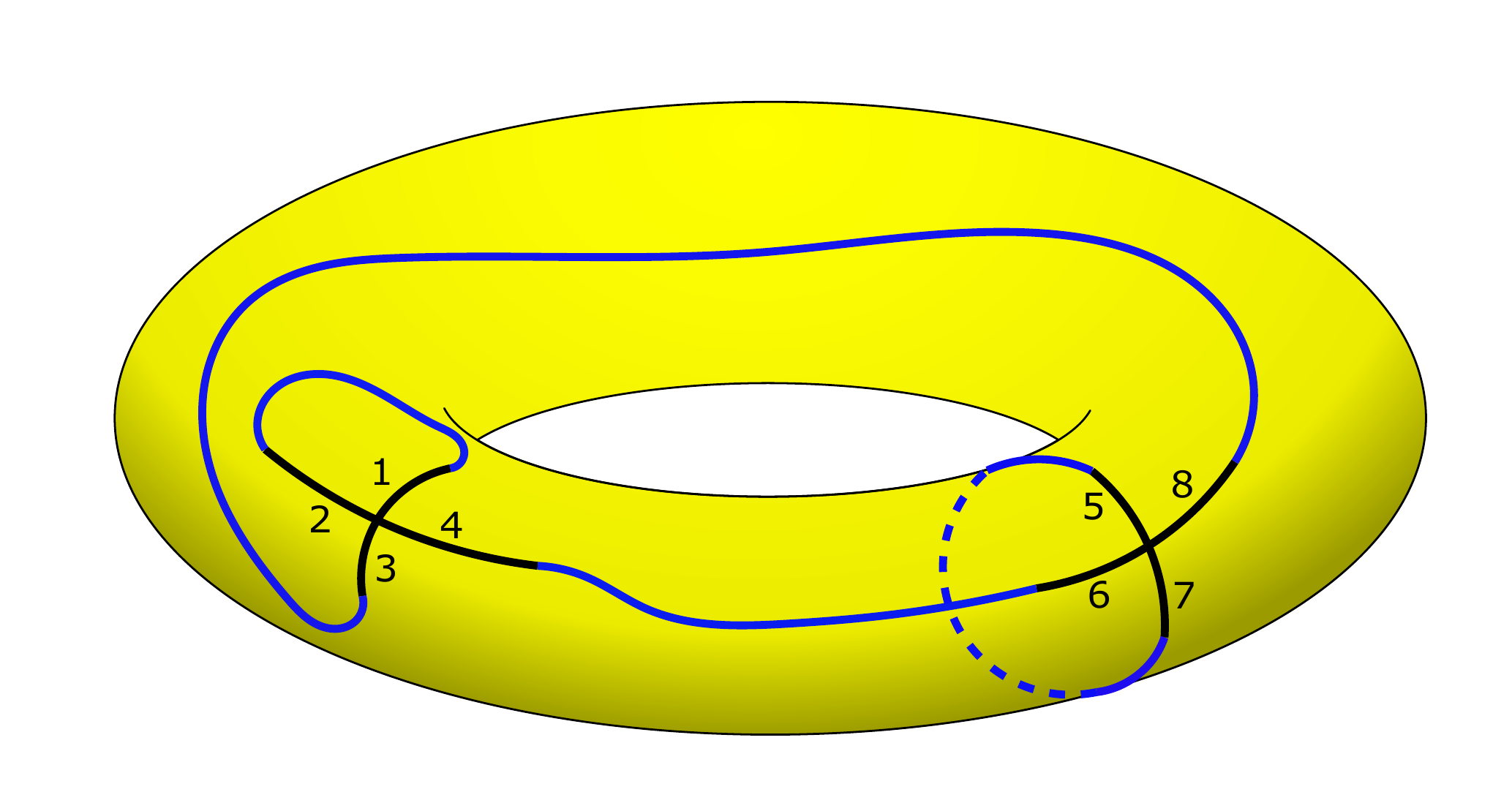}
	\end{subfigure}
	\hfil
	\begin{subfigure}{0.24\textwidth}
		\centering
		\includegraphics[width=\textwidth]{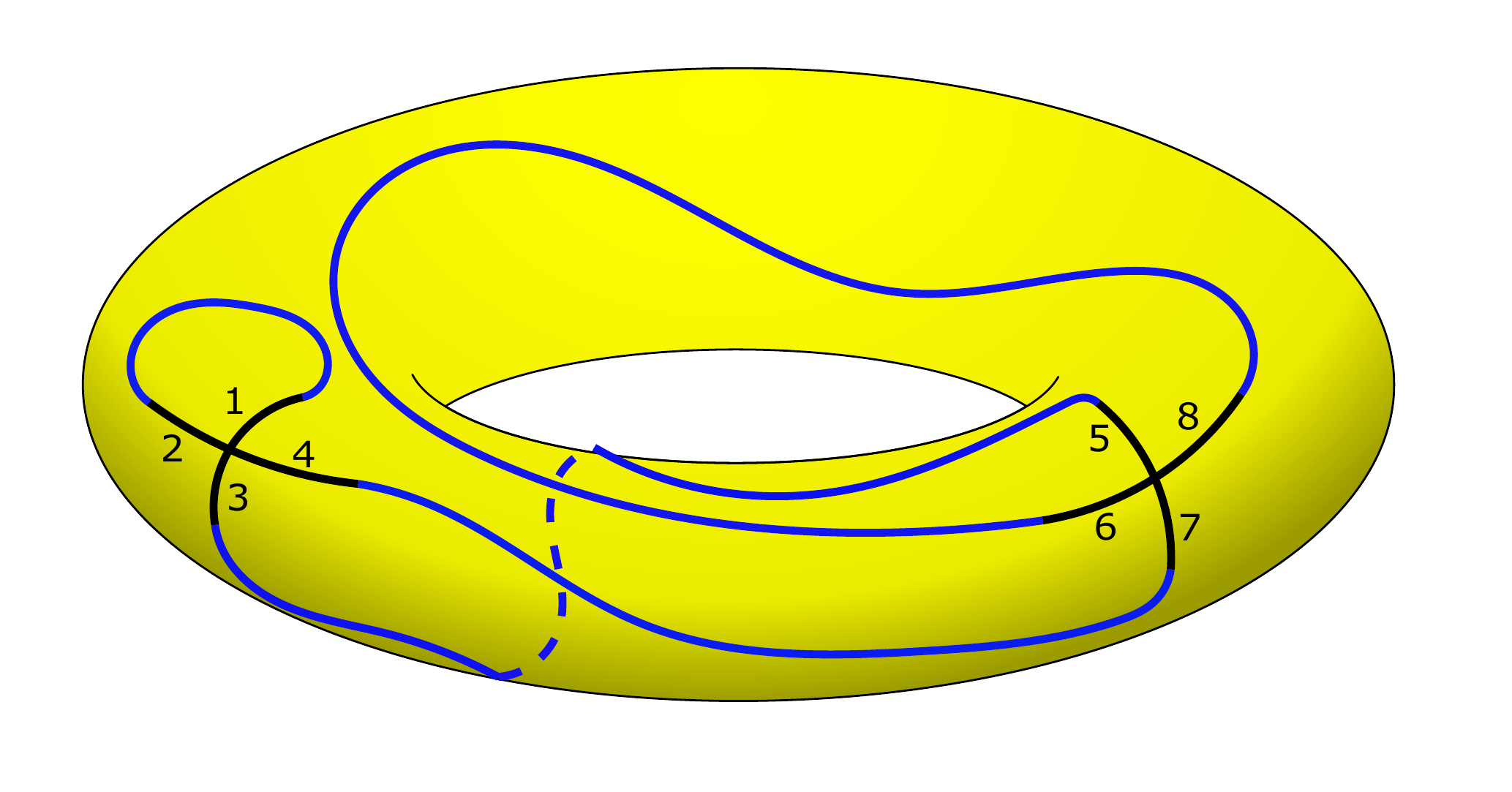}
	\end{subfigure}
	\caption{All four labeled connected 4-valent graphs with two vertices, where $e_1$ connects to $e_2$ which are not realizable on the sphere(compare with Figure \ref{fig: g=0,1 j=2, e1->e_2}). Identically, for the case where $e_1$ connects to $e_4$, there are also four distinct graphs.}
	\label{fig: g=1 j=2, e1->e_2}
\end{figure}

In Figure \ref{fig: g=0,1 j=2, e1->e_6} we have already realized two graphs with $e_1 \leftrightarrow e_6$ \& $e_2 \leftrightarrow e_3$, two graphs with $e_1 \leftrightarrow e_6$  \& $e_2 \leftrightarrow e_5$, and only one graph with $e_1 \leftrightarrow e_6$  \& $e_2 \leftrightarrow e_7$. So there remains one graph with $e_1 \leftrightarrow e_6$ \& $e_2 \leftrightarrow e_3$, one graph with $e_1 \leftrightarrow e_6$  \& $e_2 \leftrightarrow e_5$, and two graphs with $e_1 \leftrightarrow e_6$  \& $e_2 \leftrightarrow e_7$ to be realized on the torus (See the first four graphs in Figure \ref{fig: g=1 j=2, e1->e_6}). Although no graphs with $e_1 \leftrightarrow e_6$  \& $e_2 \leftrightarrow e_8$, and no graphs with $e_1 \leftrightarrow e_6$  \& $e_2 \leftrightarrow e_4$ could be realized on the sphere, one can check that all six such graphs could indeed be realized on the torus as shown in the last six graphs in Figure \ref{fig: g=1 j=2, e1->e_6}. In total, this gives 10 distinct graphs with $e_1 \leftrightarrow e_6$ which can not be realized on the sphere but can be realized on the torus. Identically, for each one of the cases $e_1 \leftrightarrow e_5$, $e_1 \leftrightarrow e_7$, and $e_1 \leftrightarrow e_8$ there also exist 10 distinct graphs that can not be realized on the sphere but can be realized on the torus. Thus far we have obtained $2\cdot4 + 4 \cdot 10 =48$ distinct graphs with two vertices that can not be realized on the sphere but can be realized on the torus based on Figures \ref{fig: g=1 j=2, e1->e_2} and \ref{fig: g=1 j=2, e1->e_6}. The only remaining case to focus on is the number of graphs with two vertices and $e_1 \leftrightarrow e_3$ realizable on the torus (notice that $e_1 \leftrightarrow e_3$ is not possible on the sphere). But this is now obvious as we describe now: Fix $e_1 \leftrightarrow e_3$ and any of the four possible destinations $e_5, e_6, e_7,$ or $e_8$ (Notice that $e_4$ can not be a destination for $e_2$ as it renders the graph disconnected). As described earlier, there are three distinct graphs with two enforced pairwise edge connections. Thus there exists $4 \cdot 3 =12$ distinct graphs with two vertices and $e_1 \leftrightarrow e_3$ which can not be realized on the sphere but can be realized on the torus. This finishes our justification for $\mathscr{N}_2(1)=48+12=60$.

\begin{figure}[!htbp]
	\centering
	\begin{subfigure}{0.24\textwidth}
		\centering
		\includegraphics[width=\textwidth]{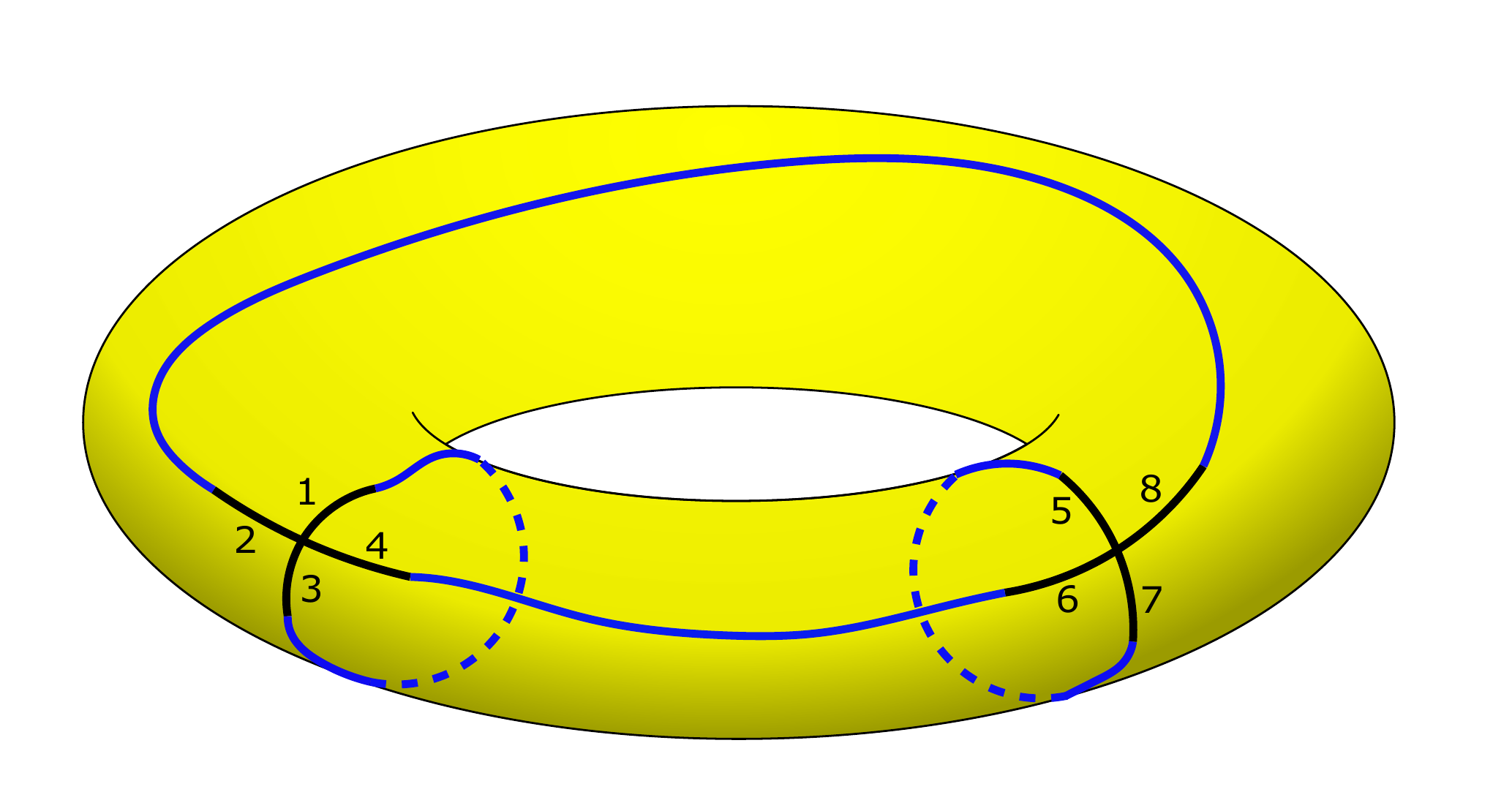}
	\end{subfigure} \hfil
	\begin{subfigure}{0.24\textwidth}
		\centering
		\includegraphics[width=\textwidth]{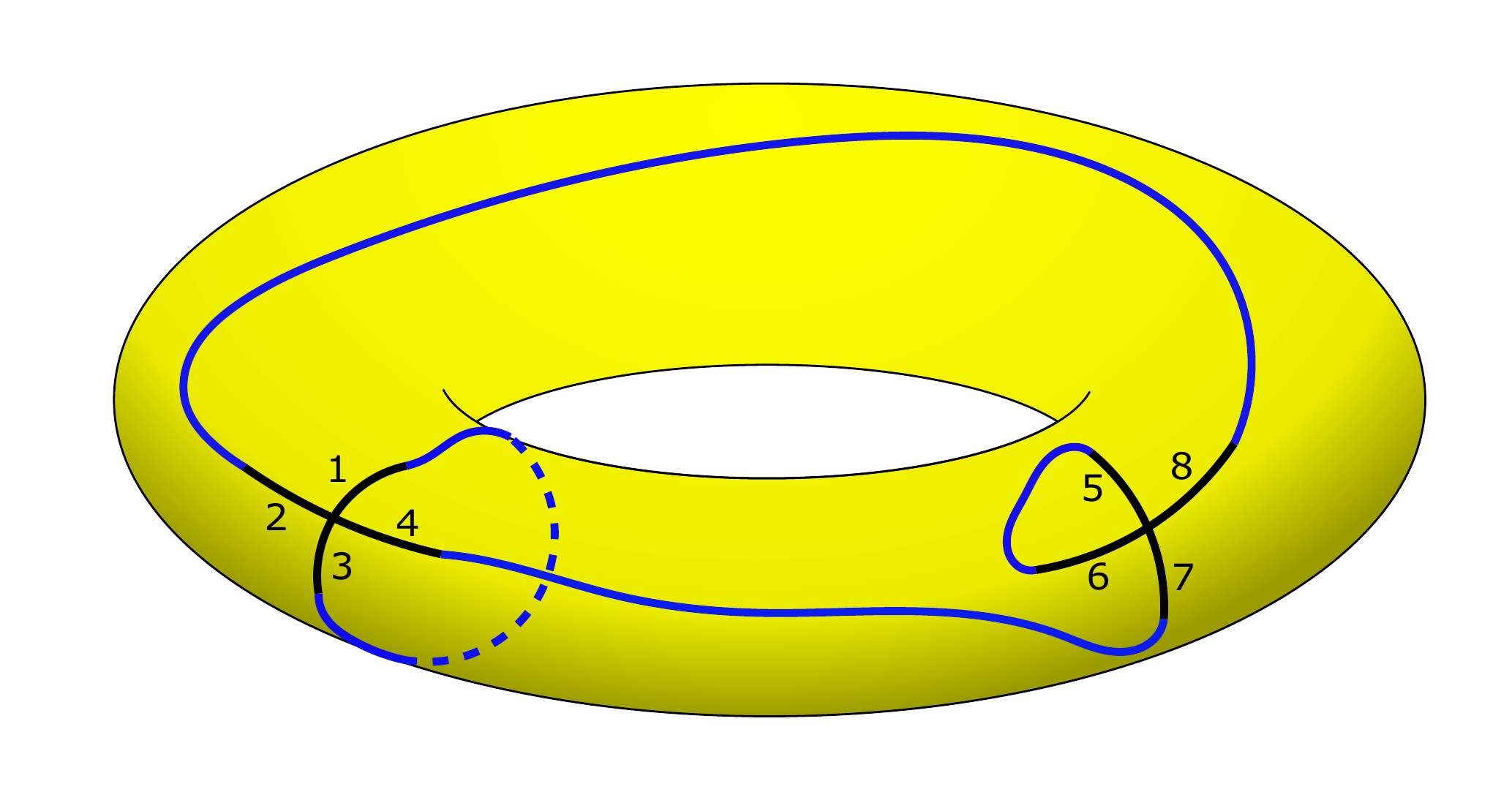}
	\end{subfigure} \hfil
	\begin{subfigure}{0.24\textwidth}
		\centering
		\includegraphics[width=\textwidth]{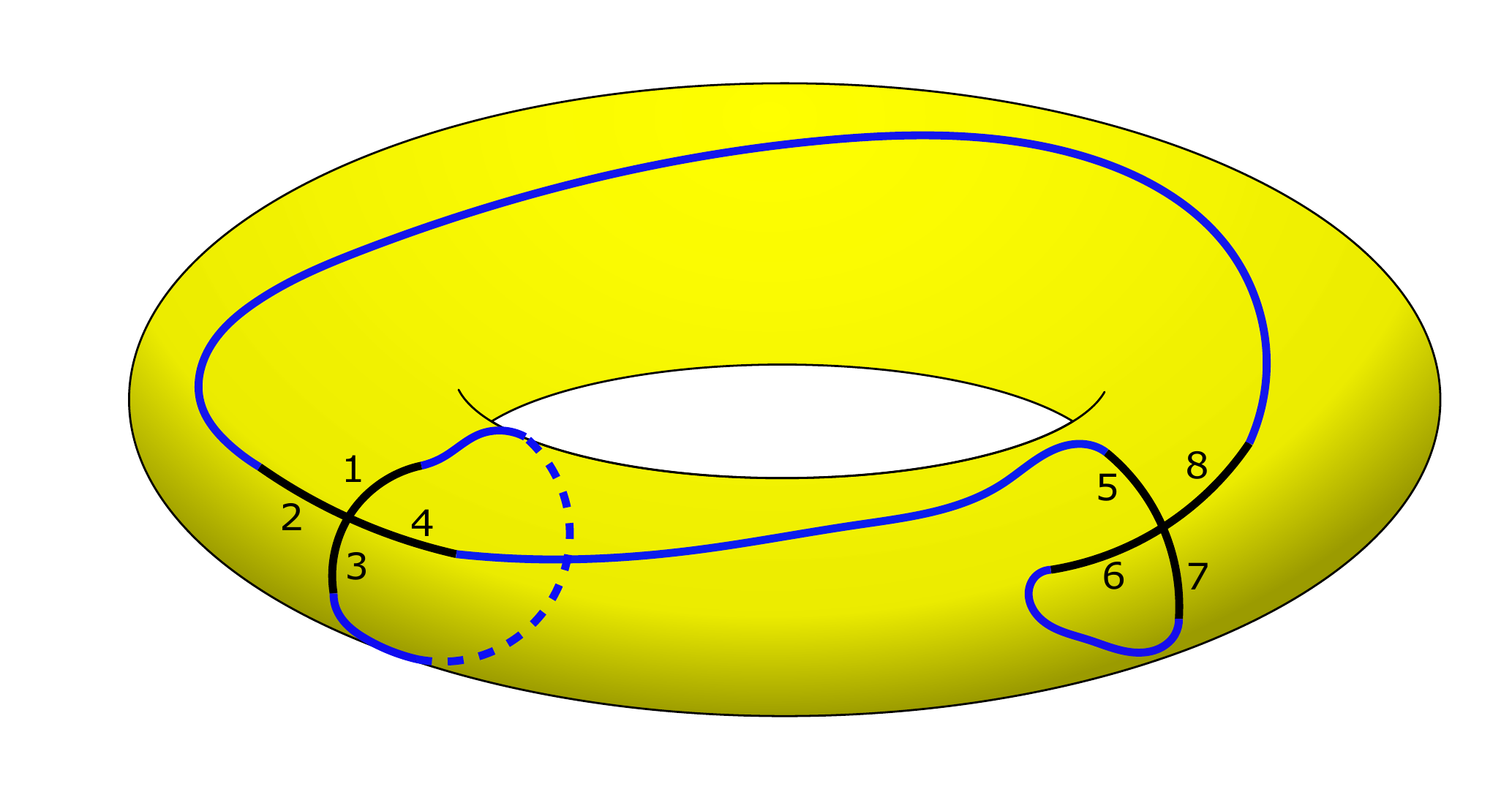}
	\end{subfigure}
	\caption{All three labeled connected 4-valent graphs with two vertices, where $e_1 \leftrightarrow e_3$ \& $e_2 \leftrightarrow e_8$. Identically, for each of the cases $e_1 \leftrightarrow e_3$ \& $e_2 \leftrightarrow e_5$, $e_1 \leftrightarrow e_3$ \& $e_2 \leftrightarrow e_7$, and $e_1 \leftrightarrow e_3$ \& $e_2 \leftrightarrow e_6$ there are three distinct graphs. Thus there exists $4 \cdot 3 =12$ distinct graphs with two vertices and $e_1 \leftrightarrow e_3$ realizable on the torus.}
	\label{fig: g=1 j=2, e1->e_3}
\end{figure}

\begin{figure}[!htbp]
	\centering
	\begin{subfigure}{0.24\textwidth}
		\centering
		\includegraphics[width=\textwidth]{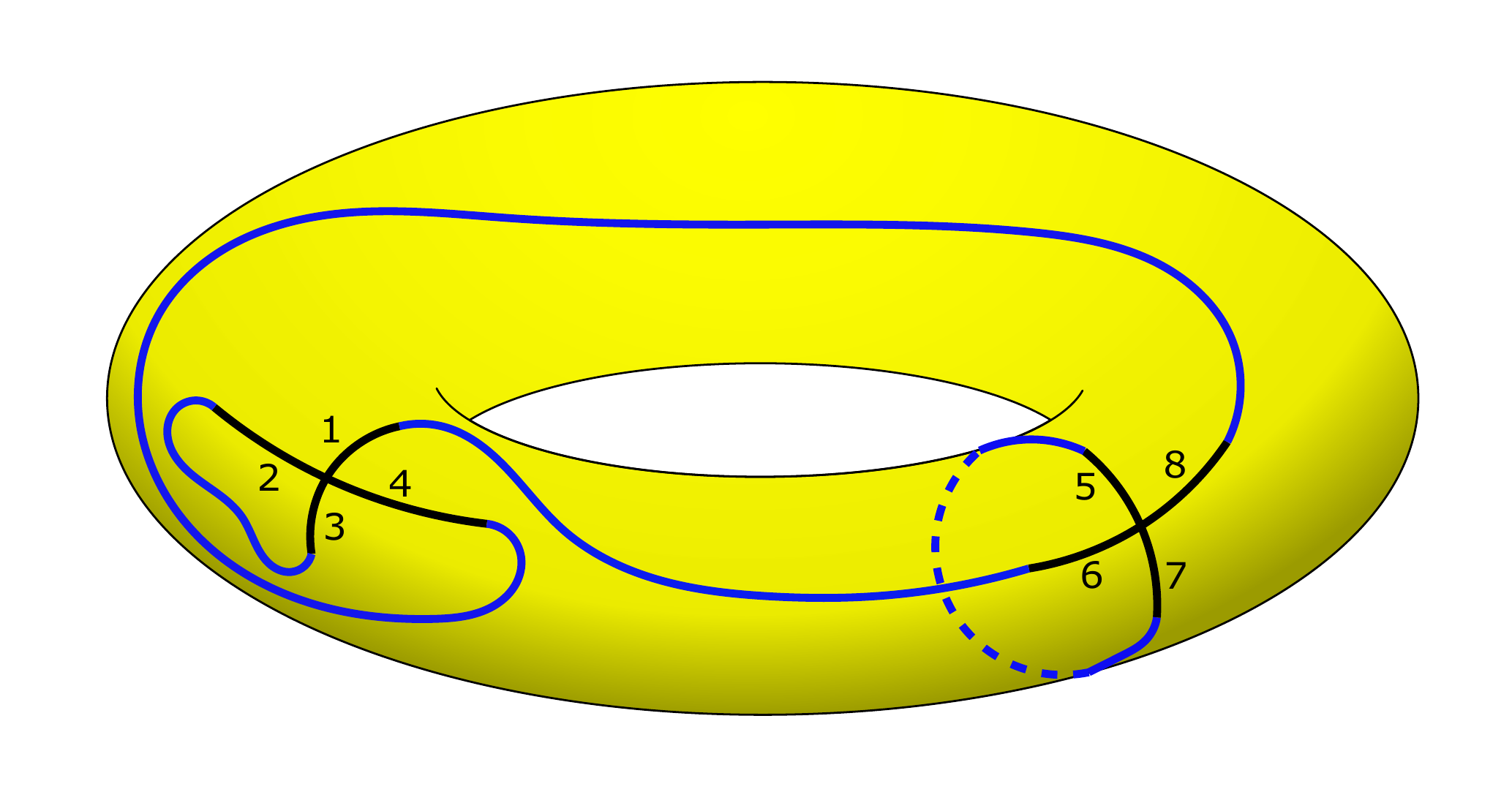}
	\end{subfigure} \hfil
	\begin{subfigure}{0.24\textwidth}
		\centering
		\includegraphics[width=\textwidth]{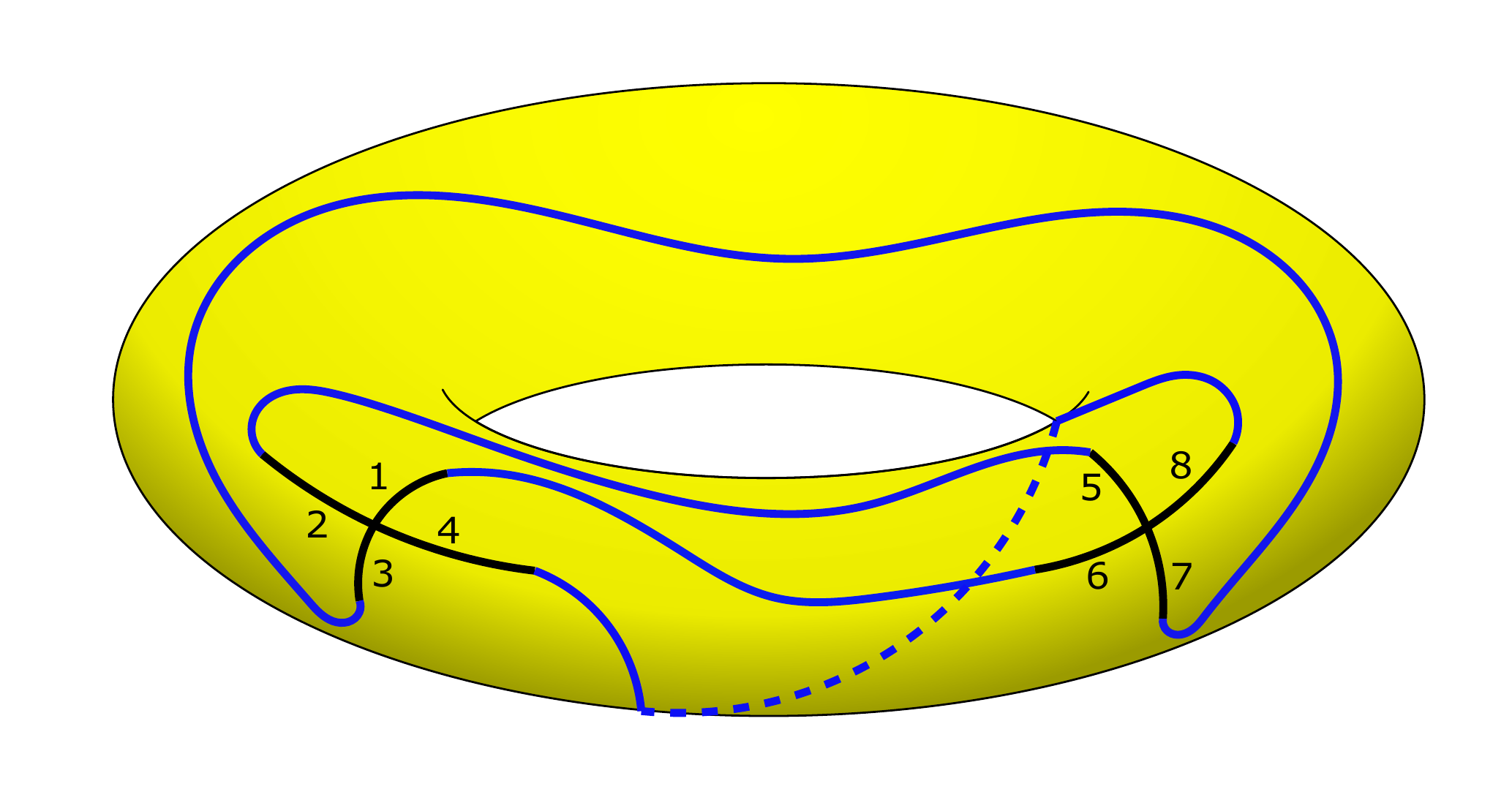}
	\end{subfigure} \hfil
	\begin{subfigure}{0.24\textwidth}
		\centering
		\includegraphics[width=\textwidth]{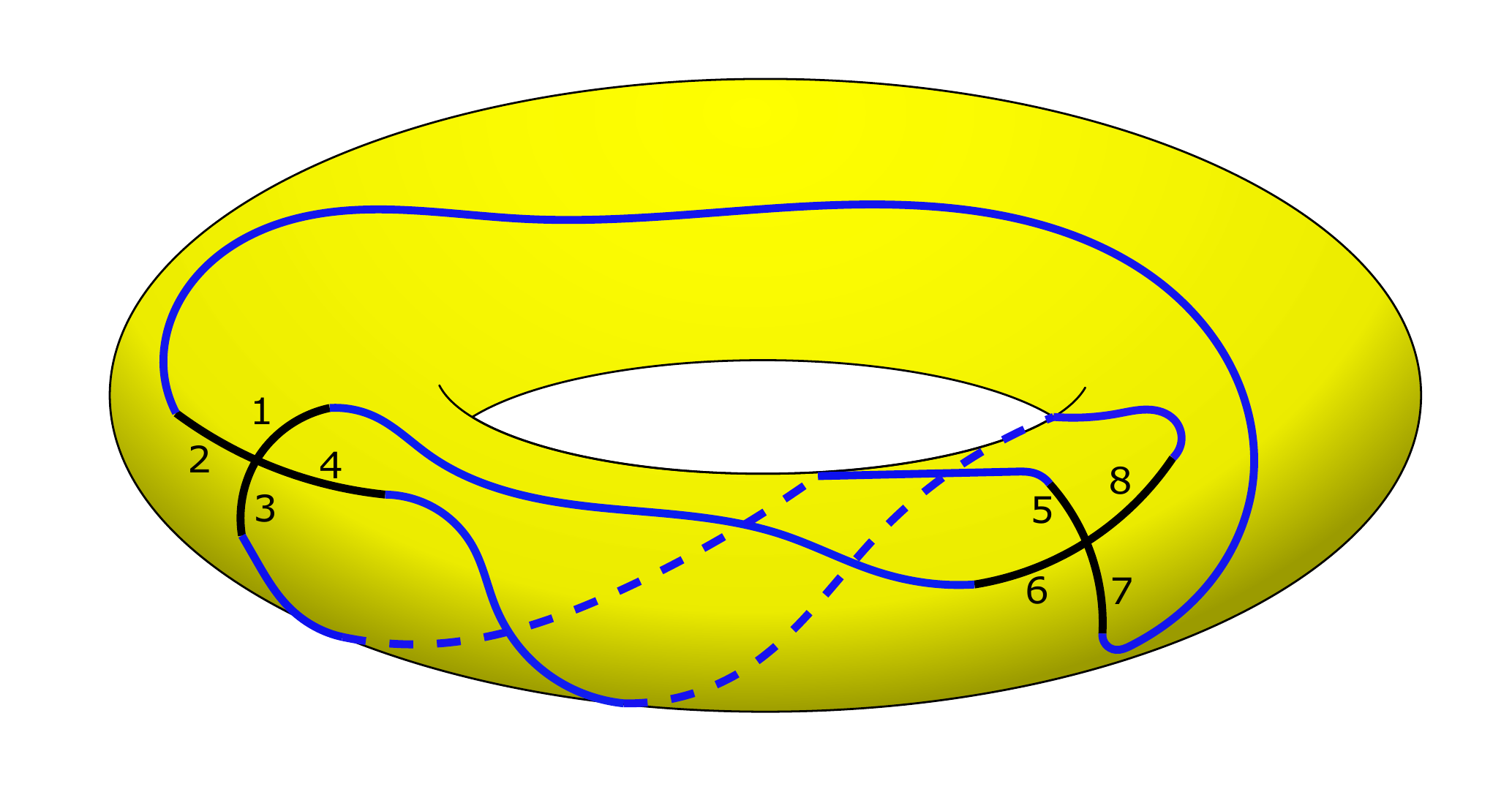}
	\end{subfigure}
	\hfil
	\begin{subfigure}{0.24\textwidth}
		\centering
		\includegraphics[width=\textwidth]{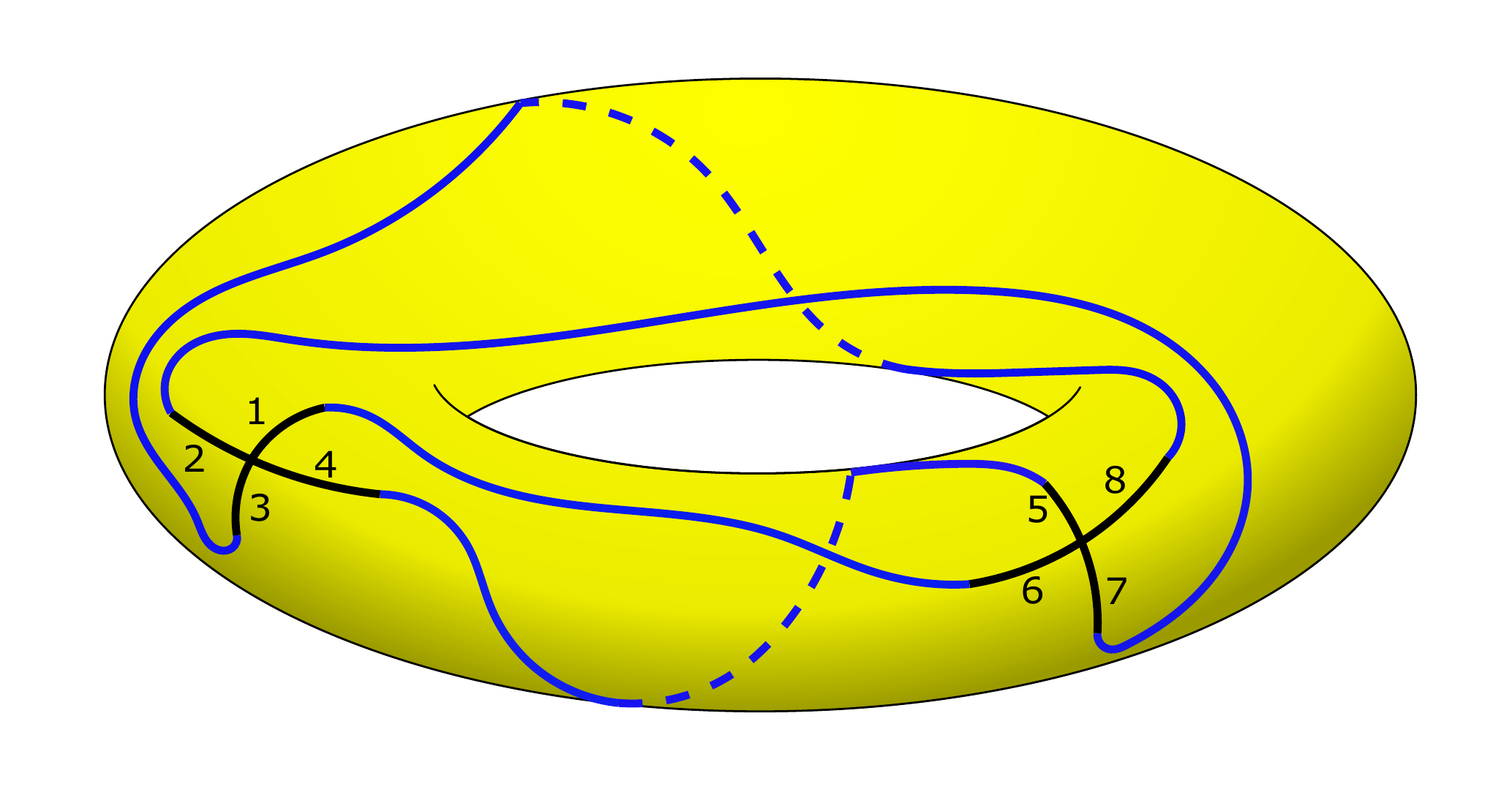}
	\end{subfigure}
	
	
	\centering
	\begin{subfigure}{0.24\textwidth}
		\centering
		\includegraphics[width=\textwidth]{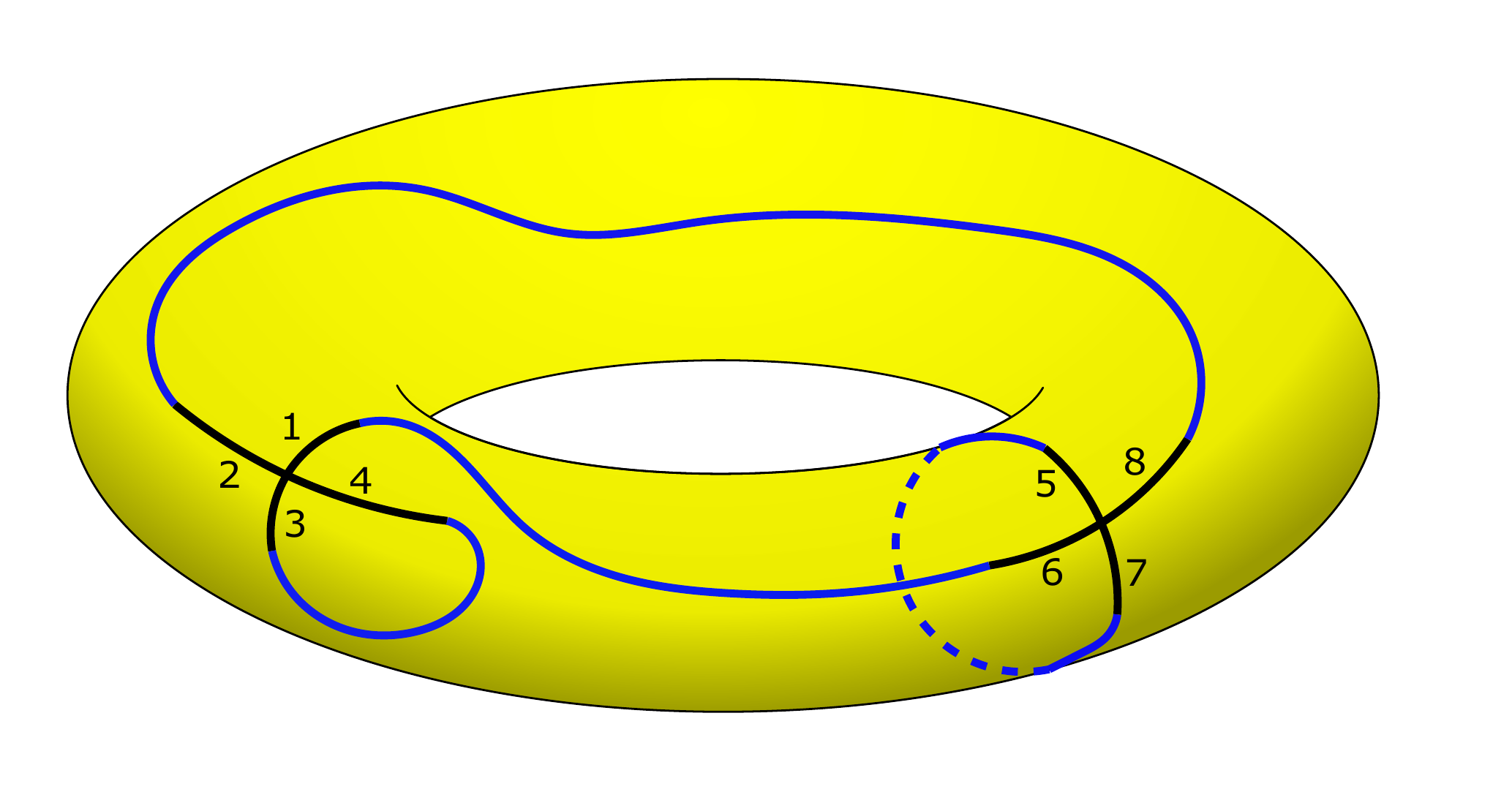}
	\end{subfigure} \hfil
	\begin{subfigure}{0.24\textwidth}
		\centering
		\includegraphics[width=\textwidth]{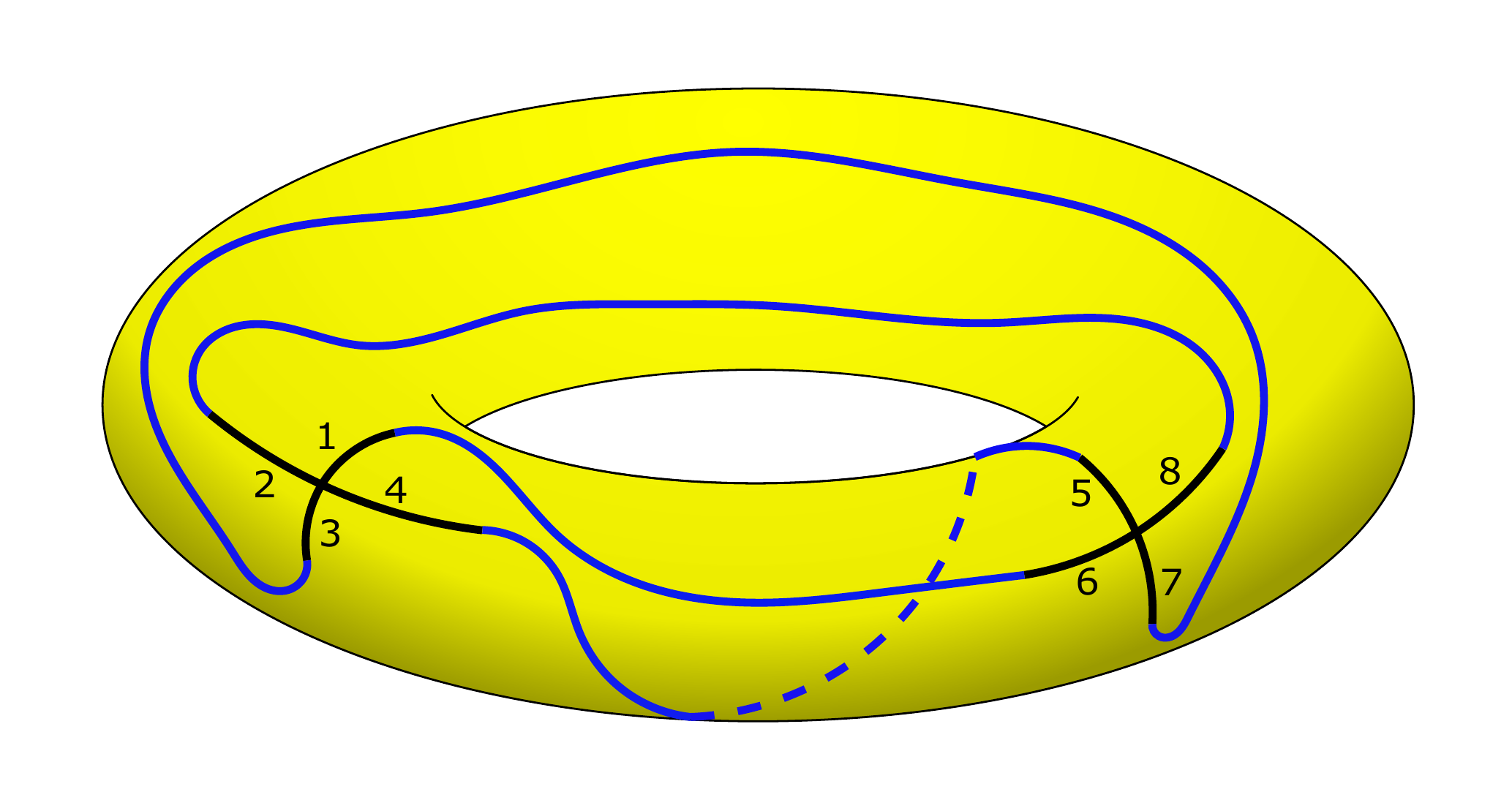}
	\end{subfigure} \hfil
	\begin{subfigure}{0.24\textwidth}
		\centering
		\includegraphics[width=\textwidth]{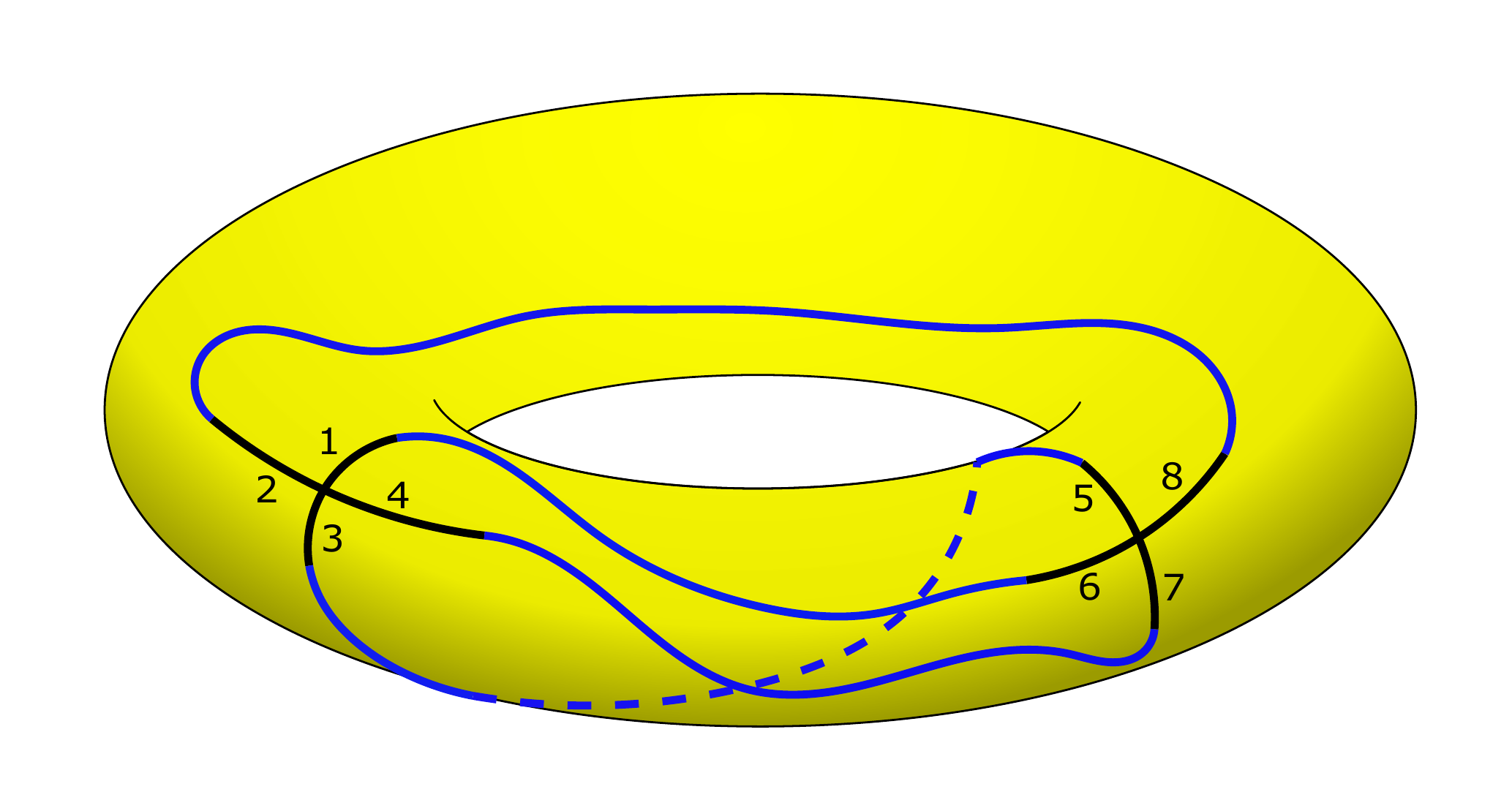}
	\end{subfigure}
	\hfil
	\begin{subfigure}{0.24\textwidth}
		\centering
		\includegraphics[width=\textwidth]{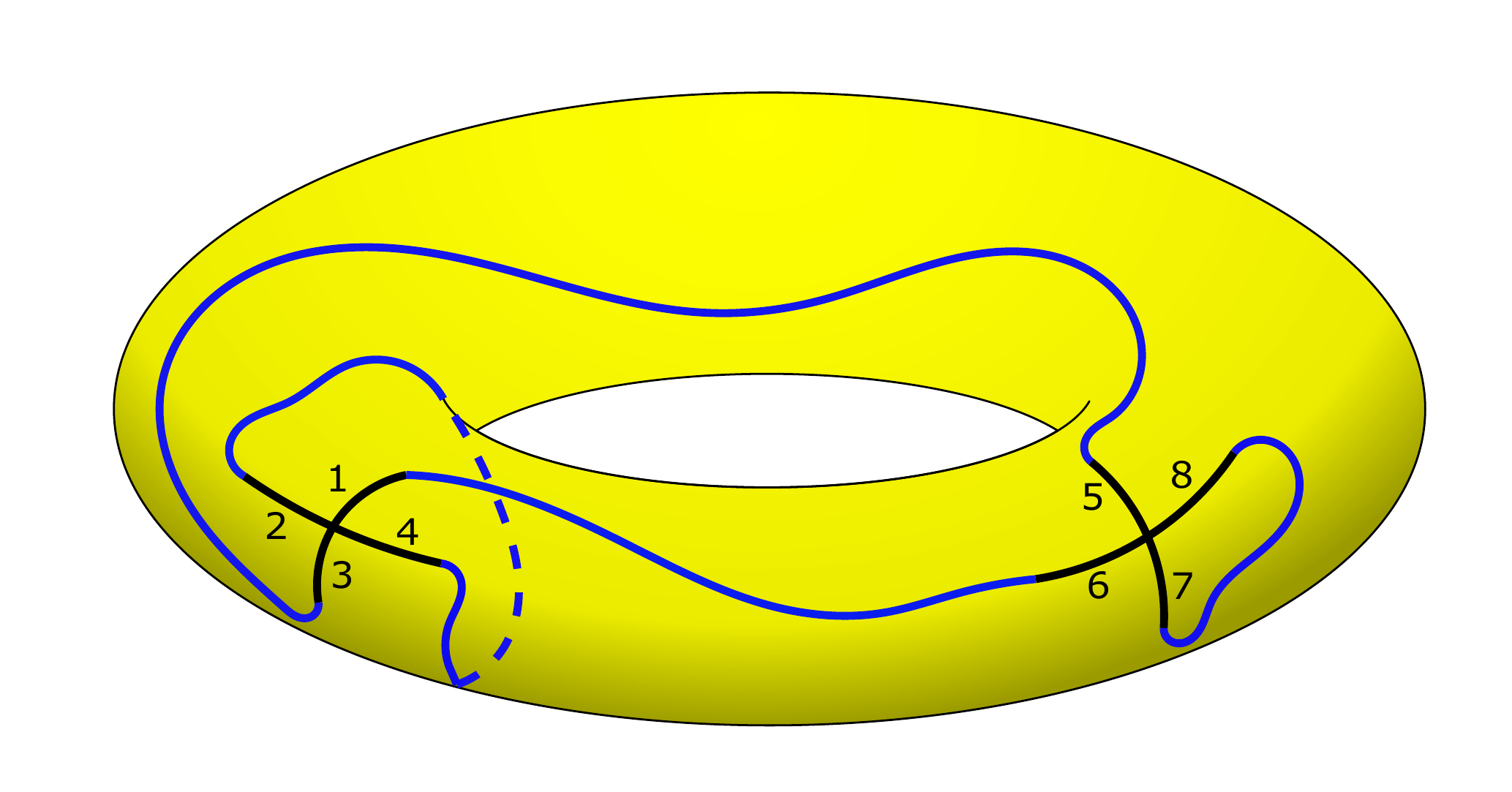}
	\end{subfigure}
	
	
	\centering
	\begin{subfigure}{0.24\textwidth}
		\centering
		\includegraphics[width=\textwidth]{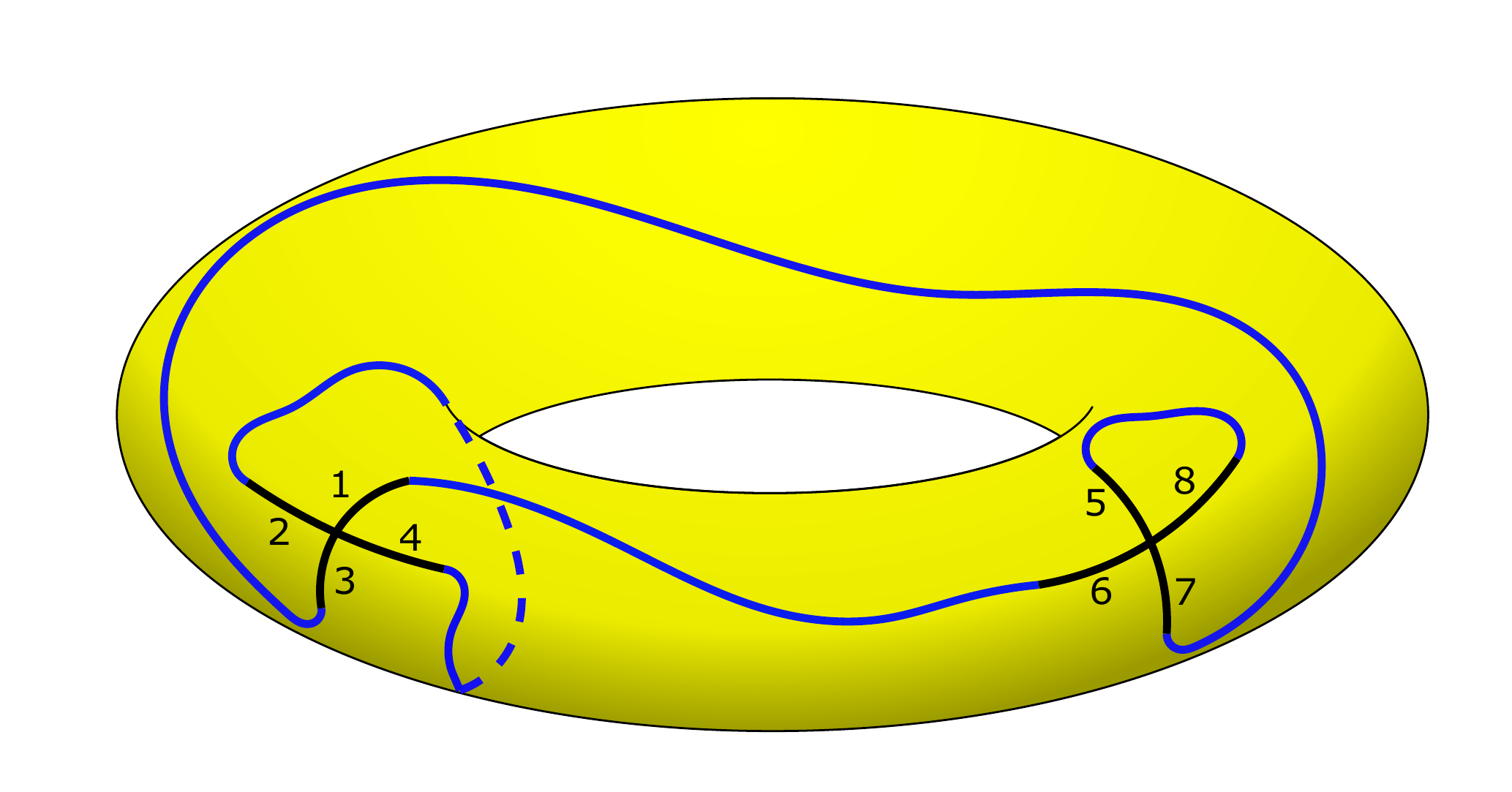}
	\end{subfigure} \hfil
	\begin{subfigure}{0.24\textwidth}
		\centering
		\includegraphics[width=\textwidth]{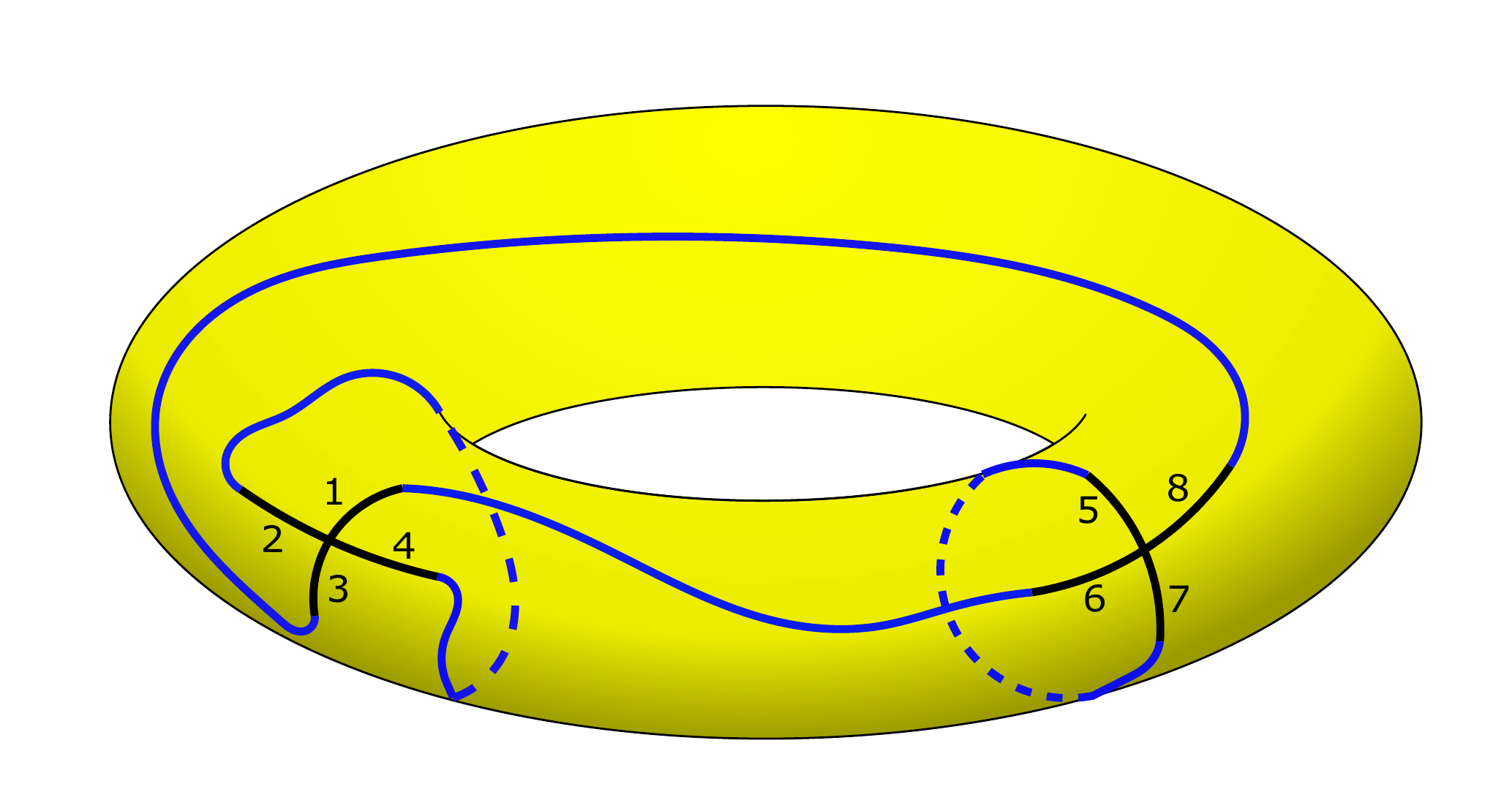}
	\end{subfigure}
	
	\caption{All ten labeled connected 4-valent graphs with two vertices, where $e_1$ connects to $e_6$ which are not realizable on the sphere(compare with Figure \ref{fig: g=0,1 j=2, e1->e_6}). Identically, for each one of the cases $e_1 \leftrightarrow e_5$, $e_1 \leftrightarrow e_7$, and $e_1 \leftrightarrow e_8$ there also exist 10 distinct graphs. This means that there are totally 40 distinct graphs, not realizable on the sphere, where $e_1$ connects to one of the edges emanating from $v_2$.}
	\label{fig: g=1 j=2, e1->e_6}
\end{figure}

\section*{Acknowledgements}
The authors would like to thank the anonymous referee for several suggestions that helped improve the manuscript. This material is based upon work supported by the National Science Foundation under Grant No. DMS-1928930. Kenneth T-R McLaughlin additionally acknowledges support from the NSF grant DMS-1733967. We gratefully acknowledge the Mathematical Sciences Research Institute, Berkeley, California and the organizers of the semester-long program \textit{Universality and Integrability in Random Matrix Theory and Interacting Particle Systems} for their support in the Fall of 2021, during which this project was completed. P.B. would like to acknowledge hospitality and support from  the Galileo Galilei Institute, Florence, Italy, the scientific program at GGI on  \textit{Randomness, Integrability, and Universality}, and the Simons Foundation program for GGI Visiting Scientists. We also thank Marco Bertola, Christophe Charlier, Philippe Di Francesco, Nicholas Ercolani,  Alexander Its, and Alexander Tovbis for their very helpful remarks and suggestions. 

\newcommand{\etalchar}[1]{$^{#1}$}


\end{document}